\documentclass[article, bibtotoc, final]{scrbook}

\usepackage{lmodern}
\usepackage[english]{babel}
\usepackage[utf8]{inputenc}

\usepackage[T1]{fontenc}

\usepackage[table]{xcolor}

\usepackage{amsthm}

\usepackage{amsmath, amsfonts, amssymb}
\usepackage{mathrsfs} 
\usepackage{stmaryrd} 

\usepackage[
	colorlinks=false,
	linkcolor=blue,
	citecolor=black,
	urlcolor=black,
	hidelinks,
	breaklinks
]{hyperref}

\usepackage[sanserif]{complexity}
\usepackage{tikz}
\usetikzlibrary{calc}
\usetikzlibrary{patterns}
\usetikzlibrary{trees}
\usetikzlibrary{fit}
\usetikzlibrary{scopes}
\usetikzlibrary{snakes}
\usetikzlibrary{hobby}
\usetikzlibrary{through}
\usetikzlibrary{arrows}
\usetikzlibrary{intersections}
\usetikzlibrary{positioning}

\tikzset{highlight/.style={rectangle,
		fill=red!15,
		blend mode = multiply,
		rounded corners = 0.5 mm,
		inner sep=1pt,
		fit = #1}}
	
\tikzset{highlightblue/.style={rectangle,
		fill=blue!15,
		blend mode = multiply,
		rounded corners = 0.5 mm,
		inner sep=1pt,
		fit = #1}}
	
\tikzset{highlightyellow/.style={rectangle,
		fill=yellow!15,
		blend mode = multiply,
		rounded corners = 0.5 mm,
		inner sep=1pt,
		fit = #1}}
	
\tikzset{block/.style = {rectangle, draw, minimum height=10mm,
		align=center}}

\tikzstyle{shapeaftershift}=[thick,dashed]
\tikzstyle{shapebeforeshift}=[gray,opacity=.5]

\usepackage{wrapfig}

\usepackage{pgfplots}

\usepackage{standalone}

\usepackage{array}

\usepackage{nicematrix}

\usepackage{bbm}

\usepackage{enumitem}

\usepackage{multirow}

\usepackage{centernot}

\usepackage{cleveref}

\usepackage[ruled]{algorithm2e}
\SetKw{Continue}{continue}
\SetKw{Reject}{reject}
\SetKw{Accept}{accept}
\SetKwFunction{decsimple}{decompose\_simple}

\graphicspath{ {./figures/} }

\SetAlFnt{\small}
\SetAlCapFnt{\large}
\SetAlCapNameFnt{\large}

\hypersetup{
	pdftitle={Variable Independence in Linear Real Arithmetic},
	pdfauthor={Alexander Mayorov}
}

\usepackage[labelfont={normal,bf}]{caption}
\setcapindent{0pt}
\usepackage{subcaption}

\usepackage[
a4paper,
top=35mm,
bottom=40mm,
inner=35mm,
outer=40mm,
bindingoffset=5mm,
marginparsep=5mm,
marginparwidth=40mm,
]{geometry}

\usepackage{amsthm}

\theoremstyle{theorem}
\newtheorem{theorem}{Theorem}[section]

\theoremstyle{definition}
\newtheorem{zbdefinition}[theorem]{Definition}

\theoremstyle{theorem}
\newtheorem{lemma}[theorem]{Lemma}

\theoremstyle{theorem}
\newtheorem{corollary}[theorem]{Corollary}

\theoremstyle{theorem}
\newtheorem{proposition}[theorem]{Proposition}

\theoremstyle{definition}
\newtheorem{zbclaim}[theorem]{Claim}

\theoremstyle{definition}
\newtheorem{zbproblem}[theorem]{Problem}

\theoremstyle{definition}
\newtheorem{zbexample}[theorem]{Example}
\AtEndEnvironment{zbexample}{\qed}

\usepackage[square,numbers]{natbib}

\def\zbdate{\leavevmode\hbox{\twodigits\day.\twodigits\month.\the\year}}
\def\twodigits#1{\ifnum#1<10 0\fi\the#1}

\definecolor{OliveGreen}{rgb}{0,0.6,0}
\definecolor{LightCyan}{rgb}{0.88,1,1}

\newcolumntype{P}[1]{>{\centering\arraybackslash}p{#1}}

\newcommand{\gen}[1]{\langle #1 \rangle}

\newcommand{\norm}[1]{\left\Vert#1\right\Vert}

\def\Npos{\mathbb{N}_{\ge 1}}
\def\N{\mathbb{N}_{\ge 0}}

\def\R{\mathbb{R}}

\def\Q{\mathbb{Q}}

\def\transp{\mathrm{T}}

\DeclareMathOperator{\ModelsOf}{Mod}
\DeclareMathOperator{\Free}{Free}
\DeclareMathOperator{\Pred}{Pred}
\DeclareMathOperator{\PredRespecting}{PredResp}
\DeclareMathOperator{\PredDisrespecting}{PredDisresp}
\DeclareMathOperator{\LinDep}{LinDep}
\DeclareMathOperator{\Disjuncts}{DNF}
\DeclareMathOperator{\DisjunctsOf}{DisjOf}

\DeclareMathOperator{\ImageOf}{im}
\DeclareMathOperator{\Sat}{Sat}
\DeclareMathOperator{\OpenCubeSymbol}{C}

\def\varphiprop{\varphi_\mathrm{prop}}

\def\Mstruct{\mathcal{M}}
\def\Uuniverse{\mathcal{U}}
\def\Iinterp{\mathcal{I}}
\def\Fixes{\mathcal{F}}

\def\DisjPhi{\Disjuncts_\varphi[\varphi]}
\def\DisjNegPhi{\Disjuncts_\varphi[\neg\varphi]}
\def\DisjTrue{\Disjuncts_\varphi[\top]}

\newcommand{\DisjOf}[1]{\DisjunctsOf_\varphi(#1)}

\newcommand{\PiSimp}[1]{\PredRespecting_\Pi(#1)}
\newcommand{\PiComp}[1]{\PredDisrespecting_\Pi(#1)}
\newcommand{\id}{\mathrm{id}}
\newcommand{\lc}{\mathrm{lc}}

\newcommand{\OpenCube}[2]{\OpenCubeSymbol_{#2}^{#1}}

\newlang{\QFLRA}{QFLRA}

\def\zbdocname/{thesis}
\def\zbsectionname/{chapter}

\newcommand{\zbthesismode}{1}
\newcommand{\zbtikzbiggerscaling}{1.4}
\newcommand{\zbtikznormalscaling}{1.025}
\newcommand{\zbtikzaggressivescaling}{1.2}
\newcommand{\zbrefsec}[1]{Chapter \ref{#1}}

\title{Variable Independence \\ in Linear Real Arithmetic}
\author{Alexander Mayorov}

\date{\today}

\begin{document}

\makeatletter

\let\chapterorig\chapter
\let\sectionorig\section
\let\subsectionorig\subsection

\renewcommand{\section}{\chapterorig}
\renewcommand{\subsection}{\sectionorig}
\renewcommand{\subsubsection}{\subsectionorig}

\makeatother

\frontmatter


\newgeometry{margin=1in}
\begin{titlepage}
	\makeatletter
	\centering
	\includegraphics[width=0.5\textwidth]{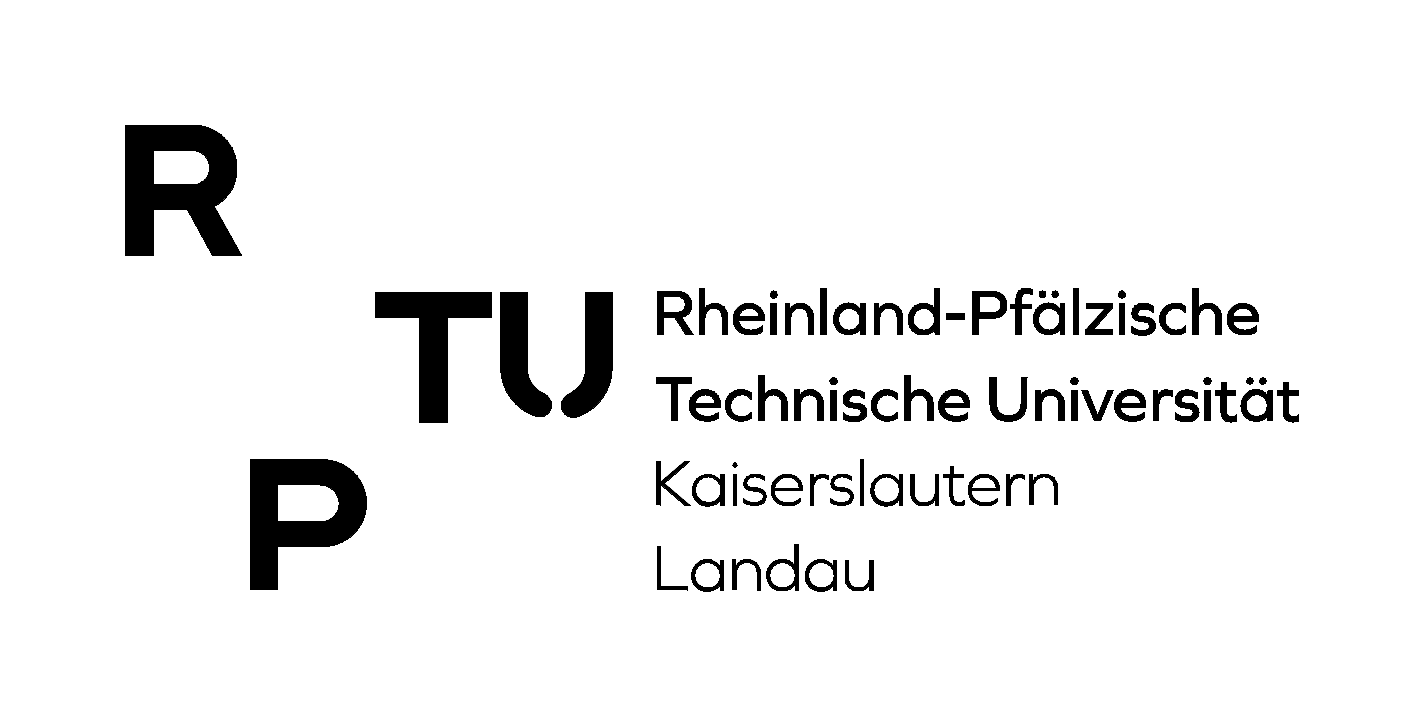} \par \vspace*{\fill}
	{\scshape\LARGE \@title} \par \vspace*{\fill}
	{\bfseries Bachelor Thesis} \par\vspace{1.5cm}
	{\Large\itshape \@author} \par\vspace{1.5cm}
	{\@date} \par\vspace*{\fill}
	{Rheinland-Pf{\"a}lzische Technische Universit{\"a}t Kaiserslautern-Landau\\
		Department of Computer Science,\\
		Gottlieb-Daimler-Straße, 67663 Kaiserslautern, Germany} \par\vspace*{\fill}
	{
		\begin{tabular}{rl}
			Supervisors: & Prof. Dr. Anthony W. Lin\\
			& Prof. Matthew Hague
		\end{tabular}
	}
	\makeatother
\end{titlepage}
\restoregeometry


\sectionorig*{Declaration}

I hereby declare that I am the sole author of this bachelor thesis and that the work contained herein is my own except where explicitly stated otherwise in the text. I further declare that I have not submitted this thesis at any other institution in order to obtain a degree.



\par\vspace*{8ex}\noindent
Kaiserslautern, \zbdate
\hfill
\begin{tabular}[t]{@{}c@{}}
	\includegraphics[scale=0.3]{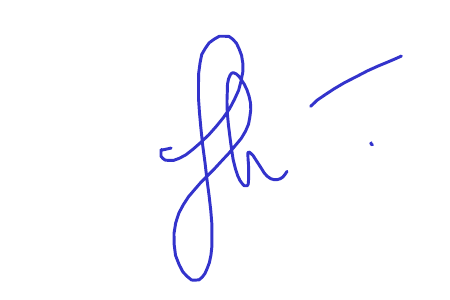}\vspace{-10pt} \\
	\makebox[15em]{\dotfill}\\
	Alexander Mayorov
\end{tabular}


\newpage
\sectionorig*{Abstract}

\begin{center}
	\begin{minipage}[t]{0.7\textwidth}
		Variable independence and decomposability are algorithmic techniques for simplifying logical formulas by tearing apart connections between free variables. These techniques were originally proposed to speed up query evaluation in constraint databases, in particular by representing the query as a Boolean combination of formulas with no interconnected variables. They also have many other applications in SMT, string analysis, databases, automata theory and other areas. However, the precise complexity of variable independence and decomposability has been left open especially for the quantifier-free theory of linear real arithmetic (LRA), which is central in database applications. We introduce a novel characterization of formulas admitting decompositions and use it to show that it is $ \coNP $-complete to decide variable decomposability over LRA. As a corollary, we obtain that deciding variable independence is in the second level of the polynomial hierarchy. These results substantially improve the best known double-exponential time algorithms for variable decomposability and independence. In many practical applications, it is important not only to decide variable independence/decomposability algorithmically but also to be able to efficiently eliminate connections between variables whenever possible. We design and implement an algorithm for this problem, which is optimal in theory, exponentially faster compared to the current state-of-the-art algorithm and efficient on various microbenchmarks. In particular, our algorithm is the first one to overcome a fundamental barrier between non-discrete and discrete first-order theories (for the latter, the problem is significantly easier, and many techniques yielding efficient algorithms are known). Formulas arising in practice often have few or even no free variables that are perfectly independent. In this case, our algorithm can compute a best-possible approximation of a decomposition, which can be used to optimize database queries by exploiting partial variable independence, which is present in almost every logical formula or database query constraint.

	\end{minipage}
\end{center}


\newpage

\sectionorig*{Acknowledgements}

\begin{center}
	\begin{minipage}[t]{0.7\textwidth}
		I wish to express my sincere thanks to my supervisor Anthony Lin for his guidance, continuous support, advice and for the helpful and inspiring discussions we had. I am also very grateful to him for providing me with all the necessary facilities for my research. I would like to thank my co-supervisor Matthew Hague and Daniel Stan for their guidance, useful feedback, comments and for the fruitful discussions we had throughout the entire period of research.

I want to also express my profound gratitude to my family and friends for providing me with unfailing support and continuous encouragement throughout the process of researching and writing this thesis. Thank you.
	\end{minipage}
\end{center}


\setcounter{tocdepth}{3}
\renewcommand{\contentsname}{Table of Contents}
\cleardoublepage                    
\pdfbookmark[0]{\contentsname}{toc}
\tableofcontents





\mainmatter

\section{Introduction}
\label{sec:intro}

Constraint databases \cite{CDB-book, gaede:1996, revesz:1998, KKR90} generalize the notion of a tuple in a relational database \cite{ramakrishnan:2003:databasemanagementsystems} to a conjunction of constraints over an appropriate language that is typically a fragment of first-order logic. By using constraints, it is possible to represent complex and a very large or even infinite amount of data in a compact and efficient way. The key idea of constraint databases is to allow users to interact transparently with data expressed using constraints as if that data were to be explicitly stored in a standard relational database. For example, in a cadastral database, one might be interested in finding intersecting polygons representing positions of private properties on a map. This query can be formulated in the language of constraint database queries roughly as follows: does there exist a point belonging to two polygons at the same time? Note that this query mentions the existence of a point, i.e., an object that is not in the database and that would be impossible to explicitly store in a relational database simply because there are infinitely many points in two-dimensional unbounded space. More generally, constraint databases have numerous applications in spatial \cite{brodsky:1993, brodsky:1995, brodsky:1999, paredaens:1994, vandeurzen:1995}, temporal \cite{chomicki:1994:temporalquerylangs, toman:1997} and spatiotemporal databases \cite{GRS01, grumbach:1999}. Many important classes of constraints have been studied extensively over the last decades, including Presburger Arithmetic \cite{haase:2018, bradley:2007, kroening:2016, griggio:2012:intlinarith, stansifer:1984, presburger:1929}, Linear Real Arithmetic (LRA) \cite{monniaux:2008, dutertre:2006:linarithsolver, dutertre:2006:simplexdpllt, kroening:2016, bradley:2007}, the theory of Real-Closed Fields \cite{caviness:2012:qecad, jirstrand:1995:cad, goldin:1996} and decidable theories over strings (in particular, automatic structures) \cite{barcelo:2019, chen:2017, jez:2022}.


\subsection{Variable independence}

Variable independence and decomposability \cite{Lib99, libkin:2003, GRS01, CGKT03, CR99, CGK96} have been identified as crucial algorithmic techniques in relational and constraint databases because they allow us to significantly optimize query evaluation whenever there is variable independence in the constraints involved. We demonstrate this on a simple example.

\begin{zbexample}
	\label{example:independence_eff_query_eval}
	Consider a relational database with two binary relations $ X(x_1, x_2) $, $ Y(y_1, y_2) $ over the real numbers. Suppose we want to evaluate a conjunctive query \[
		q \leftarrow X(x_1, x_2), Y(y_1, y_2), \varphi(x_1, x_2, y_1, y_2)
	\] where \[
		\varphi(x_1, x_2, y_1, y_2) := 4x_1 - 6x_2 + y_1 + 5y_2 < 8 \wedge 2x_1 - 3x_2 = 1
	\] is an LRA formula expressing the desired tuples (i.e., $ x_1, x_2, y_1, y_2 $ are real-valued free variables). Set $ n := \max(s, r) $, where $ s $ is the size of the database and $ r $ is the size of the desired query output. Naively evaluating this query would require computing the join of $ X $ and $ Y $, and selecting rows satisfying $ \varphi $. Clearly, this would take $ \Theta(n^2) $ time in the worst case. However, it turns out that every $ x_i $ variable is independent from every $ y_j $ variable of $ \varphi $ in the sense that $ \varphi $ can be equivalently rewritten as a conjunction \[
		\psi := \underbrace{2x_1 - 3x_2 = 1}_{\text{Uses only }x_1, x_2} \wedge \underbrace{y_1 + 5y_2 < 6}_{\text{Uses only }y_1, y_2}
	\] of formulas never using $ x_i $ and $ y_j $ at the same time. Using this knowledge, we can evaluate the same query $ q $ in $ O(n) $ time (instead of $ \Theta(n^2) $) as follows: \begin{itemize}
		\item Select tuples in $ X $ that satisfy $ 2x_1 - 3x_2 = 1 $ and let $ A \subseteq X $ denote the result;
		\item Select tuples in $ Y $ that satisfy $ y_1 + 5y_2 < 6 $ and let $ B \subseteq Y $ denote the result;
		\item Return $ A \times B $;
	\end{itemize} Moreover, the first two steps are fully independent of each other and thus can be executed in parallel. As regards the method of measuring query evaluation performance, it is worth noting that in the present example we have slightly departed from the standard notion of \textit{data complexity} used in the literature to assess the theoretical speed of query evaluation (see \cite[Definition 2.4.9]{CDB-book} and \cite{vardi:1982, chandra:1980}) in the aspect that we treat as taking linear time the operations whose time complexity is linear in the size of the desired output. This has a very simple explanation: the result of a query may be of size quadratic in the database size, so it would be otherwise impossible to develop a subquadratic algorithm for query evaluation because merely outputting the result already requires quadratic time.

	More generally, suppose in some database we have relations $ X(x_1, \dots, x_m) \subseteq \R^m $ and $ Y(y_1, \dots, y_k) \subseteq \R^k $, and we want to evaluate a conjunctive query \[
		q \leftarrow X(x_1, \dots, x_m), Y(y_1, \dots, y_k), \varphi(x_1, \dots, x_m, y_1, \dots, y_k)
	\] where $ \varphi $ is the query constraint. Suppose the variables $ x_1, \dots, x_m $ are independent from $ y_1, \dots, y_k $ in the sense that $ \varphi $ has an equivalent representation \[
		\varphi \equiv \bigvee_{i=1}^t \psi_i(x_1, \dots, x_m) \wedge \gamma_i(y_1, \dots, y_k)
	\] where $ \psi_i $ and $ \gamma_i $ are formulas. Naively evaluating this query would mean joining $ X $ and $ Y $, which takes $ \Theta(n^2) $ time. However, the same query can be evaluated in $ O(n) $ time by independently selecting tuples from $ X $ that satisfy some $ \psi_i $ formula, tuples from $ Y $ satisfying some $ \gamma_j $ formula, and outputting the union of all pairs of tuples satisfying $ \psi_i \wedge \gamma_i $ for some $ i $. 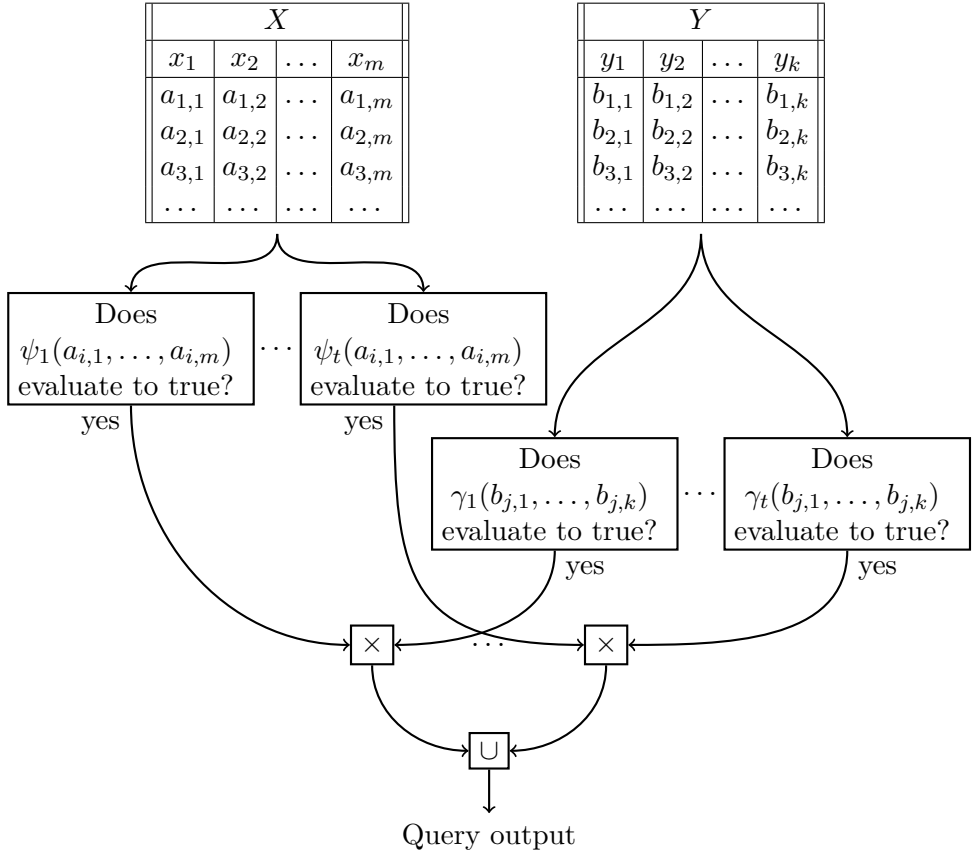
\begin{figure}[t]
		\centering
		\begin{tikzpicture}[thick]
	\node (xrelation) { 
		\begin{tabular}{||c|c|c|c||}
			\hline
			\multicolumn{4}{||c||}{$ X $} \\
			\hline
			$ x_1 $ & $ x_2 $ & $ \dots $ & $ x_m $ \\
			\hline
			$ a_{1,1} $ & $ a_{1,2} $ & $ \dots $ & $ a_{1,m} $ \\
			$ a_{2,1} $ & $ a_{2,2} $ & $ \dots $ & $ a_{2,m} $ \\
			$ a_{3,1} $ & $ a_{3,2} $ & $ \dots $ & $ a_{3,m} $ \\
			$ \dots $ & $ \dots $ & $ \dots $ & $ \dots $ \\
			\hline
		\end{tabular}
	};
	
	\node[right=5em of xrelation] (yrelation) {
		\begin{tabular}{||c|c|c|c||}
			\hline
			\multicolumn{4}{||c||}{$ Y $} \\
			\hline
			$ y_1 $ & $ y_2 $ & $ \dots $ & $ y_k $ \\
			\hline
			$ b_{1,1} $ & $ b_{1,2} $ & $ \dots $ & $ b_{1,k} $ \\
			$ b_{2,1} $ & $ b_{2,2} $ & $ \dots $ & $ b_{2,k} $ \\
			$ b_{3,1} $ & $ b_{3,2} $ & $ \dots $ & $ b_{3,k} $ \\
			$ \dots $ & $ \dots $ & $ \dots $ & $ \dots $ \\
			\hline
		\end{tabular}
	};

	\node[block, below=2em of xrelation, xshift=-5em] (xrelationfilter1) {
		\begin{varwidth}{9em}
			\centering
			Does $ \psi_1(a_{i,1}, \dots, a_{i,m}) $ evaluate to true?
		\end{varwidth}
	};

	\node[block, below=2em of xrelation, xshift=5em] (xrelationfiltert) {
		\begin{varwidth}{9em}
			\centering
			Does $ \psi_t(a_{i,1}, \dots, a_{i,m}) $ evaluate to true?
		\end{varwidth}
	};

	\node[xshift=0.05em] at ($ (xrelationfilter1.east)!0.5!(xrelationfiltert.west) $) {$ \dots $};

	\node[block, below=7em of yrelation, xshift=-5em] (yrelationfilter1) {
		\begin{varwidth}{9em}
			\centering
			Does $ \gamma_1(b_{j,1}, \dots, b_{j,k}) $ evaluate to true?
		\end{varwidth}
	};

	\node[block, below=7em of yrelation, xshift=5em] (yrelationfiltert) {
		\begin{varwidth}{9em}
			\centering
			Does $ \gamma_t(b_{j,1}, \dots, b_{j,k}) $ evaluate to true?
		\end{varwidth}
	};

	\node[xshift=0.05em] at ($ (yrelationfilter1.east)!0.5!(yrelationfiltert.west) $) {$ \dots $};


	\node[rectangle, draw, align=center, anchor=north, xshift=-4em, yshift=-7em] (combinator1) at ($ (xrelationfiltert)!0.5!(yrelationfilter1) $) {$ \times $};

	\node[rectangle, draw, align=center, anchor=north, xshift=4em, yshift=-7em] (combinatort) at ($ (xrelationfiltert)!0.5!(yrelationfilter1) $) {$ \times $};
	
	\node[xshift=0.05em] at ($ (combinator1)!0.5!(combinatort) $) {$ \dots $};

	\node[rectangle, draw, align=center, anchor=north, yshift=-3em] (union) at ($ (combinator1)!0.5!(combinatort) $) {$ \cup $};
	
	\node[below=1.5em of union] (queryresult) {Query output};
	
	\draw[->] (xrelation) to[out=-90, in=90] (xrelationfilter1);
	\draw[->] (xrelation) to[out=-90, in=90] (xrelationfiltert);
	
	\draw[->] (yrelation) to[out=-90, in=90] (yrelationfilter1);
	\draw[->] (yrelation) to[out=-90, in=90] (yrelationfiltert);
	
	\draw[->] (xrelationfilter1) to[out=-90, in=180] (combinator1);
	\draw[->] (yrelationfilter1) to[out=-90, in=0] (combinator1);
	
	\draw[->] ($ (xrelationfiltert.south) - (1em, 0) $) to[out=-90, in=180, distance=6em] (combinatort);
	\draw[->] (yrelationfiltert) to[out=-90, in=0] (combinatort);
	
	\node[anchor=north east] at (xrelationfilter1.south) {yes};
	\node[anchor=north east] at ($ (xrelationfiltert.south) - (1em, 0) $) {yes};
	\node[anchor=north west] at (yrelationfilter1.south) {yes};
	\node[anchor=north west] at (yrelationfiltert.south) {yes};
	
	\draw[->] (combinator1) to[out=-90, in=180] (union);
	\draw[->] (combinatort) to[out=-90, in=0] (union);
	
	
	\draw[->] (union) -- (queryresult);
\end{tikzpicture}
		\caption{Data flow of efficient query evaluation. Tuples of the form $ (a_{i,1}, \dots, a_{i,m}) \in \R^m $ and $ (b_{j,1}, \dots, b_{j,k}) \in \R^k $ originate in relations $ X $ and $ Y $ depicted at the top of the diagram and travel along edges. Nodes containing a question filter tuples accordingly. The ``$ \times $'' node produces, for every pair of incoming tuples $ (a_{i,1}, \dots, a_{i,m}) $ and $ (b_{j,1}, \dots, b_{j,k}) $, the tuple $ (a_{i,1}, \dots, a_{i,m}, b_{j,1}, \dots, b_{j,k}) $. The ``$ \cup $'' node simply forwards all incoming tuples.}
		\label{fig:independence_eff_query_eval}
	\end{figure} This computation is visualized in Figure~\ref{fig:independence_eff_query_eval}. Moreover, this evaluation technique can be parallelized because the decision of whether or not we select some tuple from $ X $ or $ Y $ depends only on that tuple.
\end{zbexample}

Although Example~\ref{example:independence_eff_query_eval} deals with a simple relational database, it can be easily generalized to constraint databases. In that setting, we would be working with constraints rather than tuples, and selection in $ X $ (resp. $ Y $) would correspond to adding $ 2x_1 - 3x_2 = 1 $ (resp. $ y_1 + 5y_2 < 6 $) as a conjunct to every formula in $ A $ (resp. $ B $) and removing unsatisfiable conjuncts.


As is apparent from Example~\ref{example:independence_eff_query_eval}, query evaluation optimization reduces to checking whether the constraints involved in the query can be rewritten as a disjunction of conjunctions never using certain pairs of variables at the same time, and computing such a representation if yes. More generally, the property of being able to rewrite $ \varphi(x_1, \dots, x_n) $ as a Boolean combination $ \psi $ of predicates not establishing unwanted connections between variables is known as \textit{variable decomposability} with respect to a partition $ \Pi $ of $ \{x_1, \dots, x_n\} $, where $ \Pi $ specifies the allowed and forbidden connections between variables \cite{libkin:2003, libkin:2000}. More precisely, a formula $ \varphi $ from some quantifier-free first-order theory generated by a structure $ \Mstruct $ is said to be \textit{$ \Pi $-decomposable} \cite{hague:2020} if there exists a Boolean combination $ \psi $ of predicates each having its free variables within some block of $ \Pi $, such that $ \varphi $ is equivalent to $ \psi $ in the theory of $ \Mstruct $. Such a Boolean combination is called a \textit{$ \Pi $-decomposition}. Two variables $ x_i $ and $ x_j $ are said to be \textit{independent} \cite{libkin:2003} if $ \varphi $ is $ \Pi $-decomposable for some partition $ \Pi $ containing $ x_i $ and $ x_j $ in distinct blocks. In the context of Example~\ref{example:independence_eff_query_eval}, $ \Mstruct $ is the $ (\R, +, -, 0, 1, =, <) $ structure whose theory is well-known as Linear Real Arithmetic (LRA); the transformation of $ \varphi $ into $ \psi $ shows that $ \varphi $ is $ \Pi $-decomposable for $ \Pi := \{\{x_1, x_2\}, \{y_1, y_2\}\} $. In particular, it follows that $ x_i $ and $ y_j $ are independent for all $ i, j \in \{1, 2\} $.

One of the most important special cases of variable decomposability people have extensively studied over the last decade is the notion of \textit{monadic decomposability} \cite{veanes:2017}, which asks whether the given formula can be rewritten as a Boolean combination of monadic predicates, i.e., predicates containing at most one free variable. Such a Boolean combination is known as a \textit{monadic decomposition}. In other words, a monadic decomposition of a formula $ \varphi(x_1, \dots, x_n) $ is simply a $ \Pi $-decomposition, where $ \Pi := \{\{x_1\}, \dots, \{x_n\}\} $ is the partition consisting of singleton sets. Thus, for such a partition $ \Pi $, the notion of monadic decomposability coincides with that of $ \Pi $-decomposability.

\subsection{History}

Variable decomposability was first systematically studied by Libkin about 20 years ago \cite{libkin:2000, libkin:2003}. His primary motivation was to optimize constraint database query processing and, in particular, quantifier elimination. His approach involved developing an equivalence relation capturing the notion of variable decomposability, which he used to derive a sufficient condition for the decidability of variable decomposition and independence problems \cite[pp. 438, 439]{libkin:2003}. In the same paper, he extended his decidability result to yield algorithms determining variable decomposability over Presburger arithmetic, Linear Real Arithmetic and the theory of Real-Closed Fields. The monadic decomposability problem was originally studied by Veanes et al. \cite{veanes:2017} in the context of Satisfiability Modulo Theories (SMT). They derived an efficient semi-algorithm that produces monadic decompositions regardless of the base theory under consideration, but it does not terminate if the input formula is not monadically decomposable. The authors also provided decidability results for several major SMT theories, including Presburger arithmetic and EUF, i.e., the theory of equality and uninterpreted functions. Their work also includes a comprehensive list of applications of monadic/variable decompositions, including symbolic transducers \cite{dnthony:2017:symtransducers} and string analysis.

A series of other complexity results for the variable decomposability problem are available. For example, Hague et al. \cite{hague:2020} proved that over the quantifier-free fragment of Presburger arithmetic, deciding monadic/variable decomposability is $ \coNP $-complete. Barcel{\'{o}} et al. \cite{barcelo:2019} showed that it is $ \NL $-complete (resp. $ \PSPACE $-complete) to decide the monadic decomposability over the quantifier-free theory of automatic structures when the formulas are represented as deterministic (resp. nondeterministic) synchronized multi-tape automata. A good example of a similar research direction is the recent work by Markgraf et al. \cite{markgraf:2021}, where the authors develop a polynomial-time algorithm for learning integer hypercubes and apply it to generate monadic decompositions of Presburger arithmetic formulas.

\subsection{Motivation}

It is hard to construct procedures capable of deciding variable decomposability of $ \varphi $ with respect to a partition $ \Pi $ because in case $ \varphi $ is not $ \Pi $-decomposable, any correct algorithm must have, after finitely many steps, ruled out the existence of any $ \Pi $-decomposition, whereas there are infinitely many Boolean combinations of predicates syntactically not establishing forbidden connections between variables. In other words, it is hard to find a finite witness for the fact that no formula from an infinite family thereof is equivalent to the given formula $ \varphi $. Moreover, merely checking the correctness of a given $ \Pi $-decomposition involves deciding the equivalence of two arbitrary formulas, which is computationally already as hard as deciding the validity of a formula in the first-order theory of $ \Mstruct $, which is a $ \coNP $-hard problem for most theories. Another major difficulty in deciding variable decomposability lies in the apparent absence of a normal form which could allow us to solve the problem by separately analyzing subformulas of some simpler form. For example, over linear real arithmetic, the formulas $ x = y $, $ x < -y $ and $ y+2 > -x $ are all not monadically decomposable, but their disjunction \[
	\varphi := x = y \vee x < -y \vee y+2 > -x
\] is. Therefore, we cannot simply analyze DNF\footnote{Disjunctive normal form.} terms of $ \varphi $ separately. Neither can a separate analysis of CNF\footnote{Conjunctive normal form.} clauses yield any decisive insight into whether or not the given formula is monadically or $ \Pi $-decomposable.


There are already a number of known techniques which cope with all these difficulties and successfully yield decision procedures for monadic or even variable decomposability. However, all known algorithms work either only for ``discrete'' fragments of first-order logic, or they rely on very general techniques which yield decision procedures whose complexity is either high in theory or infeasible in practice (or both). For example, Libkin's double-exponential time\footnote{In the paper, it is shown that the algorithm runs in polynomial time. However, this is because the number of free variables is fixed. Adapting his algorithm to work for a non-fixed number of variables entails a double-exponential blowup.} algorithm for variable decomposability over the theory of Real-Closed Fields \cite[Theorem 12]{libkin:2003} heavily relies on quantifier elimination complexity bounds, which unfortunately makes the algorithm only of theoretical interest. The proof technique of Hague et al. \cite{hague:2020} exploits the fact that after a value assigned to some free variable passes a certain bound, swapping its name with that of some other free variable cannot change the truth value of the given Presburger arithmetic formula if it is monadically decomposable. Their technique is beautiful, but unfortunately, it works only thanks to the discrete nature of integers because it heavily relies on the possibility of simply enumerating all possible models of the formula consisting of values not exceeding the bound. This enumeration is impossible if we have, e.g., reals instead of integers, like in linear real arithmetic. All other comparable algorithms (see, e.g., \cite{veanes:2017, markgraf:2021, barcelo:2019, bergstrasser:2022:ramsey, beimel:1998}) also, unfortunately, at their core, extensively rely on the discreteness of $ \Mstruct $'s universe; they cannot be adapted easily to tackle the same problems but for non-discrete universes.

For these reasons, the primary motivation of this \zbdocname/ is to develop a technique that does not fundamentally rely on discreteness assumptions but nevertheless yields algorithms for monadic/variable decomposability that are both optimal in theory and efficient in practice. Another motivation, of course, is the wide range of possible applications monadic/variable decompositions and consequently our results have, including constraint database query optimization (see above), string solving \cite{hague:2020}, symbolic transducers \cite{veanes:2017, dnthony:2017:symtransducers} and quantifier elimination \cite{markgraf:2021, hague:2020}. In particular, it is worth noting that most applications of monadic/variable decomposability in \cite{hague:2020} transfer to our setting of quantifier-free linear real arithmetic because they mostly rely only on the construction of $ \Pi $-decompositions, whereas $ \Pi $-decomposability of $ \varphi $ over quantifier-free linear real arithmetic implies $ \Pi $-decomposability of the same formula $ \varphi $ over (quantifier-free) Presburger arithmetic.



\subsection{Contributions}

The main contribution of this \zbdocname/ is a novel \textit{model-flooding} technique that we primarily use to establish precise complexity bounds for monadic/variable decomposability. We show that the same technique also yields optimal algorithms for computing variable decompositions, which is a central problem in constraint database applications (see, e.g., \cite{libkin:2003, CGKT03}).


We present this novel model-flooding technique in the setting of linear real arithmetic -- one of the most popular and best-investigated fragments of first-order logic \cite{monniaux:2008, dutertre:2006:linarithsolver, dutertre:2006:simplexdpllt, bozzano:2005:incrementalsatlra}. We establish the precise complexity of deciding variable decomposability and an $ \NP^\NP $ upper bound for the variable independence problem:

\begin{theorem}
	\label{thm:contrib:vardec_conp_complete_ind_sigma2}
	Given a quantifier-free linear real arithmetic formula $ \varphi $ and a partition $ \Pi $, it is $ \coNP $-complete to decide the $ \Pi $-decomposability of $ \varphi $. In particular, it follows that, over the same theory, deciding the independence of two given variables is in $ \NP^\NP $.
\end{theorem}

We also obtain a similar complexity result for monadic decomposability:

\begin{theorem}
	\label{thm:contrib:mondec_conp_complete}
	Given a quantifier-free linear real arithmetic formula $ \varphi $, it is $ \coNP $-complete to decide the monadic decomposability of $ \varphi $.
\end{theorem}

Theorems \ref{thm:contrib:vardec_conp_complete_ind_sigma2} and \ref{thm:contrib:mondec_conp_complete} answer an open question of \cite{veanes:2017}, which was subsequently reiterated in \cite{hague:2020}. In particular, these results significantly improve the best known double-exponential time algorithms for monadic/variable decomposability and independence problems \cite{libkin:2000, libkin:2003}.

We further derive a deterministic exponential-time algorithm for computing variable decompositions whenever they exist:

\begin{theorem}
	\label{thm:contrib:computing_dec_exponential_time}
	There exists an exponential-time algorithm, which, given a quantifier-free linear real arithmetic formula $ \varphi $ and a binary partition $ \Pi $, either computes a $ \Pi $-decomposition of $ \varphi $ or determines that $ \varphi $ is not $ \Pi $-decomposable.
\end{theorem}

The algorithm of Theorem~\ref{thm:contrib:computing_dec_exponential_time} is not only exponentially better compared to the best known algorithm for the problem \cite{libkin:2003} but also optimal in the setting of deterministic algorithms because the underlying decision problem is $ \coNP $-hard by Theorem~\ref{thm:contrib:vardec_conp_complete_ind_sigma2}.

We further provide an implementation of the algorithm of Theorem~\ref{thm:contrib:computing_dec_exponential_time}. Our experiments have shown that the algorithm is not just of theoretical interest, but it is also practical and outperforms existing comparable algorithms in both time and size of produced decompositions, at least for classes of formulas we considered.

In constraint database applications, constraints appearing in queries often do not admit perfect variable independence in the sense that some constraint $ \varphi $ may not be $ \Pi $-decomposable (for an appropriate partition $ \Pi $). Nevertheless, there still may be some ``partial'' variable independence in $ \varphi $. One of the major advantages of the model flooding technique we introduce is its ability to yield algorithms that efficiently construct best-possible approximations of variable decompositions for any formula. As we demonstrate in Example~\ref{example:partial_variable_independence_application}, such approximations of decompositions can be used to speed up query evaluation even if the constraints involved are not $ \Pi $-decomposable.

\begin{zbexample}
	\label{example:partial_variable_independence_application}
	Consider a database containing information about student apartments in Kaiserslautern and Berlin. More precisely, let the database have two relations -- $ \operatorname{KL}(\operatorname{Address}, s_{\operatorname{KL}}) $ and $ \operatorname{Berlin}(\operatorname{Address}, s_{\operatorname{Berlin}}) $. The $ s_{\operatorname{KL}} $ and $ s_{\operatorname{Berlin}} $ attributes of $ \operatorname{KL} $ and $ \operatorname{Berlin} $ represent the areas (in $ m^2 $) of apartments in Kaiserslautern and Berlin, respectively. Furthermore, let $ r_{\operatorname{KL}} := 18.24 $ and $ r_{\operatorname{Berlin}} := 27.17 $ be the (average) rental prices of one square meter in Kaiserslautern and Berlin, respectively.
	
	Suppose we want to find apartments from Kaiserslautern and Berlin that have approximately the same area but which differ in price significantly, say, the price of the apartment in Berlin must be at least double the price of the apartment in Kaiserslautern. More precisely, we regard two apartments as having approximately the same area whenever these areas differ by at most 5 square meters. Finding such apartments corresponds to evaluating the conjunctive query \[
	q \leftarrow \operatorname{KL}(\_, s_{\operatorname{KL}}), \operatorname{Berlin}(\_, s_{\operatorname{Berlin}}), \varphi(s_{\operatorname{KL}}, s_{\operatorname{Berlin}})
	\] where \begin{align*}
		\varphi(s_{\operatorname{KL}}, s_{\operatorname{Berlin}}) &:= \left|s_{\operatorname{KL}} - s_{\operatorname{Berlin}}\right| \le 5 \wedge r_{\operatorname{Berlin}} s_{\operatorname{Berlin}} \ge 2 r_{\operatorname{KL}} s_{\operatorname{KL}} \\
		&\wedge s_{\operatorname{KL}} > 0 \wedge s_{\operatorname{Berlin}} > 0
	\end{align*} is an LRA formula defining the desired relation between areas of apartments\footnote{In the formula, we used the absolute value operator that is not allowed in LRA. However, $ \left|s_{\operatorname{KL}} - s_{\operatorname{Berlin}}\right| \le 5 $ is expressible in LRA because this statement can be rewritten as $ 0 \le s_{\operatorname{KL}} - s_{\operatorname{Berlin}} \le 5 \vee 0 \le s_{\operatorname{Berlin}} - s_{\operatorname{KL}} \le 5 $.}. As in Example~\ref{example:independence_eff_query_eval}, a naive evaluation of $ q $ requires quadratic time in the worst case. However, using a variation of the exponential-time algorithm of Theorem~\ref{thm:contrib:computing_dec_exponential_time}, it is possible to efficiently compute the formula \[
	\psi := s_{\operatorname{KL}} \le \frac{715}{49} \wedge s_{\operatorname{KL}} > 0 \wedge s_{\operatorname{Berlin}} \le \frac{960}{49} \wedge s_{\operatorname{Berlin}} > 0
	\] such that: (1) $ \varphi \rightarrow \psi $ is valid in LRA, and (2) $ \psi $ is a best-possible approximation of a monadic decomposition of $ \varphi $. The central idea is to use $ \psi $ instead of $ \varphi $ to optimize query evaluation as follows: \begin{itemize}
		\item Select tuples in $ \operatorname{KL} $ that satisfy $ s_{\operatorname{KL}} \le \frac{715}{49} \wedge s_{\operatorname{KL}} > 0 $ and let $ K $ denote the result;
		\item Select tuples in $ \operatorname{Berlin} $ that satisfy $ s_{\operatorname{Berlin}} \le \frac{960}{49} \wedge s_{\operatorname{Berlin}} > 0 $ and let $ B $ denote the result;
		\item Return all tuples in $ K \times B $ that satisfy $ \varphi $;
	\end{itemize}\begin{figure}[t]
		\centering
		\begin{tikzpicture}[thick]
	\node (kl) { 
		\begin{tabular}{||c|c||}
			\hline
			\multicolumn{2}{||c||}{$ \operatorname{KL} $} \\
			\hline
			$ \operatorname{Address} $ & $ s_{\operatorname{KL}} $ \\
			\hline
			Richard-Wagner-Stra{\ss}e 88 & $ 14 $ \\
			Friedrich-Engels-Stra{\ss}e 5 & $ 13 $ \\
			Apfelstra{\ss}e 6 & $ 19 $ \\
			$ \dots $ & $ \dots $ \\
			\hline
		\end{tabular}
	};
	
	\node[right=2em of kl] (berlin) {
		\begin{tabular}{||c|c||}
			\hline
			\multicolumn{2}{||c||}{$ \operatorname{Berlin} $} \\
			\hline
			$ \operatorname{Address} $ & $ s_{\operatorname{Berlin}} $ \\
			\hline
			Hermannstra{\ss}e 151 & $ 19 $ \\
			Lacknerstra{\ss}e 5 & $ 24 $ \\
			Friedelstra{\ss}e 23 & $ 18 $ \\
			$ \dots $ & $ \dots $ \\
			\hline
		\end{tabular}
	};
	
	\node[block, below=1.5em of kl] (klfilter) {
		\begin{varwidth}{9em}
			\centering
			Does $ s_{\operatorname{KL}} \le \frac{715}{49} \wedge s_{\operatorname{KL}} > 0 $ evaluate to true?
		\end{varwidth}
	};
	
	\node[block, below=1.5em of berlin] (berlinfilter) {
		\begin{varwidth}{12em}
			\centering
			Does $ s_{\operatorname{Berlin}} \le \frac{960}{49} \wedge s_{\operatorname{Berlin}} > 0 $ evaluate to true?
		\end{varwidth}
	};

	\node[rectangle, draw, align=center, anchor=north, yshift=-4em] (combinator) at ($ (klfilter)!0.5!(berlinfilter) $) {$ \times $};
	
	\node[block, below=1.5em of combinator] (finalfilter) {Does $ \varphi $ evaluate to true?};
	
	\node[below=1.5em of finalfilter] (queryresult) {Query output};

	\draw[->] (kl) -- (klfilter);
	\draw[->] (berlin) -- (berlinfilter);
	\draw[->] (klfilter) to[out=-90, in=180] (combinator);
	\draw[->] (berlinfilter) to[out=-90, in=0] (combinator);
	
	\node[anchor=north east] at (klfilter.south) {yes};
	\node[anchor=north west] at (berlinfilter.south) {yes};
	
	
	\draw[->] (combinator) -- (finalfilter);
	
	\draw[->] (finalfilter) -- (queryresult) node[midway, right] {yes};
\end{tikzpicture}
		\caption{Data flow of efficient query evaluation. As in Figure~\ref{fig:independence_eff_query_eval}, tuples originate in relations $ \operatorname{KL} $ and $ \operatorname{Berlin} $ depicted at the top of the diagram and travel along edges. Nodes containing a question filter tuples accordingly. The ``$ \times $'' node produces, for every pair of incoming tuples $ (a_{\operatorname{KL}}, s_{\operatorname{KL}}) $ and $ (a_{\operatorname{Berlin}}, s_{\operatorname{Berlin}}) $, the tuple $ (a_{\operatorname{KL}}, s_{\operatorname{KL}}, a_{\operatorname{Berlin}}, s_{\operatorname{Berlin}}) $.}
		\label{fig:partial_variable_independence_application}
	\end{figure}
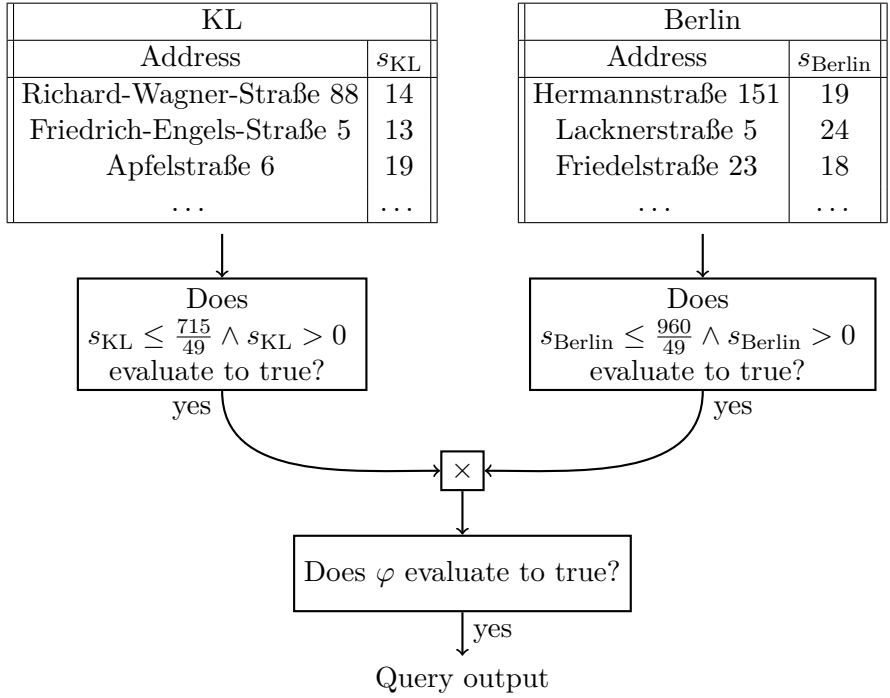 This method of evaluating $ q $ is visualized in Figure~\ref{fig:partial_variable_independence_application}. Since $ \varphi \models \psi $, the desired query result is a subset of $ K \times B $. However, $ K \times B $ may contain tuples that do not satisfy $ \varphi $. In order to ensure that we return precisely the desired set of tuples, before outputting any tuple from $ K \times B $, we simply check if that tuple satisfies $ \varphi $ and ignore it if not. Since $ \psi $ is a best-possible approximation of a monadic decomposition of $ \varphi $, it is reasonable to assume that there are at most $ O(n) $ tuples in $ \operatorname{KL} \times \operatorname{Berlin} $ that satisfy $ \psi $ but do not satisfy $ \varphi $. Hence, $ K \times B $ contains only $ O(n) $ many tuples, so the entire query can be evaluated in linear time (instead of quadratic) under the above assumption. As in Example~\ref{example:independence_eff_query_eval}, it can be further observed that the described method of evaluating $ q $ can be optimized by performing the selection of tuples from $ \operatorname{KL} $ and $ \operatorname{Berlin} $ in parallel.
\end{zbexample}

Next, we prove that the variable independence problem cannot be solved in polynomial time unless $ \P = \NP = \coNP $:

\begin{theorem}
	\label{thm:contrib:independence_conp_hard_wrt_ctt}
	Given a quantifier-free linear real arithmetic formula $ \varphi $ and two of its free variables $ x_i $ and $ x_j $, the problem of deciding the independence of $ x_i $ and $ x_j $ is $ \coNP $-hard with respect to conjunctive truth-table reductions.
\end{theorem}

We complement our results by proving that determining $ \Pi $-decomposability can be efficiently reduced to the same problem but when the partition $ \Pi $ is binary, i.e., consists of only two blocks. The reduction works by computing a collection $ S $ of binary partitions and then outputting that $ \varphi $ is $ \Pi $-decomposable whenever $ \varphi $ is $ \Pi' $-decomposable for all $ \Pi' \in S $. This method and its correctness immediately follow from the following Theorem~\ref{thm:contrib:reduction}.

\begin{theorem}
	\label{thm:contrib:reduction}
	For any (non-unary) partition $ \Pi $, there exists a collection $ S $ of binary partitions, such that any formula $ \varphi $ is $ \Pi $-decomposable if and only if $ \varphi $ is $ \Pi' $-decomposable for all $ \Pi' \in S $. Moreover, $ S $ is of size $ O(\log(|\Pi|)) $ and can be computed in polynomial time.
\end{theorem}

It was already known\footnote{This follows from \cite[p. 5]{hague:2020} read in conjunction with \cite{cosmadakis:2001}.} that testing $ \{X_1, \dots, X_n\} $-decomposability of $ \varphi $ is equivalent to checking that $ \varphi $ is $ \{X_i, \bigcup_{j\neq i} X_j\} $-decomposable for all $ 1 \le i \le n $. Thus, the novelty of Theorem~\ref{thm:contrib:reduction} is in the size of $ S $ -- we improve it from $ O(\left|\Pi\right|) $ to $ O(\log(\left|\Pi\right|)) $.

\subsection{Organization}

This \zbdocname/ is organized as follows. We shall first discuss the necessary preliminaries in \zbrefsec{sec:preliminaries}. \zbrefsec{sec:vardec} is dedicated to the development of algorithms for deciding variable decomposability and for constructing decompositions if they exist. In \zbrefsec{sec:lower_bounds}, we prove lower bounds for deciding monadic/variable decomposability and for determining variable independence. In \zbrefsec{sec:experiments}, we discuss our implementation of the algorithm for variable decomposability and present experimental results. We give some concluding remarks and discuss future work in \zbrefsec{sec:conclusions}.

\section{Preliminaries}
\label{sec:preliminaries}

\subsection{Definitions and notation}

\subsubsection{First-order logic}

We assume familiarity with the syntax and semantics of first-order logic \cite{kroening:2016, bradley:2007, schoening:2008, immerman:2012:descriptivecomplexity, libkin:2004}. Fix a first-order signature $ \Sigma $ and a $ \Sigma $-structure $ \Mstruct = (\Uuniverse, \Iinterp) $, where $ \Uuniverse \neq \varnothing $ is the universe and $ \Iinterp $ is the interpretation function. Let $ \varphi(x_1, \dots, x_n) $ be a $ \Sigma $-formula with free variables $ x_1, \dots, x_n $ and $ \Pi $ be a partition of $ \{x_1, \dots, x_n\} $. The formula $ \varphi $ is said to \textit{respect} $ \Pi $, if $ \varphi $ is a Boolean combination of formulas, each having its free variables within some element of $ \Pi $ \cite{veanes:2017}. Such a Boolean combination is also called a \textit{$ \Pi $-decomposition} \cite{hague:2020}. We say that $ \varphi $ \textit{disrespects} $ \Pi $, if $ \varphi $ does not respect $ \Pi $.

A formula $ \varphi(x_1, \dots, x_n) $ in the theory of $ \Mstruct $ is called \textit{$ \Pi $-decomposable} whenever there exists a $ \Pi $-respecting formula $ \psi(x_1, \dots, x_n) $, such that \[
	\Mstruct \models \forall x_1, \dots, x_n (\varphi \leftrightarrow \psi)
\] i.e., when $ \varphi $ is equivalent to $ \psi $ in the theory of $ \Mstruct $. Note that the difference between respecting $ \Pi $ and being $ \Pi $-decomposable is that the former notion depends only on the syntax of the formula, whereas the latter notion is semantic (it depends on $ \Mstruct $). The \textit{variable decomposition problem} is to decide, given $ \varphi $ and $ \Pi $, whether $ \varphi $ is $ \Pi $-decomposable and to compute, in case of a positive answer, a possible $ \Pi $-decomposition of $ \varphi $. Two free variables $ x_i $ and $ x_j $ appearing in a formula $ \varphi(x_1, \dots, x_n) $ are said to be \textit{independent}, if $ \varphi $ is $ \Pi $-decomposable for some partition $ \Pi $ containing $ x_i $ and $ x_j $ in distinct blocks. The \textit{variable independence problem} \cite{libkin:2003} is to decide, given $ \varphi $, $ x_i $ and $ x_j $, whether $ x_i $ and $ x_j $ are independent or not.

We use the notation $ \Free(\varphi) $ and $ \Pred(\varphi) $ to denote the set of free variables and the set of predicates appearing in the formula $ \varphi $, respectively. A valuation $ \nu $ is a function mapping every free variable to an element of $ \Uuniverse $. We write $ \Mstruct_\nu \models \varphi $ whenever $ \varphi $ evaluates to true under $ \Mstruct $ with every $ x \in \Free(\varphi) $ being interpreted as $ \nu(x) $. If this is the case, we call $ \Mstruct_\nu $ a model of $ \varphi $. Let $ \Phi, \Psi $ be sets of formulas. We write $ \Mstruct_\nu \models \Phi $ if $ \Mstruct_\nu \models \varphi $ for every $ \varphi \in \Phi $. More generally, we say that $ \Phi $ entails $ \Psi $ and write $ \Phi \models \Psi $ if for every valuation $ \nu $, $ \Mstruct_\nu \models \Phi $ implies $ \Mstruct_\nu \models \Psi $. If $ \Phi = \{\varphi\} $ is a singleton set, then $ \varphi \models \Psi $ stands for $ \{\varphi\} \models \Psi $. We allow analogous syntactic sugar also for the right-hand side of the entailment. Furthermore, we do not distinguish between single formulas and finite sets thereof. That is, for a finite set $ \Phi $ of formulas, we may also treat it directly as a single formula where appropriate. More precisely, $ \Phi $ is syntactic sugar for $ \bigwedge_{\varphi \in \Phi} \varphi $ if used as a formula.

\subsubsection{Partitions and lattices}

Let $ X \neq \varnothing $ be a set and $ \Pi, \Pi' $ be partitions of $ X $. The partition $ \Pi $ is called a \textit{refinement} of $ \Pi' $ (written: $ \Pi \sqsubseteq \Pi' $) if every element of $ \Pi $ is a subset of some set in $ \Pi' $. Clearly, the set of all partitions of $ X $ together with the refinement relation $ \sqsubseteq $ form a lattice with infimum $ \{\{x\} \mid x \in X\} $ and supremum $ \{X\} $. The \textit{meet operation} $ \sqcap $ in the lattice of $ X $'s partitions is defined as follows: the elements of $ \Pi \sqcap \Pi' $ are the nonempty sets of the form $ P \cap P' $, where $ P \in \Pi $ and $ P' \in \Pi' $. Note that $ \sqcap $ is commutative and associative.

\subsubsection{Complexity theory}
\label{sec:preliminaries:complexitytheory}

We assume familiarity with the basics of complexity theory \cite{arora:2009:computationalcomplexity, hemaspaandra:2002:complexitytheorycompanion, sipser:2012}. Among other well-known concepts, we shall, in particular, rely on the notion of a \textit{conjunctive truth-table reduction} \cite[pp. 305, 306]{hemaspaandra:2002:complexitytheorycompanion}, which we now briefly reiterate because it is less well-known compared to concepts line Turing/many-one reducibility or the polynomial hierarchy. Intuitively, truth-table reductions are simply Turing reductions, but with the additional requirement that the output must be the conjunction of answers to queries made by the reduction. More precisely, fix a finite alphabet $ \Sigma $, a symbol $ \# \notin \Sigma $, and let $ A \subseteq \Sigma^* $, $ B \subseteq \Sigma^* $ be languages. A conjunctive truth-table reduction from $ A $ to $ B $ is a polynomial-time computable function $ g : \Sigma^* \rightarrow (\Sigma \cup \{\#\})^* $, such that for all $ x \in \Sigma^* $, if \[
	g(x) = z_1 \# z_2 \# \dots \# z_l
\] is the output of $ g $ on input $ x $ where $ z_i \in \Sigma^* $, then \begin{align}
	\label{eqn:ctt_acceptance}
	x \in A \Leftrightarrow \bigwedge_{i=1}^l z_i \in B
\end{align} In this definition, the $ z_i $ strings are the queries to $ B $ made throughout the reduction, whereas (\ref{eqn:ctt_acceptance}) is the property that its output must be the conjunction of answers to the queries it makes. Note that $ \coNP $ is closed under conjunctive truth-table reductions.

\subsection{Algebraic definitions}

We expect the reader to be familiar with basic concepts in algebra such as homomorphisms\footnote{Also known as structure-preserving or linear mappings.}, vectorspaces and chain conditions \cite{axler:2015:linalg, hungerford:2012:algebra, eisenbud:2013:commalg, halmos:1995:linalg}. We also expect familiarity with the very basics of convex and affine geometry \cite{boyd:2004, paffenholz:2010:polyhedralgeometryandlinearoptimization, gallier:2001, gruber:2007:convex, ziegler:2012:lecturesonpolytopes}.

We now list some algebraic definitions we use but for which there is no consensus in the literature. Let $ \id $ denote the identity mapping. If $ A $, $ B $, $ C $ are sets, $ f : B \rightarrow C $ and $ g : A \rightarrow B $ are mappings, then we write $ f \circ g $ for the composition of $ f $ and $ g $ defined as follows: \begin{align*}
	f \circ g : A &\rightarrow C \\
	x &\mapsto f(g(x))
\end{align*} In the entire \zbdocname/, unless stated otherwise, $ \log $ denotes the binary logarithm.

If $ A \subseteq \Q^n $ and $ B \subseteq \Q^n $ are sets, then we use the notation $ A + B $ (resp. $ A - B $) to denote the \textit{Minkowski sum} (resp. \textit{difference}) \cite{ziegler:2012:lecturesonpolytopes, paffenholz:2010:polyhedralgeometryandlinearoptimization, gilbert:1988:gjk, montanari:2017:gjkimprovement} of $ A $ and $ B $: \begin{align*}
	A + B := \{a + b \mid a \in A, b \in B\} \\
	A - B := \{a - b \mid a \in A, b \in B\}
\end{align*} For $ v \in \Q^n $, the notation $ v + A $ is syntactic sugar for $ \{v\} + A $. Analogously, for $ \lambda \in \Q $, we define $ \lambda \cdot A := \{\lambda \cdot a \mid a \in A\} $.

Let $ K $ be a field and let $ B = (v_1, \dots, v_n) $ be a basis of a $ K $-vectorspace $ V $. We denote with $ \lc_B $ the linear combination isomorphism $ K^n \cong V $ given by \begin{align*}
	\lc_B : K^n &\rightarrow V \\
	(a_1, \dots, a_n)^\transp &\mapsto \sum_{i=1}^{n} a_i v_i
\end{align*} We furthermore define $ V^\bot $ to be the \textit{orthogonal complement} \[
	V^\bot := \{u \in K^n \mid \forall v \in V : u \cdot v = 0\}
\] of $ V $ in $ K^n $. In this definition and in the following, ``$ \cdot $'' stands for the dot product if written between vectors.

If $ v_1, \dots, v_m $ are vectors in a $ K $-vectorspace $ K^n $, then we use the notation $ (v_1 \mid \dots \mid v_m) $ to denote the $ n \times m $ matrix obtained by writing $ v_1, \dots, v_m $ as columns.

We say that $ U $ is a subspace of $ V $ and write $ U \le V $ whenever $ U \subseteq V $ and $ U $ is a $ K $-vectorspace. If $ U \le V $ and $ U \subsetneq V $, we say that $ U $ is a strict subspace of $ V $ and write $ U \lneq V $.

For a $ K $-vectorspace homomorphism $ h : U \rightarrow V $ we let \begin{align*}
	\ImageOf(h) &:= \{h(u) \mid u \in U\} \le V \\
	\ker(h) &:= \{u \in U \mid h(u) = 0\} \le U
\end{align*} denote the image and kernel of $ h $, respectively.

For vectors $ u, v \in \Q^n $ we furthermore define \[
	(u, v) := \{\lambda \cdot v + (1 - \lambda) \cdot u \mid \lambda \in \Q, 0 < \lambda < 1\}
\] to be the open convex hull of $ \{u, v\} $ \cite{paffenholz:2010:polyhedralgeometryandlinearoptimization, boyd:2004}. Hereinafter, unless stated otherwise, we assume $ \varepsilon $, $ \delta $ and other coefficients to be rational numbers.

\subsection{Theory-specific definitions}

In this \zbdocname/, we study the variable decomposition and independence problems over the quantifier-free theory of the $ \R_{\mathrm{lin}} := (\R, +, -, 0, 1, =, <) $ structure. Since $ \R_{\mathrm{lin}} $ is elementarily equivalent\footnote{This can be easily shown using Ehrenfeucht-Fra{\"i}ss{\'e} games \cite{immerman:2012:descriptivecomplexity, libkin:2004}} to $ \Q_{\mathrm{lin}} := (\Q, +, -, 0, 1, =, <) $ which is in turn elementarily equivalent to \[
	\Mstruct := (\Q, +, -, 0, 1, =, <, >)
\] it suffices to study the theory of $ \Mstruct $. Define $ \QFLRA $ to be the set of all quantifier-free formulas over the theory of $ \Mstruct $. Note that $ \Mstruct $ is just $ \Q_{\mathrm{lin}} $ with the predicate symbol ``$ > $'' added to it and interpreted accordingly, i.e., as the ``strictly greater than'' relation over $ \Q $. This will simplify our reasoning later because over $ \Mstruct $, unlike $ \Q_{\mathrm{lin}} $, we can assume all conjunctions of literals to be positive, i.e., be actually conjunctions of predicates. Whenever prefix notation is more convenient, we may write the predicate symbols $ P_<, P_>, P_= $ as syntactic sugar denoting the $ <, >, = $ predicates in infix notation, respectively. For convenience, define \[
	\mathbb{P} := \{P_<, P_>, P_=\}
\] to be the set of all possible predicate symbols. Hereinafter, unless stated otherwise, we assume all formulas to be quantifier-free and to have at most $ n $ free variables. Whenever dealing with predicate sets, we assume them to be finite unless explicitly stated otherwise. Due to these assumptions about formulas and structures we are restricting our attention to, we can refer to valuations $ \nu $ as models. Then it makes sense to think of a set of models of $ \varphi $ as a subset of $ \Q^n $. We use the notation $ \ModelsOf(\varphi) $ to denote the set of models of a formula $ \varphi $.

We now define some further convenient notation which we will use when dealing with sets of predicates over the theory of $ \Mstruct $. For a set $ \Gamma $ of predicates and $ P \in \mathbb{P} $, let \[
	\Gamma^P := \{p \in \Gamma \mid \text{the predicate symbol of }p\text{ is }P\}
\] be the set of $ P $-predicates appearing in $ \Gamma $. Since in $ \Mstruct $ all predicate symbols are binary, we can also define the notion of a substitution on predicates: if $ p = P(t_1, t_2) $ is an instance of the $ P $ predicate for some $ P \in \mathbb{P} $ and terms $ t_1, t_2 $, then for every $ Q \in \mathbb{P} $ we define $ p^Q := Q(t_1, t_2) $. For a predicate set $ \Gamma $ and partition $ \Pi $, we also define \[
	\PiSimp{\Gamma} := \{p \in \Gamma \mid p \text{ respects } \Pi\}
\] and \[
	\PiComp{\Gamma} := \{p \in \Gamma \mid p \text{ disrespects } \Pi\}
\] to be the sets of $ \Pi $-respecting and $ \Pi $-disrespecting predicates appearing in $ \Gamma $, respectively. If $ R $ is a set of predicate sets, then for convenience, we define \[
	\Sat(R) := \{\Gamma \in R \mid \Gamma \text{ is satisfiable}\}
\] to be the subset of satisfiable predicate sets.

\section{Upper bounds}
\label{sec:vardec}

In this \zbsectionname/, we develop algorithms for deciding variable decomposability and for solving the variable decomposition problem, i.e., when we have to additionally output a $ \Pi $-decomposition in case the input formula is $ \Pi $-decomposable.

This \zbsectionname/ is structured as follows. We start by showing that the variable decomposition problem reduces to the same problem but with the partition being binary (Section~\ref{sec:vardec:reduction_to_bin_partitions}). Then, in Section~\ref{sec:vardec:normal_form}, we derive a simple and convenient variant of the disjunctive normal form, which we will need for our subsequent reasoning. In Section~\ref{sec:vardec:reduction} we define a new problem that we call \textit{the covering problem} and show that the variable decomposition problem reduces to it. Hence, it suffices to study the covering problem, which we do in the remaining sections. In Section~\ref{sec:vardec:pisimple_picomplex} we observe that is makes sense to split the instances of the covering problem into two types. In this respect, we introduce the notions of \textit{$ \Pi $-simple} and \textit{$ \Pi $-complex} predicate sets. In order to solve the covering problem, in Section~\ref{sec:vardec:analyzing_dep_between_vars} we introduce a method for analyzing dependencies between variables. Then, in Section~\ref{sec:vardec:high_level_overview}, we introduce the \textit{covering algorithm} that solves the covering problem by first giving an intuitive and high-level overview. Then we state the algorithm precisely and fill-in all missing details in Section~\ref{sec:vardec:cover}. After that, in Section~\ref{sec:vardec:correctness}, we prove that the covering algorithm is correct and derive a double-exponential time algorithm for the variable decomposition problem. In Section~\ref{sec:vardec:example} we give a concrete example showing how to solve the variable decomposition problem by reducing it to the covering problem and running the covering algorithm. After that, in Section~\ref{sec:vardec:exponential_upper_bound}, we prove that the variable decomposition problem can be solved in exponential time, by providing a fine-grained analysis of a slightly refined version of the covering algorithm. Then we turn to the construction of algorithms for deciding variable decomposability. We discuss the bird's eye view on the main steps of our approach in Section~\ref{sec:vardec:towards_conp_upper_bound}. Relying on results of previous sections, we show that it is possible to witness the fact that a formula is not $ \Pi $-decomposable in a sound and complete way (Section~\ref{sec:vardec:nondec_proof_system}); consequently, we obtain a $ \coNEXP $ algorithm for deciding variable decomposability. However, our first witness of non-$ \Pi $-decomposability admits only an encoding of exponential length, which is the performance bottleneck of the resulting complexity bound. Thus, in order to derive a faster algorithm for deciding variable decomposability, in Section~\ref{sec:vardec:proof_compression} we show that the witnesses (i.e., proofs) of non-$ \Pi $-decomposability can actually be compressed into a string of polynomial length. Relying on this, we derive the optimal $ \coNP $ algorithm for variable decomposability.

\subsection{Reduction to binary partitions}
\label{sec:vardec:reduction_to_bin_partitions}

Let $ \varphi $ be the given formula and $ \Pi $ be the specified partition. Until the end of \zbrefsec{sec:vardec}, we assume $ \Pi = \{X, Y\} $ to be binary and the free variables of $ \varphi $ to be $ x_1, \dots, x_{|X|}, y_1, \dots, y_{|Y|} $, where $ X = \{x_1, \dots, x_{|X|}\} $ and $ Y = \{y_1, \dots, y_{|Y|}\} $. We also assume that components of all vectors representing models correspond to the \[
	x_1, \dots, x_{|X|}, y_1, \dots, y_{|Y|}
\] variables in this order. For convenience, we define \[
	\vec{z} := (x_1, \dots, x_{|X|}, y_1, \dots, y_{|Y|})^\transp
\] to be the column vector consisting of the free variables of $ \varphi $ in the correct order.

By assuming that $ \Pi $ is binary, we do not lose generality due to the following Proposition~\ref{prop:reduction_to_binary_partitions}, which yields an efficient reduction from the variable decomposition problem for arbitrary partitions to the same problem for binary ones.

\begin{proposition}
	\label{prop:reduction_to_binary_partitions}
	For any (non-unary) partition $ \Pi $, there exists a collection $ S $ of binary partitions, such that any formula $ \varphi $ is $ \Pi $-decomposable if and only if $ \varphi $ is $ \Pi' $-decomposable for all $ \Pi' \in S $. Moreover, $ S $ is of size $ O(\log(|\Pi|)) $ and can be computed in polynomial time.
\end{proposition}

More precisely, the reduction works by computing, based on $ \Pi $, a collection $ S $ of binary partitions and then returning that $ \varphi $ is $ \Pi $-decomposable whenever $ \varphi $ is $ \Pi' $-decomposable for all $ \Pi' \in S $. For the proof of Proposition~\ref{prop:reduction_to_binary_partitions}, we will need the following result by Cosmadakis et al. \cite[Theorem 1]{cosmadakis:2001}.

\begin{proposition}[\cite{cosmadakis:2001}]
	\label{prop:pi1_pi2_dec_implies_dec_wrt_meet_of_pi1_pi2}
	Let $ \varphi $ be a formula, $ \Pi_1, \Pi_2 $ be partitions and $ \Pi := \Pi_1 \sqcap \Pi_2 $ be the meet of $ \Pi_1 $ and $ \Pi_2 $. Then $ \varphi $ is $ \Pi $-decomposable if $ \varphi $ is decomposable with respect to both $ \Pi_1 $ and $ \Pi_2 $.
\end{proposition}

For example, if $ \Pi_1 = \{\{x\}, \{y, z\}\} $, $ \Pi_2 = \{\{x, y\}, \{z\}\} $ and we have a formula $ \varphi(x, y, z) $ that is decomposable with respect to both $ \Pi_1 $ and $ \Pi_2 $, Proposition~\ref{prop:pi1_pi2_dec_implies_dec_wrt_meet_of_pi1_pi2} yields that $ \varphi $ is $ \Pi $-decomposable for \[
\Pi = \Pi_1 \sqcap \Pi_2 = \{\{x\}, \{y\}, \{z\}\}
\] In other words, it follows that $ \varphi $ is monadically decomposable. Overall, Proposition~\ref{prop:pi1_pi2_dec_implies_dec_wrt_meet_of_pi1_pi2} is a powerful and convenient tool allowing us to talk about decomposability with respect to refinements of partitions, without actually decomposing the formula with respect to that refinement. We now apply this tool to prove Proposition~\ref{prop:reduction_to_binary_partitions}.

Let $ k := \left|\Pi\right| $, $ m := \lceil\log(\left|\Pi\right|)\rceil $ and write $ \Pi = \{X_1, \dots, X_k\} $. The general idea of our reduction is to construct $ S $ in such a way that blocks of different partitions have as few common elements as possible. This ensures that the meet $ \sqcap_{\Pi' \in S} \Pi' $ approaches $ \Pi $ quickly as we keep adding partitions to $ S $, which in turn allows us to apply Proposition~\ref{prop:pi1_pi2_dec_implies_dec_wrt_meet_of_pi1_pi2} and obtain that decomposability with respect to partitions from $ S $ implies $ \Pi $-decomposability. Intuitively, we express the parallel binary search of every $ X_i \in \Pi $ in the language of binary partition sets, using the meet operation of the partition lattice. We now make these intuitions precise by defining a binary tree where the stated binary search actually happens. Our definition is recursive: let $ \Pi $ be the root; for every node $ N $ already in the tree and containing at least two elements (i.e., $ \left|N\right| \ge 2 $), we partition $ N $ into equally-sized\footnote{More precisely, we allow the sizes of the parts to differ by at most one, this ensures that we can handle the case when $ |N| $ is odd.} and disjoint parts $ L \cup R = N $ and set $ L $ and $ R $ to be the left and the right children of $ N $, respectively. This construction is visualized in Figure~\ref{fig:reduction_to_binary_partitions} for $ k = 8 $.

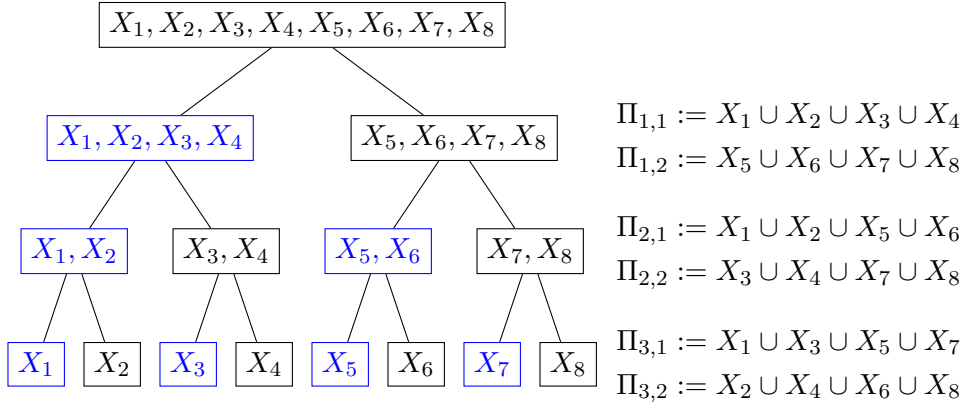
\begin{figure}[t]
	\centering
	\begin{tikzpicture}[
	every node/.style = {draw},
	grow = down,
	level distance=1.5cm,
	level 1/.style={sibling distance=4cm},
	level 2/.style={sibling distance=2cm},
	level 3/.style={sibling distance=1cm}]
	\node (root) {$ X_1, X_2, X_3, X_4, X_5, X_6, X_7, X_8 $}
	child {node [color=blue] {$ X_1, X_2, X_3, X_4 $}
		child {node [color=blue] {$ X_1, X_2 $}
			child {node[color=blue] {$ X_1 $}}
			child {node {$ X_2 $}}
		}
		child {node {$ X_3, X_4 $}
			child {node[color=blue] {$ X_3 $}}
			child {node {$ X_4 $}}
		}
	}
	child {node {$ X_5, X_6, X_7, X_8 $}
		child {node[color=blue] {$ X_5, X_6 $}
			child {node[color=blue] {$ X_5 $}}
			child {node {$ X_6 $}}
		}
		child {node {$ X_7, X_8 $}
			child {node[color=blue] {$ X_7 $}}
			child {node {$ X_8 $}}
		}
	};
	\begin{scope}[every node/.style={right}]
		\path (root-1  -| root-2-2-2) ++(5mm,0) node {$\begin{aligned}
				\Pi_{1,1} &:= X_1 \cup X_2 \cup X_3 \cup X_4\\
				\Pi_{1,2} &:= X_5 \cup X_6 \cup X_7 \cup X_8
			\end{aligned}$};
		\path (root-1-1-| root-2-2-2) ++(5mm,0) node {$\begin{aligned}
				\Pi_{2,1} &:= X_1 \cup X_2 \cup X_5 \cup X_6\\
				\Pi_{2,2} &:= X_3 \cup X_4 \cup X_7 \cup X_8
			\end{aligned}$};
		\path (root-1-1-1-| root-2-2-2) ++(5mm,0) node {$\begin{aligned}
				\Pi_{3,1} &:= X_1 \cup X_3 \cup X_5 \cup X_7\\
				\Pi_{3,2} &:= X_2 \cup X_4 \cup X_6 \cup X_8
			\end{aligned}$};
	\end{scope}
\end{tikzpicture}
	\caption{Example of the tree construction given in the proof of Proposition~\ref{prop:reduction_to_binary_partitions}, for $ k = 8 $. The set $ \Pi_{i,1} $ is visualized using blue, whereas $ \Pi_{i,2} $ is the union of the remaining $ X_j $ sets at the $ i $-th level of the tree.}
	\label{fig:reduction_to_binary_partitions}
\end{figure}

Observe that the height\footnote{In this context, the height is the length of the longest root-to-leaf path.} of this tree is $ m $. For every $ i \in \{1, \dots, m\} $, we define $ \Pi_{i,1} $ to be the union of all elements in those nodes, which are left children at the $ i $-th level\footnote{The root is at level zero.} of the tree. Similarly, let $ \Pi_{i,2} $ be the union of all elements in those nodes, which are right children at the $ i $-th level of the tree. In Figure~\ref{fig:reduction_to_binary_partitions}, $ \Pi_{i,1} $ is visualized using blue, whereas $ \Pi_{i,2} $ is composed of the remaining children at the $ i $-th level. Finally, we set \[
S := \{\{\Pi_{i,1}, \Pi_{i,2}\} \mid i \in \{1, \dots, m\}\}
\] to be the desired family of binary partitions. Note that $ \Pi_{i,j} $ is always a union of some subset of $ \Pi $, i.e., of some $ X_l $'s. We now prove the following claim, which essentially states that for every $ X $ there exists a way of intersecting the $ \Pi_{i,j} $ sets for distinct $ i $, such that we get $ X $. Since at every step of this intersection, half of the available blocks of $ \Pi $ get ruled out, we referred to this idea as ``binary search'' above.

\begin{zbclaim}
	\label{claim:x_in_pi_corresp_path_in_partition_tree}
	For every $ X \in \Pi $, there exist $ b_1, \dots, b_m \in \{1, 2\} $ such that \[
	X = \Pi_{1,b_1} \cap \dots \cap \Pi_{m,b_m}
	\]
\end{zbclaim}
\begin{proof}
	The existence of $ b_1, \dots, b_m $ satisfying $ X \subseteq \Pi_{1,b_1} \cap \dots \cap \Pi_{m,b_m} $ is clear by construction of the tree because we have ensured that $ X_1 \cup \dots \cup X_k = \Pi_{i,1} \cup \Pi_{i,2} $ holds for all $ i \in \{1, \dots, m\} $, so for every $ i $ it is possible to choose $ b_i $ such that $ X \subseteq \Pi_{i,b_i} $. Hence, it remains to show that $ X \supseteq \Pi_{1,b_1} \cap \dots \cap \Pi_{m,b_m} $. The $ b_1, \dots, b_m $ sequence can be thought of as a path in the tree, connecting its root $ \Pi $ with some leaf. We now fix $ b_1, \dots, b_m $ and argue by induction on $ i \in \{0, \dots, m\} $ that \begin{align}
		\label{eqn:pijbj_intersection_eq_n_union}
		\bigcap_{j=1}^i \Pi_{j,b_j} = \bigcup_{X \in N} X
	\end{align} where $ N $ is a unique node at the $ i $-th level of the tree. Intuitively, this means that $ \Pi_{1,b_1} \cap \dots \cap \Pi_{i,b_i} $ corresponds precisely to the node we reach after making $ i $ steps along the path defined by $ b_1, \dots, b_m $, starting at the root of the tree.
	
	\textbf{Base case}: For $ i = 0 $, the intersection is empty. Since all $ \Pi_{j,b_j} $ sets are subsets of $ \bigcup_{X \in \Pi} X $, we can treat the empty intersection as $ \bigcup_{X \in \Pi} X $ because this element is neutral.
	
	\textbf{Inductive step}: Let $ i \in \{1, \dots, m\} $ and $ N $ be the unique node at the $ (i-1) $-th level of the tree, existing by the induction hypothesis. If $ N $ is a leaf, there is nothing to prove. Otherwise, $ N $ has two children $ L $ and $ R $ such that $ \bigcup_{L' \in L} L' \subseteq \Pi_{i,1} $ and $ \bigcup_{R' \in R} R' \subseteq \Pi_{i,2} $. Let $ C \in \{L, R\} $ be the child such that $ \bigcup_{C' \in C} C' \subseteq \Pi_{i,b_i} $. Since $ C $ is non-empty and $ \Pi_{i,1} \cap \Pi_{i,2} = \varnothing $, it follows that $ C $ is unique. Moreover, by the induction hypothesis (``IH''), \begin{align*}
		\bigcap_{j=1}^i \Pi_{j,b_j} = \bigcap_{j=1}^{i-1} \Pi_{j,b_j} \cap \Pi_{i,b_i} \overset{\mathrm{IH}}{=} \Big(\bigcup_{X \in N} X\Big) \cap \Pi_{i,b_i} = \bigcup_{\substack{X \in N \\ X \subseteq \Pi_{i,b_i}}} X = \bigcup_{X \in C} X
	\end{align*} so we conclude that (\ref{eqn:pijbj_intersection_eq_n_union}) holds. Since the leaves of the tree correspond to elements of $ \Pi $, the claim follows.
\end{proof}

Since the blocks of $ \sqcap_{\Pi' \in S} \Pi' $ are non-empty intersections of the form \[
	\Pi_{1,b_1} \cap \dots \cap \Pi_{m,b_m}
\] for some $ b_1, \dots, b_m \in \{1, 2\} $, Claim~\ref{claim:x_in_pi_corresp_path_in_partition_tree} immediately implies \[
	\bigsqcap_{\Pi' \in S} \Pi' = \Pi
\] which in turn enables us to prove the correctness of the reduction:

\begin{zbclaim}
	\label{claim:phi_dec_iff_phi_dec_for_all_part_from_s}
	$ \varphi $ is $ \Pi $-decomposable if and only if $ \varphi $ is $ \Pi' $-decomposable for all $ \Pi' \in S $.
\end{zbclaim}
\begin{proof}
	The ``$ \Rightarrow $'' direction holds because $ \Pi $ is a refinement of every $ \Pi' \in S $, meaning that every $ \Pi $-decomposition is in particular a $ \Pi' $-decomposition. The converse implication follows from Proposition~\ref{prop:pi1_pi2_dec_implies_dec_wrt_meet_of_pi1_pi2} applied to partitions from $ S $, because $ \sqcap_{\Pi' \in S} \Pi' = \Pi $.
\end{proof}

Finally, observe that it is clearly possible to compute $ S $ in polynomial time. This completes the proof of Proposition~\ref{prop:reduction_to_binary_partitions}.

\begin{zbexample}
	We give a concrete example of how Proposition~\ref{prop:reduction_to_binary_partitions} can be applied to reduce the variable decomposition problem for arbitrary partitions to the same problem for binary ones.
	Let \[
		\Pi = \{\{x_1, x_2\}, \{x_3\}, \{x_4\}, \{x_5, x_6\}, \{x_7, x_8, x_9\}, \{x_{10}\}\}
	\] In order to achieve the tree construction given in the proof of Proposition~\ref{prop:reduction_to_binary_partitions}, we first assign indices to elements of $ \Pi $ as indicated in the following table: \begin{center}
		\begin{tabular}{|c||c|c|c|c|c|c|}
			\hline
			$ X \in \Pi $ & $ \{x_1, x_2\} $ & $ \{x_3\} $ & $ \{x_4\} $ & $ \{x_5, x_6\} $ & $ \{x_7, x_8, x_9\} $ & $ \{x_{10}\} $ \\
			\hline
			Index & $ 000 $ & $ 001 $ & $ 010 $ & $ 011 $ & $ 100 $ & $ 101 $ \\ \hline
		\end{tabular}
	\end{center} Note that we represent the indices in binary. Next, we look at the $ i $-th bit and set $ \Pi_{i,1} $ to be the union of those sets whose indices have $ i $-th bit equal zero. Similarly, $ \Pi_{i,2} $ is the union of those $ X \in \Pi $ which were assigned an index having a one at the $ i $-th position. That is, we obtain the collection \begin{align*}
	S = \{&\{\{x_1, x_2, x_3, x_4, x_5, x_6\}, \{x_7, x_8, x_9, x_{10}\}\},\\
		&\{\{x_1, x_2, x_3, x_7, x_8, x_9, x_{10}\}, \{x_4, x_5, x_6\}\},\\
		&\{\{x_1, x_2, x_4, x_7, x_8, x_9\}, \{x_3, x_5, x_6, x_{10}\}\}\}
	\end{align*} of binary partitions. By Proposition~\ref{prop:reduction_to_binary_partitions}, testing the $ \Pi $-decomposability of a formula $ \varphi $ is equivalent to determining whether $ \varphi $ is $ \Pi' $-decomposable for all $ \Pi' \in S $.
\end{zbexample}

\subsection{Normal form}
\label{sec:vardec:normal_form}

We now describe a simple variant of the disjunctive normal form (DNF) which we use to represent the given formula $ \varphi $. Let $ p_1, \dots, p_k $ be the predicates appearing in $ \varphi $ and define \[
	\DisjTrue := \{\{p_1^{P_1}, \dots, p_k^{P_k}\} \mid (P_1, \dots, P_k) \in \{P_<, P_>, P_=\}^k\}
\] We refer to elements of $ \DisjTrue $ as disjuncts. Clearly, $ \top \equiv \bigvee_{\Gamma \in \DisjTrue} \Gamma $. Without loss of generality, we assume the input formula to be of the form $ \varphi = \bigvee_{\Gamma \in \DisjPhi} \Gamma $ for some set $ \DisjPhi \subseteq \DisjTrue $. Define furthermore $ \DisjNegPhi := \DisjTrue \setminus \DisjPhi $. This normal form (unlike standard DNF) has the property that any two distinct disjuncts cannot agree on any model. Using this fact, we show the following property of our normal form, which we will need in subsequent proofs.

\begin{lemma}
	\label{lemma:every_disj_either_true_or_false}
	For $ \varphi \in \QFLRA $, every satisfiable disjunct $ \Theta \in \Sat(\DisjTrue) $ entails either $ \varphi $ or $ \neg\varphi $.
\end{lemma}
\begin{proof}
	Clearly, $ \Theta \wedge \varphi $ or $ \Theta \wedge \neg\varphi $ must be satisfiable. Assume first the existence of $ v \models \Theta \wedge \varphi $. Then \[
		v \models \varphi \equiv \bigvee_{\Gamma \in \DisjPhi} \Gamma
	\] and hence there must exist $ \Gamma \in \DisjPhi $ such that $ v \models \Gamma \models \varphi $. Since $ v \models \Theta $, it follows that $ \Gamma \wedge \Theta $ is satisfiable. By the definition of the normal form, no two disjuncts can agree on any model unless they are equal, so it follows that $ \Gamma = \Theta $. Hence, $ \Theta \models \varphi $. The remaining case, when $ \Theta \wedge \neg\varphi $ is satisfiable, analogously leads to the conclusion that $ \Theta \models \neg\varphi $.
\end{proof}

The introduced normal form also has the property that, for every $ \Gamma \in \DisjTrue $, changing predicate symbols of predicates appearing in $ \Gamma $ yields another disjunct in $ \DisjTrue $. Intuitively, we can think of this as the ability to travel between elements of $ \DisjTrue $ by changing predicate symbols of predicates. Later, we will also need some notation to be able to address the set of those disjuncts from $ \DisjTrue $ which agree on some model with a given predicate set $ \Gamma $. We capture this in the following definition.

\begin{zbdefinition}[Disjuncts of a predicate set]
	For a formula $ \varphi \in \QFLRA $ and a set of predicates $ \Gamma $, we call \[
	\DisjOf{\Gamma} := \Sat(\{\Gamma \cup \Omega \mid \Omega \in \DisjTrue\})
	\] the set of disjuncts of $ \Gamma $.
\end{zbdefinition}

\subsection{Reduction to a covering problem}
\label{sec:vardec:reduction}

Intuitively, it turns out that studying the variable decomposition problem reduces itself to the study of formulas defining sets that entail either $ \varphi $ or $ \neg\varphi $, assuming $ \varphi $ is $ \Pi $-decomposable. In other words, we want to study formulas $ \psi $ such that the $ \Pi $-decomposability of $ \varphi $ causes the existence of just a single model $ v \models \psi \wedge \varphi $ to make all models of $ \psi $ satisfy $ \varphi $. We call this effect ``model flooding'' and capture it precisely in the following definition.

\begin{zbdefinition}[$ (\varphi, \Pi) $-MFF]
	Let $ \varphi, \psi \in \QFLRA $ be formulas. We say that $ \psi $ is a $ (\varphi, \Pi) $-model-flooding formula ($ (\varphi, \Pi) $-MFF) whenever the implication \[
	\varphi\text{ is }\Pi\text{-decomposable} \Rightarrow (\psi \models \varphi \vee \psi \models \neg\varphi)
	\] holds.
\end{zbdefinition}

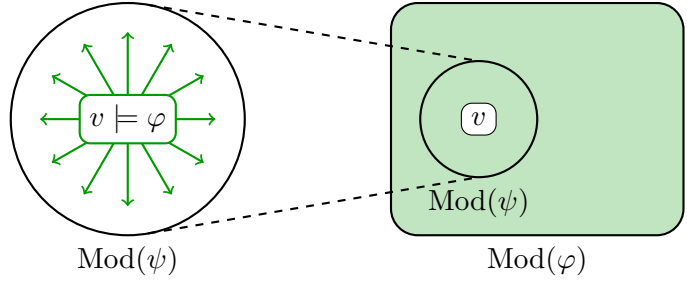
\begin{figure}[t]
	\centering
	\begin{tikzpicture}
	\begin{scope}[shift={(-4em,0)}]
		
		\coordinate (vmodphipos) at (0, 0);
		
		\node[circle, minimum size=8em, thick, draw=black, label={below:$ \ModelsOf(\psi) $}] at (vmodphipos) {};
		
		\coordinate (lefttop) at (80:4em);
		\coordinate (leftbottom) at (280:4em); 
		
		\draw[->, thick, color=OliveGreen] (vmodphipos) -- (0:3em);
		\draw[->, thick, color=OliveGreen] (vmodphipos) -- (30:3em);
		\draw[->, thick, color=OliveGreen] (vmodphipos) -- (60:3em);
		\draw[->, thick, color=OliveGreen] (vmodphipos) -- (90:3em);
		\draw[->, thick, color=OliveGreen] (vmodphipos) -- (120:3em);
		\draw[->, thick, color=OliveGreen] (vmodphipos) -- (150:3em);
		\draw[->, thick, color=OliveGreen] (vmodphipos) -- (180:3em);
		\draw[->, thick, color=OliveGreen] (vmodphipos) -- (210:3em);
		\draw[->, thick, color=OliveGreen] (vmodphipos) -- (240:3em);
		\draw[->, thick, color=OliveGreen] (vmodphipos) -- (270:3em);
		\draw[->, thick, color=OliveGreen] (vmodphipos) -- (300:3em);
		\draw[->, thick, color=OliveGreen] (vmodphipos) -- (330:3em);
		
		\draw (vmodphipos) node [rounded corners, thick, draw=OliveGreen, fill=white] {$ v \models \varphi $};
	\end{scope}
	
	\begin{scope}[shift={(8em,0)}]
		\filldraw[thick, rounded corners=10, fill=OliveGreen, fill opacity=0.25] (-3em, -4em) -- (-3em, 4em) -- (7em, 4em) -- (7em, -4em) -- cycle;
		
		\coordinate (vpos) at (0, 0);
		
		\node[circle, minimum size=4em, thick, draw=black, label={below:$ \ModelsOf(\psi) $}] at (vpos) {};
		
		\draw (vpos) node [rounded corners, draw=black, fill=white] {$ v $};
		
		\node[anchor=north,thick] at (2em, -4em) {$ \ModelsOf(\varphi) $};
	\end{scope}
	
	\coordinate (righttop) at ($ (vpos) + (0, 2em) $);
	\coordinate (rightbottom) at ($ (vpos) - (0, 2em) $);

	\draw[thick, dashed] (lefttop) -- (righttop);
	\draw[thick, dashed] (leftbottom) -- (rightbottom);
\end{tikzpicture}
	\caption{Venn diagram illustrating the behavior captured by the definition of a $ (\varphi, \Pi) $-model-flooding formula, where $ \varphi $ is $ \Pi $-decomposable. The green arrows indicate that the single model $ v \models \varphi \wedge \psi $ implies that all models of $ \psi $ must also satisfy $ \varphi $.}
	\label{fig:mff_intuition}
\end{figure} A set-theoretic way of thinking about $ (\varphi, \Pi) $-model-flooding formulas is visualized in Figure~\ref{fig:mff_intuition}. Observe that $ \psi $ is a $ (\varphi, \Pi) $-MFF if and only if $ \psi $ is a $ (\neg\varphi, \Pi) $-MFF. Building on top of this definition, we now define the central problem we will focus on in the remainder of \zbrefsec{sec:vardec} because the variable decomposition problem reduces to it.

\begin{zbproblem}[Covering problem]
	\label{problem:cover}
	Given a formula $ \varphi \in \QFLRA $, a partition $ \Pi $ and a predicate set $ \Gamma \in \DisjTrue $, compute a $ \Pi $-respecting $ (\varphi, \Pi) $-MFF $ \psi $ such that $ \Gamma \models \psi $.
\end{zbproblem}

We will later show that, for any formula $ \varphi $ and $ \Gamma \in \DisjTrue $, there indeed exists a solution $ \psi $ to the covering problem, which we call \textit{the covering} of $ \Gamma $. But first, we establish a fundamental connection between the variable decomposition and the covering problem.

\begin{theorem}
	\label{thm:reduction}
	Fix a formula $ \varphi \in \QFLRA $ and an oracle solving the covering problem. Let $ \psi_\Gamma $ be the solution to the covering problem for $ \Gamma $, produced by the oracle. Then $ \varphi $ is $ \Pi $-decomposable if and only if $ \psi_\Gamma \models \varphi $ for all $ \Gamma \in \Sat(\DisjPhi) $.
\end{theorem}
\begin{proof}
	``$ \Rightarrow $'': Since $ \Sat(\DisjPhi) \ni \Gamma \models \psi_\Gamma $, it follows that $ \psi_\Gamma \wedge \varphi $ is satisfiable for all $ \Gamma $. Hence, by definition of $ (\varphi, \Pi) $-MFF it follows that $ \psi_\Gamma \models \varphi $ for all $ \Gamma $.
	
	``$ \Leftarrow $'': The formula $ \varphi $ is $ \Pi $-decomposable because \[
		\varphi \equiv \bigvee_{\Gamma \in \Sat(\DisjPhi)} \psi_\Gamma
	\] More precisely, the $ \models $ entailment follows from the fact that $ \Gamma \models \psi_\Gamma $, whereas the converse entailment holds by the assumption that $ \psi_\Gamma \models \varphi $.
\end{proof}

Theorem~\ref{thm:reduction} gives us an alternative and novel characterization of $ \Pi $-decomposable formulas, while its proof yields the following reduction from the variable decomposition problem to the covering problem.

\begin{itemize}
	\item For each $ \Gamma \in \Sat(\DisjPhi) $, compute the covering $ \psi_\Gamma $ and check whether $ \psi_\Gamma \models \varphi $.
	\item If the entailment does not hold for some $ \psi_\Gamma $, output that the formula is not $ \Pi $-decomposable and terminate. Otherwise, $ \varphi $ is $ \Pi $-decomposable, and a possible $ \Pi $-decomposition is $ \bigvee_{\Gamma \in \Sat(\DisjPhi)} \psi_\Gamma $.
\end{itemize}

We now give a simple example illustrating how this reduction can be used to solve the variable decomposition problem. Our goal is also to give some important insight into how the covering formulas can be constructed.

\begin{zbexample}
	\label{example:reduction_ex}
	Consider $ \Pi := \{\{x\}, \{y\}\} $ and \[
		\varphi := x + y \neq 2 \vee x - y \neq 0 \in \QFLRA
	\] Bringing $ \varphi $ to the assumed normal form yields \begin{align*}
	\DisjPhi = \{\{x + y < 2, x - y < 0\}, \{x + y < 2, x - y > 0\}, \\
	\{x + y > 2, x - y < 0\}, \{x + y > 2, x - y > 0\}\}
	\end{align*} We first compute a covering for $ \Gamma := \{x + y < 2, x - y < 0\} \in \Sat(\DisjPhi) $. Since at this point we have not yet developed the necessary theory needed to prove the correctness of the covering, we first introduce a geometric way of thinking about coverings, rely on it to analyze the present formula and defer the rigorous proof to a later theorem. The set of models $ \ModelsOf(\Gamma) $ can be thought of as being the intersection of two half-spaces visualized in Figure~\ref{fig:reduction_ex:geom}. \begin{figure}[t]
	\centering
	\hfill\begin{subfigure}{0.3\textwidth}
	\centering
	\begin{tikzpicture}[scale=0.68]
	\node[blue,anchor=west,thick] at (1.5, 1.5) {$x+y<2$};
	\node[red,anchor=west,thick] at (1.5, 0.5) {$x-y<0$};
	
	
	\coordinate (topleft) at (-2, 3);
	\coordinate (topright) at (3, 3);
	\coordinate (botleft) at (-2, -1);
	\coordinate (botright) at (3, -1);
	
	\draw[->] (-2,0) -- (3,0) coordinate[label = {below:$x$}] (xmax);
	\draw[->] (0,-1) -- (0,3.2) coordinate[label = {right:$y$}] (ymax);
	
	\coordinate (redtop) at (-1, 3);
	\coordinate (redbot) at (3, -1);
	
	\draw[thick, dashed, red] (redtop) -- (redbot);
	
	\fill[fill=red,opacity=0.2] (topleft) -- (redtop) -- (redbot) -- (botleft) -- cycle;
	
	\coordinate (bluebot) at (-1, -1);
	\coordinate (bluetop) at (3, 3);
	
	\draw[thick, dashed, blue] (bluebot) -- (bluetop);
	\fill[fill=blue,opacity=0.2] (topleft) -- (bluetop) -- (bluebot) -- (botleft) -- cycle;
	
	\node[fill=white, rounded corners] (gamma) at (-1, 1) {$ \Gamma $};
	\end{tikzpicture}
	\caption{\label{fig:reduction_ex:geom}}
\end{subfigure} \hfill
\begin{subfigure}{0.3\textwidth}
	\centering
	\begin{tikzpicture}[scale=0.75]
	\coordinate (redtop) at (-1, 3);
	\coordinate (redbot) at (3, -1);
	
	\coordinate (bluebot) at (-1, -1);
	\coordinate (bluetop) at (3, 3);
	
	\coordinate (intpoint) at (1, 1);
	
	\draw[thick, dashed] (redtop) -- (redbot); 
	\draw[thick, dashed] (bluetop) -- (bluebot);
	
	
	\filldraw[ultra thin,red,pattern=north east lines,pattern color=red, draw opacity=0.7] (redtop) -- (intpoint) -- (bluebot) -- cycle;
	\node[fill=white, rounded corners] (gamma) at (-0.5,1) {$ \Gamma $};
	
	\draw[->, ultra thick] (-0.5, 1.5) .. controls +(up:7mm) and +(left:7mm) .. (0.5, 2.5);
	\draw[->, ultra thick] (-0.5, 0.5) .. controls +(down:7mm) and +(left:7mm) .. (0.5, -0.5);
	\draw[->, ultra thick] (1.5, 2.5) .. controls +(right:7mm) and +(up:7mm) .. (2.5, 1.5);
	\draw[->, ultra thick] (1.5, -0.5) .. controls +(right:7mm) and +(down:7mm) .. (2.5, 0.5);
	
	\end{tikzpicture}
	\caption{\label{fig:reduction_ex:gamma}}
\end{subfigure} \hfill
\begin{subfigure}{0.3\textwidth}
	\centering
	\begin{tikzpicture}[scale=0.75]
	\coordinate (redtop) at (-1, 3);
	\coordinate (redbot) at (3, -1);
	
	\coordinate (bluebot) at (-1, -1);
	\coordinate (bluetop) at (3, 3);
	
	\coordinate (intpoint) at (1, 1);
	\coordinate (mfdltop) at (-1, 3);
	\coordinate (mfdlbot) at (-1, -1);
	\coordinate (mfdrtop) at (3, 3);
	\coordinate (mfdrbot) at (3, -1);
	
	\draw[thick, dashed] (redtop) -- (redbot); 
	\draw[thick, dashed] (bluetop) -- (bluebot);
	
	\filldraw[ultra thin,red,pattern=north east lines,pattern color=red, draw opacity=0.7] (mfdltop) -- (mfdrtop) -- (mfdrbot) -- (mfdlbot) -- cycle;
	
	\filldraw[red] (intpoint) circle (0.15);
	\filldraw[blue] (intpoint) circle (0.13);
	
	\node[fill=white, rounded corners] at (-0.25,1) {$ \psi_\Gamma $};
	\end{tikzpicture}
	\caption{\label{fig:reduction_ex:covering}}
\end{subfigure}\hfill
	\caption{Figure~\ref{fig:reduction_ex:geom} is a visualization of the sets defined by $ x + y < 2 $ and $ x - y < 0 $. The set of models of $ \Gamma := \{x + y < 2, x - y < 0\} $ is visualized in Figure~\ref{fig:reduction_ex:gamma}. The arrows illustrate the model flooding $ \Gamma $ causes. Figure~\ref{fig:reduction_ex:covering} visualizes a possible correct covering $ \psi_\Gamma = x \neq 1 \vee y \neq 1 $ for $ \Gamma $, which excludes the blue point.}
	\label{fig:reduction_ex}
\end{figure} At this point, it is worth noting that, in general, a conjunction of LRA predicates defines an \textit{evenly convex polyhedral set} \cite{klee:2007:econvexsetsbasics, rodriguez:2017}; thus, the geometry we are dealing with in the present example is trivial compared to the general case. Our goal is to compute a covering $ \psi_\Gamma $ satisfying the properties listed in the following table, which also explains the way one can think about each of the properties geometrically.
	
	\begin{center}
		\begin{tabular}{|P{0.3\textwidth}|P{0.65\textwidth}|}
			\hline
			Property of $ \psi_\Gamma $ & Geometric intuition \\
			\hline\hline
			$ \psi_\Gamma $ respects $ \Pi $ & $ \ModelsOf(\psi_\Gamma) $ is a finite union of (closed or open) rectangles such that every edge is aligned along either the $ x $ or the $ y $ axis \\
			\hline
			$ \Gamma \models \psi_\Gamma $ & $ \psi_\Gamma $ defines a shape containing $ \ModelsOf(\Gamma) $ \\
			\hline
			$ \psi_\Gamma $ is a $ (\varphi, \Pi) $-MFF & If $ \varphi $ is $ \Pi $-decomposable, then the entire shape defined by $ \psi_\Gamma $ must be contained in $ \ModelsOf(\varphi) $ \\ \hline
		\end{tabular}
	\end{center}

	Clearly, $ \psi_\Gamma = \top $ satisfies the first two conditions. However, $ \top $ is not a $ (\varphi, \Pi) $-MFF because $ \varphi $ is $ \Pi $-decomposable but is not a tautology. The general idea is to start with $ \top $ and iteratively transform the current formula in a way that excludes from its set of models those ``parts'' that need not be included in $ \ModelsOf(\varphi) $ for $ \varphi $ to be $ \Pi $-decomposable. In other words, in order to obtain a correct covering $ \psi_\Gamma $, it is crucial to ensure that, if $ \psi_\Gamma $ agrees with some $ \Lambda \in \DisjTrue $ on a model, then $ \Gamma \models \varphi $ indeed implies $ \Lambda \models \varphi $ under the assumption that $ \varphi $ is $ \Pi $-decomposable. Whenever this implication \[
	\varphi\text{ is }\Pi\text{-decomposable} \wedge \Gamma \models \varphi \implies \Lambda \models \varphi
	\] holds, we say that $ \Gamma $ \textit{causes model flooding into} $ \Lambda $. Note that, by definition and Lemma~\ref{lemma:every_disj_either_true_or_false}, if $ \psi_\Gamma $ is a $ (\varphi, \Pi) $-MFF such that $ \Gamma \models \psi_\Gamma $, then $ \Gamma $ causes model flooding into all $ \Lambda \in \DisjTrue $ agreeing on some model with $ \psi_\Gamma $.
	
	Turning to the example at hand, we can geometrically observe that $ \Gamma $ causes model flooding into every $ \Lambda \in \Sat(\DisjTrue) \setminus \{\{x + y = 2, x - y = 0\}\} $, due to the following intuitive reasons. First, $ \Gamma $ causes model flooding into $ \{x + y < 2, x - y = 0\} $ and $ \{x + y = 2, x - y < 0\} $ (see Figures \ref{fig:reduction_ex:geom} and \ref{fig:reduction_ex:gamma}) because the lines defined by $ x + y = 2 $ and $ x - y = 0 $ are not aligned along the $ x $ or $ y $ axis, meaning that it is impossible to ``distinguish'' the mentioned predicate sets from $ \Gamma $ using a $ \Pi $-respecting formula, as any such formula can only define a set which is a finite union of rectangles with edges aligned either along the $ x $ axis, or along the $ y $ axis. Then, by the same argumentation, the above predicate sets cause model flooding into $ \{x + y < 2, x - y > 0\} $ and $ \{x + y > 2, x - y < 0\} $, respectively. This model flooding effect is visualized using the two black arrows in the left half of Figure~\ref{fig:reduction_ex:gamma}. By the same reasoning, model flooding continues in the right part of the plot (see the two arrows in the right half of Figure~\ref{fig:reduction_ex:gamma}). Hence, \[
		\{x + y = 2, x - y = 0\} \in \DisjTrue
	\] is the only disjunct which $ \Gamma $ does not cause model flooding into. This disjunct corresponds to the point $ (1, 1) $. Note that the $ \Pi $-decomposability of $ \varphi $ does not depend on whether $ \varphi $ evaluates to true or false under this point. In this case $ (1, 1)^\transp \not\models \varphi $, so it is particularly important to exclude this point from $ \top $, which we can do by simply writing \[
		\psi_\Gamma := x \neq 1 \vee y \neq 1
	\] Due to the simplicity of the present example, no further steps are required, and this formula is a correct covering for $ \Gamma $ (see Figure~\ref{fig:reduction_ex:covering}). Moreover, we do not need to cover disjuncts of $ \varphi $ other than $ \Gamma $ because every such disjunct already entails $ \psi_\Gamma $. Since $ \psi_\Gamma \models \varphi $, by Theorem~\ref{thm:reduction} we conclude that $ \varphi $ is $ \Pi $-decomposable and a possible $ \Pi $-decomposition is $ \psi_\Gamma = (x \neq 1 \vee y \neq 1) $.
\end{zbexample}

To sum up, Theorem~\ref{thm:reduction} gives a method for reducing the variable decomposition problem to the covering problem. Thus, in order to obtain algorithms for the former, in the remainder of \zbrefsec{sec:vardec} we mostly study the latter.

\subsection{$ \Pi $-simple and $ \Pi $-complex predicate sets}
\label{sec:vardec:pisimple_picomplex}

The easiest way to produce a correct covering $ \psi_\Gamma $ of $ \Gamma $ is to simply come up with a $ \Pi $-respecting formula defining precisely $ \Gamma $. Clearly, this is not always possible, but whenever it is, it is certainly advantageous to produce such a covering $ \psi_\Gamma \equiv \Gamma $ because in this case we do not have to perform any analysis of other disjuncts, which, as we will see later, is a nontrivial and costly operation. This motivates the following definition of a $ \Pi $-simple predicate set, which can be thought of as an under-approximation of the set of $ \Pi $-decomposable predicate sets.

\begin{zbdefinition}[$ \Pi $-simple, $ \Pi $-complex]
	Let $ \Gamma $ be a predicate set and $ \Pi = \{X, Y\} $ be a partition of $ \{x_1, \dots, x_n\} $. We say that \begin{itemize}
		\item $ \Gamma $ fixes $ x_i $ if $ \Gamma $ is satisfiable and all models of $ \Gamma $ agree on $ x_i $, that is, assign $ x_i $ the same value;
		\item $ \Gamma $ fixes a set $ Z \subseteq \{x_1, \dots, x_n\} $ of variables if $ \Gamma $ fixes every $ x_i \in Z $;
		\item $ \Gamma $ is $ \Pi $-simple if $ \Gamma $ is unsatisfiable or if $ \Gamma $ fixes $ X $ or $ Y $;
		\item $ \Gamma $ is $ \Pi $-complex if $ \Gamma $ is not $ \Pi $-simple, that is, if $ \Gamma $ is satisfiable but fixes neither $ X $ nor $ Y $;
	\end{itemize} and let $ \Fixes(\Gamma) $ denote the set of all variables fixed by $ \Gamma $, if $ \Gamma $ is satisfiable, and $ \Fixes(\Gamma) := \varnothing $ otherwise.
\end{zbdefinition}

We now argue that, as motivated above, $ \Gamma $ being $ \Pi $-simple indeed implies that $ \Gamma $ is $ \Pi $-decomposable. Moreover, it is possible to efficiently test whether $ \Gamma $ is $ \Pi $-simple or not. First check whether $ \Gamma $ is satisfiable and conclude that $ \Gamma $ is $ \Pi $-simple if not. Otherwise, let $ \Mstruct_\nu $ be an arbitrary model of $ \Gamma $. Each variable $ x $ can be identified as fixed by testing whether $ \Gamma \land x \neq \nu(x) $ is unsatisfiable, so repeating this analysis for every variable gives us a method of deciding whether $ \Gamma $ is $ \Pi $-simple or not.
We now argue that this method has a polynomial running time. Clearly, checking unsatisfiability of $ \Gamma \land x \neq \nu(x) $ is equivalent to testing whether $ \Gamma \land x < \nu(x) $ and $ \Gamma \land x > \nu(x) $ are both unsatisfiable. These conjunctions are linear programs, for which there are well-known polynomial-time algorithms based on interior-point and ellipsoid methods (see Karmarkar's algorithm \cite[Chapter 15]{karmarkar:1984, schrijver:1998:theoryofintegerandlinearprogramming} and Khachiyan's algorithm \cite[Theorem 4.18]{khachiyan:1979, korte:2012:combinatorialoptimization}). More precisely, these results talk about linear programs where only non-strict inequalities are allowed, but the stated algorithms can be adapted to support strict inequalities via, for example, the reduction given in \cite[pp. 217, 218]{bradley:2007}.

If $ \Gamma $ is $ \Pi $-simple and fixes $ X \in \Pi $ or $ Y \in \Pi $, then it is possible to compute a $ \Pi $-respecting formula $ \psi $ by substituting $ \nu(x) $ into $ x $ and adding the $ x = \nu(x) $ predicate, for every variable $ x $ fixed by $ \Gamma $. Clearly, this transformation preserves equivalence. For a $ \Pi $-simple predicate set $ \Gamma $, we define $ \decsimple(\Pi, \Gamma) $ to be the $ \Pi $-decomposition of $ \Gamma $ obtained as just described.

It follows that it suffices to study the covering problem only for the case when $ \Gamma $ is $ \Pi $-complex. Accordingly, we hereinafter assume that, unless stated otherwise, $ \Gamma $ is $ \Pi $-complex.

\subsection{Analyzing dependencies between variables}
\label{sec:vardec:analyzing_dep_between_vars}

In order to construct provably correct coverings in general, we need to study under what conditions one predicate set causes model flooding into another one. For this, we need to precisely understand what ``connections'' between variables can potentially be established by some conjunction of predicates. More precisely, the idea is to define a metric which allows us to compare arbitrary satisfiable predicate sets, say, $ \Gamma' $ and $ \Gamma'' $. This metric must have the property that $ \Gamma' $ being equivalent to $ \Gamma'' $ with respect to this metric implies that $ \Gamma' $ causes model flooding into $ \Gamma'' $ under certain additional assumptions.


\subsubsection{Linear dependencies}

It turns out that over quantifier-free linear real arithmetic it suffices to analyze only those ``connections'' between variables which are established by equality predicates. We call these connections linear dependencies because they can be enforced by linear constraints and characterized using vectorspaces.

\begin{zbdefinition}
	\label{def:lindep}
	Let $ \Gamma $ be a satisfiable set of predicates, $ V $ be a $ \Q $-vectorspace and $ h : \Q^n \rightarrow V $ be a mapping. For every $ v \in \Q^n $ we define \[
	\LinDep_{h, v}(\Gamma) := \{U \le V \mid h(\ModelsOf(\Gamma) - v) \subseteq U^{\bot}\}
	\] where $ U^{\bot} $ is the orthogonal complement of $ U $ in $ V $.
\end{zbdefinition}

The intuition behind this definition is as follows. Think of $ \ModelsOf(\Gamma) $ as a subset of $ \Q^n $. Our goal is to analyze the linear dependencies ``contained'' in $ \Gamma $, which we do with respect to some mapping $ h $ and vector $ v $. The mapping is needed to be able to apply a transformation to the set before analyzing its linear dependencies, whereas $ v $ can be thought of as an offset, by which we translate $ \ModelsOf(\Gamma) $. 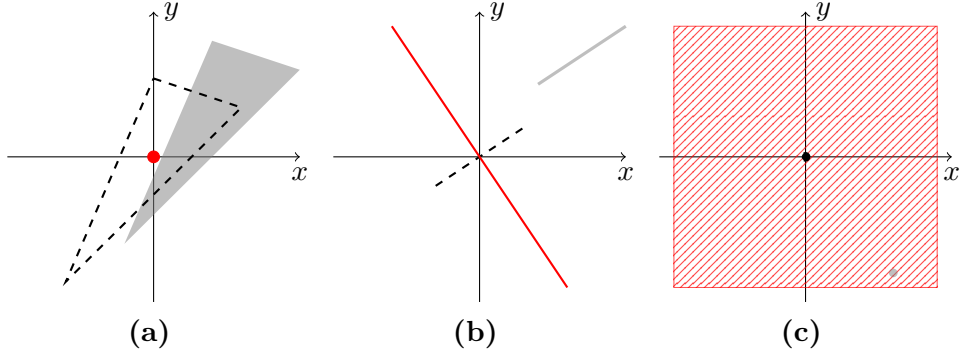
\begin{figure}[t]
	\centering
	\newcommand{\drawaxis}{
  \draw[->] (-5em,0) -- (5em,0) coordinate[label = {below:$x$}] (xmax);
  \draw[->] (0,-5em) -- (0,5em) coordinate[label = {right:$y$}] (ymax);
}

\hfill\begin{subfigure}[b]{0.3\textwidth}
\centering
\begin{tikzpicture}
	\drawaxis
	
	\fill[shapebeforeshift]
	(2em,4em) -- (5em,3em) -- (-1em, -3em) -- cycle;
	\draw[shapeaftershift,transform canvas={xshift=-2em,yshift=-1.3em}]
	(2em,4em) -- (5em,3em) -- (-1em, -3em) -- cycle;
	
	\filldraw[red] (0, 0) circle (0.2em);
\end{tikzpicture}
\caption{~}
\label{fig:lindep_def:a}
\end{subfigure} \hfill
\begin{subfigure}[b]{0.3\textwidth}
\centering
\begin{tikzpicture}
	\drawaxis
	
	\draw[very thick, shapebeforeshift]
		(2em,2.5em) -- (5em,4.5em);
	\draw[shapeaftershift, transform canvas={xshift=-3.5em,yshift=-2em}]
		(2em,1em) -- (5em,3em);
	\draw[thick, red] (3em, -4.5em) -- (-3em, 4.5em);
\end{tikzpicture}

\caption{~}
\label{fig:lindep_def:b}
\end{subfigure} \hfill
\begin{subfigure}[b]{0.3\textwidth}
\centering
\begin{tikzpicture}
	\drawaxis
	
	\filldraw[ultra thin,red,pattern=north east lines,pattern color=red, draw opacity=0.7] 
	(-4.5em, -4.5em) -- (-4.5em, 4.5em) -- (4.5em, 4.5em) -- (4.5em, -4.5em) -- cycle;
	
	\node[shapebeforeshift, draw,circle,inner sep=1,fill] at (3em,-4em) {};
	\node[shapeaftershift, draw, circle,inner sep=1,fill] at (0,0) {};
\end{tikzpicture}

\caption{~}
\label{fig:lindep_def:c}
\end{subfigure}\hfill
	\caption{Three examples showing the intuition behind the way we capture linear dependencies contained in $ \Gamma $, with respect to $ h = \id $ and $ v \models \Gamma $, when we have only two dimensions. The gray shapes represent $ \ModelsOf(\Gamma) $, while the dashed versions depict the translation $ \ModelsOf(\Gamma) - v $. The red regions visualize the largest vectorspace in $ \LinDep_{\id, v}(\Gamma) $ with respect to inclusion.}
	\label{fig:lindep_def}
\end{figure} In Figure~\ref{fig:lindep_def} we give three examples illustrating the way we capture linear dependencies with respect to $ h := \id $ and $ v \models \Gamma $. The only vectorspace containing $ \ModelsOf(\Gamma) - v $ (i.e., the dashed shape) of Figure~\ref{fig:lindep_def:a} is $ \Q^2 $, so \[
	\LinDep_{\id, v}(\Gamma) = \{\gen{0}\}
\] Intuitively, this means that the triangle in 2D does not establish any linear dependencies between variables. More precisely, we can think of $ \gen{0} \in \LinDep_{\id, v}(\Gamma) $ as an element corresponding to a predicate entailed by $ \Gamma $, which is in this case a trivial tautology predicate. Indeed, it is impossible to come up with a nontrivial equality predicate $ p \not\equiv \top $ such that $ \Gamma \models p $. Intuitively, the orthogonal complement in Definition~\ref{def:lindep} is needed to obtain this correspondence between predicates entailed by $ \Gamma $ and elements of $ \LinDep_{\id, v}(\Gamma) $.

Unlike the triangle in Figure~\ref{fig:lindep_def:a}, the line segment depicted in Figure~\ref{fig:lindep_def:b} does establish a linear dependency. More precisely, $ \ModelsOf(\Gamma) - v $ is contained in two vectorspaces -- $ \Q^2 $ and the one-dimensional vectorspace whose orthogonal complement is the red line. If we denote this orthogonal complement by $ U $, then it follows that \[
	\LinDep_{\id, v}(\Gamma) = \{\gen{0}, U\}
\] Once again, $ \gen{0} \in \LinDep_{\id, v}(\Gamma) $ corresponds to the fact that $ \Gamma $ entails any predicate that is a tautology, while $ U \in \LinDep_{\id, v}(\Gamma) $ contains the information about a nontrivial equality predicate $ q $ entailed by $ \Gamma $. Intuitively, $ q $ asserts that the gray line segment depicted in Figure~\ref{fig:lindep_def:b} is orthogonal to the red line, that is, $ U $.

Finally, the third example of Figure~\ref{fig:lindep_def:c} shows what happens when $ \Gamma $ defines a point: we translate it to the origin and, since every vectorspace contains zero, it follows that \[
	\LinDep_{\id, v}(\Gamma) = \{U \le \Q^2\}
\] To sum up, these examples show that $ \LinDep_{\id, v}(\Gamma) $ is, in essence, a measure of how ``flat'' $ \ModelsOf(\Gamma) $ is: the ``flatter'' the shape defined by $ \Gamma $ is, the more elements $ \LinDep_{\id, v}(\Gamma) $ contains.

For the sake of convenience, we define $ \LinDep_h(\Gamma) $ to be $ \LinDep_{h, v}(\Gamma) $ for some $ v \models \Gamma $ (where $ \Gamma $ is a satisfiable predicate set). Indeed, as we prove in Appendix~\ref{sec:app:lindep_facts} (Lemma~\ref{lemma:lindep_gamma_v1_equals_lindep_gamma_v2}), this definition is well-defined in the sense that $ \LinDep_{h, v}(\Gamma) $ does not depend on the particular choice of $ v \models \Gamma $.

For a block $ Z \in \Pi = \{X, Y\} $, let $ \pi_Z: \Q^n \rightarrow \Q^{|Z|} $ be the projection map outputting precisely those components of the input vector, which correspond to $ Z \in \Pi $. That is, $ \pi_X $ and $ \pi_Y $ are the homomorphisms defined by the \[
	\Pi_X := \begin{pNiceMatrix}[first-col, first-row, extra-margin=2pt, code-for-first-col=\scriptscriptstyle, code-for-first-row=\scriptscriptstyle, columns-width=2em]
		  & 1 & 2 & \cdots & \left|X\right| & 1 & \cdots & \left|Y\right| \\
		1 & 1 & 0 & \cdots & 0 & 0 & \cdots & 0 \\
		2 & 0 & 1 & \cdots & 0 & 0 & \cdots & 0 \\
		\vdots & \vdots & \vdots & \ddots & \vdots & \vdots & \ddots & \vdots \\
		\left|X\right| & 0 & 0 & \cdots & 1 & 0 & \cdots & 0
		\CodeAfter \tikz \node [highlight = (1-1) (4-4)] {} ;
	\end{pNiceMatrix} \in \Q^{\left|X\right| \times n}
\] and \[
	\Pi_Y := \begin{pNiceMatrix}[first-col, first-row, extra-margin=2pt, code-for-first-col=\scriptscriptstyle, code-for-first-row=\scriptscriptstyle, columns-width=2em]
		& 1 & \cdots & \left|X\right| & 1 & 2 & \cdots & \left|Y\right| \\
		1 & 0 & \cdots & 0 & 1 & 0 & \cdots & 0 \\
		2 & 0 & \cdots & 0 & 0 & 1 & \cdots & 0 \\
		\vdots & \vdots & \ddots & \vdots & \vdots & \vdots & \ddots & \vdots \\
		\left|Y\right| & 0 & \cdots & 0 & 0 & 0 & \cdots & 1
		\CodeAfter \tikz \node [highlight = (1-4) (4-7)] {} ;
	\end{pNiceMatrix} \in \Q^{\left|Y\right| \times n}
\] matrices, respectively. The highlighted parts are identity sub-matrices, while all remaining entries are zero. For example, if $ X = \{x_1, x_2\} $, $ Y = \{y_1, y_2, y_3\} $ and $ v := (1, -3, -7, 2, 5, 3)^\transp $, then $ \pi_X(v) = (1, -3)^\transp $ simply outputs the first and second components, whereas $ \pi_Y(v) = (-7, 2, 5, 3)^\transp $ returns the remaining entries.

Since in the following we will be extensively studying the $ \LinDep_{h}(\Gamma) $ sets for various satisfiable predicate sets $ \Gamma $ but (in most cases) restricted to $ h = \pi_Z $, it makes sense to give $ \LinDep_{\pi_Z}(\Gamma) $ a separate name: for $ Z \in \Pi $, we refer to $ \LinDep_{\pi_Z}(\Gamma) $ as the set of \textit{$ Z $-dependencies} of $ \Gamma $. Whenever \[
	\LinDep_{\pi_Z}(\Gamma) \subseteq \LinDep_{\pi_Z}(\Gamma')
\] holds for $ Z \in \Pi $ and satisfiable predicate sets $ \Gamma $ and $ \Gamma' $, we say that $ \Gamma' $ has at least as many $ Z $-dependencies as $ \Gamma $. If strict set inclusion holds, then we say that $ \Gamma' $ has more $ Z $-dependencies than $ \Gamma $. In this case we may also say that $ \Gamma $ has fewer $ Z $-dependencies than $ \Gamma' $.


\subsubsection{Overspilling}

We now describe the first application of the $ \LinDep_h $ metric. Recall that one of the central properties of the normal form we introduced is that one can always ``travel'' between any two disjuncts by merely changing predicate symbols. We now want to study what happens when we make only one step during this traveling process, that is, when we change the symbol of just a single predicate. More precisely, we determine the conditions under which one predicate set causes model flooding into another one, but do this only for predicate sets $ \Theta_1 $ and $ \Theta_2 $ differing in exactly one element. Moreover, we impose even stronger assumptions on $ \Theta_1 $ and $ \Theta_2 $ which we formalize in the following definition.

\begin{zbdefinition}[``$ p $-next to'']
	Let $ \Theta_1 $ and $ \Theta_2 $ be predicate sets. We say that $ \Theta_1 $ is $ p $-next to $ \Theta_2 $ if there exist $ p \in \Theta_1^= $ and $ P \in \{<, >\} $ such that \begin{itemize}
		\item $ \Theta_1 = \Theta \cup \{p\} $
		\item $ \Theta_2 = \Theta \cup \{p^P\} $
	\end{itemize} holds for some predicate set $ \Theta $.
\end{zbdefinition}

Roughly speaking, a predicate set $ \Theta_1 $ is $ p $-next to $ \Theta_2 $ if $ \Theta_1 $ agrees with $ \Theta_2 $ on all predicates except exactly one equality $ p \in \Theta_1^= $. If both predicate sets are satisfiable, then it is also useful to think about $ \Theta_1 $ being $ p $-next to $ \Theta_2 $ geometrically as shown in Figure~\ref{fig:overspilling_pnextintuition}.

\begin{figure}[t]
	\centering
	\input{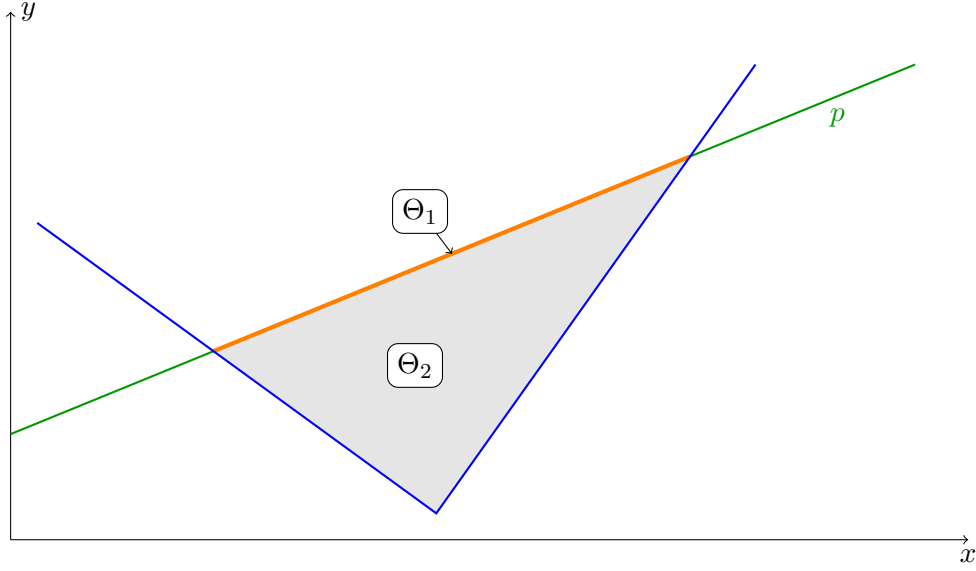}
	\caption{A two-dimensional example of how the sets defined by predicate sets $ \Theta_1 $ and $ \Theta_2 $ may look like geometrically, in the case when $ \Theta_1 $ is $ p $-next to $ \Theta_2 $. Blue lines depict the borders defined by linear constraints on which $ \Theta_1 $ and $ \Theta_2 $ agree.}
	\label{fig:overspilling_pnextintuition}
\end{figure}

Recall that in Example~\ref{example:reduction_ex}, we argued that $ \Gamma $ causes model flooding into all $ \Lambda \in \Sat(\DisjTrue) \setminus \{\{x + y = 2, x - y = 0\}\} $. In essence, our reasoning relied on the vague geometric intuition that ``non-aligned'' borders cannot be expressed using $ \Pi $-respecting formulas. At this point, we have finally introduced all the necessary concepts needed to make this intuition precise by giving the exact conditions under which one predicate set causes model flooding into another one. However, we first show a slightly different statement, which is stronger because it is independent of whether $ \Theta_1 $ and $ \Theta_2 $ are elements of $ \DisjTrue $ or not.

\begin{theorem}[Overspilling]
	\label{thm:overspilling}
	Let $ \varphi \in \QFLRA $ be a formula, $ \Pi $ be a binary partition and $ \Theta_1, \Theta_2 $ be predicate sets such that
	\begin{itemize}
		\item $ \Theta_1 $ and $ \Theta_2 $ are both $ \Pi $-complex
		\item $ \Theta_1 $ is $ p $-next to $ \Theta_2 $ on a predicate $ p \in \Theta_1^= $
		\item $ \LinDep_{\pi_Z}(\Theta_1) = \LinDep_{\pi_Z}(\Theta_2) $ holds for all $ Z \in \Pi $
	\end{itemize}
	Then $ \varphi \wedge \Theta_2 $ is satisfiable if $ \varphi $ is $ \Pi $-decomposable and $ \Theta_1 \models \varphi $.
\end{theorem}
\begin{proof}
	See Appendix~\ref{sec:app:overspilling_proof}.
\end{proof}

Using this result and the fact that every disjunct entails either $ \varphi $ or $ \neg\varphi $ (Lemma~\ref{lemma:every_disj_either_true_or_false}), it is easy to show that $ \Theta_1 \in \DisjTrue $ causes model flooding into $ \Theta_2 \in \DisjTrue $ and vice versa if these predicate sets satisfy the requirements listed in the Overspilling Theorem~\ref{thm:overspilling}. Indeed, assuming $ \varphi $ is $ \Pi $-decomposable and $ \Theta_1 \models \varphi $ implies by the Overspilling Theorem~\ref{thm:overspilling} that $ \varphi \wedge \Theta_2 $ is satisfiable, whereupon we conclude that $ \Theta_2 \models \neg\varphi $ is impossible, so by Lemma~\ref{lemma:every_disj_either_true_or_false} $ \Theta_2 \models \varphi $. Similar reasoning is applicable in the other direction: if $ \Theta_2 \models \varphi $, then $ \Theta_1 \models \neg\varphi $ is impossible because this would imply (by the Overspilling Theorem~\ref{thm:overspilling}) the satisfiability of $ \neg\varphi \wedge \Theta_2 $ directly contradicting $ \Theta_2 \models \varphi $. Hence, $ \Theta_1 \models \varphi $ holds by Lemma~\ref{lemma:every_disj_either_true_or_false}. We conclude that the conditions listed in the Overspilling Theorem~\ref{thm:overspilling} provide a sufficient condition under which one disjunct causes model flooding into another one. This theorem is a central tool we will later use to prove that the coverings our algorithm constructs are indeed $ (\varphi, \Pi) $-model-flooding-formulas.

\subsection{Towards the covering algorithm}
\label{sec:vardec:high_level_overview}

Having done the preparatory work introducing and motivating the necessary concepts and results, we now give an intuitive high-level overview of the covering algorithm for solving the covering problem. Let $ \Gamma $ be the given $ \Pi $-complex predicate set.

\begin{enumerate}[label=(\alph*)]
	\item \label{step:cover:1}We first remove all $ \Pi $-disrespecting predicates from $ \Gamma $, which leaves us with $ \Theta := \PiSimp{\Gamma} $. In the remainder of the algorithm, we will build the covering by extending $ \Theta $ with further predicates.
	\item \label{step:cover:2}Next, we compute a set of predicates together enforcing all $ X $- and $ Y $-dependencies of $ \Gamma $ and add that set to $ \Theta $. Intuitively, this is needed to ensure that every disjunct of $ \Theta $ cannot have fewer $ X $- or $ Y $-dependencies than $ \Gamma $.
	\item \label{step:cover:3}We initialize the covering to be $ \Theta $.
	\item \label{step:cover:4}At this point, some disjunct \[
		\Omega \in \DisjOf{\Theta}
	\] may have more $ Z $-dependencies than $ \Gamma $, for some $ Z \in \Pi $. While this is the case for some $ \Omega $ and $ Z $, the current covering has ``parts'' which may be distinguished from $ \Gamma $ in the language of $ \Pi $-decompositions and consequently the covering may not be a $ (\varphi, \Pi) $-MFF as we want it to be (recall in this context Example~\ref{example:reduction_ex} where such a ``part'' was the $ \{x + y = 2, x - y = 0\} $ disjunct). To circumvent this problem, we ``separate'' $ \Omega $ from $ \Gamma $ by synthesizing a certain predicate $ q $ such that $ \Omega \models q $. We call $ q $ a \textit{separating predicate}. Hence, adding $ \neg q $ as a conjunct to the covering ensures that it cannot agree with $ \Omega $ on any model. However, the $ \neg q $ conjunct may rule out some models of $ \Gamma $, namely those under which $ \Gamma \wedge q $ is true. Thus, by adding $ \neg q $ to the covering, we have made a major step towards ensuring that in the end we indeed get a $ (\varphi, \Pi) $-MFF, but, unfortunately, the $ \neg q $ conjunct breaks the correctness of the covering in the aspect that $ \Gamma $ may no longer entail it. To compensate this loss of models, we recursively cover $ \Gamma \cup \{q\} $ and add the result as a disjunct to the final covering the algorithm returns.
	\item \label{step:cover:5}We convert the covering constructed so far into a disjunction over a set of predicate sets and remove those disjuncts which have no common models with $ \Gamma $. Intuitively, this is needed to ensure that $ \Gamma $ can indeed start model flooding into every disjunct agreeing on some model with the covering.
\end{enumerate}

The crux of the covering algorithm can be summarized as follows. First, the algorithm synthesizes a predicate set $ \Theta $ enforcing all $ Z $-dependencies of $ \Gamma $. This ensures that every disjunct of $ \Theta $ cannot have fewer $ Z $-dependencies than $ \Gamma $. Next, the algorithm ensures that the produced covering never agrees with a disjunct of $ \Theta $ having more $ Z $-dependencies than $ \Gamma $, which is needed to guarantee that model flooding indeed occurs in the produced covering. In order to achieve this, the algorithm synthesizes special separating predicates and adds them to the covering in a way which ensures that model flooding starts. However, adding such predicates breaks the correctness of the covering in the aspect that $ \Gamma $ may no longer entail it. This problem is resolved by a special recursive compensation step, which ensures that no models of $ \Gamma $ get lost.


\subsection{The covering algorithm}
\label{sec:vardec:cover}

In the high-level overview of the covering algorithm above, we have omitted a lot of details about how exactly the individual steps work. In particular, we have not yet answered the following key questions about the algorithm. \begin{itemize}
	\item In step \ref{step:cover:2}, how do we actually compute the set of predicates enforcing the $ X $- and $ Y $-dependencies of $ \Gamma $?
	\item In step \ref{step:cover:4}, how do we algorithmically determine whether $ \Omega $ has more $ Z $-dependencies than $ \Gamma $?
	\item In step \ref{step:cover:4}, how do we synthesize the $ q $ predicate separating $ \Omega $ from $ \Gamma $ in the event that $ \Omega $ has more $ Z $-dependencies than $ \Gamma $?
	\item Why does the algorithm terminate?
	\item What is the running time of the algorithm?
\end{itemize} We now fill in the details by addressing these questions one-by-one in the context of a detailed specification of the covering procedure given in Algorithm~\ref{alg:cover}. \begin{algorithm}
	\caption{Covering algorithm}
	\label{alg:cover}
	\SetAlgoLined
	\KwIn{A formula $ \varphi $, partition $ \Pi $ and a predicate set $ \Gamma $}
	\KwOut{A $ \Pi $-respecting formula solving the covering problem}
	\SetKwFunction{cover}{cover}
	\SetKwProg{coverproc}{Function}{}{end}
	\coverproc{\cover{$ \Pi $, $ \Gamma $}}{
		\nl\label{alg:cover:line:pi_simple_check}\If{$ \Gamma $ is $ \Pi $-simple}{
			\nl \Return $ \decsimple(\Pi, \Gamma) $\;
		}
		\tcp{Remove all $ \Pi $-disrespecting predicates from $ \Gamma $}
		\nl\label{alg:cover:line:predresp}$ \Theta := \PiSimp{\Gamma} $\;
		\tcp{Solve the system of linear equations $ \Gamma^= $, in order to be able to analyze the $ Z $-dependencies of $ \Gamma $}
		\nl\label{alg:cover:line:gamma_lindep_repr}Compute $ a \in \Q^n $ and a basis $ A := (a_1, \dots, a_l) $ such that $ \ModelsOf(\Gamma^=) = a + \gen{a_1, \dots, a_l} $\;
		\tcp{For every $ Z \in \Pi $, synthesize a set of $ \Pi $-respecting predicates together enforcing all $ Z $-dependencies of $ \Gamma $ and add those predicates to $ \Theta $}
		\nl\label{alg:cover:line:first_loop}\ForEach{$ Z \in \Pi $}{
			\nl\label{alg:cover:line:first_loop:dep_basis}Compute a basis $ (v_1, \dots, v_s) $ of $ \ImageOf(\pi_Z \circ \lc_A)^{\bot} \le \Q^{|Z|} $\;
			\nl\label{alg:cover:line:first_loop:theta_upd}For every $ w \in \{v_1, \dots, v_s\} $, update $ \Theta := \Theta \cup \{\pi_Z(\vec{z}) \cdot w = \pi_Z(a) \cdot w\} $\;
		}
		\nl $ \Delta := \bot $\;
		\nl $ \Upsilon := \varnothing $\;
		\tcp{The goal of this loop is to ensure that the produced covering does not agree on any model with every $ \Omega \in \DisjOf{\Theta} $ having more $ Z $-dependencies than $ \Gamma $}
		\nl\label{alg:cover:line:second_loop}\ForEach{$ \Omega \in \DisjOf{\Theta} $}{
			\tcp{Solve the system of linear equations $ \Omega^= $ in order to be able to analyze $ Z $-dependencies of $ \Omega $}
			\nl\label{alg:cover:line:second_loop:dep_basis}Compute $ b \in \Q^n $ and a basis $ B := (b_1, \dots, b_r) $ such that $ \ModelsOf(\Omega^=) = b + \gen{b_1, \dots, b_r} $\;
			\tcp{Check whether $ \Omega $ has more $ Z $-dependencies than $ \Gamma $, for some $ Z \in \Pi $}
			\nl\label{alg:cover:line:second_loop:main_if}\If{$ \ImageOf(\pi_Z \circ \lc_B)^{\bot} \cap \ImageOf(\pi_Z \circ \lc_A) \neq \gen{0} $ for some $ Z \in \Pi $}{
				\tcp{$ \Omega $ has more $ Z $-dependencies than $ \Gamma $, so we synthesize an inequality separating $ \Omega $ from $ \Gamma $ and add it to $ \Upsilon $}
				\nl\label{alg:cover:line:second_loop:witness_vec}Compute a vector $ 0 \neq w \in \ImageOf(\pi_Z \circ \lc_B)^{\bot} \cap \ImageOf(\pi_Z \circ \lc_A) \le \Q^{|Z|} $\;
				\nl\label{alg:cover:line:second_loop:upsilon_upd}$ \Upsilon := \Upsilon \cup \{\pi_Z(\vec{z}) \cdot w \neq \pi_Z(b) \cdot w\} $\;
				\tcp{Due to the added inequality, the models of $ \Gamma \cup \{\pi_Z(\vec{z}) \cdot w = \pi_Z(b) \cdot w\} $ no longer satisfy the covering. To compensate for this, we need to recursively cover.}
				\nl\label{alg:cover:line:second_loop:rec_call}$ \Delta := \Delta \vee \cover(\Pi, \Gamma \cup \{\pi_Z(\vec{z}) \cdot w = \pi_Z(b) \cdot w\}) $\;
			}
		}
		\tcp{Restrict $ \Upsilon $ in a way which ensures that model flooding starts}
		\nl\label{alg:cover:line:d_set}Compute a set $ D $ of predicate sets, such that $ \Upsilon \equiv \bigvee_{\Upsilon' \in D} \Upsilon' $, by rewriting every inequality $ q \in \Upsilon $ as $ q^< \vee q^> $ and applying distributivity\;
		\nl\label{alg:cover:line:return}\Return $ \Delta \vee \Theta \wedge \bigvee \{\Upsilon' \in D \mid \Upsilon' \cup \Gamma \text{ is satisfiable}\} $\;
	}
\end{algorithm} The comments provided in the pseudocode connect the implementation of every step with the intuitions explained in Section~\ref{sec:vardec:high_level_overview} above. In particular, step \ref{step:cover:1} corresponds to Line~\ref{alg:cover:line:predresp}; the computation of a predicate set enforcing all $ X $- and $ Y $-dependencies of $ \Gamma $ (step \ref{step:cover:2}) is implemented in the first loop at Line~\ref{alg:cover:line:first_loop}; step \ref{step:cover:4} corresponds to the second loop at Line~\ref{alg:cover:line:second_loop}; the final adjustment of the covering described in step \ref{step:cover:5} is implemented at Lines \ref{alg:cover:line:d_set} and \ref{alg:cover:line:return}.

\subsubsection{Synthesis of predicates enforcing $ Z $-dependencies}

Suppose we have a satisfiable predicate set $ \Gamma $ and we want to ``extract'' its set of $ Z $-dependencies with the goal of expressing it as a finite conjunction of $ \Pi $-respecting predicates. At first sight, this may seem an infeasible task for an algorithm because $ \LinDep_{h}(\Gamma) $ is, in general, infinite\footnote{Recall in this context the example given in Figure~\ref{fig:lindep_def:c}.}. Thus, any brute-force approach is doomed to fail and it is essential to come up with a more clever way of ``extracting'' $ Z $-dependencies. Indeed, with the help of algebra, we can prove that the set of linear dependencies has a finite representation, which can be computed efficiently and be used to obtain the desired predicate set. The first step in that direction is to prove the following result.

\begin{theorem}[Algebraic characterization of $ \LinDep_{h, v} $]
	\label{thm:lindep_algebraic_char}
	Let $ \Gamma $ be a satisfiable set of predicates, $ V $ be a $ \Q $-vectorspace and $ h : \Q^n \rightarrow V $ be a $ \Q $-vectorspace homomorphism. Let furthermore $ B := (b_1, \dots, b_r) $ be a basis such that $ \ModelsOf(\Gamma^=) = b + \gen{b_1, \dots, b_r} $ for some $ b \in \Q^n $. Then \[
	\LinDep_{h, v}(\Gamma) = \{U \le V \mid U \le \ImageOf(h \circ \lc_B)^{\bot}\}
	\] holds for all $ v \models \Gamma^= $.
\end{theorem}
\begin{proof}
	See Section~\ref{sec:app:lindep_facts:proof_lindep_algebraic_char} of Appendix~\ref{sec:app:lindep_facts}.
\end{proof}

In essence, Theorem~\ref{thm:lindep_algebraic_char} establishes that $ \ImageOf(h \circ \lc_B)^{\bot} $ is the supremum of $ \LinDep_{h, v}(\Gamma) $ in the lattice of $ V $'s subspaces. Intuitively, this means that $ \ImageOf(h \circ \lc_B)^{\bot} $ contains all the information about linear dependencies present in $ \Gamma $. In the context of Algorithm~\ref{alg:cover}, it follows that $ \ImageOf(\pi_Z \circ \lc_A)^{\bot} $ can be thought of an algebraic object capturing all $ Z $-dependencies of $ \Gamma $. Moreover, the same intuition applies to the basis of $ \ImageOf(\pi_Z \circ \lc_A)^{\bot} $, which is the finite representation of the set of $ Z $-dependencies we use to synthesize predicates enforcing them (see Line~\ref{alg:cover:line:first_loop:dep_basis} of the algorithm).

Observe that the vectors $ a $ and $ a_1, \dots, a_l $ satisfying \[
	\ModelsOf(\Gamma^=) = a + \gen{a_1, \dots, a_l}
\] can be computed in polynomial time via Gaussian elimination applied to $ \Gamma^= $ (see Line~\ref{alg:cover:line:gamma_lindep_repr}). Since $ \gen{a_1, \dots, a_l} = \ImageOf(\lc_A) $, it follows that \begin{align*}
	\ImageOf(\pi_Z \circ \lc_A) &= \ImageOf(\Pi_Z \cdot (a_1 \mid \dots \mid a_l)) \\
	&= \ImageOf\big(\Pi_Z \cdot a_1 \mid \dots \mid \Pi_Z \cdot a_l\big)
\end{align*} and thus \begin{align}
	\label{eqn:im_pizlca_orth_eq_ker_piz_a1al_transp}
	\ImageOf(\pi_Z \circ \lc_A)^{\bot} = \ker\big(\Pi_Z \cdot a_1 \mid \dots \mid \Pi_Z \cdot a_l\big)^\transp
\end{align} holds for all $ Z \in \Pi $. Hence, at Line~\ref{alg:cover:line:first_loop:dep_basis} the basis $ v_1, \dots, v_s $ of $ \ImageOf(\pi_Z \circ \lc_B)^{\bot} $ can be computed by applying Gaussian elimination to compute the kernel of $ (\Pi_Z \cdot a_1 \mid \dots \mid \Pi_Z \cdot a_l)^\transp $.

We now informally explain why the construction of predicates given at Line~\ref{alg:cover:line:first_loop:theta_upd} is correct (the rigorous proof is in Appendix~\ref{sec:app:cover_analysis}). Essentially, Line~\ref{alg:cover:line:first_loop:theta_upd} translates the $ v_1, \dots, v_s $ basis vectors (spanning the set of all $ Z $-dependencies of $ \Gamma $) into $ \Pi $-respecting predicates and collects them in $ \Theta $. By construction of $ \Theta $, \begin{align}
	\label{eqn:theta_pizxminusa_dot_vi_eq_zero}
	\Theta \models \pi_Z(\vec{z} - a) \cdot v_i = 0
\end{align} holds for all $ Z \in \Pi $ and $ i \in \{1, \dots, s\} $. Any $ Z $-dependency of $ \Gamma $ corresponds to a linear combination of the $ v_i $ vectors, which is enforced by $ \Theta $ because we can replace $ v_i $ in (\ref{eqn:theta_pizxminusa_dot_vi_eq_zero}) by that linear combination and the entailment would still hold by the linearity of the dot product.

\subsubsection{Comparison of dependencies in predicate sets}

In the second loop of the algorithm at Line~\ref{alg:cover:line:second_loop}, our goal is not just to ``extract'' the $ Z $-dependencies of $ \Omega $, but to compare them with those of $ \Gamma $. We have already discussed that, intuitively, the basis of $ \ImageOf(\pi_Z \circ \lc_A)^{\bot} $ contains the information about all $ Z $-dependencies present in $ \Gamma $. By the same reasoning, if $ b $ and $ B = (b_1, \dots, b_r) $ are as defined at Line~\ref{alg:cover:line:second_loop:dep_basis}, then the basis of $ \ImageOf(\pi_Z \circ \lc_B)^{\bot} $ captures the $ Z $-dependencies present in $ \Omega $. We defer the rigorous proof to Appendix~\ref{sec:app:cover_analysis} and focus here only on the high-level intuition. The main idea is to prove that the basis of $ \ImageOf(\pi_Z \circ \lc_B)^{\bot} \cap \ImageOf(\pi_Z \circ \lc_A) $ corresponds to precisely those $ Z $-dependencies which are present in $ \Omega $ but absent in $ \Gamma $. Intuitively, this is the case because the $ \ImageOf(\pi_Z \circ \lc_B)^{\bot} $ component ensures that we are looking at the $ Z $-dependencies of $ \Omega $, whereas the $ \ImageOf(\pi_Z \circ \lc_A) $ part can be thought of as the complement of $ \Gamma $'s $ Z $-dependencies: think of $ \ImageOf(\pi_Z \circ \lc_A) $ as the orthogonal complement of the $ \ImageOf(\pi_Z \circ \lc_A)^\bot $ vectorspace capturing the $ Z $-dependencies of $ \Gamma $ by Theorem~\ref{thm:lindep_algebraic_char}.

Given this intuition about the $ \ImageOf(\pi_Z \circ \lc_B)^{\bot} \cap \ImageOf(\pi_Z \circ \lc_A) $ intersection, it is not surprising that this vectorspace contains a non-zero vector if and only if $ \Omega $ has more $ Z $-dependencies than $ \Gamma $. Since $ \Omega $ is a disjunct of $ \Theta $, it follows that $ \ImageOf(\pi_Z \circ \lc_B)^{\bot} \cap \ImageOf(\pi_Z \circ \lc_A) = \gen{0} $ holds if and only if $ \LinDep_{\pi_Z}(\Gamma) = \LinDep_{\pi_Z}(\Omega) $. This is a central property we will extensively rely on in the rigorous proof of correctness, specifically in the proof that the formula produced by the covering algorithm is a $ (\varphi, \Pi) $-MFF.

We now discuss how to compute the basis of $ \ImageOf(\pi_Z \circ \lc_B)^{\bot} \cap \ImageOf(\pi_Z \circ \lc_A) $. We do this in three following steps. \begin{enumerate}
	\item Applying the same reasoning as in the derivation of (\ref{eqn:im_pizlca_orth_eq_ker_piz_a1al_transp}) above (but to $ \Omega $ instead of $ \Gamma $) yields the \begin{align}
		\label{eqn:im_pizlcb_orth_eq_ker_piz_b1br_transp}
		\ImageOf(\pi_Z \circ \lc_B)^{\bot} = \ker\big(\Pi_Z \cdot b_1 \mid \dots \mid \Pi_Z \cdot b_r\big)^\transp
	\end{align} characterization, which in turn gives us a way to compute the basis of $ \ImageOf(\pi_Z \circ \lc_B)^{\bot} $ by computing the kernel of $ (\Pi_Z \cdot b_1 \mid \dots \mid \Pi_Z \cdot b_r)^\transp $ via Gaussian elimination.
	\item It is immediate that \[
		\ImageOf(\pi_Z \circ \lc_A) = \ImageOf(\Pi_Z \cdot (a_1 \mid \dots \mid a_l))
	\] meaning that the basis of $ \ImageOf(\pi_Z \circ \lc_A) $ is $ \Pi_Z \cdot a_1, \dots, \Pi_Z \cdot a_l $.
	\item Once the bases of $ \ImageOf(\pi_Z \circ \lc_B)^{\bot} $ and $ \ImageOf(\pi_Z \circ \lc_A) $ are computed, it is a standard operation in linear algebra to compute the basis of the intersection. This can be achieved in polynomial time, for example, via the well-known Zassenhaus sum-intersection algorithm \cite[pp. 207 -- 210]{fisher:2012}.
\end{enumerate}

\subsubsection{Synthesis of a separating predicate}

We now address the question of how to synthesize a predicate separating $ \Omega $ from $ \Gamma $ in case $ \Omega $ has more $ Z $-dependencies than $ \Gamma $. The answer is given in Lines \ref{alg:cover:line:second_loop:witness_vec} and \ref{alg:cover:line:second_loop:upsilon_upd} -- our approach here is similar to the way we used $ \Theta $ to enforce all $ Z $-dependencies of $ \Gamma $ at Lines \ref{alg:cover:line:first_loop:dep_basis} and \ref{alg:cover:line:first_loop:theta_upd} (see the corresponding discussion above). The only difference is that in the second loop it suffices to compute just a single vector $ w $ and to translate it into the separating predicate via the dot product construction. Intuitively, the $ \pi_Z(\vec{z}) \cdot w = \pi_Z(b) \cdot w $ predicate indeed separates $ \Omega $ from $ \Gamma $ in the sense that \[
	\Omega \models \pi_Z(\vec{z}) \cdot w = \pi_Z(b) \cdot w
\] because this is equivalent to $ \Omega \models \pi_Z(\vec{z} - b) \cdot w = 0 $ and it holds that $ w \in \ImageOf(\pi_Z \circ \lc_B)^{\bot} $, where, as already noted above, $ \ImageOf(\pi_Z \circ \lc_B)^{\bot} $ can be thought of as the span of all $ Z $-dependencies contained in $ \Omega $. In essence, the $ \pi_Z(\vec{z} - b) \cdot w = 0 $ dot product asserts that $ w \in \ImageOf(\pi_Z \circ \lc_B)^{\bot} $ is orthogonal to $ \ImageOf(\pi_Z \circ \lc_B) $. For a full and rigorous proof of correctness, see Appendix~\ref{sec:app:cover_analysis}.

\subsubsection{Termination and time complexity}
\label{sec:vardec:cover:termination_time_complexity}

We now give an intuitive explanation of why the covering algorithm always terminates (for a detailed proof, we once again refer the reader to Appendix~\ref{sec:app:cover_analysis}). Observe that the recursive call at Line~\ref{alg:cover:line:second_loop:rec_call} is the only part of the algorithm raising termination concerns; all linear-algebraic operations trivially terminate since the domain we are working over is always a vectorspace of dimension not exceeding $ n = \left|X\right| + \left|Y\right| $. In short, the algorithm terminates because every recursive call adds a new $ Z $-dependency to $ \Gamma $ and at some point we are guaranteed to run out of $ Z $-dependencies which are not redundant. Once such a state has been reached, $ \Gamma $ is guaranteed to be $ \Pi $-simple, whereupon the nested recursion stops due to the check at Line~\ref{alg:cover:line:pi_simple_check}. 

More precisely, at Line~\ref{alg:cover:line:second_loop:rec_call} it is guaranteed that $ \Gamma \cup \{\pi_Z(\vec{z}) \cdot w = \pi_Z(b) \cdot w\} $ enforces more $ Z $-dependencies than $ \Gamma $, for some $ Z \in \Pi $. Hence, the recursion depth of the covering algorithm is bounded by the number of times it is possible to add a novel $ Z $-dependency. The idea is to argue that for a fixed $ Z \in \Pi $, this amount of times cannot exceed $ \left|Z\right| $ because any $ Z $-dependency is a subspace of the $ \left|Z\right| $-dimensional vectorspace $ \Q^{\left|Z\right|} $ and keeping adding novel $ Z $-dependencies corresponds to constructing an ascending chain of subspaces in $ \Q^{\left|Z\right|} $. Thus, the fundamental mathematical reason why we are able to obtain this upper bound $ \left|Z\right| $ is because $ \Q^{\left|Z\right|} $ as a $ \Q $-module satisfies the ascending chain condition (i.e., is Noetherian).

Overall, applying the above reasoning to both $ X $ and $ Y $ yields the upper bound $ n = \left|X\right| + \left|Y\right| $ for the recursion depth. Furthermore, since the dimensions of vectorspaces the algorithm deals with are all bounded by $ n $, linear-algebraic operations can all be implemented to run in polynomial time. Hence, the height of the tree of recursive calls is bounded by $ n $, and every node has at most exponentially many children (see the second loop at Line~\ref{alg:cover:line:second_loop}). It follows that the overall number of calls to the covering algorithm is in \[
	O\Big(\big(2^{\poly(\left|\varphi\right|)}\big)^n\Big) = O\big(2^{\poly(\left|\varphi\right|)}\big)
\] Now note that the time complexity of the algorithm without recursion is double-exponential because in the worst-case scenario there may be exponentially many disjuncts $ \Omega \in \DisjOf{\Theta} $ that yield separating predicates, meaning that $ \Upsilon $ is of exponential size and we consequently get a double-exponential blowup at Lines \ref{alg:cover:line:d_set} and \ref{alg:cover:line:return}. Accordingly, we conclude that the overall time complexity of the covering algorithm is \[
	 O\big(2^{\poly(\left|\varphi\right|)}\big) \cdot O\big(2^{2^{\poly(\left|\varphi\right|)}}\big) = O\big(2^{2^{\poly(\left|\varphi\right|)}}\big)
\]

\subsection{Correctness of the covering algorithm}
\label{sec:vardec:correctness}

We now prove that the covering algorithm indeed correctly solves the covering problem in the sense that any formula $ \psi := \cover(\Pi, \Gamma) $ produced by it is a $ \Pi $-respecting $ (\varphi, \Pi) $-MFF satisfying $ \Gamma \models \psi $, where $ \Gamma \in \Sat(\DisjPhi) $ is $ \Pi $-complex. Indeed, it can be immediately observed that, by construction of every predicate created within the algorithm, it can output only $ \Pi $-respecting formulas. In Appendix~\ref{sec:app:cover_analysis} (Lemma~\ref{lemma:gamma_entails_psi}), we show that the $ \Gamma \models \psi $ requirement also holds; intuitively, this follows from the first loop's invariant $ \Gamma \models \Theta $ and the compensation step incorporated into the second loop at line 15 (see the discussion of step \ref{step:cover:4} in Section~\ref{sec:vardec:high_level_overview} above). Thus, to establish correctness of the algorithm, it suffices to prove that $ \psi $ is a $ (\varphi, \Pi) $-MFF, which is the focus of the remainder of Section~\ref{sec:vardec:correctness}.

\subsubsection{Properties of the produced covering}
\label{sec:vardec:properties_of_produced_covering}

In order to show that model flooding indeed occurs in $ \psi $, we need to precisely understand the set of models this formula defines. We achieve this by bringing $ \psi $ into DNF and studying the properties of every term. More precisely, we can assume (due to the normal form we introduced) the existence of a set $ \Psi $ of predicate sets such that \begin{align}
	\label{eqn:psi_disj_lambda_psi_lambda}
	\psi \equiv \bigvee_{\Lambda \in \Psi} \Lambda
\end{align} In other words, every $ \Lambda $ simply corresponds to a DNF term of $ \psi $. 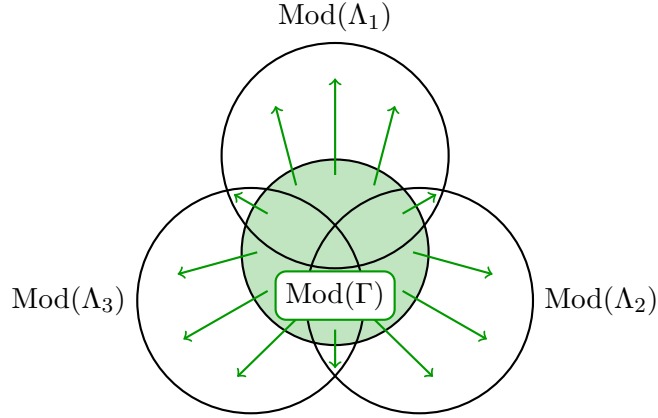
\begin{figure}[t]
	\centering
	\begin{tikzpicture}[
		scale=\zbtikznormalscaling,
		lambdaset/.style = {circle, minimum size=85, thick, draw=black}
	]
	
	\coordinate (lambda1set) at (90:1.25);
	\coordinate (lambda2set) at (330:1.25);
	\coordinate (lambda3set) at (210:1.25);
	
	\node[circle, minimum size=70, thick, draw=black, fill=OliveGreen, fill opacity=0.25] at (0, 0) {};
	
	\node[lambdaset, label={above:$ \ModelsOf(\Lambda_1) $}] at (lambda1set) {};
	\node[lambdaset, label={right:$ \ModelsOf(\Lambda_2) $}] at (lambda2set) {};
	\node[lambdaset, label={left:$ \ModelsOf(\Lambda_3) $}] at (lambda3set) {};
	
	\draw[->, thick, color=OliveGreen] (60:1) -- ($ (lambda1set) + (40:1) $);
	\draw[->, thick, color=OliveGreen] (90:1) -- ($ (lambda1set) + (90:1) $);
	\draw[->, thick, color=OliveGreen] (120:1) -- ($ (lambda1set) + (140:1) $);
	\draw[->, thick, color=OliveGreen] (300:1) -- ($ (lambda2set) + (280:1) $);
	\draw[->, thick, color=OliveGreen] (330:1) -- ($ (lambda2set) + (330:1) $);
	\draw[->, thick, color=OliveGreen] (360:1) -- ($ (lambda2set) + (380:1) $);
	\draw[->, thick, color=OliveGreen] (180:1) -- ($ (lambda3set) + (160:1) $);
	\draw[->, thick, color=OliveGreen] (210:1) -- ($ (lambda3set) + (210:1) $);
	\draw[->, thick, color=OliveGreen] (240:1) -- ($ (lambda3set) + (260:1) $);
	\draw[->, thick, color=OliveGreen] (30:1) -- (30:1.5);
	\draw[->, thick, color=OliveGreen] (150:1) -- (150:1.5);
	\draw[->, thick, color=OliveGreen] (270:1) -- (270:1.5);
	
	\draw (0, -0.56) node [rounded corners, thick, draw=OliveGreen, fill=white] {$ \ModelsOf(\Gamma) $};
	
\end{tikzpicture}
	\caption{Venn diagram illustrating the relationship between $ \ModelsOf(\Gamma) $ (the green set in the center) and $ \ModelsOf(\Lambda_i) $ for $ \Lambda_i \in \Psi = \{\Lambda_1, \Lambda_2, \Lambda_3\} $. In particular, observe that $ \Gamma \models \psi \equiv \Lambda_1 \vee \Lambda_2 \vee \Lambda_3 $. The model flooding $ \Gamma $ causes into every $ \Lambda_i \in \Psi $ (which we aim to prove) is visualized using green arrows. The region corresponding to $ \Gamma $ is colored green to illustrate that $ \Gamma \models \varphi $.}
	\label{fig:psi_domain_flower}
\end{figure} An overview of relationships between $ \Gamma $ and $ \Lambda \in \Psi $ sets as well as of the desired model flooding effect we aim to prove is given in Figure~\ref{fig:psi_domain_flower}.

Observe that due to (\ref{eqn:psi_disj_lambda_psi_lambda}), in order to prove that $ \psi $ is a $ (\varphi, \Pi) $-MFF, it suffices to show that $ \Gamma $ causes model flooding into every $ \Lambda \in \Psi $. To prove the latter, we need to study certain properties of every individual $ \Lambda \in \Psi $, which we capture in the following Lemma~\ref{lemma:psi_properties}.

\begin{lemma}[Properties of $ \psi $]
	\label{lemma:psi_properties}
	Let $ \varphi \in \QFLRA $ be a formula, $ \Pi $ be a partition, $ \Gamma \in \Sat(\DisjTrue) $ be a $ \Pi $-complex predicate set, $ \psi := \cover(\Pi, \Gamma) $ and $ \Psi $ be the set of predicate sets corresponding to the DNF terms of $ \psi $. Then for any $ \Lambda \in \Psi $ there exists a set $ \Theta $ computed in some (possibly nested) call to $ \cover $ such that: \begin{enumerate}[label=(\alph*)]
		\item\label{lemma:psi_properties:a} $ \Gamma \wedge \Lambda $ is satisfiable.
		\item\label{lemma:psi_properties:b} For every $ \Gamma_1, \Gamma_2 \in \DisjOf{\Theta} $, the satisfiability of $ \Gamma_1 \wedge \Lambda $ and $ \Gamma_2 \wedge \Lambda $ implies \[
		\LinDep_{\pi_Z}(\Gamma_1) = \LinDep_{\pi_Z}(\Gamma_2)
		\] where $ Z \in \Pi $ is arbitrary.
		\item\label{lemma:psi_properties:c} For every $ \Omega \in \DisjOf{\Theta} $, $ \Omega $ is $ \Pi $-complex if $ \Omega \wedge \Lambda $ is satisfiable.
	\end{enumerate} 
\end{lemma}

Before we prove Lemma~\ref{lemma:psi_properties}, we intuitively explain every property of $ \psi $ established by this lemma. 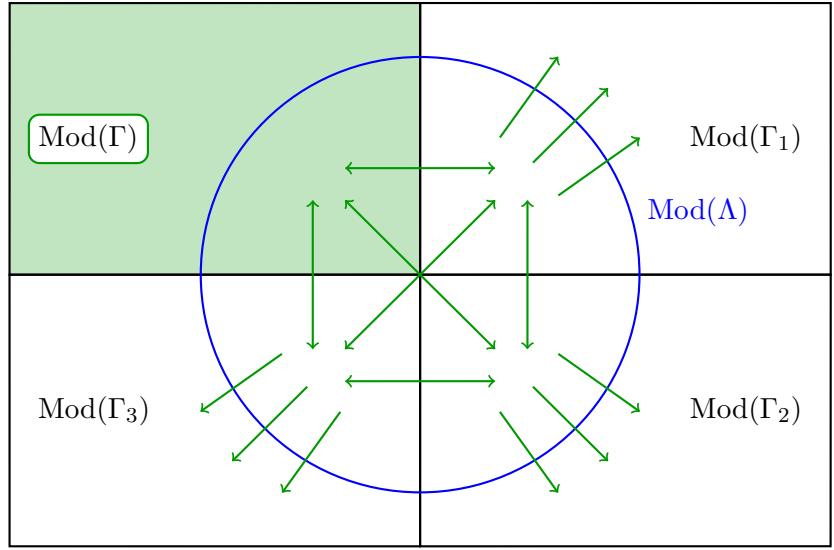
\begin{figure}[t]
	\centering
	\begin{tikzpicture}[scale=\zbtikzaggressivescaling]
	
	\coordinate (thetatl) at (-4.5, 3);
	\coordinate (thetabr) at (4.5, -3);
	
	\coordinate (thetatr) at ($ (thetatl-|thetabr) $);
	\coordinate (thetabl) at ($ (thetabr-|thetatl) $);
	
	\coordinate (thetatc) at ($ (thetatl)!0.5!(thetatr) $);
	\coordinate (thetabc) at ($ (thetabl)!0.5!(thetabr) $);
	\coordinate (thetalc) at ($ (thetabl)!0.5!(thetatl) $);
	\coordinate (thetarc) at ($ (thetabr)!0.5!(thetatr) $);
	
	\coordinate (arrowtl) at (-1, 1);
	\coordinate (arrowbr) at (1, -1);
	\coordinate (arrowtr) at ($ (arrowtl-|arrowbr) $);
	\coordinate (arrowbl) at ($ (arrowbr-|arrowtl) $);
	
	\filldraw[thick, fill=OliveGreen, fill opacity=0.25] (thetatl) -- (thetatc) -- (0, 0) -- (thetalc) -- cycle;
	\draw[thick] (thetatc) -- (thetatr) -- (thetarc) -- (0, 0) -- cycle;
	\draw[thick] (0, 0) -- (thetarc) -- (thetabr) -- (thetabc) -- cycle;
	\draw[thick] (thetalc) -- (0, 0) -- (thetabc) -- (thetabl) -- cycle;

	\node[circle, minimum size=15em, thick, draw=blue, label={[color=blue]10:$ \ModelsOf(\Lambda) $}] at (0, 0) {};
	
	\draw[<->, thick, draw=OliveGreen] ($ (arrowtl) + (45:0.25) $) -- ($ (arrowtr) + (135:0.25) $);
	\draw[<->, thick, draw=OliveGreen] ($ (arrowtl) + (225:0.25) $) -- ($ (arrowbl) + (135:0.25) $);
	\draw[<->, thick, draw=OliveGreen] ($ (arrowbl) + (315:0.25) $) -- ($ (arrowbr) + (225:0.25) $);
	\draw[<->, thick, draw=OliveGreen] ($ (arrowtr) + (315:0.25) $) -- ($ (arrowbr) + (45:0.25) $);
	\draw[<->, thick, draw=OliveGreen] ($ (arrowtl) + (315:0.25) $) -- ($ (arrowbr) + (135:0.25) $);
	\draw[<->, thick, draw=OliveGreen] ($ (arrowbl) + (45:0.25) $) -- ($ (arrowtr) + (225:0.25) $);
	\draw[->, thick, color=OliveGreen] (30:1.75) -- ($ (arrowtr) + (20:1.5) $);
	\draw[->, thick, color=OliveGreen] (45:1.75) -- ($ (arrowtr) + (45:1.5) $);
	\draw[->, thick, color=OliveGreen] (60:1.75) -- ($ (arrowtr) + (70:1.5) $);
	\draw[->, thick, color=OliveGreen] (300:1.75) -- ($ (arrowbr) + (290:1.5) $);
	\draw[->, thick, color=OliveGreen] (315:1.75) -- ($ (arrowbr) + (315:1.5) $);
	\draw[->, thick, color=OliveGreen] (330:1.75) -- ($ (arrowbr) + (340:1.5) $);
	\draw[->, thick, color=OliveGreen] (210:1.75) -- ($ (arrowbl) + (200:1.5) $);
	\draw[->, thick, color=OliveGreen] (225:1.75) -- ($ (arrowbl) + (225:1.5) $);
	\draw[->, thick, color=OliveGreen] (240:1.75) -- ($ (arrowbl) + (250:1.5) $);
	
	\node[anchor=west, rounded corners, thick, draw=OliveGreen, fill=white] at ($ (thetatl)!0.5!(thetalc) + (0.2, 0) $) {$ \ModelsOf(\Gamma) $};
	\node[anchor=east] at ($ (thetatr)!0.5!(thetarc) - (0.2, 0) $) {$ \ModelsOf(\Gamma_1) $};
	\node[anchor=east] at ($ (thetarc)!0.5!(thetabr) - (0.2, 0) $) {$ \ModelsOf(\Gamma_2) $};
	\node[anchor=west] at ($ (thetalc)!0.5!(thetabl) + (0.2, 0) $) {$ \ModelsOf(\Gamma_3) $};
	
\end{tikzpicture}
	\caption{Venn diagram illustrating the relationship between $ \ModelsOf(\Lambda) $ and the sets defined by those disjuncts of $ \Theta $ (in this case $ \Gamma $, $ \Gamma_1 $, $ \Gamma_2 $, $\Gamma_3 $) which agree on at least one model with $ \Lambda $. Observe that these disjuncts define pairwise disjoint sets, and every $ v \models \Lambda \models \Theta $ must satisfy some disjunct. The model flooding that $ \Gamma $ causes into $ \Lambda $ and $ \Lambda $ subsequently into every $ \Gamma_i \in \{\Gamma_1, \Gamma_2, \Gamma_3\} $ is visualized using green arrows. The region corresponding to $ \Gamma $ is colored green to illustrate that $ \Gamma \models \varphi $.}
	\label{fig:theta_disjuncts_intersecting_lambda}
\end{figure} Fix some $ \Lambda \in \Psi $. This $ \Lambda $ must have resulted from some (possibly nested) call to the covering algorithm meaning that there must be a $ \Theta $ set computed at some point in the covering algorithm such that $ \Lambda $ is $ \Theta $ with (possibly) some additional strict inequalities, which originated in the second loop of the covering algorithm at Line~\ref{alg:cover:line:second_loop}. Think about the disjuncts of that $ \Theta $ as a grid partitioning $ \ModelsOf(\Theta) $ into (disjoint) sets. As demonstrated in Figure~\ref{fig:theta_disjuncts_intersecting_lambda}, the set defined by $ \Lambda $ can be thought of as a subset of $ \ModelsOf(\Theta) $. In essence, property \ref{lemma:psi_properties:a} says that $ \Lambda $ must agree with $ \Gamma $ on some model (note in this connection that in Figure~\ref{fig:theta_disjuncts_intersecting_lambda}, $ \ModelsOf(\Lambda) $ has a nonempty intersection with $ \ModelsOf(\Gamma) $). Intuitively, we will need this property to argue that model flooding indeed starts, in the sense that $ \Lambda $ agrees on some model with $ \Gamma $, which causes all models of $ \Lambda $ to satisfy $ \varphi $ if $ \Gamma \models \varphi $. Hence, the satisfiability of $ \Gamma \wedge \varphi $ implies that $ \psi \models \varphi $ if $ \varphi $ is $ \Pi $-decomposable.


Intuitively, property \ref{lemma:psi_properties:b} says that the disjuncts of $ \Theta $ agreeing on some model with $ \Lambda $ all pairwise have the same set of $ Z $-dependencies. We will need this later to be able to apply the Overspilling Theorem~\ref{thm:overspilling} which is at the core of the model flooding effect we are proving. Property \ref{lemma:psi_properties:c} is another technical statement we will need to make sure that all conditions of the Overspilling Theorem~\ref{thm:overspilling} are fulfilled.

\subsubsection{Proof of Lemma~\ref{lemma:psi_properties}}
\label{sec:vardec:lemma_psi_properties_proof}

We now prove Lemma~\ref{lemma:psi_properties}. By structural induction, we can assume that the stated properties hold if $ \Lambda $ comes from $ \Delta $, i.e., from the covering produced by any recursive call. Thus, it suffices to consider only the case when $ \Lambda $ is a DNF term of $ \Theta \wedge \bigvee \{\Upsilon' \in D \mid \Upsilon' \cup \Gamma \text{ is satisfiable}\} $ (see Line~\ref{alg:cover:line:return}).

Let $ \Upsilon' \in D $ be the set such that $ \Theta \cup \Upsilon' = \Lambda $. Since $ \Upsilon' \wedge \Gamma $ is satisfiable by construction (see Line~\ref{alg:cover:line:return}) and $ \Gamma \models \Theta $ holds by Lemma~\ref{lemma:gamma_entails_theta} (see Appendix~\ref{sec:app:cover_analysis}), it follows that $ \Gamma \wedge \Upsilon' \wedge \Theta $ is satisfiable. This shows property \ref{lemma:psi_properties:a}.

As regards property \ref{lemma:psi_properties:b}, we rigorously prove it in Appendix~\ref{sec:app:cover_analysis} (Lemma~\ref{lemma:disjuncts_touching_lambda_have_same_z_dependencies}) and discuss here only the high-level intuition. To prove that $ \Gamma_1, \Gamma_2 $ have the same set of $ Z $-dependencies for every $ Z \in \Pi $, we compare the $ Z $-dependencies of $ \Omega $ with those of $ \Gamma $, where $ \Omega \in \DisjOf{\Theta} $ is a disjunct such that $ \Omega \wedge \Lambda $ is satisfiable. After all, this is exactly what the covering algorithm essentially does in the second loop at Line~\ref{alg:cover:line:second_loop}: if $ \Omega $ has more $ Z $-dependencies than $ \Gamma $, then $ \Omega $ gets separated from $ \Gamma $. The fact that $ \Omega $ cannot have fewer $ Z $-dependencies (compared to $ \Gamma $) follows from the definition of $ \Theta $, which was constructed so that it has exactly the same $ Z $-dependencies as $ \Gamma $, meaning, in particular, that $ \Theta $ enforces all those $ Z $-dependencies present in $ \Gamma $. At the same time, $ \Omega $ cannot have more $ Z $-dependencies than $ \Gamma $ because otherwise the algorithm would have found that $ \Omega $ in the second loop at Line~\ref{alg:cover:line:second_loop} and would have ensured that no $ \Lambda \in \Psi $ agrees with $ \Omega $ on any model, which contradicts the assumption that $ \Omega \wedge \Lambda $ is satisfiable.

We now prove statement \ref{lemma:psi_properties:c}. Let $ \Omega \in \DisjOf{\Theta} $ be a disjunct such that $ \Omega \wedge \Lambda $ is satisfiable. By property \ref{lemma:psi_properties:a} and the fact that $ \Gamma \models \Theta $ holds by Lemma~\ref{lemma:gamma_entails_theta} (see Appendix~\ref{sec:app:cover_analysis}), $ \Gamma \cup \Theta \cup \Lambda $ is satisfiable. Hence, we can apply statement \ref{lemma:psi_properties:b} to $ \Gamma_1 := \Gamma \cup \Theta $ and $ \Gamma_2 := \Omega $, which gives us that \begin{align}
	\label{eqn:lindep_piz_gammauniontheta_eq_lindep_piz_omega}
	\LinDep_{\pi_Z}(\Gamma \cup \Theta) = \LinDep_{\pi_Z}(\Omega)
\end{align} holds for all $ Z \in \Pi $. Since $ \Gamma \models \Theta $, $ \Gamma \equiv \Gamma \cup \Theta $, so (\ref{eqn:lindep_piz_gammauniontheta_eq_lindep_piz_omega}) yields that \begin{align}
	\label{eqn:lindep_piz_gamma_eq_lindep_piz_omega}
	\LinDep_{\pi_Z}(\Gamma) = \LinDep_{\pi_Z}(\Omega)
\end{align} holds for all $ Z \in \Pi $. Since $ \Gamma $ is $ \Pi $-complex, it suffices to show that (\ref{eqn:lindep_piz_gamma_eq_lindep_piz_omega}) implies $ \Fixes(\Gamma) = \Fixes(\Omega) $. By (\ref{eqn:lindep_piz_gamma_eq_lindep_piz_omega}), proving $ \Fixes(\Gamma) = \Fixes(\Omega) $ reduces to showing that for any satisfiable predicate set $ \Upsilon $, $ Z \in \Pi $ and variable $ z_i \in Z $ \begin{align}
	\label{eqn:ezi_lindep_piz_upsilon_iff_zi_in_fixes_upsilon}
	\gen{e_{Z,i}} \in \LinDep_{\pi_Z}(\Upsilon) \Leftrightarrow z_i \in \Fixes(\Upsilon)
\end{align} where $ e_{Z,i} $ is the $ i $-th canonical basis vector in $ \Q^{\left|Z\right|} $. More precisely, we assume that $ z_i = x_i \in X $ if $ Z = X $ and $ z_i = y_i \in Y $ if $ Z = Y $. Indeed, for some fixed $ v \models \Upsilon $ it holds that \begin{align*}
	\gen{e_{Z,i}} \in \LinDep_{\pi_Z}(\Upsilon) &\Leftrightarrow \pi_Z(\ModelsOf(\Upsilon) - v) \subseteq \gen{e_{Z,i}}^\bot \\
	&\Leftrightarrow \pi_Z(\ModelsOf(\Upsilon) - v) \cdot e_{Z,i} = \{0\} \\
	&\Leftrightarrow \forall u \models \Upsilon : \pi_Z(u - v) \cdot e_{Z,i} = 0 \\
	&\Leftrightarrow \forall u \models \Upsilon : \pi_Z(u) \cdot e_{Z,i} = \pi_Z(v) \cdot e_{Z,i} \\
	&\Leftrightarrow z_i \in \Fixes(\Upsilon)
\end{align*} This shows (\ref{eqn:ezi_lindep_piz_upsilon_iff_zi_in_fixes_upsilon}) and thus completes the proof of property \ref{lemma:psi_properties:c}. This concludes the proof of Lemma~\ref{lemma:psi_properties}.

\subsubsection{Model flooding preliminaries}

For the proof that $ \Gamma $ causes model flooding into every $ \Lambda \in \Psi $, we will also need the following Lemma~\ref{lemma:predicate_convexity}, Theorem~\ref{thm:union_of_same_lindep_has_that_lindep} and Lemma~\ref{lemma:adding_neq_predicates_cannot_make_pi_complex_set_pi_simple}.

\begin{lemma}
	\label{lemma:predicate_convexity}
	Let $ \Lambda $ be a predicate set and $ p $ be an equality predicate. Then $ \Lambda \cup \{p\} $ is satisfiable if $ \Lambda \cup \{p^<\} $ and $ \Lambda \cup \{p^>\} $ are both satisfiable.
\end{lemma}
\begin{proof}
	Let $ v_< \models \Lambda \cup \{p^<\} $, $ v_> \models \Lambda \cup \{p^>\} $ be models and consider the convex hull of $ \{v_<, v_>\} $, that is, the line $ L := (v_<, v_>) $ between $ v_< $ and $ v_> $. Since $ \ModelsOf(\Lambda) $ is a convex set and $ v_<, v_> \in \ModelsOf(\Lambda) $, it follows that $ L \subseteq \ModelsOf(\Lambda) $. Now observe that $ p $ defines a hyperplane (in affine space), whereas $ p^< $ and $ p^> $ define the open half-spaces of that hyperplane. 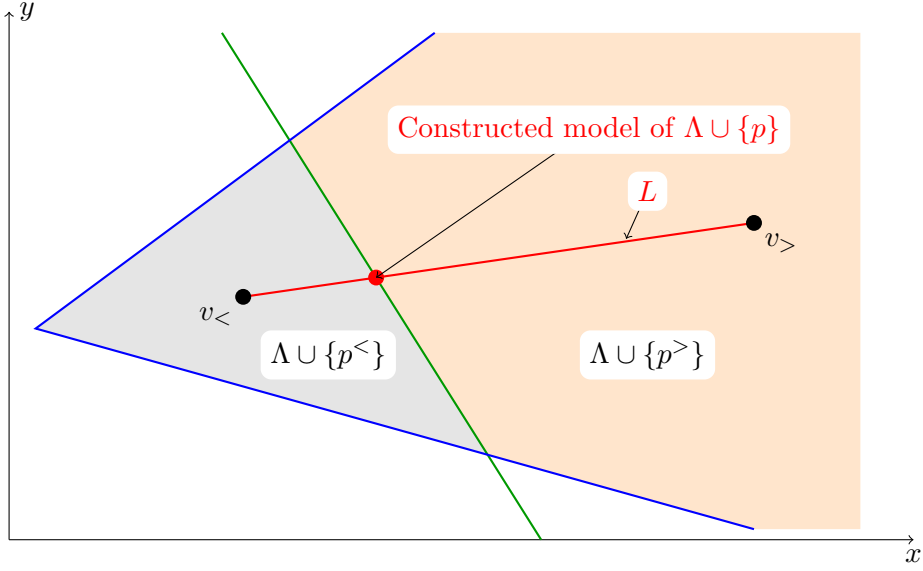
\begin{figure}[t]
		\centering
		\begin{tikzpicture}[scale=\zbtikzbiggerscaling]
	\draw[->] (0,0) -- (8.5,0) coordinate[label = {below:$x$}] (xmax);
	\draw[->] (0,0) -- (0,5) coordinate[label = {right:$y$}] (ymax);
	
	\coordinate (mainlinetl) at (2, 4.8); 
	\coordinate (mainlinebr) at (5, 0); 
	
	\coordinate (trtr) at (4, 4.8); 
	\coordinate (trbl) at (0.25, 2); 
	\coordinate (trbr) at (7, 0.1); 
	
	\coordinate (ltdot) at (2.2, 2.3); 
	\coordinate (gtdot) at (7, 3); 
	
	\coordinate (intup) at (intersection of mainlinetl--mainlinebr and trtr--trbl);
	\coordinate (intdown) at (intersection of mainlinetl--mainlinebr and trbr--trbl);
	\fill[fill=gray,opacity=0.2]
	(trbl) -- (intup) -- (intdown) -- cycle;
	
	\coordinate (regtr) at (8, 4.8);
	\coordinate (regbr) at (8, 0.1);
	
	\fill[fill=orange,opacity=0.2]
	(trtr) -- (regtr) -- (regbr) -- (trbr) -- (intdown) -- (intup) -- cycle;
	
	\draw[thick, OliveGreen] (mainlinetl) -- (mainlinebr);
	
	\draw[thick, blue] (trtr) -- (trbl) -- (trbr);
	
	\draw[red, thick] (ltdot) -- (gtdot);
	
	\fill[black] (ltdot) circle (0.075);
	\fill[black] (gtdot) circle (0.075);
	
	\coordinate (lineint) at (intersection of ltdot--gtdot and mainlinetl--mainlinebr);
	\fill[red] (lineint) circle (0.075);
	
	
	\node[anchor=north east,thick] at (ltdot) {$ v_< $};
	\node[anchor=north west,thick] at (gtdot) {$ v_> $};
	
	\node[rounded corners, fill=white] at (3, 1.75) {$ \Lambda \cup \{p^<\} $};
	\node[rounded corners, fill=white] at (6, 1.75) {$ \Lambda \cup \{p^>\} $};
	
	\coordinate (llabel) at (6, 3.3);
	\draw[->] (llabel) -- (5.8, 2.85);
	\node[rounded corners, red, fill=white] at (llabel) {$ L $};
	
	\coordinate (intlabel) at ($ (lineint) + (2, 1.4) $);
	\draw[->] (intlabel) -- (lineint);
	\node[rounded corners, red, fill=white] at (intlabel) {Constructed model of $ \Lambda \cup \{p\} $};
\end{tikzpicture}
		\caption{Two-dimensional geometric intuition behind the convexity argument in the proof of Lemma~\ref{lemma:predicate_convexity}. The line defined by the $ p $ predicate is colored green. The gray and orange regions together with their common border correspond to $ \ModelsOf(\Lambda) $. The half-space defined by $ p^< $ (resp. $ p^> $) is to the left (resp. right) of the green line. Blue lines depict possible borders defined by constraints in $ \Lambda $.}
		\label{fig:predicate_convexity}
	\end{figure} Since $ L $ is a line connecting two points, each of them lying in different half-spaces, $ L $ must intersect with the hyperplane defined by $ p $ (see Figure~\ref{fig:predicate_convexity}). Translating from geometry to logic, this means that $ L \cap \ModelsOf(p) \neq \varnothing $. Combining this with $ L \subseteq \ModelsOf(\Lambda) $ yields that $ \Lambda \cup \{p\} $ is satisfiable.
\end{proof}

Intuitively, the following Uniform union Theorem~\ref{thm:union_of_same_lindep_has_that_lindep} establishes that a disjunction of predicate sets having pairwise identical linear dependencies cannot be equivalent to a predicate set having a distinct set of linear dependencies. In other words, no unexpected linear dependencies can suddenly arise or vanish when considering a union of sets defined by predicate sets.

\begin{theorem}[Uniform union]
	\label{thm:union_of_same_lindep_has_that_lindep}
	Let $ \Lambda, \Lambda_1, \dots, \Lambda_k $ be satisfiable predicate sets such that $ \Lambda \equiv \bigvee_{i = 1}^k \Lambda_i $, $ V $ be a $ \Q $-vectorspace and $ h : \Q^n \rightarrow V $ be a $ \Q $-vectorspace homomorphism. Then, \[
	\LinDep_{h}(\Lambda_i) = \LinDep_{h}(\Lambda_j)
	\] for all $ i, j \in \{1, \dots, k\} $ implies \[
	\LinDep_{h}(\Lambda) = \LinDep_{h}(\Lambda_m)
	\] for all $ m \in \{1, \dots, k\} $.
\end{theorem}
\begin{proof}
	See Section~\ref{sec:app:lindep_facts:proof_uniform_union_theorem} of Appendix~\ref{sec:app:lindep_facts}.
\end{proof}

In essence, the following Lemma~\ref{lemma:adding_neq_predicates_cannot_make_pi_complex_set_pi_simple} simply states that adding inequality predicates cannot make a predicate set $ \Pi $-simple in a nontrivial manner, i.e., without simply making the set unsatisfiable.

\begin{lemma}
	\label{lemma:adding_neq_predicates_cannot_make_pi_complex_set_pi_simple}
	Let $ \Gamma $ be a $ \Pi $-complex predicate set and $ p $ be an inequality predicate such that $ \Gamma \cup \{p\} $ is satisfiable. Then $ \Gamma \cup \{p\} $ is $ \Pi $-complex.
\end{lemma}
\begin{proof}
	See Section~\ref{sec:app:lindep_facts:proof_adding_neq_predicates_cannot_make_pi_complex_set_pi_simple} of Appendix~\ref{sec:app:lindep_facts}.
\end{proof}

\subsubsection{Model flooding}

Having done the above preparatory work, we now prove the covering algorithm correct by showing that $ \Gamma $ causes model flooding into every $ \Lambda \in \Psi $. We formalize this in the following Theorem~\ref{thm:model_flooding}, which captures the model flooding we visualized using green arrows in Figure~\ref{fig:theta_disjuncts_intersecting_lambda} above.

\begin{theorem}[Model flooding]
	\label{thm:model_flooding}
	Let $ \varphi \in \QFLRA $ be a formula, $ \Pi $ be a binary partition, $ \Gamma \in \Sat(\DisjPhi) $ be a $ \Pi $-complex predicate set, $ \psi := \cover(\Pi, \Gamma) $ and $ \Psi $ be the set of predicate sets corresponding to the DNF terms of $ \psi $. Then every $ \Lambda \in \Psi $ entails $ \varphi $ if $ \varphi $ is $ \Pi $-decomposable.
\end{theorem}

We prove Theorem~\ref{thm:model_flooding}. Let $ p_1, \dots, p_k $ be the elements of \begin{align}
	\label{eqn:gamma_disresp_eq_p1_upto_pk}
	\PiComp{\Gamma} = \{p_1, \dots, p_k\}
\end{align} and fix an arbitrary $ \Lambda \in \Psi $. Let furthermore $ \Theta $ the be corresponding set computed in the call to the covering algorithm which produced $ \Lambda $. Suppose $ \varphi $ were $ \Pi $-decomposable. Since $ \Gamma $ is $ \Pi $-complex and $ \Gamma \models \varphi $, it suffices to show that for any $ \Gamma_1, \Gamma_2 \in \DisjOf{\Lambda} $, $ \Gamma_1 \models \varphi $ is equivalent to $ \Gamma_2 \models \varphi $. To be able to prove this by induction on the number of predicates on which $ \Gamma_1 $ and $ \Gamma_2 $ agree, we parametrize the above statement and show the following Claim~\ref{claim:model_flooding} instead (note that setting $ m := k+1 $ makes both statements identical).

\begin{zbclaim}
	\label{claim:model_flooding}
	For every $ m \in \{1, \dots, k + 1\} $ and $ Q = (Q_m, \dots, Q_k) \in \mathbb{P}^{k-m+1} $, if \[
		\Gamma_1, \Gamma_2 \in \DisjOf{\Lambda \cup \{p_m^{Q_m}, \dots, p_k^{Q_k}\}}
	\] are disjuncts, then $ \Gamma_1 \models \varphi $ is equivalent to $ \Gamma_2 \models \varphi $.
\end{zbclaim}

\noindent We now prove Claim~\ref{claim:model_flooding} by induction on $ m $, starting with $ m = 1 $ and going up to $ m = k + 1 $.

\textbf{Base case}: Let $ m = 1 $ and $ Q = (Q_1, \dots, Q_k) \in \mathbb{P}^k $ be arbitrary but fixed. Due to the initialization of $ \Theta $ at Line~\ref{alg:cover:line:predresp} of the covering algorithm, \[
	\PiSimp{\Gamma} \subseteq \Theta \subseteq \Lambda \subseteq \Gamma_i
\] holds for all $ i \in \{1, 2\} $. Combining this with (\ref{eqn:gamma_disresp_eq_p1_upto_pk}) and \[
	\{p_1^{Q_1}, \dots, p_k^{Q_k}\} \subseteq \Gamma_i
\] yields \[
	\DisjTrue \ni \PiSimp{\Gamma} \cup \{p_1^{Q_1}, \dots, p_k^{Q_k}\} \subseteq \Gamma_i
\] for all $ i \in \{1, 2\} $. Hence, $ \Gamma_1 $ agrees with $ \Gamma_2 $ on the disjunct from $ \DisjPhi $ which $ \Gamma_1 $ and $ \Gamma_2 $ both contain. Since $ \Gamma_1 $ and $ \Gamma_2 $ agree on all remaining elements, \[
	\Gamma_1, \Gamma_2 \in \DisjOf{\Lambda \cup \{p_1^{Q_1}, \dots, p_k^{Q_k}\}} = \{\Lambda \cup \{p_1^{Q_1}, \dots, p_k^{Q_k}\}\}
\] Hence, $ \Gamma_1 = \Gamma_2 $, so it becomes trivial that $ \Gamma_1 \models \varphi $ is equivalent to $ \Gamma_2 \models \varphi $.

\textbf{Inductive step}: Let $ Q = (Q_{m+1}, \dots, Q_k) \in \mathbb{P}^{k-m} $ and \[
	\Gamma_1, \Gamma_2 \in \DisjOf{\Lambda \cup \{p_{m+1}^{Q_{m+1}}, \dots, p_k^{Q_k}\}}
\] be disjuncts. For convenience, we define \[
	\Lambda_S := \Lambda \cup \{p_m^S, p_{m+1}^{Q_{m+1}}, \dots, p_k^{Q_k}\}
\] for every $ S \in \mathbb{P} $. Let $ S_1, S_2 \in \mathbb{P} $ be (unique) predicate symbols such that $ \Gamma_1 \in \DisjOf{\Lambda_{S_1}} $ and $ \Gamma_2 \in \DisjOf{\Lambda_{S_2}} $. \begin{figure}[t]
	\centering
	\input{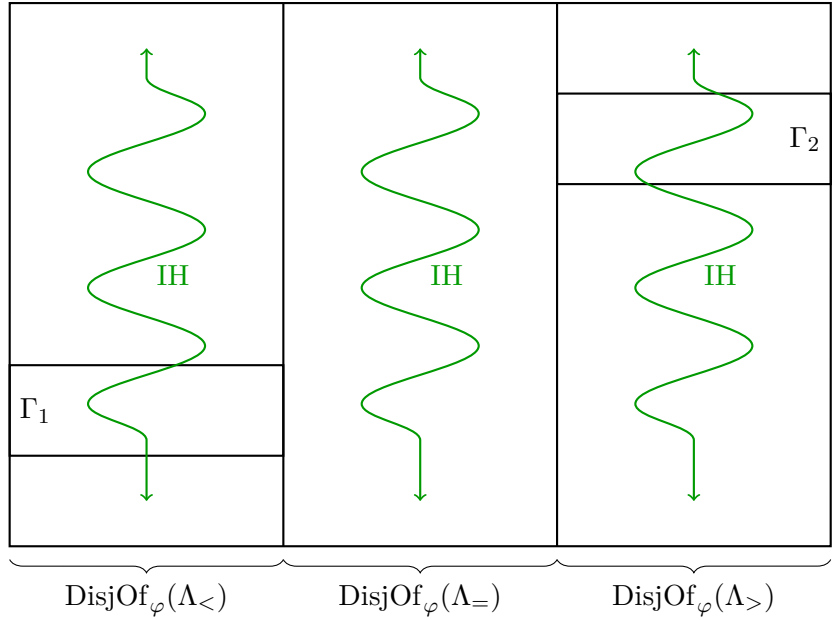}
	\caption{Venn diagram illustrating the way we partition $ \DisjOf{\Lambda \cup \{p_{m+1}^{Q_{m+1}}, \dots, p_k^{Q_k}\}} $ into groups of disjuncts. The regions labeled $ \Gamma_1 $ and $ \Gamma_2 $ in the diagram are to be interpreted as elements (i.e., not subsets) of $ \DisjOf{\Lambda_<} $ and $ \DisjOf{\Lambda_>} $, respectively. The green arrows illustrate the model flooding happening by the induction hypothesis (``IH''). More precisely, for all $ Q \in \mathbb{P} $, the disjuncts of $ \Lambda_Q $ either all entail $ \varphi $, or all entail $ \neg\varphi $. The disjuncts $ \Gamma_1 $ and $ \Gamma_2 $ are visualized for the case when $ p_m^< \in \Gamma_1 $ and $ p_m^> \in \Gamma_2 $. Thus, in this example $ S_1 = P_< $ and $ S_2 = P_> $; observe also that $ \Gamma_1 \models \Lambda_< $ and $ \Gamma_2 \models \Lambda_> $.}
	\label{fig:model_flooding_inductive_step}
\end{figure} A way of thinking about the induction hypothesis is visualized in Figure~\ref{fig:model_flooding_inductive_step}.

\begin{zbclaim}
	\label{claim:lambdas1_and_lambdas2_both_picomplex}
	$ \Lambda_{S_1} $ and $ \Lambda_{S_2} $ are both $ \Pi $-complex.
\end{zbclaim}
\begin{proof}
	Note that $ \Gamma_1 $ and $ \Gamma_2 $ can each be written as a union of a predicate set without equalities and a disjunct of $ \Theta $, which is $ \Pi $-complex by Lemma~\ref{lemma:psi_properties}~\ref{lemma:psi_properties:c}. Since adding strict inequalities cannot make a predicate set $ \Pi $-simple if the resulting set is satisfiable (Lemma~\ref{lemma:adding_neq_predicates_cannot_make_pi_complex_set_pi_simple}), applying this lemma inductively yields that $ \Gamma_1 $ and $ \Gamma_2 $ are both $ \Pi $-complex. Since $ \Gamma_1 \models \Lambda_{S_1} $ and $ \Gamma_2 \models \Lambda_{S_2} $, the claim follows.
\end{proof}

We now show that, in essence, it is easy to handle the case when $ \Lambda_= $ is unsatisfiable.

\begin{zbclaim}
	\label{claim:gammaeq_unsat_implies_s1_eq_s2}
	If $ \Lambda_= $ is unsatisfiable, then $ S_1 = S_2 $.
\end{zbclaim}
\begin{proof}
	Suppose $ \Lambda_= $ were unsatisfiable. First observe that $ S_1 \neq P_= $ and $ S_2 \neq P_= $. This is because $ S_i = P_= $ would imply $ \Gamma_i \models \Lambda_{S_i} \equiv \Lambda_= $ which, given that $ \Gamma_i $ is satisfiable, contradicts the unsatisfiability of $ \Lambda_= $. Now suppose $ S_1 \neq S_2 $ or, in other words, $ \{S_1, S_2\} = \{<, >\} $. Hence, the satisfiability of both $ \Lambda_{S_1} $ and $ \Lambda_{S_2} $ implies, by Lemma~\ref{lemma:predicate_convexity}, that $ \Lambda_= $ is satisfiable, which is a contradiction.
\end{proof}

\noindent Note that if $ \Lambda_= $ is unsatisfiable, then by the above Claim~\ref{claim:gammaeq_unsat_implies_s1_eq_s2} we have \[
\Gamma_1, \Gamma_2 \in \DisjOf{\Lambda_{S_1}} = \DisjOf{\Lambda_{S_2}}
\] and hence Claim~\ref{claim:model_flooding} holds by the induction hypothesis.

Thus, in the remainder of the proof, we can without loss of generality focus only on the case when $ \Lambda_= $ is satisfiable. It follows that $ \Lambda_= $ must have at least one satisfiable disjunct that must be $ \Pi $-complex by Lemma~\ref{lemma:psi_properties} \ref{lemma:psi_properties:c}. This is because any satisfiable disjunct of $ \Lambda_= $ implies (due to $ \Theta \subseteq \Lambda \subseteq \Lambda_= $) the existence of a disjunct of $ \Theta $, which is satisfiable when taken in conjunction with $ \Lambda $, and adding strict inequalities cannot make a predicate set $ \Pi $-simple (Lemma~\ref{lemma:adding_neq_predicates_cannot_make_pi_complex_set_pi_simple}). Hence, $ \Lambda_= $ is $ \Pi $-complex.

At this point, the overall idea is to apply the Overspilling Theorem~\ref{thm:overspilling} to the disjuncts of $ \Lambda_S $ visualized as columns in Figure~\ref{fig:model_flooding_inductive_step}. For this we need to establish that the $ \Lambda_S $ predicate sets have the same $ Z $-dependencies for all $ Z \in \Pi $ and $ S \in \mathbb{P} $ (for which $ \Lambda_S $ is satisfiable), which is precisely what we now focus on. Intuitively, the Overspilling Theorem~\ref{thm:overspilling} taken together with the induction hypothesis implies that a green arrow illustrating model flooding can also be drawn between the columns in Figure~\ref{fig:model_flooding_inductive_step}. In the remainder of the proof, we rigorously verify this intuition.

\begin{zbclaim}
	\label{claim:lindep_lambdaeq_eq_lindep_lambdasi}
	For every $ i \in \{1, 2\} $, it holds that \[
		\LinDep_{\pi_Z}(\Lambda_=) = \LinDep_{\pi_Z}(\Lambda_{S_i})
	\] where $ Z \in \Pi $ is arbitrary.
\end{zbclaim}
\begin{proof}
	Let $ \Omega_= \in \DisjOf{\Lambda_=} $ and $ \Omega_{S_i} \in \DisjOf{\Lambda_{S_i}} $ be arbitrary. Note that these disjuncts can be written as \begin{align*}
		\Omega_= &= \Omega'_= \cup (\Omega_= \setminus \Omega'_=) \\
		\Omega_{S_i} &= \Omega'_{S_i} \cup (\Omega_{S_i} \setminus \Omega'_{S_i})
	\end{align*} for some $ \Omega'_=, \Omega'_{S_i} \in \DisjOf{\Theta} $. Since $ \Omega_= \setminus \Omega'_= $ and $ \Omega_{S_i} \setminus \Omega'_{S_i} $ contain only strict inequality predicates, applying Lemma~\ref{lemma:psi_properties} \ref{lemma:psi_properties:b} and Theorem~\ref{thm:only_equality_predicates_can_establish_lindep_strong} (see Appendix~\ref{sec:app:lindep_facts}) yields that \[
		\LinDep_{\pi_Z}(\Omega_=) \overset{\ref{thm:only_equality_predicates_can_establish_lindep_strong}}{=} \LinDep_{\pi_Z}(\Omega'_=) \overset{\ref{lemma:psi_properties} \ref{lemma:psi_properties:b}}{=} \LinDep_{\pi_Z}(\Omega'_{S_i}) \overset{\ref{thm:only_equality_predicates_can_establish_lindep_strong}}{=} \LinDep_{\pi_Z}(\Omega_{S_i})
	\] holds for all $ Z \in \Pi $. Since $ \Omega_= $ and $ \Omega_{S_i} $ were chosen arbitrarily, it follows that the disjuncts of $ \Lambda_= $ and $ \Lambda_{S_i} $ all have the same set of $ Z $-dependencies. Hence, applying the Uniform union Theorem~\ref{thm:union_of_same_lindep_has_that_lindep} to both $ \Lambda_= \equiv \bigvee \DisjOf{\Lambda_=} $ and $ \Lambda_{S_i} \equiv \bigvee \DisjOf{\Lambda_{S_i}} $ yields the claim.
\end{proof}

Having done this preparatory work, we can now apply the Overspilling Theorem~\ref{thm:overspilling} to prove the following central claim establishing that ``model flooding'' between the columns in Figure~\ref{fig:model_flooding_inductive_step} indeed occurs. In other words, we now show the correctness of green arrows inside the circle corresponding to $ \Lambda $ in Figure~\ref{fig:theta_disjuncts_intersecting_lambda}.

\begin{zbclaim}
	\label{claim:model_flooding_overspilling}
	For every $ i \in \{1, 2\} $, $ \Lambda_{S_i} \models \varphi $ is equivalent to $ \Lambda_= \models \varphi $.
\end{zbclaim}
\begin{proof}
	We assume $ S_i \in \{<, >\} $, because the case when $ S_i = P_= $ is trivial. First, note that, as discussed above, $ \Lambda_{S_i} $ and $ \Lambda_= $ are both $ \Pi $-complex (see, in particular, Claim~\ref{claim:lambdas1_and_lambdas2_both_picomplex}). Next, observe that $ \Lambda_= $ is $ p_m^= $-next to $ \Lambda_{S_i} $ by definition. Furthermore, as we have shown in Claim~\ref{claim:lindep_lambdaeq_eq_lindep_lambdasi} above, both predicate sets have the same set of $ Z $-dependencies. Hence, all conditions of the Overspilling Theorem~\ref{thm:overspilling} (except $ \Theta_1 \models \varphi $) are satisfied. We consider the following two cases separately.
	
	\textbf{Case 1.} Suppose $ \varphi \wedge \Lambda_= $ were satisfiable. This means that at least one disjunct of $ \Lambda_= $ cannot entail $ \neg\varphi $, so by Lemma~\ref{lemma:every_disj_either_true_or_false}, $ \Lambda_= $ must have a disjunct entailing $ \varphi $. Since, by the induction hypothesis, all disjuncts of $ \Lambda_= $ entail $ \varphi $ if and only if at least one of them does so, it follows that $ \Lambda_= \models \varphi $. Hence, by the Overspilling Theorem~\ref{thm:overspilling}, $ \varphi \wedge \Lambda_{S_i} $ is satisfiable. We again argue that there must exist a disjunct of $ \Lambda_{S_i} $ which does not entail $ \neg\varphi $. Consequently, by Lemma~\ref{lemma:every_disj_either_true_or_false}, that disjunct must entail $ \varphi $. By the induction hypothesis, it follows that every disjunct of $ \Lambda_{S_i} $ must entail $ \varphi $ and hence $ \Lambda_{S_i} \models \varphi $. To sum up, we have shown the following implication chain. \begin{align}
		\label{eqn:phi_thetaeq_sat_imp_chain2}
		\varphi \wedge \Lambda_= \text{ is satisfiable} \Rightarrow \Lambda_= \models \varphi \Rightarrow \Lambda_{S_i} \models \varphi
	\end{align}
	
	\textbf{Case 2.} Suppose $ \neg\varphi \wedge \Lambda_= $ were satisfiable. Since the set of $ \Pi $-decomposable formulas is closed under negation and the definition of $ \DisjTrue $ remains equivalent when negating $ \varphi $, the same reasoning as for Case 1 above, but applied to $ \neg\varphi $ instead of $ \varphi $, yields the \begin{align}
		\label{eqn:negphi_thetaeq_sat_imp_chain2}
		\neg\varphi \wedge \Lambda_= \text{ is satisfiable} \Rightarrow \Lambda_= \models \neg\varphi \Rightarrow \Lambda_{S_i} \models \neg\varphi
	\end{align} implication chain.
	
	Since $ \Lambda_= $ is satisfiable, we are guaranteed to land in one of the two cases above. Hence, by (\ref{eqn:phi_thetaeq_sat_imp_chain2}) and (\ref{eqn:negphi_thetaeq_sat_imp_chain2}), we have $ \Lambda_= \models \varphi $ and $ \Lambda_S \models \varphi $, or it holds that $ \Lambda_= \models \neg\varphi $ and $ \Lambda_S \models \neg\varphi $. We conclude that in both cases $ \Lambda_{S_i} \models \varphi $ is equivalent to $ \Lambda_= \models \varphi $.
\end{proof}

Finally, note that by the induction hypothesis (``IH''), all satisfiable disjuncts of $ \Lambda_{S_i} $ entail $ \varphi $ if and only if at least one of them does so (as visualized in Figure~\ref{fig:model_flooding_inductive_step} using green arrows). In particular, $ \Gamma_i \models \varphi $ holds if and only if $ \Lambda_{S_i} \models \varphi $. Combining this with Claim~\ref{claim:model_flooding_overspilling} yields \begin{align*}
	\Gamma_1 \models \varphi \overset{\mathrm{IH}}{\Leftrightarrow} \Lambda_{S_1} \models \varphi \overset{\ref{claim:model_flooding_overspilling}}{\Leftrightarrow} \Lambda_= \models \varphi \overset{\ref{claim:model_flooding_overspilling}}{\Leftrightarrow} \Lambda_{S_2} \models \varphi \overset{\mathrm{IH}}{\Leftrightarrow} \Gamma_2 \models \varphi
\end{align*} \begin{figure}[t]
	\centering
	\input{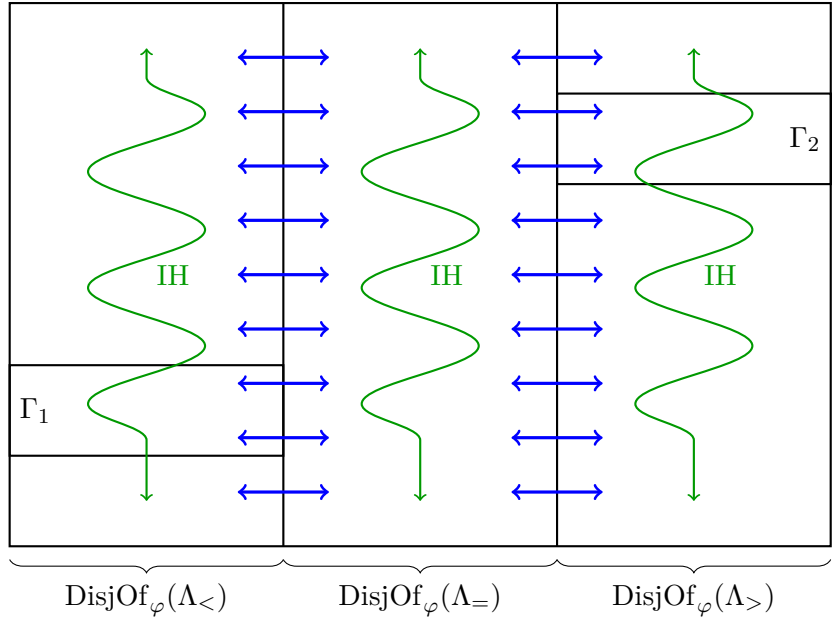}
	\caption{Figure~\ref{fig:model_flooding_inductive_step} extended with additional blue arrows illustrating the model flooding we have proven in the inductive step.}
	\label{fig:model_flooding_inductive_step_done}
\end{figure} These equivalences establish the desired model flooding (see Figure~\ref{fig:model_flooding_inductive_step_done}). This completes the proof of Claim~\ref{claim:model_flooding} and hence of Theorem~\ref{thm:model_flooding}. Thus, the formulas produced by the covering algorithm are indeed $ (\varphi, \Pi) $-MFF, so we have shown that the covering algorithm correctly solves the covering problem:

\begin{theorem}[Correctness of the covering algorithm]
	\label{thm:cover_alg_correctness}
	For a formula $ \varphi \in \QFLRA $, a binary partition $ \Pi $ and a predicate set $ \Gamma \in \DisjPhi $, the covering algorithm runs in time double-exponential in the size of $ \varphi $, and any formula $ \psi := \cover(\Pi, \Gamma) $ produced by it correctly solves the covering problem.
\end{theorem}

Combining Theorem~\ref{thm:cover_alg_correctness} with Theorem~\ref{thm:reduction} reducing the variable decomposition problem to the covering problem (see Section~\ref{sec:vardec:reduction} above) yields the following result.

\begin{theorem}
	\label{thm:vardec_double_exponential}
	Over quantifier-free linear real arithmetic, the variable decomposition problem (for binary partitions) can be solved in time double-exponential in the size of the given formula $ \varphi $. In particular, if $ \Pi $ is a binary partition, then it follows that any $ \Pi $-decomposable formula $ \varphi \in \QFLRA $ has a $ \Pi $-decomposition in DNF of double-exponential size.
\end{theorem}
\begin{proof}
	We solve the variable decomposition problem via the reduction of Theorem~\ref{thm:reduction} discussed in Section~\ref{sec:vardec:reduction}. That is, we use the covering algorithm to compute $ \psi_\Gamma := \cover(\Pi, \Gamma) $ for every $ \Gamma \in \Sat(\DisjPhi) $ whereupon the algorithm outputs that $ \varphi $ is $ \Pi $-decomposable whenever $ \psi_\Gamma \models \varphi $ holds for all $ \Gamma \in \Sat(\DisjPhi) $. The correctness of this procedure immediately follows from Theorems \ref{thm:cover_alg_correctness} and \ref{thm:reduction}, so it remains to prove that the running time of the obtained algorithm is indeed at most double-exponential in the size of $ \varphi $. We have already shown that the covering algorithm runs in double-exponential time (see Section~\ref{sec:vardec:cover:termination_time_complexity} and Theorem~\ref{thm:cover_alg_correctness}), so proving the stated complexity bound reduces to showing that testing $ \psi_\Gamma \models \varphi $ is decidable in double-exponential time. Indeed, this can be done in the following way. \begin{enumerate}
		\item Compute a CNF representation of $ \varphi \equiv C_1 \wedge \dots \wedge C_k $. Here, every $ C_k $ is a clause, i.e., a disjunction of literals.
		\item Let the DNF representation of $ \psi_\Gamma $ be $ \psi_\Gamma \equiv \psi_1 \vee \dots \vee \psi_l $ (note that $ \psi_\Gamma = \cover(\Pi, \Gamma) $ is by construction already in DNF).
		\item Output that $ \psi_\Gamma \models \varphi $ whenever \[
			\gamma_{i, j} := \psi_i \wedge \neg C_j
		\] is unsatisfiable for all $ i \in \{1, \dots, l\} $ and $ j \in \{1, \dots, k\} $.
	\end{enumerate} We now prove that this method for testing $ \psi_\Gamma \models \varphi $ is correct: \begin{align*}
		\psi_\Gamma \models \varphi &\Leftrightarrow \forall i \in \{1, \dots, l\} : \psi_i \models \varphi \\
		&\Leftrightarrow \forall i \in \{1, \dots, l\} : \psi_i \wedge \neg\varphi\text{ is unsatisfiable} \\
		&\Leftrightarrow \forall i \in \{1, \dots, l\} : \psi_i \wedge (\neg C_1 \vee \dots \vee \neg C_k) \text{ is unsatisfiable} \\
		&\Leftrightarrow \forall i \in \{1, \dots, l\} : (\psi_i \wedge \neg C_1) \vee \dots \vee (\psi_i \wedge \neg C_k)\text{ is unsatisfiable} \\
		&\Leftrightarrow \forall i \in \{1, \dots, l\}, j \in \{1, \dots, k\} : \gamma_{i, j}\text{ is unsatisfiable}
	\end{align*} Regarding the running time, the first step takes $ O(2^{\poly(n)}) $ time. The second step is trivial because $ \psi_\Gamma $ is already a DNF formula. As regards the third step, note that since $ C_j $ is a disjunction of literals, $ \neg C_j $ can be rewritten as a conjunction of literals (via De Morgan laws). Hence, the satisfiability of $ \gamma_{i, j} $ can be decided in time polynomial in the size of $ \gamma_{i, j} $, that is, in $ O(2^{2^{\poly(n)}}) $ time, where $ n $ is the size of $ \varphi $. This is because conjunctions of linear real arithmetic predicates are linear programs, for which there are well-known polynomial-time algorithms based on interior-point and ellipsoid methods (see Karmarkar's algorithm \cite[Chapter 15]{karmarkar:1984, schrijver:1998:theoryofintegerandlinearprogramming} and Khachiyan's algorithm \cite[Theorem 4.18]{khachiyan:1979, korte:2012:combinatorialoptimization}). More precisely, these results talk about linear programs where only non-strict inequalities are allowed, but the stated algorithms can be adapted to support strict inequalities via, for example, the reduction given in \cite[pp. 217, 218]{bradley:2007}. Hence, the time complexity of the third step is \[
		l \cdot k \cdot O(2^{2^{\poly(n)}}) \subseteq O(2^{2^{\poly(n)}}) \cdot O(2^{\poly(n)}) \cdot O(2^{2^{\poly(n)}}) = O(2^{2^{\poly(n)}})
	\] We conclude that the described algorithm for testing $ \psi_\Gamma \models \varphi $ runs in time double-exponential in the size of $ \varphi $.
\end{proof}

\subsubsection{Minimality of decompositions}

One important feature of the covering algorithm is that it produces coverings which are in a certain sense best-possible and minimal. This ensures that if $ \psi_\Gamma $ is the solution to the covering problem for $ \Gamma \in \Sat(\DisjPhi) $ and partition $ \Pi $ produced by the covering algorithm, then $ \bigvee_{\Gamma \in \Sat(\DisjPhi)} \psi_\Gamma $ is the best-possible approximation of a $ \Pi $-decomposition of $ \varphi $ in the sense that it must entail any other $ \Pi $-decomposition of any superset of $ \ModelsOf(\varphi) $ definable using only predicates appearing in $ \varphi $, or, to be more precise, as a disjunction of some subset of $ \DisjTrue $. As we demonstrated in Example~\ref{example:partial_variable_independence_application}, this allows us to use algorithms for the variable decomposition problem to speed up query evaluation even if $ \varphi $ is not $ \Pi $-decomposable. This minimality property follows from the fact that our proof of Theorem~\ref{thm:model_flooding} actually establishes a stronger model-flooding effect formulated in the following Theorem~\ref{thm:model_flooding_strong}.

\begin{theorem}[Strong model flooding]
	\label{thm:model_flooding_strong}
	Let $ \varphi \in \QFLRA $ be a formula, $ \Pi $ be a binary partition, $ \Gamma \in \Sat(\DisjPhi) $ be a $ \Pi $-complex predicate set, $ \psi := \cover(\Pi, \Gamma) $, $ \Psi $ be the set of predicate sets corresponding to the DNF terms of $ \psi $ and $ \gamma $ be a disjunction of the form $ \gamma = \bigvee_{\Upsilon \in S} \Upsilon $ for some $ S \subseteq \DisjTrue $ such that $ \Gamma \in S $. Then every $ \Lambda \in \Psi $ entails $ \gamma $ if $ \gamma $ is $ \Pi $-decomposable.
\end{theorem}

\begin{corollary}[Minimality of decompositions]
	\label{cor:minimality_of_dec_for_double_exponential_algo}
	Let $ \varphi \in \QFLRA $ be a formula and $ \Pi $ be a binary partition. For every $ \Gamma \in \Sat(\DisjPhi) $, let $ \psi_\Gamma := \cover(\Pi, \Gamma) $ denote the covering of $ \Gamma $. Let $ \gamma $ be a $ \Pi $-decomposable disjunction of the form $ \gamma = \bigvee_{\Upsilon \in S} \Upsilon $ for some $ S \subseteq \DisjTrue $ such that $ \varphi \models \gamma $. Then \[
		\bigvee_{\Gamma \in \Sat(\DisjPhi)} \psi_\Gamma \models \gamma
	\] In words, the decomposition $ \psi := \bigvee_{\Gamma \in \Sat(\DisjPhi)} \psi_\Gamma $ produced by the variable decomposition algorithm of Theorem~\ref{thm:vardec_double_exponential} is minimal in the sense that $ \psi \models \gamma $.
\end{corollary}
\begin{proof}
	For every $ \Gamma \in \Sat(\DisjPhi) $, let $ \Psi_\Gamma $ denote the set of predicate sets corresponding to the DNF terms of $ \psi_\Gamma $. By Theorem~\ref{thm:model_flooding_strong}, $ \Lambda \models \gamma $ holds for all $ \Gamma \in \Sat(\DisjPhi) $ and $ \Lambda \in \Psi_\Gamma $. Hence, \[
		\bigvee_{\Gamma \in \Sat(\DisjPhi)} \psi_\Gamma \equiv \bigvee_{\Gamma \in \Sat(\DisjPhi)} \bigvee_{\Lambda \in \Psi_\Gamma} \Lambda \models \gamma
	\]
\end{proof}

\subsubsection{Example}

\begin{zbexample}
	\label{example:reduction_ex_continued}
	Consider the formula $ \varphi $ and partition $ \Pi $ from Example~\ref{example:reduction_ex}. We analyze the $ \Pi $-decomposability of $ \varphi $ via the algorithm of Theorem~\ref{thm:vardec_double_exponential}. Let \[
		\Gamma := \{x + y < 2, x - y < 0\} \in \Sat(\DisjPhi)
	\] be the first disjunct to be covered. We execute $ \cover(\Pi, \Gamma) $. Since $ \Gamma $ is $ \Pi $-complex, we directly start with the construction of $ \Theta $, which in this case is $ \Theta = \varnothing $ because $ \Gamma $ contains only strict inequality predicates. This corresponds to a covering of $ \Theta \equiv \top $ that is not a $ (\varphi, \Pi) $-MFF because $ \varphi $ is $ \Pi $-decomposable but is not a tautology. In order to obtain a correct covering, we, therefore, start with the analysis of $ \DisjOf{\Theta} $ (see the second \texttt{foreach} loop at Line~\ref{alg:cover:line:second_loop} of the algorithm). As already discussed in Section~\ref{sec:vardec:high_level_overview}, the goal of this analysis is to rule out from the resulting covering those $ \Omega \in \DisjOf{\Theta} $ which have more $ Z $-dependencies than $ \Gamma $. In the present example, it is not hard to see that \[
		\Omega := \{x + y = 2, x - y = 0\} \in \DisjOf{\Theta}
	\] is the only disjunct of $ \Theta $ having strictly more $ X $- or $ Y $-dependencies. \begin{figure}[t]
		\centering
		\hfill\begin{subfigure}{0.3\textwidth}
	\centering
	\begin{tikzpicture}[scale=0.68]
		\node[blue,anchor=west,thick] at (1.5, 1.5) {$x+y<2$};
		\node[red,anchor=west,thick] at (1.5, 0.5) {$x-y<0$};
		
		
		\coordinate (topleft) at (-2, 3);
		\coordinate (topright) at (3, 3);
		\coordinate (botleft) at (-2, -1);
		\coordinate (botright) at (3, -1);
		
		\draw[->] (-2,0) -- (3,0) coordinate[label = {below:$x$}] (xmax);
		\draw[->] (0,-1) -- (0,3.2) coordinate[label = {right:$y$}] (ymax);
		
		\coordinate (redtop) at (-1, 3);
		\coordinate (redbot) at (3, -1);
		
		\draw[thick, dashed, red] (redtop) -- (redbot);
		
		\fill[fill=red,opacity=0.2] (topleft) -- (redtop) -- (redbot) -- (botleft) -- cycle;
		
		\coordinate (bluebot) at (-1, -1);
		\coordinate (bluetop) at (3, 3);
		
		\draw[thick, dashed, blue] (bluebot) -- (bluetop);
		\fill[fill=blue,opacity=0.2] (topleft) -- (bluetop) -- (bluebot) -- (botleft) -- cycle;
		
		\node[fill=white, rounded corners] (gamma) at (-1, 1) {$ \Gamma $};
	\end{tikzpicture}
	\caption{\label{fig:reduction_ex_continued:geom}}
\end{subfigure} \hfill
\begin{subfigure}{0.3\textwidth}
	\centering
	\begin{tikzpicture}[scale=0.75]
		\coordinate (redtop) at (-1, 3);
		\coordinate (redbot) at (3, -1);
		
		\coordinate (bluebot) at (-1, -1);
		\coordinate (bluetop) at (3, 3);
		
		\coordinate (intpoint) at (1, 1);
		
		\draw[thick, dashed] (redtop) -- (redbot); 
		\draw[thick, dashed] (bluetop) -- (bluebot);
		
		
		\filldraw[ultra thin,red,pattern=north east lines,pattern color=red, draw opacity=0.7] (redtop) -- (intpoint) -- (bluebot) -- cycle;
		\node[fill=white, rounded corners] (gamma) at (-0.5,1) {$ \Gamma $};
		
	\end{tikzpicture}
	\caption{\label{fig:reduction_ex_continued:gamma}}
\end{subfigure} \hfill
\begin{subfigure}{0.3\textwidth}
	\centering
	\begin{tikzpicture}[scale=0.75]
		\coordinate (redtop) at (-1, 3);
		\coordinate (redbot) at (3, -1);
		
		\coordinate (bluebot) at (-1, -1);
		\coordinate (bluetop) at (3, 3);
		
		\coordinate (mfdltop) at (-1, 3);
		\coordinate (mfdlbot) at (-1, -1);
		\coordinate (mfdrtop) at (1, 3);
		\coordinate (mfdrbot) at (1, -1);
		
		\draw[thick, dashed] (redtop) -- (redbot); 
		\draw[thick, dashed] (bluetop) -- (bluebot);
		
		\filldraw[ultra thin,red,pattern=north east lines,pattern color=red, draw opacity=0.7] (mfdltop) -- (mfdrtop) -- (mfdrbot) -- (mfdlbot) -- cycle;
		
		\node[fill=white, rounded corners] at (-0.25,1) {$ \psi_\Gamma $};
	\end{tikzpicture}
	\caption{\label{fig:reduction_ex_continued:covering}}
\end{subfigure}\hfill
		\caption{Figure~\ref{fig:reduction_ex_continued:geom} is a visualization of the sets defined by $ x + y < 2 $ and $ x - y < 0 $. The set of models of $ \Gamma $ is visualized in Figure~\ref{fig:reduction_ex_continued:gamma}. Figure~\ref{fig:reduction_ex_continued:covering} depicts the covering $ \psi_\Gamma = x < 1 $ for $ \Gamma $.}
		\label{fig:reduction_ex_continued}
	\end{figure} Indeed, the only solution to $ \Omega $ is $ b := (1, 1)^\transp $ (see Figure~\ref{fig:reduction_ex_continued}); the $ B $ basis computed at Line~\ref{alg:cover:line:second_loop:dep_basis} spans the zero vectorspace. Hence, \[
		\ImageOf(\pi_Z \circ \lc_B)^\bot \cap \ImageOf(\pi_Z \circ \lc_A) = \Q \cap \Q = \Q
	\] holds for all $ Z \in \Pi $. Consequently, we need to choose a concrete $ Z \in \Pi $ and synthesize a separating predicate with respect to that $ Z $. Choosing $ Z := X := \{x\} $ yields $ w = 1 $ and the separating predicate \[
		\underbrace{\pi_X(\vec{z}) \cdot w}_{x} \neq \underbrace{\pi_X(b) \cdot w}_{1}
	\] Since $ \Gamma \wedge x = 1 $ is unsatisfiable (as can immediately be observed geometrically, see Figure~\ref{fig:reduction_ex_continued}), the recursive call at Line~\ref{alg:cover:line:second_loop:rec_call} immediately returns the covering $ \bot $ that does not affect the covering of $ \Gamma $ we are constructing. Overall, the second loop of the covering algorithm produces $ \Upsilon = \{x \neq 1\} $. Rewriting the inequalities in $ \Upsilon $ yields $ D = \{x < 1, x > 1\} $ (see Line~\ref{alg:cover:line:d_set}). Since $ \Gamma \wedge x < 1 $ is satisfiable, but $ \Gamma \wedge x > 1 $ is not, the algorithm outputs $ \psi_\Gamma := x < 1 $ as the resulting covering of $ \Gamma $. As specified in the reduction of Theorem~\ref{thm:reduction}, we check whether $ \psi_\Gamma \models \varphi $. In the present example, the entailment holds, so we proceed with the computation of coverings for other disjuncts.
	
	The coverings of other disjuncts $ \Gamma' \in \Sat(\DisjPhi) $ can be computed in a very similar way; once they are computed, it remains to check that they all entail $ \varphi $. In the present example this is the case, so we conclude that $ \varphi $ is $ \Pi $-decomposable by Theorem~\ref{thm:reduction}; a possible $ \Pi $-decomposition of $ \varphi $ we can obtain via the covering algorithm is \[
		x < 1 \vee x > 1 \vee y < 1 \vee y > 1
	\]
\end{zbexample}

\subsection{Example}
\label{sec:vardec:example}

We give a more detailed example illustrating how the algorithm of Theorem~\ref{thm:vardec_double_exponential} can be used to analyze the $ \Pi $-decomposability of formulas and to construct $ \Pi $-decompositions whenever they exist. Let $ X := \{x_1, x_2, x_3\} $, $ Y := \{y_1, y_2, y_3\} $ and $ \Pi := \{X, Y\} $. Define predicates \begin{align*}
	p_1 &:= 12x_1 - 2x_2 + 16x_3 + 28y_1 + 35y_2 - 7y_3 = 0 \\
	p_2 &:= 6x_1 - x_2 + 8x_3 + 16y_1 + 20y_2 - 4y_3 = 0 \\
	p_3 &:= 4x_1 - x_2 + 4x_3 + y_1 + y_2 = 1 \\
	p_4 &:= 6x_1 - 3x_2 + 2y_1 + y_2 + y_3 = 0 \\
	p_5 &:= 3x_1 - x_2 + 2x_3 + 2y_1 + y_2 + y_3 = 0
\end{align*} and consider the formula \[
\varphi := p_1^= \wedge p_2^= \wedge (p_3^< \vee p_4^>) \wedge (p_3^= \rightarrow p_5^< \vee p_5^>)
\] We analyze the $ \Pi $-decomposability of $ \varphi $. First, note that every disjunct in $ \DisjTrue $ can by addressed by providing the five predicate symbols to be put instead of equalities in $ p_1, p_2, p_3, p_4, p_5 $. In order to make our notation for addressing disjuncts simple and succinct, it makes sense to define \[
\varphi[P_1P_2P_3P_4P_5] := \{p_1^{P_1}, p_2^{P_2}, p_3^{P_3}, p_4^{P_4}, p_5^{P_5}\} \in \DisjTrue
\] for all $ P_1, \dots, P_5 \in \{=, <, >\} $. We now start computing coverings of various disjuncts in $ \Sat(\DisjPhi) $. \begin{figure}[t]
	\centering
	\includegraphics[width=\textwidth]{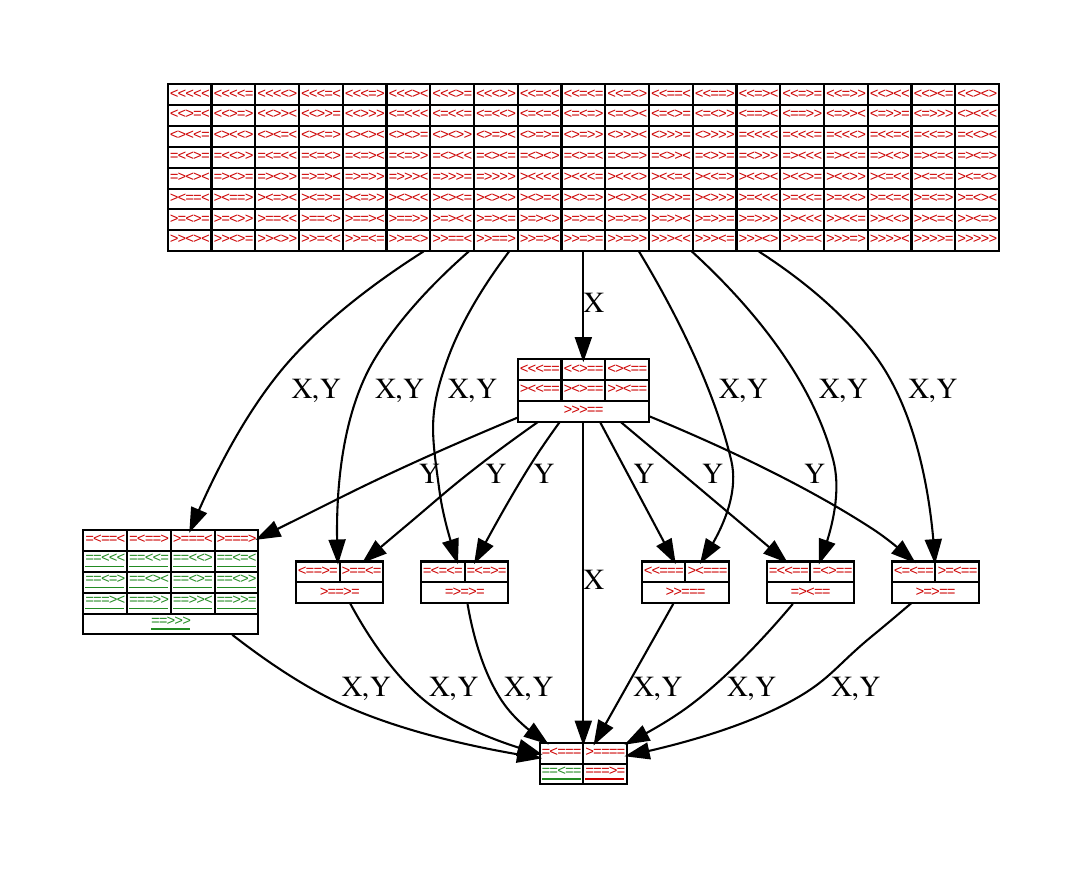}
	\caption{Illustration of $ \Sat(\DisjTrue) $. Every disjunct $ \Gamma \in \Sat(\DisjTrue) $ is identified by five predicate symbols $ P_1, \dots, P_2 $ such that $ \Gamma = \varphi[P_1P_2P_3P_4P_5] $. Disjuncts entailing $ \varphi $, that is, elements of $ \DisjPhi $, are colored green, whereas predicate sets in $ \DisjNegPhi $ have red labels. Some disjuncts are grouped into a single block of disjuncts, so that for $ Z \in \Pi $ there exists a path from $ \Gamma $ to $ \Gamma' $ with all edges labeled $ Z $ if and only if $ \Gamma' $ has more $ Z $-dependencies than $ \Gamma $.}
	\label{fig:cover_ex:disj_all}
\end{figure} The set of all disjuncts $ \Sat(\DisjPhi) $ and the relationships between their sets of $ X $- and $ Y $-dependencies are illustrated in Figure~\ref{fig:cover_ex:disj_all}.

\subsubsection{First covering}

We choose \[
\Gamma_1 := \varphi[==>><] \in \Sat(\DisjPhi)
\] as the first disjunct to be covered and execute the covering algorithm on $ \Gamma_1 $ while connecting every step with the corresponding line of Algorithm~\ref{alg:cover}.

\paragraph{Preliminary check.} We first check whether $ \Gamma_1 $ is $ \Pi $-simple (Line~\ref{alg:cover:line:pi_simple_check}) because then it would be possible to immediately return a correct covering. Unfortunately, $ \Gamma_1 $ is $ \Pi $-complex.

\paragraph{Computation of $ \Theta $.} Next, our goal is to synthesize a set of predicates $ \Theta $ enforcing all $ X $- and $ Y $-dependencies of $ \Gamma_1 $. As a first step, simply removing $ \Pi $-disrespecting predicates from $ \Gamma_1 $ yields $ \Theta := \varnothing $. Next, we solve the system of equations $ \Gamma_1^= = \{p_1, p_2\} $ by writing the coefficients of $ p_1 $ and $ p_2 $ in the rows of a matrix and computing its kernel: \begin{align}
	\label{eqn:hard_ex:gamma_1_eq_sol}
	\ker \begin{pNiceMatrix}[first-row, extra-margin=2pt, code-for-first-row=\scriptscriptstyle, columns-width=1.25em]
		x_1 & x_2 & x_3 & y_1 & y_2 & y_3 \\
		12 & -2 & 16 & 28 & 35 & -7 \\
		6 & -1 & 8 & 16 & 20 & -4
	\end{pNiceMatrix} = \ImageOf \begin{pNiceMatrix}[first-col, extra-margin=2pt, code-for-first-col=\scriptscriptstyle, columns-width=1.25em, vlines]
		x_1 & 1 & -4 &  0 &  0 \\
		x_2 & 6 &  0 &  0 &  0 \\
		x_3 & 0 &  3 &  0 &  0 \\
		y_1 & 0 &  0 & -5 &  1 \\
		y_2 & 0 &  0 &  4 &  0 \\
		y_3 & 0 &  0 &  0 &  4
		\CodeAfter
		\UnderBrace[yshift=3pt]{1-1}{last-1}{\scriptstyle a_1}
		\UnderBrace[yshift=3pt]{1-2}{last-2}{\scriptstyle a_2}
		\UnderBrace[yshift=3pt]{1-3}{last-3}{\scriptstyle a_3}
		\UnderBrace[yshift=3pt]{1-4}{last-4}{\scriptstyle a_4}
		\tikz \node [highlight = (1-1) (3-1)] {};
		\tikz \node [highlight = (1-2) (3-2)] {};
		\tikz \node [highlightblue = (4-3) (6-3)] {};
		\tikz \node [highlightblue = (4-4) (6-4)] {};
	\end{pNiceMatrix}
\end{align}\[ \\[12pt] \] This step corresponds to Line~\ref{alg:cover:line:gamma_lindep_repr}; in this case $ a := 0 $ and the basis $ A := (a_1, a_2, a_3, a_4) $ consists of the columns of the solution matrix in (\ref{eqn:hard_ex:gamma_1_eq_sol}). We now use the characterization (\ref{eqn:im_pizlca_orth_eq_ker_piz_a1al_transp}) to obtain a representation of the set of all $ X $- and $ Y $-dependencies of $ \Gamma_1 $: \begin{align*}
	\ImageOf(\pi_X \circ \lc_A)^\bot &\overset{(\ref{eqn:im_pizlca_orth_eq_ker_piz_a1al_transp})}{=} \ker\big(\Pi_X \cdot a_1 \mid \Pi_X \cdot a_2 \mid \Pi_X \cdot a_3 \mid \Pi_X \cdot a_4\big)^\transp \\
	&= \ker \begin{pNiceMatrix}[first-row, extra-margin=2pt, code-for-first-row=\scriptscriptstyle, columns-width=1.5em]
		x_1 & x_2 & x_3 \\
		1 & 6 & 0 \\
		-4 & 0 & 3
		\CodeAfter
		\tikz \node [highlight = (1-1) (1-3)] {};
		\tikz \node [highlight = (2-1) (2-3)] {};
	\end{pNiceMatrix} = \ImageOf \begin{pNiceMatrix}[first-col, extra-margin=2pt, code-for-first-col=\scriptscriptstyle, columns-width=1.25em, vlines]
		x_1 &  6 \\
		x_2 & -1 \\
		x_3 &  8
		\CodeAfter
		\UnderBrace[yshift=3pt]{1-1}{last-1}{\scriptstyle v_1}
	\end{pNiceMatrix}
\end{align*}\[ \\[12pt] \]\begin{align*}
	\ImageOf(\pi_Y \circ \lc_A)^\bot &\overset{(\ref{eqn:im_pizlca_orth_eq_ker_piz_a1al_transp})}{=} \ker\big(\Pi_Y \cdot a_1 \mid \Pi_Y \cdot a_2 \mid \Pi_Y \cdot a_3 \mid \Pi_Y \cdot a_4\big)^\transp \\
	&= \ker \begin{pNiceMatrix}[first-row, extra-margin=2pt, code-for-first-row=\scriptscriptstyle, columns-width=1.5em]
		y_1 & y_2 & y_3 \\
		-5 & 4 & 0 \\
		1 & 0 & 4
		\CodeAfter
		\tikz \node [highlightblue = (1-1) (1-3)] {};
		\tikz \node [highlightblue = (2-1) (2-3)] {};
	\end{pNiceMatrix} = \ImageOf \begin{pNiceMatrix}[first-col, extra-margin=2pt, code-for-first-col=\scriptscriptstyle, columns-width=1.25em, vlines]
		y_1 &  -4 \\
		y_2 & -5 \\
		y_3 &  1
		\CodeAfter
		\UnderBrace[yshift=3pt]{1-1}{last-1}{\scriptstyle v_1}
	\end{pNiceMatrix}
\end{align*}\[ \\[12pt] \] Translating the vectors generating $ \ImageOf(\pi_X \circ \lc_A)^\bot $ and $ \ImageOf(\pi_Y \circ \lc_A)^\bot $ into predicates via the dot product construction (see Line~\ref{alg:cover:line:first_loop:theta_upd}) yields \[
\Theta = \{6x_1 - x_2 + 8x_3 = 0, -4y_1 - 5y_2 + y_3 = 0\}
\]

\subsubsection{First covering: analysis of $ \Theta $'s disjuncts}

We now analyze the disjuncts of $ \Theta $, i.e., those elements of $ \DisjTrue $ which share at least one model with $ \Theta $. Precisely those disjuncts are underlined in Figure~\ref{fig:cover_ex:disj_all}. As we have already discussed in Section~\ref{sec:vardec:high_level_overview}, the goal of this analysis is to rule out from the resulting covering those $ \Omega \in \DisjOf{\Theta} $ which have more $ Z $-dependencies than $ \Gamma_1 $. Indeed, since $ \Theta $ enforces all $ X $- and $ Y $-dependencies present in $ \Gamma_1 $, every its disjunct $ \Omega \in \DisjOf{\Theta} $ must also do so. Hence, if we rule out all $ \Omega \in \DisjOf{\Theta} $ having more $ Z $-dependencies than $ \Gamma_1 $, then all $ \Omega \in \DisjOf{\Theta} $ agree on the set of $ Z $-dependencies (see Lemma~\ref{lemma:psi_properties} \ref{lemma:psi_properties:b}), which, as we have seen in Section~\ref{sec:vardec:correctness}, is crucial for the resulting covering to be a $ (\varphi, \Pi) $-MFF.

Turning back to the present example, we study the \[
\Omega := \varphi[===>=] \cup \Theta \in \DisjOf{\Theta}
\] disjunct, because it has more both $ X $- and $ Y $-dependencies than $ \Gamma_1 $ (see Figure~\ref{fig:cover_ex:disj_all}, in particular, the edge between $ \Gamma_1 = \varphi[==>><] $ and $ \Omega = \varphi[===>=] $). Before we proceed with the computation of a separating predicate, we explain how exactly this can be determined computationally. Note that, in general, we are interested in all disjuncts that have more $ X $- or $ Y $-dependencies compared to $ \Gamma_1 $.

\paragraph{Solving $ \Omega^= $.} We consider the equalities in $ \Omega^= = \{p_1, p_2, p_3, p_5\} \cup \Theta $, write the corresponding coefficients in the rows of a matrix and compute its kernel: \begin{align}
	\label{eqn:hard_ex:gamma_1_omega_1_eq_sol}
	\ker \begin{pNiceMatrix}[first-row, extra-margin=2pt, code-for-first-row=\scriptscriptstyle, columns-width=1.5em]
		x_1  & x_2 & x_3 & y_1 & y_2 & y_3 \\
		12  &  -2 &  16 &  28 &  35 &  -7 \\
		6  &  -1 &   8 &  16 &  20 &  -4 \\
		4  &  -1 &   4 &   1 &   1 &   0 \\
		3  &  -1 &   2 &   2 &   1 &   1 \\
		6  &  -1 &   8 &   0 &   0 &   0 \\
		0  &   0 &   0 &  -4 &  -5 &   1
	\end{pNiceMatrix} = \ImageOf \begin{pNiceMatrix}[first-col, extra-margin=2pt, code-for-first-col=\scriptscriptstyle, columns-width=1.25em, vlines]
		x_1 & -2 & 0 \\
		x_2 & -4 & 0 \\
		x_3 & 1  & 0 \\
		y_1 & 0  & -1 \\
		y_2 & 0  & 1 \\
		y_3 & 0  & 1
		\CodeAfter
		\UnderBrace[yshift=3pt]{1-1}{last-1}{\scriptstyle b_1}
		\UnderBrace[yshift=3pt]{1-2}{last-2}{\scriptstyle b_2}
		\tikz \node [highlight = (1-1) (3-1)] {};
		\tikz \node [highlightblue = (4-2) (6-2)] {};
	\end{pNiceMatrix}
\end{align}\[ \\[12pt] \] Note that $ \Omega^= $ is an inhomogeneous system of linear equations, so the above image does not yet precisely equal $ \ModelsOf(\Omega^=) $; we have to figure out what the value of a correct affine offset is, which can most easily be achieved by just computing a model of $ \Omega^= $, such as \[
b := (-2/3, -4, 0, 0, -1/3, -5/3)^\transp
\] Hence, $ \ModelsOf(\Omega^=) = b + \gen{b_1, b_2} $, where $ b_1 $ and $ b_2 $ are the columns of the solution matrix in (\ref{eqn:hard_ex:gamma_1_omega_1_eq_sol}). These steps correspond to Line~\ref{alg:cover:line:second_loop:dep_basis} of the covering algorithm, where $ B = (b_1, b_2) $.

\paragraph{Comparing $ \Omega $'s $ X $- and $ Y $-dependencies with those of $ \Gamma_1 $.} Having computed the above solution to $ \Omega^= $, we can now compare the $ X $- and $ Y $-dependencies of $ \Omega $ with those of $ \Gamma_1 $. We use (\ref{eqn:im_pizlcb_orth_eq_ker_piz_b1br_transp}) to obtain a representation of the set of all $ X $- and $ Y $-dependencies of $ \Omega $: \begin{align}
	\label{eqn:hard_ex:gamma_1_omega_1_lindep_x}
	\ImageOf(\pi_X \circ \lc_B)^\bot &\overset{(\ref{eqn:im_pizlcb_orth_eq_ker_piz_b1br_transp})}{=} \ker\big(\Pi_X \cdot b_1 \mid \Pi_X \cdot b_2\big)^\transp \\
	\label{eqn:hard_ex:gamma_1_omega_1_lindep_x_cont}
	&= \ker \begin{pNiceMatrix}[first-row, extra-margin=2pt, code-for-first-row=\scriptscriptstyle, columns-width=1.5em]
		x_1 & x_2 & x_3 \\
		-2 & -4 & 1
		\CodeAfter
		\tikz \node [highlight = (1-1) (1-3)] {};
	\end{pNiceMatrix} = \ImageOf \begin{pNiceMatrix}[first-col, extra-margin=2pt, code-for-first-col=\scriptscriptstyle, columns-width=1.25em, vlines]
		x_1 & -2 & 1 \\
		x_2 &  1 & 0 \\
		x_3 &  0 & 2
	\end{pNiceMatrix} \\
	\label{eqn:hard_ex:gamma_1_omega_1_lindep_y}
	\ImageOf(\pi_Y \circ \lc_B)^\bot &\overset{(\ref{eqn:im_pizlcb_orth_eq_ker_piz_b1br_transp})}{=} \ker\big(\Pi_Y \cdot b_1 \mid \Pi_Y \cdot b_2\big)^\transp \\
	\label{eqn:hard_ex:gamma_1_omega_1_lindep_y_cont}
	&= \ker \begin{pNiceMatrix}[first-row, extra-margin=2pt, code-for-first-row=\scriptscriptstyle, columns-width=1.5em]
		y_1 & y_2 & y_3 \\
		-1 & 1 & 1
		\CodeAfter
		\tikz \node [highlightblue = (1-1) (1-3)] {};
	\end{pNiceMatrix} = \ImageOf \begin{pNiceMatrix}[first-col, extra-margin=2pt, code-for-first-col=\scriptscriptstyle, columns-width=1.25em, vlines]
		y_1 & 1 & 1 \\
		y_2 & 1 & 0 \\
		y_3 & 0 & 1
	\end{pNiceMatrix}
\end{align} Now, we can combine (\ref{eqn:hard_ex:gamma_1_eq_sol}) with (\ref{eqn:hard_ex:gamma_1_omega_1_lindep_x}--\ref{eqn:hard_ex:gamma_1_omega_1_lindep_x_cont}), (\ref{eqn:hard_ex:gamma_1_omega_1_lindep_y}--\ref{eqn:hard_ex:gamma_1_omega_1_lindep_y_cont}) to get the following basis of $ \ImageOf(\pi_X \circ \lc_B)^{\bot} \cap \ImageOf(\pi_X \circ \lc_A) $ (see Line~\ref{alg:cover:line:second_loop:main_if}). \begin{align*}
	\ImageOf(\pi_X \circ \lc_B)^{\bot} \cap \ImageOf(\pi_X \circ \lc_A) &= \ImageOf \begin{pNiceMatrix}[first-col, extra-margin=2pt, code-for-first-col=\scriptscriptstyle, columns-width=1.25em, vlines]
		x_1 & -2 & 1 \\
		x_2 &  1 & 0 \\
		x_3 &  0 & 2
	\end{pNiceMatrix} \cap \ImageOf \begin{pNiceMatrix}[first-col, extra-margin=2pt, code-for-first-col=\scriptscriptstyle, columns-width=1.25em, vlines]
		x_1 & 1 & -4 \\
		x_2 & 6 &  0 \\
		x_3 & 0 &  3
	\end{pNiceMatrix} \\
	&= \ImageOf \begin{pNiceMatrix}[first-col, extra-margin=2pt, code-for-first-col=\scriptscriptstyle, columns-width=1.25em, vlines]
		x_1 & -\frac{31}{22} \\
		x_2 & 1 \\
		x_3 & \frac{13}{11}
	\end{pNiceMatrix}
\end{align*} This computation shows that $ \Omega $ has more $ X $-dependencies than $ \Gamma_1 $. At this point, $ \Omega $ actually also has more $ Y $-dependencies than $ \Gamma_1 $, but we leave it as an exercise for a curious reader to verify this. After all, in order to synthesize a separating predicate, it suffices to find just a single $ Z \in \Pi $ such that $ \Omega $ has more $ Z $-dependencies than $ \Gamma_1 $.

\paragraph{Synthesis of a separating predicate.} We pick (see Line~\ref{alg:cover:line:second_loop:witness_vec}) \[
w := \begin{pNiceMatrix}[first-col, extra-margin=2pt, code-for-first-col=\scriptscriptstyle, columns-width=1.25em, vlines]
	x_1 & -\frac{31}{22} \\
	x_2 & 1 \\
	x_3 & \frac{13}{11}
\end{pNiceMatrix} \in \ImageOf(\pi_X \circ \lc_B)^{\bot} \cap \ImageOf(\pi_X \circ \lc_A)
\] to be the vector containing information about a certain $ X $-dependency which is present in $ \Omega $ but absent in $ \Gamma $. Thus, we can separate $ \Omega $ from $ \Gamma $ by adding \[
\pi_Z(\vec{z}) \cdot w \neq \pi_Z(b) \cdot w
\] to $ \Upsilon $ (see Line~\ref{alg:cover:line:second_loop:upsilon_upd}). By computing the dot products, we obtain \[
-\frac{31}{22}x_1 + x_2 + \frac{13}{11}x_3 \neq -\frac{101}{33}
\] It can be verified that this inequality indeed ensures that $ \Upsilon $ (and consequently the resulting covering) cannot agree with $ \Omega $ on any model because \[
\Omega \models -\frac{31}{22}x_1 + x_2 + \frac{13}{11}x_3 = -\frac{101}{33}
\] In this case, \[
	\Gamma_1 \cup \{-\frac{31}{22}x_1 + x_2 + \frac{13}{11}x_3 = -\frac{101}{33}\}
\] in unsatisfiable and consequently $ \Pi $-simple, so the recursive call at Line~\ref{alg:cover:line:second_loop:rec_call} immediately returns $ \bot $. At this point, we are going to slightly deviate from the pseudocode of the covering algorithm because the following simple observation can spare us the analysis of all other disjuncts of $ \Theta $. Note that \[
	\Gamma_1 \models -\frac{31}{22}x_1 + x_2 + \frac{13}{11}x_3 < -\frac{101}{33}
\] holds, meaning that \[
	\psi_{\Gamma_1} := \Theta \wedge -\frac{31}{22}x_1 + x_2 + \frac{13}{11}x_3 < -\frac{101}{33}
\] is a good candidate for being a correct covering of $ \Gamma_1 $ because $ \Gamma_1 \models \Theta $ taken together with the above entailment automatically ensures that $ \Gamma_1 \models \psi_{\Gamma_1} $. Hence, it only remains to check whether $ \psi_{\Gamma_1} \models \varphi $. Indeed, in this case we are lucky and $ \psi_{\Gamma_1} $ is a correct covering of $ \Gamma_1 $ because $ \psi_{\Gamma_1} \models \varphi $ holds. In general, one could integrate this test into the algorithm and simply resume the execution of the second loop at Line~\ref{alg:cover:line:second_loop} if we are not lucky and the covering obtained as just described does not entail $ \varphi $. We will further develop and implement this optimization idea in \zbrefsec{sec:experiments}.

\subsubsection{Second and third coverings}

In order to construct a $ \Pi $-decomposition for $ \varphi $, one would now proceed with the computation of a covering for some further disjunct in $ \Sat(\DisjPhi) $. For example, one could choose \[
	\Gamma_2 := \varphi[==>>>] \in \Sat(\DisjPhi)
\] to be the next disjunct to be covered. Since we have already shown how all major steps of the covering algorithm work on a concrete example, and the computations of coverings for the remaining disjuncts do not involve any fundamentally new ideas, we omit the details. It is worth noting, however, that for the analysis of the present formula $ \varphi $ it suffices to call the covering algorithm only five times, including recursive calls. This performance can be achieved by choosing only those disjuncts in $ \Sat(\DisjPhi) $ which do not entail the disjunction of the already computed coverings. For example, suppose $ \psi_{\Gamma_2} $ is a correct covering of $ \Gamma_2 $. Then, we need to execute the covering algorithm for a further disjunct $ \Gamma_3 \in \Sat(\DisjPhi) $ only if $ \Gamma_3 \not\models \psi_{\Gamma_1} \vee \psi_{\Gamma_2} $. In the present example, we could choose $ \Gamma_3 := \varphi[==<><] \in \Sat(\DisjPhi) $ and compute a covering $ \psi_{\Gamma_3} $ so that no further calls to the covering algorithm are necessary. We would then observe that $ \psi_{\Gamma_i} \models \varphi $ holds for all $ i \in \{1, 2, 3\} $, so by Theorem~\ref{thm:reduction} we would conclude that $ \varphi $ is $ \Pi $-decomposable, and a possible $ \Pi $-decomposition is $ \psi_{\Gamma_1} \vee \psi_{\Gamma_2} \vee \psi_{\Gamma_3} $.

\subsection{Exponential upper bound}
\label{sec:vardec:exponential_upper_bound}


Recall that in Section~\ref{sec:vardec:correctness} we have, in particular, proven that the covering algorithm of Section~\ref{sec:vardec:cover} runs in double-exponential time (Theorem~\ref{thm:cover_alg_correctness}). We then used this result to obtain a double-exponential time algorithm for the variable decomposition problem (Theorem~\ref{thm:vardec_double_exponential}). However, it turns out that this upper bound is not tight. In this section, we provide a more fine-grained analysis of the covering algorithm and show that it actually runs in exponential time if slightly refined. As a corollary, we obtain an exponential-time algorithm for the variable decomposition problem. This running time is optimal in the setting of deterministic algorithms because, as we will see in \zbrefsec{sec:lower_bounds}, deciding variable decomposability is $ \coNP $-hard.


\subsubsection{Tackling the double-exponential blowup}

Observe that almost all steps of the covering algorithm trivially run in exponential time unless they are recursive calls themselves. The only exception is the computation of the $ D $ set at Line~\ref{alg:cover:line:d_set}. In other words, the only actual source of the double-exponential blowup is the need to consider all possible combinations of ways to replace the inequalities in elements of $ \Upsilon $ by the $ P_< $ and $ P_> $ predicate symbols, when computing the set $ D $ as specified at Line~\ref{alg:cover:line:d_set}. Indeed, this enumeration of possibilities requires time exponential in the size of $ \Upsilon $, whereas at this point we only have an exponential upper bound for $ \left|\Upsilon\right| $ arising from the observation that the second loop at Line~\ref{alg:cover:line:second_loop} makes at most exponentially many iterations. We now give a high-level overview of the approach we use to avoid this double-exponential blowup. \begin{itemize}
	\item First, we prove the existence of an exponential-size set $ D $ satisfying $ \Upsilon \equiv \bigvee_{\Upsilon' \in D} \Upsilon' $ as specified at Line~\ref{alg:cover:line:d_set}. We achieve this result by constructing $ D $ in the obvious way and then deleting all unsatisfiable predicate sets therein; relying on discrete geometry, we prove an exponential upper bound for the size of the resulting set.
	\item Then, we show that it is possible to compute $ D $ in time polynomial in the value of the upper bound for a certain compressed, but still correct, version of $ D $. Consequently, we get an exponential-time algorithm for computing $ D $ and fully avoid the double-exponential blowup in the running time of the covering algorithm.
\end{itemize}

\subsubsection{Upper bound for the size of $ D $}

For the sake of convenience, until the end of Section~\ref{sec:vardec:exponential_upper_bound}, we will include inequalities of the form $ t_1 \neq t_2 $ into the notion of a predicate and use the superscript notation $ q^Q $ to denote the inequality $ q $ with the predicate symbol replaced by $ Q $. For example, $ (2x-z \neq y)^< = (2x - z < y) $. Let $ p_1, \dots, p_n $ be some linearization of $ \Upsilon = \{p_1, \dots, p_n\} $. Clearly, \[
	D := \{\{p_1^{P_1}, \dots, p_n^{P_n}\} \mid (P_1, \dots, P_n) \in \{<, >\}^n\}
\] is a correct definition of $ D $ within the meaning of Line~\ref{alg:cover:line:d_set}. Although the size of $ D $ is $ 2^n $, where $ n $ is exponential in the size of $ \varphi $, it turns out that the size of \begin{align}
	\label{eqn:d_set_compressed_def}
	D := \Sat(\{\{p_1^{P_1}, \dots, p_n^{P_n}\} \mid (P_1, \dots, P_n) \in \{<, >\}^n\})
\end{align} is actually polynomial in $ n $ and consequently exponential in the size of $ \varphi $, whereas this set $ D $ is obviously also correct within the meaning of Line~\ref{alg:cover:line:d_set}. To prove this upper bound, we will need the following Lemma~\ref{lemma:sat_reg_count_upper_bound}.

\begin{lemma}
	\label{lemma:sat_reg_count_upper_bound}
	Let $ n, d \in \Npos $, $ p_1, \dots, p_n $ be predicates each having their set of free variables among $ x_1, \dots, x_d $. Then \[
		\left|\Big\{(P_1, \dots, P_n) \in \{<, >\}^n \mid \{p_1^{P_1}, \dots, p_n^{P_n}\} \text{ is satisfiable}\Big\}\right| \le n^{d+1} + d + 1
	\]
\end{lemma}
\begin{proof}
	We start by briefly recalling some well-known concepts from discrete geometry. A \textit{hyperplane arrangement} $ \mathcal{A} $ is a finite set of affine hyperplanes in $ \R^d $ \cite{matousekjiri:2013, stanley:2004:hyperplane_arrangements, miller:2005:combinatorialcommalg}. A \textit{region} of a hyperplane arrangement $ \mathcal{A} $ is a maximal\footnote{With respect to inclusion.} convex subset of \[
	\R^d \setminus \bigcup_{H \in \mathcal{A}} H
	\] (see \cite[p. 126]{matousekjiri:2013}, \cite[p. 3]{stanley:2004:hyperplane_arrangements} and \cite[Chapter 18]{gruenbaum:2003}). 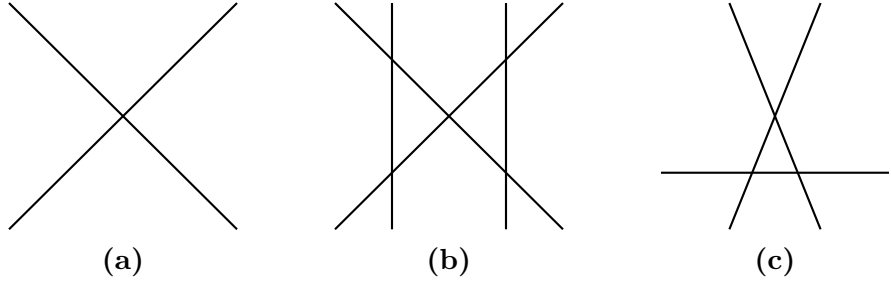
\begin{figure}[t]
		\centering
		\hfill\begin{subfigure}{0.3\textwidth}
	\centering
	\begin{tikzpicture}[scale=0.75]
		\coordinate (tl) at (-2, 2);
		\coordinate (tr) at (2, 2);
		\coordinate (bl) at (-2, -2);
		\coordinate (br) at (2, -2);
		
		\draw[thick] (tl) -- (br);
		\draw[thick] (bl) -- (tr);
	\end{tikzpicture}
	\caption{\label{fig:hyperplane_arrangement_examples:ex1}}
\end{subfigure} \hfill
\begin{subfigure}{0.3\textwidth}
	\centering
	\begin{tikzpicture}[scale=0.75]
		\coordinate (tl) at (-2, 2);
		\coordinate (tr) at (2, 2);
		\coordinate (bl) at (-2, -2);
		\coordinate (br) at (2, -2);
		
		\draw[thick] (tl) -- (br);
		\draw[thick] (bl) -- (tr);
		\draw[thick] (-1, 2) -- (-1, -2);
		\draw[thick] (1, 2) -- (1, -2);
	\end{tikzpicture}
	\caption{\label{fig:hyperplane_arrangement_examples:ex2}}
\end{subfigure} \hfill
\begin{subfigure}{0.3\textwidth}
	\centering
	\begin{tikzpicture}[scale=0.75]
		\coordinate (tl) at (-2, 2);
		\coordinate (tr) at (2, 2);
		\coordinate (bl) at (-2, -2);
		\coordinate (br) at (2, -2);
		
		\draw[thick] ($ (tl)!0.75!(bl) $) -- ($ (tr)!0.75!(br) $);
		\draw[thick] ($ (tl)!0.7!(tr) $) -- ($ (bl)!0.3!(br) $);
		\draw[thick] ($ (tl)!0.3!(tr) $) -- ($ (bl)!0.7!(br) $);
	\end{tikzpicture}
	\caption{\label{fig:hyperplane_arrangement_examples:ex3}}
\end{subfigure}\hfill
		\caption{Three examples of hyperplane arrangements in $ \R^2 $.}
		\label{fig:hyperplane_arrangement_examples}
	\end{figure} Three concrete examples of hyperplane arrangements are visualized in Figure~\ref{fig:hyperplane_arrangement_examples}. More precisely, the arrangements shown in Figures \ref{fig:hyperplane_arrangement_examples:ex1}, \ref{fig:hyperplane_arrangement_examples:ex2} and \ref{fig:hyperplane_arrangement_examples:ex3} have 4, 10 and 7 regions, respectively. It is well-known in discrete geometry that the maximum number of regions in an affine arrangement of size $ n $ is achieved when it is in the so-called \textit{general position}\footnote{Also known as \textit{generic position} \cite{edelsbrunner:1993:zonetheorem}.} \cite{orlik:1992:arrangementsofhyperplanes, matousekjiri:2013, stanley:2004:hyperplane_arrangements, edelsbrunner:1993:zonetheorem, miller:2005:combinatorialcommalg}. Since $ R := \{p_1^{P_1}, \dots, p_n^{P_n}\} $ defines a single region of some arrangement of $ n $ hyperplanes if $ R $ is satisfiable, using known formulas for the number of regions in general position (see \cite[Proposition 6.1.1]{matousekjiri:2013}, \cite[Proposition 2.4]{stanley:2004:hyperplane_arrangements} and \cite[p. 1]{orlik:1992:arrangementsofhyperplanes}) yields the \begin{align*}
		\left|\Big\{(P_1, \dots, P_n) \in \{<, >\}^n \mid \{p_1^{P_1}, \dots, p_n^{P_n}\} \text{ is satisfiable}\Big\}\right| \le \sum_{i=0}^d \binom{n}{i}
	\end{align*} upper bound. Further bounding \[
		\sum_{i=0}^d \binom{n}{i} \le \sum_{i=0}^d n^i \le n^{d+1} + d + 1
	\] by a geometric series yields the result.

\end{proof}

Let $ D $ be as defined in (\ref{eqn:d_set_compressed_def}), $ d $ be the number of free variables in $ \varphi $ and $ m $ be the size of $ \varphi $. By Lemma~\ref{lemma:sat_reg_count_upper_bound}, \begin{align}
	\label{eqn:d_size_le_onepluedenpowd}
	\left|D\right| \le n^{d+1} + d + 1
\end{align} Since $ \Upsilon $ is of size exponential in $ m $, there exists a constant $ c \in \Npos $ such that $ n \le 2^{m^c} + c $. Substituting this upper bound for $ n $ and the trivial upper bound $ d \le m $ in (\ref{eqn:d_size_le_onepluedenpowd}) yields \[
	\left|D\right| \le (2^{m^c} + c)^{m+1} + m + 1
\] Hence, we obtain the desired exponential upper bound \begin{align}
	\label{eqn:d_size_in_o_two_pow_polym}
	\left|D\right| \in O(2^{\poly(m)})
\end{align} for the size of $ D $.

\subsubsection{Designing an algorithm for computing $ D $ efficiently}
\label{sec:vardec:exponential_upper_bound:towards_algo_for_computing_d}

Having established an $ O(2^{\poly(m)}) $ upper bound for the size of $ D $ as defined in (\ref{eqn:d_set_compressed_def}), we now start working towards showing that, surprisingly, $ D $ can also be computed in $ O(2^{\poly(m)}) $ time, notwithstanding the fact that there are exponentially many predicates in $ \Upsilon $. Clearly, every element of the search space \begin{align*}
	S := \{\{p_1^{P_1}, \dots, p_n^{P_n}\} \mid (P_1, \dots, P_n) \in \{<, >\}^n\}
\end{align*} corresponds to a unique sequence $ (P_1, \dots, P_n) \in \{<, >\}^n $ of predicate symbols. Every such sequence can be thought of as the result of $ n $ \textit{decisions} to replace the predicate symbol of $ p_i $ by either $ P_< $ or $ P_> $, where $ i \in \{1, \dots, n\} $. 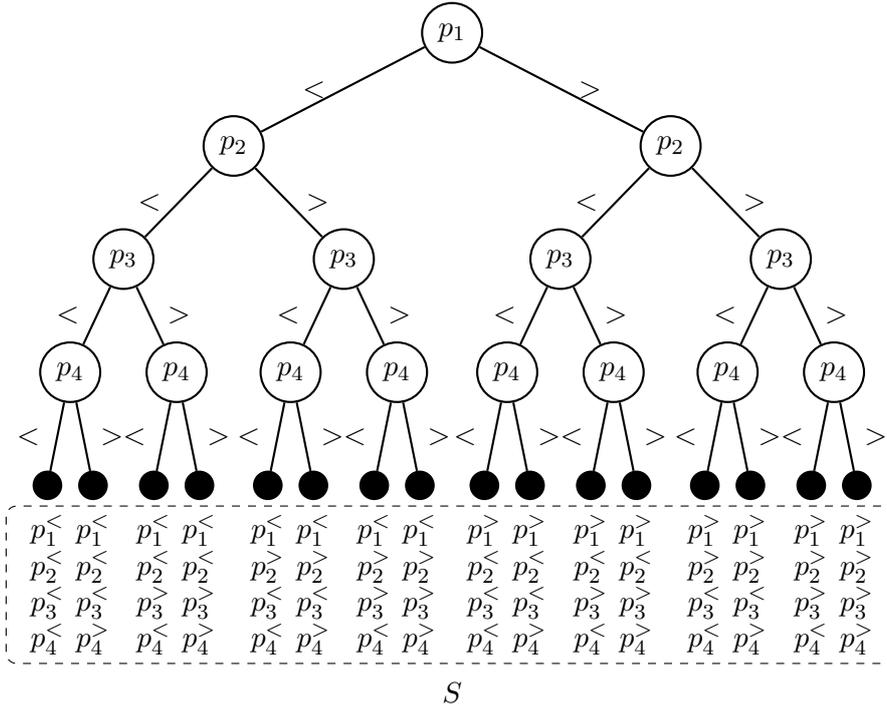
\begin{figure}[t]
	\centering
	\begin{tikzpicture}[
	every node/.style = {draw, circle, thick},
	arrdown/.style={thick},
	arrup/.style={thick},
	normaledge/.style={thick},
	leaf/.style={fill},
	configsat/.style={},
	configunsat/.style={},
	grow = down,
	level distance=1.5cm,
	level 1/.style={sibling distance=5.75cm},
	level 2/.style={sibling distance=2.9cm},
	level 3/.style={sibling distance=1.4cm},
	level 4/.style={sibling distance=0.6cm}]
	\node[configsat] (root) {$ p_1 $}
	child {node[configsat] {$ p_2 $}
		child {node[configunsat] {$ p_3 $}
			child {node[configunsat] {$ p_4 $}
				child {node[leaf, configunsat, label={[align=center, name=leftmostleaf]below:$ p_1^< $\\$ p_2^< $\\$ p_3^< $\\$ p_4^< $}] {} edge from parent [arrdown] node[left,draw=none] {$ < $}}
				child {node[leaf, configunsat, label={[align=center]below:$ p_1^< $\\$ p_2^< $\\$ p_3^< $\\$ p_4^> $}] {} edge from parent [arrdown] node[right,draw=none] {$ > $}}
				edge from parent [arrdown] node[left,draw=none] {$ < $}
			}
			child {node[configunsat] {$ p_4 $}
				child {node[leaf, configunsat, label={[align=center]below:$ p_1^< $\\$ p_2^< $\\$ p_3^> $\\$ p_4^< $}] {} edge from parent [arrdown] node[left,draw=none] {$ < $}}
				child {node[leaf, configunsat, label={[align=center]below:$ p_1^< $\\$ p_2^< $\\$ p_3^> $\\$ p_4^> $}] {} edge from parent [arrdown] node[right,draw=none] {$ > $}}
				edge from parent [arrdown] node[right,draw=none] {$ > $}
			}
			edge from parent [normaledge] node[left,draw=none] {$ < $}
		}
		child {node[configsat] {$ p_3 $}
			child {node[configunsat] {$ p_4 $}
				child {node[leaf, configunsat, label={[align=center]below:$ p_1^< $\\$ p_2^> $\\$ p_3^< $\\$ p_4^< $}] {} edge from parent [arrdown] node[left,draw=none] {$ < $}}
				child {node[leaf, configunsat, label={[align=center]below:$ p_1^< $\\$ p_2^> $\\$ p_3^< $\\$ p_4^> $}] {} edge from parent [arrdown] node[right,draw=none] {$ > $}}
				edge from parent [normaledge] node[left,draw=none] {$ < $}
			}
			child {node[configsat] {$ p_4 $}
				child {node[leaf, configsat, label={[align=center]below:$ p_1^< $\\$ p_2^> $\\$ p_3^> $\\$ p_4^< $}] {} edge from parent [arrup] node[left,draw=none] {$ < $}}
				child {node[leaf, configunsat, label={[align=center]below:$ p_1^< $\\$ p_2^> $\\$ p_3^> $\\$ p_4^> $}] {} edge from parent [normaledge] node[right,draw=none] {$ > $}}
				edge from parent [arrup] node[right,draw=none] {$ > $}
			}
			edge from parent [arrup] node[right,draw=none] {$ > $}
		}
		edge from parent [arrup] node[left,draw=none] {$ < $}
	}
	child {node[configsat] {$ p_2 $}
		child {node[configsat] {$ p_3 $}
			child {node[configsat] {$ p_4 $}
				child {node[leaf, configsat, label={[align=center]below:$ p_1^> $\\$ p_2^< $\\$ p_3^< $\\$ p_4^< $}] {} edge from parent [arrup] node[left,draw=none] {$ < $}}
				child {node[leaf, configunsat, label={[align=center]below:$ p_1^> $\\$ p_2^< $\\$ p_3^< $\\$ p_4^> $}] {} edge from parent [normaledge] node[right,draw=none] {$ > $}}
				edge from parent [arrup] node[left,draw=none] {$ < $}
			}
			child {node[configsat] {$ p_4 $}
				child {node[leaf, configunsat, label={[align=center]below:$ p_1^> $\\$ p_2^< $\\$ p_3^> $\\$ p_4^< $}] {} edge from parent [normaledge] node[left,draw=none] {$ < $}}
				child {node[leaf, configsat, label={[align=center]below:$ p_1^> $\\$ p_2^< $\\$ p_3^> $\\$ p_4^> $}] {} edge from parent [arrup] node[right,draw=none] {$ > $}}
				edge from parent [arrup] node[right,draw=none] {$ > $}
			}
			edge from parent [arrup] node[left,draw=none] {$ < $}
		}
		child {node[configsat] {$ p_3 $}
			child {node[configsat] {$ p_4 $}
				child {node[leaf, configsat, label={[align=center]below:$ p_1^> $\\$ p_2^> $\\$ p_3^< $\\$ p_4^< $}] {} edge from parent [arrup] node[left,draw=none] {$ < $}}
				child {node[leaf, configsat, label={[align=center]below:$ p_1^> $\\$ p_2^> $\\$ p_3^< $\\$ p_4^> $}] {} edge from parent [arrup] node[right,draw=none] {$ > $}}
				edge from parent [arrup] node[left,draw=none] {$ < $}
			}
			child {node[configunsat] {$ p_4 $}
				child {node[leaf, configunsat, label={[align=center]below:$ p_1^> $\\$ p_2^> $\\$ p_3^> $\\$ p_4^< $}] {} edge from parent [arrdown] node[left,draw=none] {$ < $}}
				child {node[leaf, configunsat, label={[align=center, name=rightmostleaf]below:$ p_1^> $\\$ p_2^> $\\$ p_3^> $\\$ p_4^> $}] {} edge from parent [arrdown] node[right,draw=none] {$ > $}}
				edge from parent [normaledge] node[right,draw=none] {$ > $}
			}
			edge from parent [arrup] node[right,draw=none] {$ > $}
		}
		edge from parent [arrup] node[right,draw=none] {$ > $}
	};
	\coordinate (leafrecttl) at ($ (leftmostleaf.north west) + (0.25, 0.25) $);
	\coordinate (leafrectbr) at ($ (rightmostleaf.south east) -(0.25, 0.25) $);
	\coordinate (leafrectbl) at (leafrectbr-|leafrecttl);
	\draw[dashed, rounded corners] (leafrecttl) rectangle (leafrectbr);
	\node[anchor=north, draw=none] at ($ (leafrectbl)!0.5!(leafrectbr) $) {$ S $};
	
\end{tikzpicture}
	\caption{Illustration of the decision tree $ T $ corresponding to $ S $ for $ n = 4 $. Every leaf of the tree has a column underneath it, containing elements of the corresponding set in $ S $.}
	\label{fig:d_disjunct_combination_space_tree}
\end{figure} We visualize the set of all possible sequences of such decisions using a tree $ T $ with nodes labeled $ p_1, \dots, p_n $ (see Figure~\ref{fig:d_disjunct_combination_space_tree}). The labeling of edges between a node labeled $ p_i $ and its children indicates the two possible decisions we can make about the predicate symbol by which the inequality in $ p_i \in \Upsilon $ is to be replaced. In other words and as illustrated in Figure~\ref{fig:d_disjunct_combination_space_tree}, every element of $ S $ corresponds to precisely one leaf of the tree $ T $.

\begin{figure}[t]
	\centering
	\begin{tikzpicture}[
	every node/.style = {draw, circle, thick},
	arrreddown/.style={thick, dashed, -latex},
	arrgreendown/.style={ultra thick, -latex},
	normaledge/.style={thick},
	leaf/.style={fill, solid},
	configsat/.style={color=OliveGreen},
	configunsat/.style={color=red},
	grow = down,
	level distance=1.5cm,
	level 1/.style={sibling distance=6.25cm},
	level 2/.style={sibling distance=3cm},
	level 3/.style={sibling distance=1.5cm},
	level 4/.style={sibling distance=0.6cm}]
	\node[configsat] (root) {$ p_1 $}
	child {node[configsat] {$ p_2 $}
		child {node[configunsat] {$ p_3 $}
			child {node[configunsat] {$ p_4 $}
				child {node[leaf, configunsat] {} edge from parent [arrreddown] node[left,draw=none] {$ < $}}
				child {node[leaf, configunsat] {} edge from parent [arrreddown] node[right,draw=none] {$ > $}}
				edge from parent [arrreddown] node[left,draw=none] {$ < $}
			}
			child {node[configunsat] {$ p_4 $}
				child {node[leaf, configunsat] {} edge from parent [arrreddown] node[left,draw=none] {$ < $}}
				child {node[leaf, configunsat] {} edge from parent [arrreddown] node[right,draw=none] {$ > $}}
				edge from parent [arrreddown] node[right,draw=none] {$ > $}
			}
			edge from parent [normaledge] node[left,draw=none] {$ < $}
		}
		child {node[configsat] {$ p_3 $}
			child {node[configunsat] {$ p_4 $}
				child {node[leaf, configunsat] {} edge from parent [arrreddown] node[left,draw=none] {$ < $}}
				child {node[leaf, configunsat] {} edge from parent [arrreddown] node[right,draw=none] {$ > $}}
				edge from parent [normaledge] node[left,draw=none] {$ < $}
			}
			child {node[configsat] {$ p_4 $}
				child {node[leaf, configsat] {} edge from parent [arrgreendown] node[left,draw=none] {$ < $}}
				child {node[leaf, configunsat] {} edge from parent [normaledge] node[right,draw=none] {$ > $}}
				edge from parent [arrgreendown] node[right,draw=none] {$ > $}
			}
			edge from parent [arrgreendown] node[right,draw=none] {$ > $}
		}
		edge from parent [arrgreendown] node[left,draw=none] {$ < $}
	}
	child {node[configsat] {$ p_2 $}
		child {node[configsat] {$ p_3 $}
			child {node[configsat] {$ p_4 $}
				child {node[leaf, configsat] {} edge from parent [arrgreendown] node[left,draw=none] {$ < $}}
				child {node[leaf, configunsat] {} edge from parent [normaledge] node[right,draw=none] {$ > $}}
				edge from parent [arrgreendown] node[left,draw=none] {$ < $}
			}
			child {node[configsat] {$ p_4 $}
				child {node[leaf, configunsat] {} edge from parent [normaledge] node[left,draw=none] {$ < $}}
				child {node[leaf, configsat] {} edge from parent [arrgreendown] node[right,draw=none] {$ > $}}
				edge from parent [arrgreendown] node[right,draw=none] {$ > $}
			}
			edge from parent [arrgreendown] node[left,draw=none] {$ < $}
		}
		child {node[configsat] {$ p_3 $}
			child {node[configsat] {$ p_4 $}
				child {node[leaf, configsat] {} edge from parent [arrgreendown] node[left,draw=none] {$ < $}}
				child {node[leaf, configsat] {} edge from parent [arrgreendown] node[right,draw=none] {$ > $}}
				edge from parent [arrgreendown] node[left,draw=none] {$ < $}
			}
			child {node[configunsat] {$ p_4 $}
				child {node[leaf, configunsat] {} edge from parent [arrreddown] node[left,draw=none] {$ < $}}
				child {node[leaf, configunsat] {} edge from parent [arrreddown] node[right,draw=none] {$ > $}}
				edge from parent [normaledge] node[right,draw=none] {$ > $}
			}
			edge from parent [arrgreendown] node[right,draw=none] {$ > $}
		}
		edge from parent [arrgreendown] node[right,draw=none] {$ > $}
	};
\end{tikzpicture}
	\caption{Illustration of the decision tree $ T $ corresponding to $ S $ for $ n = 4 $. The coloring depicts a possible way the satisfiability of predicates sets of the form $ \{p_1^{P_1}, \dots, p_i^{P_i}\} $ may behave, for various $ i $ and $ P_1, \dots, P_i \in \{<, >\} $. More precisely, a node reachable from the root via a path labeled $ P_1 P_2 \dots P_i $ is colored green if $ \{p_1^{P_1}, \dots, p_i^{P_i}\} $ is satisfiable and red if it is unsatisfiable. In particular, it follows that the leaves colored green correspond precisely to elements of $ D \subseteq S $.}
	\label{fig:d_disjunct_combination_space}
\end{figure}
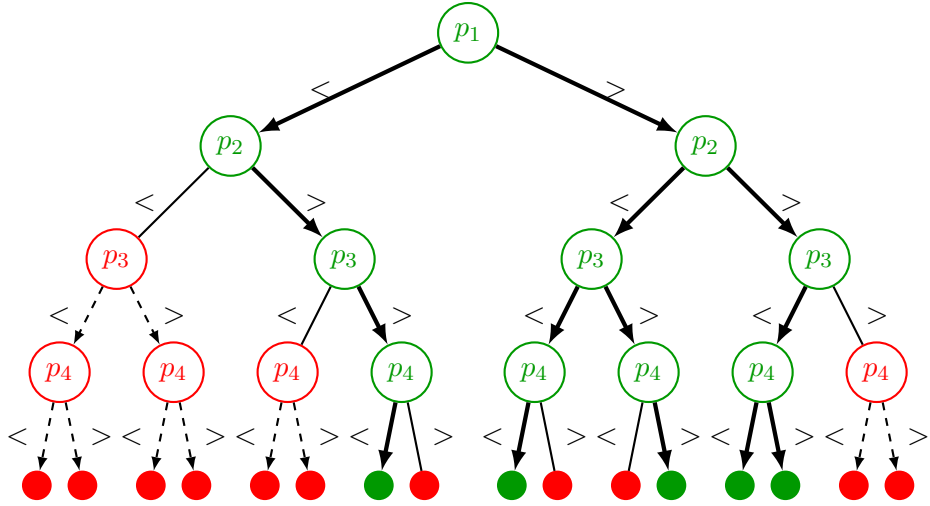

Clearly, the set $ D \subseteq S $ corresponds to a subset of the leaves of $ T $ (see Figure~\ref{fig:d_disjunct_combination_space}). Since our goal is to compute $ D $ in $ O(2^{\poly(m)}) $ time, it is essential to come up with a clever way of traversing only those parts of the tree $ T $ which lead to leaves corresponding to $ D $. We cannot afford simply traversing all possible leaves of the tree, as this would take $ O(2^{2^{\poly(m)}}) $ time because the height of $ T $ is $ n \in O(2^{\poly(m)}) $. We start by explaining the key ideas behind the approach we use to efficiently traverse $ T $.

\begin{itemize}
	\item First, think of every node in the tree $ T $ as a partial assignment of either $ P_< $ or $ P_> $ to every $ p_i $, but possibly not for all $ i \in \{1, \dots, n\} $. That is, intuitively, every node of the tree corresponds to a partial decision to replace the predicate symbol of $ p_i $ by either $ P_< $ or $ P_> $. More precisely, we regard a node as corresponding to $ \{p_1^{P_1}, \dots, p_i^{P_i}\} $ whenever the path from the root to that node is labeled $ P_1 P_2 \dots P_i $. For the sake of simplicity, we will no longer distinguish between the node reachable from the root via a path labeled $ P_1 P_2 \dots P_i $ and the set $ \{p_1^{P_1}, \dots, p_i^{P_i}\} $. We can do this because there is obviously a one-to-one correspondence between paths starting in the root and sets of the form $ \{p_1^{P_1}, \dots, p_i^{P_i}\} $.
	\item Next, we study how the satisfiability of a node $ N := \{p_1^{P_1}, \dots, p_i^{P_i}\} $ in $ T $ influences the satisfiability of $ N $'s children and vice versa. Obviously, the satisfiability of $ N $ implies the satisfiability of all $ N $'s ancestors (see Figure~\ref{fig:d_disjunct_combination_space} and observe that every ancestor of a green node is also green). It is also immediate that the unsatisfiability of $ N $ implies the unsatisfiability of all descendants of $ N $ (likewise note that it is not a coincidence that in Figure~\ref{fig:d_disjunct_combination_space} all descendants of every red node are also red; this behavior is visualized using dashed arrows). Based on these observations, we conclude that every node on a path from the root to an element of $ D $ is colored green. Hence, it is always possible to compute $ D $ by traversing only the green nodes of the tree $ T $.
	\item The above observation that it suffices to traverse only green nodes of the tree in order to compute $ D $, unfortunately, is not enough to construct an $ O(2^{\poly(m)}) $-time algorithm for computing $ D $ by traversing $ T $. 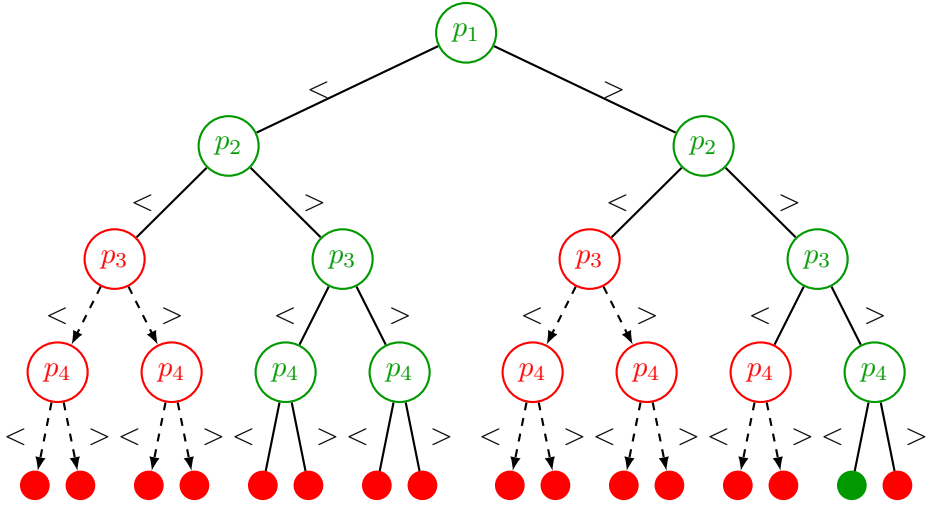
\begin{figure}[t]
		\centering
		\begin{tikzpicture}[
	every node/.style = {draw, circle, thick},
	arrreddown/.style={thick, dashed, -latex},
	arrgreendown/.style={thick},
	normaledge/.style={thick},
	leaf/.style={fill, solid},
	configsat/.style={solid, color=OliveGreen},
	configunsat/.style={solid, color=red},
	grow = down,
	level distance=1.5cm,
	level 1/.style={sibling distance=6.25cm},
	level 2/.style={sibling distance=3cm},
	level 3/.style={sibling distance=1.5cm},
	level 4/.style={sibling distance=0.6cm}]
	\node[configsat] (root) {$ p_1 $}
	child {node[configsat] {$ p_2 $}
		child {node[configunsat] {$ p_3 $}
			child {node[configunsat] {$ p_4 $}
				child {node[leaf, configunsat] {} edge from parent [arrreddown] node[left,draw=none] {$ < $}}
				child {node[leaf, configunsat] {} edge from parent [arrreddown] node[right,draw=none] {$ > $}}
				edge from parent [arrreddown] node[left,draw=none] {$ < $}
			}
			child {node[configunsat] {$ p_4 $}
				child {node[leaf, configunsat] {} edge from parent [arrreddown] node[left,draw=none] {$ < $}}
				child {node[leaf, configunsat] {} edge from parent [arrreddown] node[right,draw=none] {$ > $}}
				edge from parent [arrreddown] node[right,draw=none] {$ > $}
			}
			edge from parent [normaledge] node[left,draw=none] {$ < $}
		}
		child {node[configsat] {$ p_3 $}
			child {node[configsat] {$ p_4 $}
				child {node[leaf, configunsat] {} edge from parent [normaledge] node[left,draw=none] {$ < $}}
				child {node[leaf, configunsat] {} edge from parent [normaledge] node[right,draw=none] {$ > $}}
				edge from parent [normaledge] node[left,draw=none] {$ < $}
			}
			child {node[configsat] {$ p_4 $}
				child {node[leaf, configunsat] {} edge from parent [normaledge] node[left,draw=none] {$ < $}}
				child {node[leaf, configunsat] {} edge from parent [normaledge] node[right,draw=none] {$ > $}}
				edge from parent [arrgreendown] node[right,draw=none] {$ > $}
			}
			edge from parent [arrgreendown] node[right,draw=none] {$ > $}
		}
		edge from parent [arrgreendown] node[left,draw=none] {$ < $}
	}
	child {node[configsat] {$ p_2 $}
		child {node[configunsat] {$ p_3 $}
			child {node[configunsat] {$ p_4 $}
				child {node[leaf, configunsat] {} edge from parent [arrreddown] node[left,draw=none] {$ < $}}
				child {node[leaf, configunsat] {} edge from parent [arrreddown] node[right,draw=none] {$ > $}}
				edge from parent [arrreddown] node[left,draw=none] {$ < $}
			}
			child {node[configunsat] {$ p_4 $}
				child {node[leaf, configunsat] {} edge from parent [arrreddown] node[left,draw=none] {$ < $}}
				child {node[leaf, configunsat] {} edge from parent [arrreddown] node[right,draw=none] {$ > $}}
				edge from parent [arrreddown] node[right,draw=none] {$ > $}
			}
			edge from parent [normaledge] node[left,draw=none] {$ < $}
		}
		child {node[configsat] {$ p_3 $}
			child {node[configunsat] {$ p_4 $}
				child {node[leaf, configunsat] {} edge from parent [arrreddown] node[left,draw=none] {$ < $}}
				child {node[leaf, configunsat] {} edge from parent [arrreddown] node[right,draw=none] {$ > $}}
				edge from parent [normaledge] node[left,draw=none] {$ < $}
			}
			child {node[configsat] {$ p_4 $}
				child {node[leaf, configsat] {} edge from parent [normaledge] node[left,draw=none] {$ < $}}
				child {node[leaf, configunsat] {} edge from parent [normaledge] node[right,draw=none] {$ > $}}
				edge from parent [normaledge] node[right,draw=none] {$ > $}
			}
			edge from parent [normaledge] node[right,draw=none] {$ > $}
		}
		edge from parent [normaledge] node[right,draw=none] {$ > $}
	};
\end{tikzpicture}
		\caption{Illustration of an impossible situation when the decision tree $ T $ contains a green dead end, that is, a green node with no green leaf descendants. This particular tree contains four green dead ends: $ \{p_1^<\} $, $ \{p_1^<, p_2^>\} $, $ \{p_1^<, p_2^>, p_3^<\} $ and $ \{p_1^<, p_2^>, p_3^>\} $. All other green nodes are not green dead ends because they have a green leaf descendant $ \{p_1^>, p_2^>, p_3^>, p_4^<\} $.}
		\label{fig:d_disjunct_combination_space_tree_green_dead_end}
	\end{figure} The problem is that, even if we only traverse green nodes, there may exist green dead ends in $ T $, i.e., nodes $ N $ without green leaf descendants (see Figure~\ref{fig:d_disjunct_combination_space_tree_green_dead_end} for a concrete example). In these green dead ends the traversing algorithm would have to backtrack, which we cannot afford doing because this would make the algorithm possibly spend double-exponential time computing even some single $ \Upsilon' \in D $.
	\item We tackle the above problem by proving that, in fact, no green dead ends are possible (see Figure~\ref{fig:d_disjunct_combination_space_tree_green_dead_end}). In other words, every green node of the tree is guaranteed to lie on a path to a green leaf (see Figure~\ref{fig:d_disjunct_combination_space} for a visualization of this property using bold arrows). Consequently, the standard depth-first search (DFS) traversal of the green nodes of $ T $ is guaranteed to visit a green leaf (or to terminate) after visiting at most $ n $ many green non-leaf nodes following the visit of the previous green leaf. In other words, due to the fact that no green dead ends are possible, the DFS traversal will never get stuck searching for a leaf. Here, being ``stuck'' means visiting more than $ n $ nodes in a row, with none of them being a green leaf. Hence, if $ f(m) $ is the time complexity of visiting a single node of $ T $, then computing $ D $ via the DFS traversal of the green nodes of $ T $ takes \[
		O(\left|D\right| \cdot n \cdot f(m))
	\] time. Using the fact that $ n \in O(2^{\poly(m)}) $ and substituting the upper bound for the size of $ D $ derived in (\ref{eqn:d_size_in_o_two_pow_polym}), we obtain the \begin{align}
		\label{eqn:generic_complexity_of_dfs_algo_computing_d}
		O(2^{\poly(m)} \cdot 2^{\poly(m)} \cdot f(m)) = O(2^{\poly(m)} \cdot f(m))
	\end{align} time complexity for the DFS algorithm computing $ D $.
\end{itemize}

\subsubsection{No green dead ends are possible}

We now start filling in the missing details in the above reasoning by formulating the property that no green dead ends are possible and proving it.

\begin{zbclaim}[No green dead ends are possible]
	Let $ i \in \{1, \dots, n-1\} $ and $ (P_1, \dots, P_i) \in \{<, >\}^i $ be such that $ \{p_1^{P_1}, \dots, p_i^{P_i}\} $ is satisfiable. Then, there exists $ P_{i+1} \in \{<, >\} $ such that $ \{p_1^{P_1}, \dots, p_{i+1}^{P_{i+1}}\} $ is satisfiable.
\end{zbclaim}
\begin{proof}
	Suppose, for the sake of contradiction, that $ \{p_1^{P_1}, \dots, p_{i+1}^{P_{i+1}}\} $ were unsatisfiable for all $ P_{i+1} \in \{<, >\} $. Since \[
	\{p_1^{P_1}, \dots, p_i^{P_i}\} \models \top \equiv p_{i+1}^< \vee p_{i+1}^= \vee p_{i+1}^>
	\] it follows that $ \{p_1^{P_1}, \dots, p_i^{P_i}\} \models p_{i+1}^= $. Let $ v + W $ be the affine vectorspace of solutions to $ p_{i+1}^= $ such that \begin{align}
		\label{eqn:v_entails_p1p1uptopipi_entails_piplusoneeq}
		v \models \{p_1^{P_1}, \dots, p_i^{P_i}\} \models p_{i+1}^=
	\end{align} Clearly, due to $ p_{i+1}^= \not\equiv \top $, we have $ \dim(W) = d - 1 $. Hence, \[
	\dim(\Q^d / W) = d - (d - 1) = 1
	\] Since $ \Q^d / W \cong W^\bot $, it follows that $ \dim(W^\bot) = 1 $ and hence there exists a vector $ 0 \neq w \in W^\bot $ satisfying $ W \le \gen{w}^\bot $. Expressing this in terms of $ \LinDep_{\id, v} $ yields \begin{align}
		\label{eqn:gen0_neq_genw_in_lindep_idv_piplusoneeq}
		\gen{0} \neq \gen{w} \in \LinDep_{\id, v}(\{p_{i+1}^=\})
	\end{align} The entailment in (\ref{eqn:v_entails_p1p1uptopipi_entails_piplusoneeq}) immediately implies that \begin{align}
		\label{eqn:lindep_idv_piplusoneeq_subset_lindep_idv_p1p1uptopipi}
		\LinDep_{\id, v}(\{p_{i+1}^=\}) \subseteq \LinDep_{\id, v}(\{p_1^{P_1}, \dots, p_i^{P_i}\})
	\end{align} Since $ P_i $ is not an equality predicate for every $ i $, by Theorem~\ref{thm:only_equality_predicates_can_establish_lindep_strong} (see Appendix~\ref{sec:app:lindep_facts}) it follows that \[
	\LinDep_{\id, v}(\{p_1^{P_1}, \dots, p_i^{P_i}\}) \overset{\ref{thm:only_equality_predicates_can_establish_lindep_strong}}{=} \LinDep_{\id, v}(\varnothing) = \{\gen{0}\}
	\] Combining this with (\ref{eqn:lindep_idv_piplusoneeq_subset_lindep_idv_p1p1uptopipi}) yields \[
	\LinDep_{\id, v}(\{p_{i+1}^=\}) = \{\gen{0}\}
	\] This is a contradiction to (\ref{eqn:gen0_neq_genw_in_lindep_idv_piplusoneeq}).
\end{proof}

\subsubsection{Algorithm for computing $ D $ efficiently}

\begin{algorithm}
	\caption{A procedure for computing the set $ D $ at Line~\ref{alg:cover:line:d_set} of the covering algorithm (Algorithm~\ref{alg:cover})}
	\label{alg:compute_d}
	\SetAlgoLined
	\KwIn{A formula $ \varphi $, a partition $ \Pi $ and predicate sets $ \Upsilon $, $ N $}
	\KwOut{A set $ D $ of predicate sets}
	\SetKwFunction{computed}{compute\_D}
	\SetKwProg{computedproc}{Function}{}{end}
	\computedproc{\computed{$ \Upsilon $, $ N $}}{
		\nl \If{$ \Upsilon = \varnothing $}{
			\tcp{We have reached a leaf of the decision tree $ T $}
			\nl \Return $ \{N\} $\;
		}
		\nl $ D := \varnothing $\;
		\nl Pick $ (t_1 \neq t_2) \in \Upsilon $\;
		\If{$ N \cup \{t_1 < t_2\} $ is satisfiable}{
			\tcp{The left child of $ N $ in the decision tree $ T $ is green, so we need to traverse the left subtree}
			\nl $ D := D \cup \computed(\Upsilon \setminus \{t_1 \neq t_2\}, N \cup \{t_1 < t_2\}) $\;
		}
		\If{$ N \cup \{t_1 > t_2\} $ is satisfiable}{
			\tcp{The right child of $ N $ in the decision tree $ T $ is green, so we need to traverse the right subtree}
			\nl $ D := D \cup \computed(\Upsilon \setminus \{t_1 \neq t_2\}, N \cup \{t_1 > t_2\}) $\;
		}
		\nl \Return $ D $\;
	}
\end{algorithm}

In Algorithm~\ref{alg:compute_d}, we provide the pseudocode explaining how to compute $ D $ via the DFS traversal of the tree $ T $.  Intuitively, the algorithm uses the $ N $ set to track the green node currently being visited; visiting a child of the current node $ N $ corresponds to removing an inequality from $ \Upsilon $ and adding one to $ N $; once $ N $ contains $ n $ predicates and $ \Upsilon $ becomes empty, this means that we have reached a leaf. The comments in the pseudocode connect the implementation with the ideas explained in Section~\ref{sec:vardec:exponential_upper_bound:towards_algo_for_computing_d}. The $ \computed $ procedure outputting the desired predicate set $ D $ can easily be integrated into the covering algorithm by replacing Line~\ref{alg:cover:line:d_set} with \[
	D := \computed(\Upsilon, \varnothing)
\] Clearly, the time complexity of $ \computed $ without recursion is $ O(2^{\poly(m)}) $ because checking the satisfiability of a conjunction of strict inequalities can be done in time polynomial in the size of the conjunction, which in this case is $ n \in O(2^{\poly(m)}) $. Hence, substituting a suitable bound $ f(m) \in O(2^{\poly(m)}) $ into (\ref{eqn:generic_complexity_of_dfs_algo_computing_d}) yields that the overall running time of the $ \computed $ algorithm is \[
	O(2^{\poly(m)} \cdot 2^{\poly(m)}) = O(2^{\poly(m)})
\] We conclude that the set $ D $ at Line~\ref{alg:cover:line:d_set} of the covering algorithm can be computed in time exponential in the size of $ \varphi $.

\subsubsection{Putting everything together}

Note that, as we have already shown in Section~\ref{sec:vardec:cover:termination_time_complexity} above, the covering algorithm makes $ O(2^{\poly(m)}) $ recursive calls. We have just shown that the time complexity of the covering algorithm without recursion is exponential, so we conclude that the overall running time of the covering algorithm is also exponential in the size of $ \varphi $:

\begin{theorem}
	\label{thm:refined_cover_alg_correctness}
	For a formula $ \varphi \in \QFLRA $, a binary partition $ \Pi $ and a predicate set $ \Gamma \in \DisjPhi $, the covering algorithm with the $ \computed $ procedure integrated at Line~\ref{alg:cover:line:d_set} runs in time exponential in the size of $ \varphi $ and any formula $ \psi := \cover(\Pi, \Gamma) $ produced by it correctly solves the covering problem.
\end{theorem}

Combining Theorem~\ref{thm:refined_cover_alg_correctness} with Theorem~\ref{thm:reduction} reducing the variable decomposition problem to the covering problem yields the following result.

\begin{theorem}
	\label{thm:vardec_exponential}
	Over quantifier-free linear real arithmetic, the variable decomposition problem (for binary partitions) can be solved in time exponential in the size of the given formula $ \varphi $. In particular, if $ \Pi $ is a binary partition, then it follows that any $ \Pi $-decomposable formula $ \varphi \in \QFLRA $ has a $ \Pi $-decomposition in DNF of exponential size.
\end{theorem}
\begin{proof}
	Analogous to the proof of Theorem~\ref{thm:vardec_double_exponential}. The only difference is that since the size of the computed covering $ \psi_\Gamma := \cover(\Pi, \Gamma) $ of $ \Gamma $ is exponential in the size of $ \varphi $, the satisfiability of $ \gamma_{i, j} $ can also be decided in $ O(2^{\poly(n)}) $ time, where $ n $ is the size of $ \varphi $. Hence, the time complexity of deciding whether $ \psi_\Gamma \models \varphi $ holds is \[
		l \cdot k \cdot O(2^{\poly(n)}) \subseteq O(2^{\poly(n)}) \cdot O(2^{\poly(n)}) \cdot O(2^{\poly(n)}) = O(2^{\poly(n)})
	\] (see the proof of Theorem~\ref{thm:vardec_double_exponential}). Consequently, the variable decomposition problem can be solved in $ O(2^{\poly(n)}) $ time.
\end{proof}

Note that the integration of the $ \computed $ procedure at Line~\ref{alg:cover:line:d_set} of the covering algorithm does not break any technical properties of the algorithm we have established previously. More precisely, Lemma~\ref{lemma:psi_properties}, Theorems \ref{thm:model_flooding}, \ref{thm:cover_alg_correctness}, \ref{thm:model_flooding_strong}, and Corollary~\ref{cor:minimality_of_dec_for_double_exponential_algo} hold also for the algorithm of Theorem~\ref{thm:vardec_exponential}. In particular, it follows that the algorithm of Theorem~\ref{thm:vardec_exponential} produces decompositions that are minimal within the meaning of Corollary~\ref{cor:minimality_of_dec_for_double_exponential_algo}.

\subsection{Towards a $ \coNP $ upper bound}
\label{sec:vardec:towards_conp_upper_bound}

Theorems \ref{thm:vardec_double_exponential} and \ref{thm:vardec_exponential} establish complexity bounds for the variable decomposition problem where, by definition, we additionally have to output a $ \Pi $-decomposition whenever the input formula $ \varphi $ is $ \Pi $-decomposable. In Sections \ref{sec:vardec:nondec_proof_system} and \ref{sec:vardec:proof_compression} below, we will see that computing decompositions is a bottleneck of the covering-based algorithm presented above, in the sense that Theorem~\ref{thm:vardec_exponential} can be further strengthened if we are interested only in deciding variable decomposability. Concretely, in the remainder of \zbrefsec{sec:vardec} we aim to prove that, surprisingly, the problem of deciding variable decomposability is in $ \coNP $. We achieve this complexity result in two high-level steps:


\begin{itemize}
	\item First, in Section~\ref{sec:vardec:nondec_proof_system} below, we construct a sound and complete proof system for establishing non-$ \Pi $-decomposability, where proofs are encodable using exponentially many bits and verifiable in time polynomial in the size of the alleged proof. That is, we describe a witness for the fact that a formula $ \varphi $ is not $ \Pi $-decomposable and show that such a witness exists if and only if the formula is indeed not $ \Pi $-decomposable.
	\item Due to the exponential length of the proofs, from the previous step we get only a $ \coNEXP $ upper bound for variable decomposability (which is not so interesting, given that a much stronger $ \EXP $ upper bound already follows from Theorem~\ref{thm:vardec_exponential}). In order to strengthen this result, in Section~\ref{sec:vardec:proof_compression} below, we show how to compress the proofs of non-$ \Pi $-decomposability so that they become of polynomial size. Consequently, we obtain the desired $ \coNP $ upper bound. The central idea behind our compression technique is to construct for an appropriate predicate set $ \Lambda $ of exponential size, another predicate set $ \Lambda' $ that is of polynomial size and which preserves all necessary properties. The main property to be guaranteed is the entailment $ \Lambda' \models \Lambda $. Overall, our compression technique is completely novel (to the best of our knowledge) and may be of independent interest.
\end{itemize}

\newpage

\subsection{Proving non-decomposability}
\label{sec:vardec:nondec_proof_system}

\subsubsection{Overspilling-based proof system}

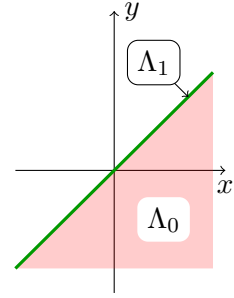
\begin{wrapfigure}{r}{0.25\textwidth}
	\centering
	\if\zbthesismode1
	\vspace{-15pt}
	\else
	\vspace{-20pt}
	\fi
	\begin{tikzpicture}[scale=0.65]	
	\draw[->] (-2,0) -- (2.25,0) coordinate[label = {below:$x$}] (xmax);
	\draw[->] (0,-2.5) -- (0,3.25) coordinate[label = {right:$y$}] (ymax);
	
	\coordinate (bottoml) at (-2, -2);
	\coordinate (bottomr) at (2, -2);
	\coordinate (topr) at (2, 2);
	
	\fill[fill=red,opacity=0.2] (bottoml) -- (bottomr) -- (topr) -- cycle;
	\node[fill=white, rounded corners] (gamma) at (1, -1) {$ \Lambda_0 $};
	
	\draw[very thick, OliveGreen] (bottoml) -- (topr);
	
	\coordinate (lambda1labelendpoint) at (1.5, 1.5);
	\coordinate (lambda1label) at ($ (lambda1labelendpoint) + (-0.7, 0.7) $);
	\draw[->] (lambda1label) -- (lambda1labelendpoint);
	\draw (lambda1label) node [rounded corners, draw=black, fill=white] {$ \Lambda_1 $};
\end{tikzpicture}
	\caption{Illustration of the sets defined by the $ \Lambda_0 $ and $ \Lambda_1 $ predicate sets we use to argue that $ \varphi $ is not monadically decomposable. $ \Lambda_1 $ is colored green to illustrate that $ \Lambda_1 \models \varphi $, while the red color depicts the $ \Lambda_0 \models \neg\varphi $ entailment.}
	\label{fig:xeqy_nondec_proof}
	\if\zbthesismode1
	\vspace{-10pt}
	\else
	\vspace{-20pt}
	\fi
\end{wrapfigure}

We start by explaining the intuitions leading to the proof system. Consider the formula $ \varphi := x = y $. It is easy to see that $ \varphi $ is not monadically decomposable, but how can we capture this rigorously? In this case it is trivial that for any distinct $ a, b \in \Q $ and any $ \{\{x\}, \{y\}\} $-respecting conjunction of literals $ \psi $ having models $ (a, a)^\transp \models \psi $ and $ (b, b)^\transp \models \psi $, $ \psi $ must also have some model (point) it is actually not allowed to have, such as $ (a, b)^\transp \models \psi $. Thus, assuming towards a contradiction the existence of a monadic decomposition (without loss of generality) brought to DNF implies that any DNF term can express only at most one point on the line defined by $ x = y $.
This line obviously contains infinitely many points, so we get a contradiction to the correctness of the allegedly existing monadic decomposition. We now use this example to show that the Overspilling Theorem~\ref{thm:overspilling} can be thought of as providing a framework for reasoning similar to (but significantly stronger than) this argumentation. Clearly, $ \Lambda_1 := \{x = y\} $ is $ (x=y) $-next to $ \Lambda_0 := \{x > y\} $, those predicate sets are both $ \Pi := \{\{x\}, \{y\}\} $-complex and \begin{align*}
	\LinDep_{\pi_X}(\Lambda_1) &= \{\gen{0}\} = \LinDep_{\pi_X}(\Lambda_0) \\
	\LinDep_{\pi_Y}(\Lambda_1) &= \{\gen{0}\} = \LinDep_{\pi_Y}(\Lambda_0)
\end{align*} holds where $ X := \{x\} $ and $ Y := \{y\} $ (see Figure~\ref{fig:xeqy_nondec_proof}). Hence, $ \varphi $ is not $ \Pi $-decomposable because assuming the contrary would imply by the Overspilling Theorem~\ref{thm:overspilling} the satisfiability of \[
	\varphi \wedge \Lambda_0 \equiv x = y \wedge x > y
\] which is a contradiction.

We now generalize this reasoning with the goal of constructing an encodable and verifiable proof of the fact that some formula is not $ \Pi $-decomposable. Observe that the above argumentation depends only on the two predicate sets $ \Lambda_0 $ and $ \Lambda_1 $. It is trivial to algorithmically check that $ \Lambda_1 $ is $ p $-next to $ \Lambda_0 $ on some $ p \in \Lambda_1^= $ and, as we have already discussed in Section~\ref{sec:vardec:pisimple_picomplex}, it can be tested in polynomial time whether a predicate set (in this case $ \Lambda_0 $ and $ \Lambda_1 $) is $ \Pi $-complex. We now argue that it is also possible to efficiently verify that \begin{align}
	\label{eqn:lindep_piz_lambda0_eq_lindep_piz_lambda1}
	\LinDep_{\pi_Z}(\Lambda_0) = \LinDep_{\pi_Z}(\Lambda_1)
\end{align} holds for all $ Z \in \Pi $. Let $ A := (a_1, \dots, a_l) $ and $ B := (b_1, \dots, b_r) $ be bases such that \begin{align*}
	\ModelsOf(\Lambda_0^=) &= a + \gen{a_1, \dots, a_l} \\
	\ModelsOf(\Lambda_1^=) &= b + \gen{b_1, \dots, b_r}
\end{align*} holds for some affine offsets $ a, b \in \Q^n $. In other words, $ B_0 $ and $ B_1 $ are bases of the affine vectorspaces corresponding to $ \Lambda_0 $ and $ \Lambda_1 $, respectively. These bases can be easily computed in polynomial time via Gaussian elimination. By Theorem~\ref{thm:lindep_algebraic_char}, checking (\ref{eqn:lindep_piz_lambda0_eq_lindep_piz_lambda1}) reduces to verifying whether \begin{align}
	\label{eqn:im_piz_lca_eq_im_piz_lcb}
	\ImageOf(\pi_Z \circ \lc_A) = \ImageOf(\pi_Z \circ \lc_B)
\end{align} holds for all $ Z \in \Pi $. As we have already seen in the discussion of the covering algorithm in Section~\ref{sec:vardec:cover}, \begin{align*}
	\ImageOf(\pi_Z \circ \lc_A) &= \ImageOf\big(\Pi_Z \cdot a_1 \mid \dots \mid \Pi_Z \cdot a_l\big) \\
	\ImageOf(\pi_Z \circ \lc_B) &= \ImageOf\big(\Pi_Z \cdot b_1 \mid \dots \mid \Pi_Z \cdot b_r\big)
\end{align*} so (\ref{eqn:im_piz_lca_eq_im_piz_lcb}) reduces to checking \begin{align}
	\label{eqn:im_piz_as_eq_im_piz_bs}
	\ImageOf\big(\Pi_Z \cdot a_1 \mid \dots \mid \Pi_Z \cdot a_l\big) = \ImageOf\big(\Pi_Z \cdot b_1 \mid \dots \mid \Pi_Z \cdot b_r\big)
\end{align} Indeed, (\ref{eqn:im_piz_as_eq_im_piz_bs}) holds if and only if every $ \Pi_Z \cdot a_i $ is expressible as a linear combination of $ \Pi_Z \cdot b_1, \dots, \Pi_Z \cdot b_r $, and vice versa, every $ \Pi_Z \cdot b_j $ is some linear combination of $ \Pi_Z \cdot a_1, \dots, \Pi_Z \cdot a_l $. It is trivial that this can be decided via Gaussian elimination in polynomial time.

Now observe that if $ \Lambda_0, \Lambda_1 \in \DisjTrue $ are disjuncts of a formula $ \varphi $ as in the running example, then it is easy to check whether $ \Lambda_i $ entails $ \varphi $ or $ \neg\varphi $ by computing a model $ v_i \models \Lambda_i $ and then checking whether $ v_i \models \varphi $. If $ v_i \models \varphi $, then $ \Lambda_i \models \varphi $ by Lemma~\ref{lemma:every_disj_either_true_or_false}, otherwise $ v_i \models \neg\varphi $ meaning that (again by Lemma~\ref{lemma:every_disj_either_true_or_false}) $ \Lambda_i \models \neg\varphi $. Thus, it can be checked in polynomial time whether the Overspilling Theorem~\ref{thm:overspilling} indeed yields a contradiction assuming that $ \varphi $ is $ \Pi $-decomposable.

We conclude that it can be verified in polynomial time whether certain $ \Lambda_0, \Lambda_1 \in \DisjTrue $ witness the non-$ \Pi $-decomposability of $ \varphi $ in the sense that applying the above reasoning to $ \Lambda_0 $ and $ \Lambda_1 $ indeed yields the conclusion that $ \varphi $ is not $ \Pi $-decomposable due to the Overspilling Theorem~\ref{thm:overspilling}. Whenever this is the case, we call $ (\Lambda_0, \Lambda_1) $ an \textit{overspilling proof} of non-$ \Pi $-decomposability. Note that overspilling proofs are sound in the sense that the existence of such a proof indeed implies by the Overspilling Theorem~\ref{thm:overspilling} that $ \varphi $ is not $ \Pi $-decomposable. However, in order to achieve the desired $ \coNP $ upper bound for variable decomposability, we need to construct a sound and complete proof system for non-decomposability, where proofs are encodable using polynomially many bits and verifiable in polynomial time. We have already observed soundness and the polynomial size of the proofs and time to verify them. Hence, the key question to ask oneself at this point is whether the proof system we just constructed is complete. In other words, does there always exist an overspilling proof of non-$ \Pi $-decomposability of a formula $ \varphi $ if $ \varphi $ is not $ \Pi $-decomposable?

Unfortunately, the answer turns out to be negative. For example, consider the formula \[
	\varphi := x = y \wedge 2x = 2y
\] and the partition $ \Pi := \{\{x\}, \{y\}\} $. Clearly, $ \varphi \equiv x = y $ is not $ \Pi $-decomposable. The only satisfiable disjuncts in $ \DisjTrue $ are \[
	\Sat(\DisjTrue) = \{\{x < y, 2x < 2y\}, \{x = y, 2x = 2y\}, \{x > y, 2x > 2y\}\}
\] but all of them are pairwise not $ p $-next to each other for every $ p $. Hence, the Overspilling Theorem~\ref{thm:overspilling} is simply not applicable to any pair of disjuncts from $ \Sat(\DisjTrue) $. As a result, $ \varphi $ has no overspilling proof of non-$ \Pi $-decomposability.

\subsubsection{A better proof system}

In the present example, we immediately see that the above problem arises due to the fact that the $ 2x > 2y $ (resp. $ 2x = 2y $, $ 2x < 2y $) predicate is redundant in the sense that it is implied by $ x > y $ (resp. $ x = y $, $ x < y $). This leads us to the idea of circumventing this problem by adjusting overspilling proofs to no longer require $ \Lambda_0 $ and $ \Lambda_1 $ to be disjuncts and defining $ \Lambda_0' $ and $ \Lambda_1' $ so that these predicate sets are disjuncts equivalent to $ \Lambda_0 $ and $ \Lambda_1 $, respectively. In other words, the idea is to defer the task of being disjuncts to separate predicate sets $ \Lambda_0' $ and $ \Lambda_1' $, whose goal is to ensure the possibility of verifying in polynomial time whether $ \Lambda_i' \models \varphi $ or $ \Lambda_i' \models \neg\varphi $ holds by computing just a single model and applying Lemma~\ref{lemma:every_disj_either_true_or_false}. More precisely, we define \begin{align*}
	\Lambda_0' &:= \{x > y, 2x > 2y\} \in \Sat(\DisjTrue) \\
	\Lambda_1' &:= \{x = y, 2x = 2y\} \in \Sat(\DisjTrue) \\
	\Lambda_0 &:= \{x > y\} \\
	\Lambda_1 &:= \{x = y\}
\end{align*} and argue that $ \varphi $ is not $ \Pi $-decomposable as follows. Note that $ \Lambda_0 $ and $ \Lambda_1 $ are both $ \Pi $-complex, $ \Lambda_1 $ is $ (x=y) $-next to $ \Lambda_0 $ and, just as in the $ x=y $ example above, $ \Lambda_0 $ has the same $ Z $-dependencies as $ \Lambda_1 $ for all $ Z \in \Pi $. Towards a contradiction, assume that $ \varphi $ were $ \Pi $-decomposable. Since $ \Lambda_1' \in \DisjPhi $ and $ \Lambda_1 \equiv \Lambda_1' $, it follows that $ \Lambda_1 \models \varphi $. Hence, the Overspilling Theorem~\ref{thm:overspilling} yields that $ \varphi \wedge \Lambda_0 $ is satisfiable. Since $ \Lambda_0 \equiv \Lambda_0' $, $ \varphi \wedge \Lambda_0' $ is satisfiable. However, by computing any model $ v \models \Lambda_0' $ (e.g., $ v = (2, 1)^\transp $), checking whether $ v \models \neg\varphi $ holds and applying Lemma~\ref{lemma:every_disj_either_true_or_false}, we conclude that $ \Lambda_0' \models \neg\varphi $. This is a contradiction. Generalizing this reasoning to arbitrary formulas leads us to the following notion of a $ \Lambda $-proof of non-$ \Pi $-decomposability, which rigorously captures all the information needed to apply the above argumentation, but additionally allows the $ \Lambda_0 $, $ \Lambda_1 $, $ \Lambda_0' $, $ \Lambda_1' $ sets to have further predicates from a certain fixed set $ \Lambda $, which is a convenient feature that we will use shortly.

\begin{zbdefinition}[$ \Lambda $-proof of non-$ \Pi $-decomposability]
	Let $ \varphi \in \QFLRA $ be a formula, $ \Pi $ be a partition, $ \Lambda $ be a predicate set and $ \Lambda_0', \Lambda_1' \in \DisjOf{\Lambda} $. Let furthermore $ \Lambda_0 $ and $ \Lambda_1 $ be predicate sets such that \begin{itemize}
		\item $ \Lambda_0' \equiv \Lambda_0 \subseteq \Lambda_0' $ and $ \Lambda_1' \equiv \Lambda_1 \subseteq \Lambda_1' $
		\item $ \Lambda_0 $ and $ \Lambda_1 $ are both $ \Pi $-complex
		\item For all $ Z \in \Pi $ it holds that $ \LinDep_{\pi_Z}(\Lambda_0) = \LinDep_{\pi_Z}(\Lambda_1) $
		\item $ \Lambda_1 $ is $ p $-next to $ \Lambda_0 $ on some predicate $ p \in \Lambda_1^= $
		\item $ \Lambda_i' \models \varphi $ and $ \Lambda_{1-i}' \models \neg\varphi $ for some $ i \in \{0, 1\} $
	\end{itemize} We call the tuple $ (\Lambda_0, \Lambda_1, \Lambda_0', \Lambda_1') $ a $ \Lambda $-proof of non-$ \Pi $-decomposability of $ \varphi $.
\end{zbdefinition}

Observe that the above argumentation showing that $ \varphi $ is not $ \Pi $-decomposable immediately generalizes to showing that the existence of a $ \Lambda $-proof implies that $ \varphi $ is indeed not $ \Pi $-decomposable. Hence, $ \Lambda $-proofs are sound. \begin{algorithm}
	\caption{A verifier for $ \Lambda $-proofs}
	\label{alg:lambda_proof_verifier}
	\SetAlgoLined
	\KwIn{A formula $ \varphi $, a partition $ \Pi $ and predicate sets $ \Lambda_0, \Lambda_1, \Lambda_0', \Lambda_1' $}
	\KwOut{``$ \Accept $'' if $ (\Lambda_0, \Lambda_1, \Lambda_0', \Lambda_1') $ is a $ \Lambda $-proof of non-$ \Pi $-decomposability of $ \varphi $ and ``$ \Reject $'' otherwise}
	\SetKwFunction{lambdaproofverifier}{verify\_lambda\_proof}
	\SetKwProg{lambdaproofverifierproc}{Function}{}{end}
	\lambdaproofverifierproc{\lambdaproofverifier{$ \Lambda_0 $, $ \Lambda_1 $, $ \Lambda_0' $, $ \Lambda_1' $}}{
		\tcp{Verify that $ \Lambda_0' $ and $ \Lambda_1' $ are indeed disjuncts of $ \Lambda $}
		\nl \If{for some $ i \in \{0, 1\} $, there exists $ p \in \Pred(\varphi) $ such that $ \{p^<, p^=, p^>\} \cap \Lambda_i' = \varnothing $}{
			\nl \Return \Reject\;
		}
		\tcp{Verify that $ \Lambda_0' \equiv \Lambda_0 \subseteq \Lambda_0' $ and $ \Lambda_1' \equiv \Lambda_1 \subseteq \Lambda_1' $}
		\nl \If{$ \Lambda_0 \subsetneq \Lambda_0' $ or $ \Lambda_1 \subsetneq \Lambda_1' $}{
			\nl \Return \Reject\;
		}
		\nl\label{alg:lambda_proof_verifier:line:equiv_check}\If{$ \Lambda_0 \not\equiv \Lambda_0' $ or $ \Lambda_1 \not\equiv \Lambda_1' $}{
			\nl \Return \Reject\;
		}
		\tcp{Verify that $ \Lambda_0 $ and $ \Lambda_1 $ fulfill the requirements of the Overspilling Theorem~\ref{thm:overspilling}}
		\nl \If{$ \Lambda_1 $ is not $ p $-next to $ \Lambda_0 $ on some predicate $ p \in \Lambda_1^= $}{
			\nl \Return \Reject\;
		}
		\nl \If{$ \Lambda_0 $ or $ \Lambda_1 $ is $ \Pi $-simple}{
			\nl \Return \Reject\;
		}
		\tcp{Solve the systems of linear equations $ \Lambda_0^= $ and $ \Lambda_1^= $, in order to be able to analyze the $ Z $-dependencies of $ \Lambda_0 $ and $ \Lambda_1 $}
		\nl Compute $ a \in \Q^n $ and a basis $ A := (a_1, \dots, a_l) $ such that $ \ModelsOf(\Lambda_0^=) = a + \gen{a_1, \dots, a_l} $\;
		\nl Compute $ b \in \Q^n $ and a basis $ B := (a_1, \dots, a_l) $ such that $ \ModelsOf(\Lambda_1^=) = b + \gen{b_1, \dots, b_r} $\;
		\tcp{Reject if $ \LinDep_{\pi_Z}(\Lambda_0) \neq \LinDep_{\pi_Z}(\Lambda_1) $}
		\nl\label{alg:lambda_proof_verifier:line:lindep_check}\If{for some $ Z \in \Pi $, $ \ImageOf(\Pi_Z(a_1) \mid \dots \mid \Pi_Z(a_l)) \neq \ImageOf(\Pi_Z(b_1) \mid \dots \mid \Pi_Z(b_r)) $}{
			\nl \Return \Reject\;
		}
		\tcp{Check whether $ \Lambda_i' \models \varphi $ and $ \Lambda_{1-i}' \models \neg\varphi $ holds for some $ i \in \{0, 1\} $}
		\nl Compute models $ v_0 \models \Lambda_0 $ and $ v_1 \models \Lambda_1 $\;
		\nl\label{alg:lambda_proof_verifier:line:final_check}\eIf{for some $ i \in \{0, 1\} $, $ v_i \models \varphi $ and $ v_{1-i} \models \neg\varphi $}{
			\tcp{$ \Lambda_0 $ and $ \Lambda_1 $ indeed show that $ \varphi $ is not $ \Pi $-decomposable}
			\nl \Return \Accept;
		}{
			\nl \Return \Reject\;
		}
	}
\end{algorithm} In Algorithm~\ref{alg:lambda_proof_verifier}, we give the pseudocode explaining how to verify alleged $ \Lambda $-proofs. In essence, the verification procedure extends and refines the one for overspilling proofs discussed above. In particular, we have already proven the correctness of the $ \LinDep_{\pi_Z}(\Lambda_0) = \LinDep_{\pi_Z}(\Lambda_1) $ check at Line~\ref{alg:lambda_proof_verifier:line:lindep_check}. Furthermore, note that if the check at Line~\ref{alg:lambda_proof_verifier:line:final_check} succeeds, then it is guaranteed that $ \Lambda_i' \models \varphi $ and $ \Lambda_{1-i}' \models \neg\varphi $. Indeed, this follows from the fact that \begin{align*}
	v_i &\models \Lambda_i \wedge \varphi \equiv \Lambda_i' \wedge \varphi \\
	v_{1-i} &\models \Lambda_{1-i} \wedge \neg\varphi \equiv \Lambda_{1-i}' \wedge \neg\varphi
\end{align*} which, by Lemma~\ref{lemma:every_disj_either_true_or_false}, is equivalent to saying that $ \Lambda_i' \models \varphi $ and $ \Lambda_{1-i}' \models \neg\varphi $, as it is stated in the definition of a $ \Lambda $-proof. More precisely, we can apply Lemma~\ref{lemma:every_disj_either_true_or_false} because $ \Lambda_0', \Lambda_1' \in \DisjOf{\Lambda} $ are disjuncts. The correctness of all remaining checks is trivial. As regards the running time of the verification procedure, it is easy to see that it is polynomial in the size of the alleged proof; the only part where this is not immediate is the \texttt{if} statement at Line~\ref{alg:lambda_proof_verifier:line:equiv_check}. This equivalence check can be implemented efficiently by testing (for $ i \in \{0, 1\} $) that $ \Lambda_i \models p $ for all $ p \in \Lambda_i' $ and vice versa, $ \Lambda_i' \models q $ for all $ q \in \Lambda_i $. Since checking these entailments is equivalent to testing the unsatisfiability of conjunctions $ \Lambda_i \wedge \neg p $ and $ \Lambda_i' \wedge \neg q $, respectively, we conclude that the check at Line~\ref{alg:lambda_proof_verifier:line:equiv_check} can be implemented to run in polynomial time.

\subsubsection{Soundness and completeness of $ \Lambda $-proofs}

Having shown that $ \Lambda $-proofs are sound and verifiable in polynomial time, we now ask ourselves whether this newly constructed proof system is complete. It turns out that, unlike in the case of overspilling proofs, any non-$ \Pi $-decomposable formula $ \varphi $ has a $ \Lambda $-proof of non-$ \Pi $-decomposability. We rigorously capture this fact in the following Theorem~\ref{thm:lambda_proof_soundness_completeness}.

\begin{theorem}[Soundness and Completeness]
	\label{thm:lambda_proof_soundness_completeness}
	Let $ \varphi \in \QFLRA $ be a formula, $ \Pi $ be a binary partition and $ \Psi_\Gamma $ be the set of predicate sets corresponding to the DNF terms of $ \psi_\Gamma := \cover(\Pi, \Gamma) $. Then $ \varphi $ is not $ \Pi $-decomposable if and only if there exists a $ \Lambda $-proof of non-$ \Pi $-decomposability of $ \varphi $ for some $ \Gamma \in \Sat(\DisjPhi) $ and $ \Lambda \in \Psi_\Gamma $.
\end{theorem}

\subsubsection{Convexity}

For the proof of Theorem~\ref{thm:lambda_proof_soundness_completeness} below, we will need the following fact about linear real arithmetic, which we call \textit{convexity} because it is, in essence, a modified and strengthened version of the well-known notion of convexity introduced by Nelson and Oppen in the context of constructing decision procedures for a combination of first-order theories \cite{nelson_oppen:1979, oppen:1980, bruttomesso:2009:nelson_oppen, kroening:2016}.

\begin{lemma}[Convexity]
	\label{lemma:convexity}
	Let $ \Gamma $ be a predicate set, $ p $ be a predicate and $ Q, R $ be predicate symbols such that \[
		\Gamma \models p^Q \vee p^R
	\] Then $ \Gamma \models p^Q $ or $ \Gamma \models p^R $.
\end{lemma}
\begin{proof}
	Without loss of generality, we assume that $ Q \neq R $ and that $ \Gamma $ is satisfiable. We distinguish between the following cases.
	
	\textbf{Case 1.} Suppose $ \{Q, R\} = \{<, >\} $ and, for the sake of contradiction, that $ \Gamma \wedge p^< $ and $ \Gamma \wedge p^> $ are both satisfiable. Then, by Lemma~\ref{lemma:predicate_convexity}, $ \Gamma \wedge p^= $ is satisfiable, which contradicts the assumption $ \Gamma \models p^Q \vee p^R $.
	
	\textbf{Case 2.} Suppose $ \{Q, R\} = \{<, =\} $. Towards a contradiction, assume that $ \Gamma \wedge p^< $ and $ \Gamma \wedge p^= $ are both satisfiable and consider models $ v_< \models \Gamma \wedge p^< $, $ v_= \models \Gamma \wedge p^= $. 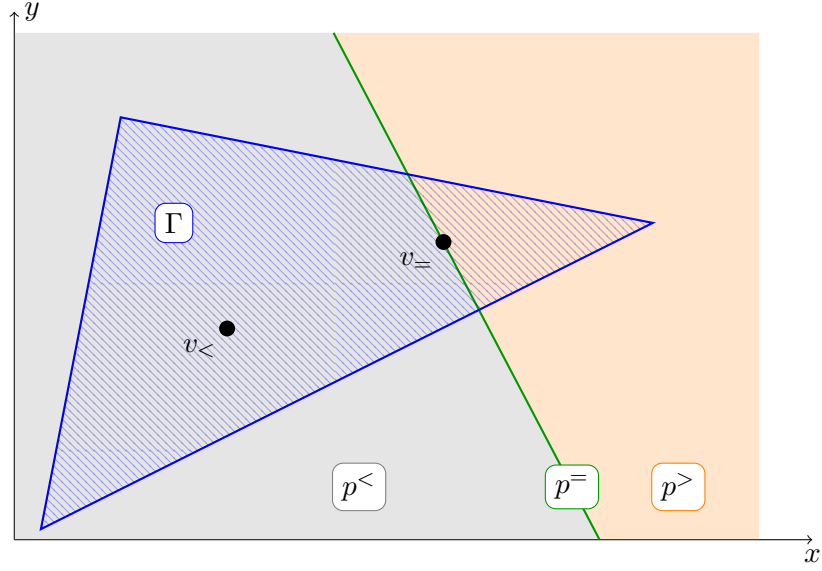
\begin{figure}[t]
		\centering
		\begin{tikzpicture}[scale=\zbtikzbiggerscaling]
	\draw[->] (0,0) -- (7.5,0) coordinate[label = {below:$x$}] (xmax);
	\draw[->] (0,0) -- (0,5) coordinate[label = {right:$y$}] (ymax);
	
	\coordinate (mainlinetl) at (3, 4.8); 
	\coordinate (mainlinebr) at (5.5, 0); 
	
	\coordinate (trtl) at (1, 4); 
	\coordinate (trright) at (6, 3); 
	\coordinate (trbl) at (0.25, 0.1); 
	
	\coordinate (ltdot) at (2, 2); 
	\coordinate (eqdotinterpstart) at (intersection of trtl--trright and mainlinetl--mainlinebr);
	\coordinate (eqdotinterpend) at (intersection of trbl--trright and mainlinetl--mainlinebr);
	\coordinate (eqdot) at ($ (eqdotinterpstart)!0.5!(eqdotinterpend) $); 
	
	\coordinate (reglabellefty) at (0, 0.5);
	\coordinate (reglabelrighty) at (7, 0.5);
	
	\coordinate (regtl) at (0, 4.8);
	\coordinate (regbl) at (0, 0);
	\coordinate (regtr) at (7, 4.8);
	\coordinate (regbr) at (7, 0);
	
	\fill[fill=gray,opacity=0.2] (regtl) -- (mainlinetl) -- (mainlinebr) -- (regbl) -- cycle;
	\fill[fill=orange,opacity=0.2] (mainlinetl) -- (regtr) -- (regbr) -- (mainlinebr) -- cycle;
	
	\draw[thick, OliveGreen] (mainlinetl) -- (mainlinebr);
	
	\filldraw[ultra thin,pattern=north west lines,pattern color=blue, draw opacity=0.3] (trtl) -- (trright) -- (trbl) -- cycle;
	\draw[thick, blue] (trtl) -- (trright) -- (trbl) -- cycle;
	
	\fill[black] (ltdot) circle (0.075);
	\fill[black] (eqdot) circle (0.075);
	
	\node[anchor=north east,thick] at (ltdot) {$ v_< $};
	\node[anchor=north east,thick] at (eqdot) {$ v_= $};
	
	\draw (1.5, 3) node [rounded corners, draw=blue, fill=white] {$ \Gamma $};
	
	\coordinate (reglabelmainlineanchor) at (intersection of reglabellefty--reglabelrighty and mainlinetl--mainlinebr);
	
	\draw ($ (reglabelmainlineanchor) - (2, 0) $) node [rounded corners, draw=gray, fill=white] {$ p^< $};
	\draw (reglabelmainlineanchor) node [rounded corners, draw=OliveGreen, fill=white] {$ p^= $};
	\draw ($ (reglabelmainlineanchor) + (1, 0) $) node [rounded corners, draw=orange, fill=white] {$ p^> $};
\end{tikzpicture}
		\caption{An illustration of how the sets defined by $ \Gamma $, $ p^< $, $ p^= $ and $ p^> $ may look like geometrically in two dimensions, assuming the existence of $ v_< \models \Gamma \wedge p^< $ and $ v_= \models \Gamma \wedge p^= $. The blue region is defined by $ \Gamma $, whereas the green line and its half-spaces colored gray and orange are defined by $ p^= $, $ p^< $ and $ p^> $, respectively. Note that the region defined by $ \Gamma $ is open, i.e., it does not include the points on the blue border.}
		\label{fig:convexity_case_2}
	\end{figure} For a visualization of how the sets defined by $ \Gamma $, $ p^< $, $ p^= $ and $ p^> $ may look like geometrically, see Figure~\ref{fig:convexity_case_2}. Let $ v_= + W $ be the affine vectorspace of solutions to $ \Gamma^= $ and let $ V $ be a vectorspace such that $ \ModelsOf(p^=) = v_= + V $. Note that \[
		\gen{0} \neq \gen{v_= - v_<} = \gen{v_< - v_=} \le W
	\] holds because $ v_< \in v_= + W $ and $ v_< \neq v_= $. Since $ \Gamma \setminus \Gamma^= $ contains only strict inequality predicates, there exists an $ \varepsilon > 0 $ such that \begin{align}
		\label{eqn:veq_minus_epsvltveq_veq_plus_epsvltveq_subset_gamma}
		(v_= - \varepsilon \cdot (v_< - v_=), v_= + \varepsilon \cdot (v_< - v_=)) \subseteq \ModelsOf(\Gamma)
	\end{align} Now our plan is to obtain a contradiction by showing that $ \Gamma \wedge p^> $ is satisfiable. In order to show this, by (\ref{eqn:veq_minus_epsvltveq_veq_plus_epsvltveq_subset_gamma}) it suffices to prove the existence of $ \sigma \in \{1, -1\} $ such that \begin{align}
		\label{eqn:veq_veqplussigmaepsilonvltminusveq_subset_pgt}
		(v_=, v_= + \sigma \cdot \varepsilon \cdot (v_< - v_=)) \subseteq \ModelsOf(p^>)
	\end{align} which we do in the remainder of the proof. First observe the following.
	
	\begin{zbclaim}
		\label{claim:veq_plus_sigmaepsilonvltminusveq_disjoint_peqmodels}
		For every $ \sigma \in \{1, -1\} $, \[
			(v_=, v_= + \sigma \cdot \varepsilon \cdot (v_< - v_=)) \cap \ModelsOf(p^=) = \varnothing
		\]
	\end{zbclaim}
	\begin{proof}
		Towards a contradiction, suppose there exists $ 0 < \delta < \varepsilon $ such that \[
			v_= + \sigma \cdot \delta \cdot (v_< - v_=) \in \underbrace{v_= + V}_{\ModelsOf(p^=)}
		\] is true for some $ \sigma \in \{1, -1\} $. This means that $ \gen{v_< - v_=} \le V $ and thus \[
			v_< = v_= + (\underbrace{v_< - v_=}_{\in V}) \in v_= + V = \ModelsOf(p^=)
		\] This is a contradiction to $ v^< \models p^< $.
	\end{proof}

	\noindent By Claim~\ref{claim:veq_plus_sigmaepsilonvltminusveq_disjoint_peqmodels}, \begin{align}
		\label{eqn:veqcveqsigmaepsilonvltveq_subset_modplt_union_modpgt}
		(v_=, v_= + \sigma \cdot \varepsilon \cdot (v_< - v_=)) \subseteq \ModelsOf(p^<) \cup \ModelsOf(p^>)
	\end{align} holds for every $ \sigma \in \{1, -1\} $. Observe that (for every $ \sigma \in \{1, -1\} $), \[
		(v_=, v_= + \sigma \cdot \varepsilon \cdot (v_< - v_=)) \subseteq \ModelsOf(p^Q)
	\] must hold for some $ Q \in \{<, >\} $, because $ (v_=, v_= + \sigma \cdot \varepsilon \cdot (v_< - v_=)) $ is a convex set and thus assuming the contrary would imply by (\ref{eqn:veqcveqsigmaepsilonvltveq_subset_modplt_union_modpgt}) that $ \ModelsOf(p^<) \cup \ModelsOf(p^>) $ is closed under taking the convex hull of a point from $ \ModelsOf(p^<) $ and from $ \ModelsOf(p^>) $, but this is not true due to $ v_= \models p^= $. Hence, in order to show (\ref{eqn:veq_veqplussigmaepsilonvltminusveq_subset_pgt}) and thus complete the proof of the present case, it suffices to show that \begin{align*}
		(v_=, v_= + \varepsilon \cdot (v_< - v_=)) &\subseteq \ModelsOf(p^Q) \\
		(v_=, v_= - \varepsilon \cdot (v_< - v_=)) &\subseteq \ModelsOf(p^Q)
	\end{align*} is impossible for every $ Q \in \{<, >\} $. Assume it is possible for some $ Q \in \{<, >\} $ and observe that indeed, since $ p^Q $ defines a convex set, the convex hull of $ (v_=, v_= + \varepsilon \cdot (v_< - v_=)) $ and $ (v_=, v_= - \varepsilon \cdot (v_< - v_=)) $ must be contained in $ \ModelsOf(p^Q) $. That is, \[
		v_= \in (v_= - \varepsilon \cdot (v_< - v_=), v_= + \varepsilon \cdot (v_< - v_=)) \subseteq \ModelsOf(p^Q)
	\] Hence, $ v_= \models p^Q $, which is a contradiction.
	
	\textbf{Case 3.} Suppose $ \{Q, R\} = \{>, =\} $. This case is completely analogous to Case 2 (replace in the above reasoning every ``$ < $'' by ``$ > $'' and vice versa).
\end{proof}

\subsubsection{Completeness of $ \Lambda $-proofs}

Having done the above preparatory work, we now prove Theorem~\ref{thm:lambda_proof_soundness_completeness}. We have already shown soundness, so it suffices to prove completeness, i.e., the ``$ \Rightarrow $'' direction. Suppose $ \varphi $ were not $ \Pi $-decomposable. Combining Theorem~\ref{thm:reduction} with Theorem~\ref{thm:cover_alg_correctness} yields the existence of $ \Gamma \in \Sat(\DisjPhi) $ such that \[
\bigvee_{\Lambda \in \Psi_\Gamma} \Lambda \equiv \psi_\Gamma \not\models \varphi
\] Hence, there exists $ \Lambda \in \Psi_\Gamma $ such that $ \Lambda \not\models \varphi $.

We now use an inductive proof technique similar to the one we used in the proof of Theorem~\ref{thm:model_flooding}. Let $ p_1, \dots, p_k $ be some linearization of \[
	\PiComp{\Gamma} = \{p_1, \dots, p_k\}
\]

\begin{zbclaim}
	\label{claim:existence_of_lambda_proof_conditional}
	Let $ m \in \{1, \dots, k + 1\} $ and $ Q = (Q_m, \dots, Q_k) \in \mathbb{P}^{k-m+1} $ be arbitrary. Suppose there exist \[
	\Gamma_0, \Gamma_1 \in \DisjOf{\Lambda \cup \{p_m^{Q_m}, \dots, p_k^{Q_k}\}}
	\] such that $ \Gamma_i \models \varphi $ and $ \Gamma_{1-i} \models \neg\varphi $ for some $ i \in \{0, 1\} $. Then there exists a $ \Lambda $-proof of non-$ \Pi $-decomposability of $ \varphi $.
\end{zbclaim}

In essence, Claim~\ref{claim:existence_of_lambda_proof_conditional} provides a condition which, if fulfilled, implies the desired existence of a $ \Lambda $-proof. Thus, in order to reduce the proof to showing Claim~\ref{claim:existence_of_lambda_proof_conditional}, we need to prove that the condition holds, i.e., that there exist $ \Gamma_0, \Gamma_1 \in \DisjOf{\Lambda} $ such that $ \Gamma_i \models \varphi $ and $ \Gamma_{1-i} \models \neg\varphi $ for some $ i \in \{0, 1\} $. Indeed, define $ \Gamma_1 := \Gamma \cup \Lambda $ and observe that by Lemma~\ref{lemma:psi_properties} \ref{lemma:psi_properties:a}, \[
\DisjOf{\Lambda} \overset{\ref{lemma:psi_properties} \ref{lemma:psi_properties:a}}{\ni} \Gamma_1 \models \Gamma \models \varphi
\] Since $ \Lambda \not\models \varphi $, there must also exist a disjunct $ \Gamma_0 \in \DisjOf{\Lambda} $ such that $ \Gamma_0 \not\models \varphi $. Hence, by Lemma~\ref{lemma:every_disj_either_true_or_false}, $ \Gamma_0 \models \neg\varphi $.

For this reason, to prove the theorem it suffices to show Claim~\ref{claim:existence_of_lambda_proof_conditional}, which we do in the remainder of the proof. More precisely, we prove Claim~\ref{claim:existence_of_lambda_proof_conditional} by induction on $ m $, starting with $ m = 1 $ and going up to $ m = k + 1 $.

\textbf{Base case}: If $ m = 1 $, then by the same reasoning as in the proof of Claim~\ref{claim:model_flooding} above, there is only one disjunct of $ \Lambda \cup \{p_m^{Q_m}, \dots, p_k^{Q_k}\} $, so $ \Gamma_0 = \Gamma_1 $, meaning that (for every $ i $) it is impossible that $ \Gamma_i \models \neg\varphi $ and $ \Gamma_{1-i} \models \varphi $. Hence, the claim for $ m = 1 $ is trivial.

\textbf{Inductive step}: Let $ Q = (Q_{m+1}, \dots, Q_k) \in \mathbb{P}^{k-m} $ and \[
\Gamma_0, \Gamma_1 \in \DisjOf{\Lambda \cup \{p_{m+1}^{Q_{m+1}}, \dots, p_k^{Q_k}\}}
\] be disjuncts such that $ \Gamma_i \models \varphi $ and $ \Gamma_{1-i} \models \neg\varphi $ for some $ i \in \{0, 1\} $. 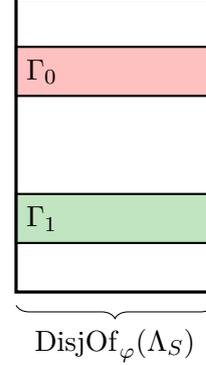
\begin{wrapfigure}{r}{0.3\textwidth}
	\centering
	\if\zbthesismode1
	\vspace{-10pt}
	\else
	\vspace{-20pt}
	\fi
	\begin{tikzpicture}[scale=0.65]
	\coordinate (tl) at (-2, 3);
	\coordinate (br) at (2, -3);
	
	\coordinate (gamma0topy) at (0, 2);
	\coordinate (gamma0boty) at (0, 1);
	\coordinate (gamma1topy) at (0, -1);
	\coordinate (gamma1boty) at (0, -2);
	
	\coordinate (tr) at ($ (tl-|br) $);
	\coordinate (bl) at ($ (br-|tl) $);
	
	\coordinate (gamma0tl) at ($ (gamma0topy-|tl) $);
	\coordinate (gamma0tr) at ($ (gamma0topy-|tr) $);
	\coordinate (gamma0br) at ($ (gamma0boty-|tr) $);
	\coordinate (gamma0bl) at ($ (gamma0boty-|tl) $);
	\filldraw[thick, fill=red, fill opacity=0.25] (gamma0tl) -- (gamma0tr) -- (gamma0br) -- (gamma0bl) -- cycle;
	\coordinate (gamma1tl) at ($ (gamma1topy-|tl) $);
	\coordinate (gamma1tr) at ($ (gamma1topy-|tr) $);
	\coordinate (gamma1br) at ($ (gamma1boty-|tr) $);
	\coordinate (gamma1bl) at ($ (gamma1boty-|tl) $);
	\filldraw[thick, fill=OliveGreen, fill opacity=0.25] (gamma1tl) -- (gamma1tr) -- (gamma1br) -- (gamma1bl) -- cycle;
	
	\draw[very thick] (tl) -- (tr) -- (br) -- (bl) -- cycle;
	
	\node[anchor=west] at ($ (gamma0tl)!0.5!(gamma0bl) $) {$ \Gamma_0 $};
	\node[anchor=west] at ($ (gamma1tl)!0.5!(gamma1bl) $) {$ \Gamma_1 $};
	
	\draw [decorate, decoration = {brace, raise=5pt, amplitude=5pt, mirror}] (bl) --  (br) node[pos=0.5,below=10pt,black]{$ \DisjOf{\Lambda_S} $};
\end{tikzpicture}
	\caption{Illustration of a situation when we can immediately conclude that by the induction hypothesis, there exists a $ \Lambda $-proof of non-$ \Pi $-decomposability of $ \varphi $. The green (resp. red) color is used to illustrate that $ \Gamma_1 $ (resp. $ \Gamma_0 $) entails $ \varphi $ (resp. $ \neg\varphi $).}
	\label{fig:existence_of_lambda_proof_inductive_step_lambda_s_entailment}
	\if\zbthesismode1
	\vspace{-50pt}
	\else
	\vspace{-20pt}
	\fi
\end{wrapfigure} For convenience, we define \[
\Lambda_S := \Lambda \cup \{p_m^S, p_{m+1}^{Q_{m+1}}, \dots, p_k^{Q_k}\}
\] for every $ S \in \mathbb{P} $.

\begin{zbclaim}
	\label{claim:wlog_lambda_s_phi_true_or_false}
	At this point, we can without loss of generality assume that for every $ S \in \mathbb{P} $, $ \Lambda_S \models \varphi $ or $ \Lambda_S \models \neg\varphi $ holds.
\end{zbclaim}
\begin{proof}
	Suppose that for some $ S \in \mathbb{P} $ we have $ \Lambda_S \not\models \varphi $ and $ \Lambda_S \not\models \neg\varphi $. Then there exist disjuncts $ \Gamma_0, \Gamma_1 \in \DisjOf{\Lambda_S} $ such that $ \Gamma_1 \wedge \varphi $ and $ \Gamma_0 \wedge \neg\varphi $ are both satisfiable. Hence, by Lemma~\ref{lemma:every_disj_either_true_or_false}, $ \Gamma_1 \models \varphi $ and $ \Gamma_0 \models \neg\varphi $ (see Figure~\ref{fig:existence_of_lambda_proof_inductive_step_lambda_s_entailment}), so Claim~\ref{claim:existence_of_lambda_proof_conditional} holds by the induction hypothesis.
\end{proof}

\begin{zbclaim}
	\label{claim:wlog_lambda_i_disj_lambdaeq_lambda_1mi_disj_lambdap}
	Furthermore, we can without loss of generality assume that $ \Gamma_1 \in \DisjOf{\Lambda_=} $ and $ \Gamma_0 \in \DisjOf{\Lambda_P} $ holds for some $ P \in \{<, >\} $.
\end{zbclaim}
\begin{proof}
	Let $ S_0, S_1 $ be predicate symbols such that $ \Gamma_0 \in \DisjOf{\Lambda_{S_0}} $ and $ \Gamma_1 \in \DisjOf{\Lambda_{S_1}} $. We can assume that $ S_0 \neq S_1 $ because otherwise Claim~\ref{claim:existence_of_lambda_proof_conditional} immediately follows from the induction hypothesis. If $ \{S_0, S_1\} = \{<, >\} $, then $ \Lambda_< $ and $ \Lambda_> $ are both satisfiable (because $ \Gamma_0 \models \Lambda_{S_0} $ and $ \Gamma_1 \models \Lambda_{S_1} $), meaning that by Lemma~\ref{lemma:predicate_convexity}, $ \Lambda_= $ is satisfiable. Hence, $ \Lambda_= $ must have a satisfiable disjunct. Depending on whether this disjunct entails $ \varphi $ or $ \neg\varphi $ (no other scenario is possible by Lemma~\ref{lemma:every_disj_either_true_or_false}), we redefine $ \Lambda_0 $ or $ \Lambda_1 $ to be that disjunct, while ensuring that $ \Gamma_i \models \neg\varphi $ and $ \Gamma_{1-i} \models \varphi $ remains true (for some $ i \in \{0, 1\} $).
\end{proof}

Let $ P \in \{<, >\} $ be the predicate and $ \Gamma_1 \in \DisjOf{\Lambda_=} $, $ \Gamma_0 \in \DisjOf{\Lambda_P} $ be the disjuncts provided by Claim~\ref{claim:wlog_lambda_i_disj_lambdaeq_lambda_1mi_disj_lambdap}. \begin{figure}[t]
	\centering
	\begin{tikzpicture}[scale=\zbtikzaggressivescaling]
	
	\coordinate (thetatl) at (-4.5, 3);
	\coordinate (thetabr) at (4.5, -3);
	
	\coordinate (thetatr) at ($ (thetatl-|thetabr) $);
	\coordinate (thetabl) at ($ (thetabr-|thetatl) $);
	
	\coordinate (sepleftposx) at (-1.5, 0);
	\coordinate (seprightposx) at (1.5, 0);
	
	\coordinate (septl) at ($ (thetatl-|sepleftposx) $);
	\coordinate (sepbl) at ($ (thetabl-|sepleftposx) $);
	
	\coordinate (septr) at ($ (thetatl-|seprightposx) $);
	\coordinate (sepbr) at ($ (thetabl-|seprightposx) $);
	
	\coordinate (gamma0topy) at (0, 2);
	\coordinate (gamma0boty) at (0, 1);
	\coordinate (gamma1topy) at (0, -1);
	\coordinate (gamma1boty) at (0, -2);
	
	\fill[OliveGreen, opacity=0.25] (thetatl) -- (septl) -- (sepbl) -- (thetabl) -- cycle;
	\fill[red, opacity=0.25] (septl) -- (septr) -- (sepbr) -- (sepbl) -- cycle;
	\fill[OliveGreen, opacity=0.25] (septr) -- (thetatr) -- (thetabr) -- (sepbr) -- cycle;
	
	\draw[thick] (septl) -- (sepbl);
	\draw[thick] (septr) -- (sepbr);
	
	\draw[thick] (thetatl) -- (thetatr) -- (thetabr) -- (thetabl) -- cycle;
	
	\coordinate (gamma0tl) at ($ (gamma0topy-|thetatl) $);
	\coordinate (gamma0tr) at ($ (gamma0topy-|septl) $);
	\coordinate (gamma0br) at ($ (gamma0boty-|septl) $);
	\coordinate (gamma0bl) at ($ (gamma0boty-|thetatl) $);
	\draw[thick] (gamma0tl) -- (gamma0tr) -- (gamma0br) -- (gamma0bl) -- cycle;
	\coordinate (gamma1tl) at ($ (gamma1topy-|septl) $);
	\coordinate (gamma1tr) at ($ (gamma1topy-|septr) $);
	\coordinate (gamma1br) at ($ (gamma1boty-|septr) $);
	\coordinate (gamma1bl) at ($ (gamma1boty-|septl) $);
	\draw[thick] (gamma1tl) -- (gamma1tr) -- (gamma1br) -- (gamma1bl) -- cycle;
	
	\node[anchor=west] at ($ (gamma0tl)!0.5!(gamma0bl) $) {$ \Gamma_0 $};
	\node[anchor=east] at ($ (gamma1tr)!0.5!(gamma1br) $) {$ \Gamma_1 $};
	
	\coordinate(lambdalttc) at ($ (thetatl)!0.5!(septl) $);
	\coordinate(lambdaltbc) at ($ (thetabl)!0.5!(sepbl) $);
	\coordinate(lambdaeqtc) at ($ (septl)!0.5!(septr) $);
	\coordinate(lambdaeqbc) at ($ (sepbl)!0.5!(sepbr) $);
	\coordinate(lambdagttc) at ($ (septr)!0.5!(thetatr) $);
	\coordinate(lambdagtbc) at ($ (sepbr)!0.5!(thetabr) $);
	\draw ($ (lambdalttc)!0.5!(lambdaltbc) $) node [rounded corners, draw=OliveGreen, fill=white] {$ \Lambda_< \models \varphi $};
	\draw ($ (lambdaeqtc)!0.5!(lambdaeqbc) $) node [rounded corners, draw=red, fill=white] {$ \Lambda_= \models \neg\varphi $};
	\draw ($ (lambdagttc)!0.5!(lambdagtbc) $) node [rounded corners, draw=OliveGreen, fill=white] {$ \Lambda_> \models \varphi $};
	
	\draw [decorate, decoration = {brace, raise=5pt, amplitude=5pt, mirror}] (thetabl) --  (sepbl) node[pos=0.5,below=10pt,black]{$ \DisjOf{\Lambda_<} $};
	\draw [decorate, decoration = {brace, raise=5pt, amplitude=5pt, mirror}] (sepbl) --  (sepbr) node[pos=0.5,below=10pt,black]{$ \DisjOf{\Lambda_=} $};
	\draw [decorate, decoration = {brace, raise=5pt, amplitude=5pt, mirror}] (sepbr) --  (thetabr) node[pos=0.5,below=10pt,black]{$ \DisjOf{\Lambda_>} $};
	
\end{tikzpicture}
	\caption{Venn diagram illustrating the way we partition $ \DisjOf{\Lambda \cup \{p_{m+1}^{Q_{m+1}}, \dots, p_k^{Q_k}\}} $ into groups of disjuncts and the assumption we without loss of generality make about $ \Gamma_0 $ and $ \Gamma_1 $ in Claim~\ref{claim:wlog_lambda_i_disj_lambdaeq_lambda_1mi_disj_lambdap}. The regions labeled $ \Gamma_0 $ and $ \Gamma_1 $ in the diagram are to be interpreted as elements (i.e., not subsets) of $ \DisjOf{\Lambda_<} $ and $ \DisjOf{\Lambda_=} $, respectively. The colors depict a possible way disjuncts of $ \Lambda $ can entail $ \varphi $ or $ \neg\varphi $, in case Claim~\ref{claim:existence_of_lambda_proof_conditional} does not already follow from the induction hypothesis by Claim~\ref{claim:wlog_lambda_s_phi_true_or_false}. More precisely, a disjunct entails $ \varphi $ (resp. $ \neg\varphi $) if it is located in a region colored green (resp. red).}
	\label{fig:existence_of_lambda_proof_inductive_step}
\end{figure} A way of thinking about the disjuncts of $ \Lambda_S $ and about what Claims \ref{claim:wlog_lambda_s_phi_true_or_false} and \ref{claim:wlog_lambda_i_disj_lambdaeq_lambda_1mi_disj_lambdap} intuitively say is visualized in Figure~\ref{fig:existence_of_lambda_proof_inductive_step}. Applying Claim~\ref{claim:wlog_lambda_s_phi_true_or_false} yields the following property, which is crucial for the correctness of the $ \Lambda $-proof we are about to construct.

\begin{zbclaim}
	\label{claim:lambda_eq_disagrees_with_lambda_p_on_phi_entailment}
	Precisely one of the following two statements is true. \begin{itemize}
		\item $ \Lambda_= \models \varphi $ and $ \Lambda_P \models \neg\varphi $
		\item $ \Lambda_= \models \neg\varphi $ and $ \Lambda_P \models \varphi $
	\end{itemize}
\end{zbclaim}
\begin{proof}
	By Claim~\ref{claim:wlog_lambda_s_phi_true_or_false}, $ \Gamma_1 \models \varphi $ is equivalent to $ \Lambda_= \models \varphi $ and, similarly, $ \Gamma_1 \models \neg\varphi $ is equivalent to $ \Lambda_= \models \neg\varphi $. Likewise, $ \Gamma_0 \models \varphi $ is equivalent to $ \Lambda_P \models \varphi $ and $ \Gamma_0 \models \neg\varphi $ is equivalent to $ \Lambda_P \models \neg\varphi $. Since $ \Gamma_i \models \neg\varphi $ and $ \Gamma_{1-i} \models \varphi $ holds for some $ i \in \{0, 1\} $, the claim follows.
\end{proof}

Having done the above preparatory work, we now start with the actual construction of the $ \Lambda $-proof of non-$ \Pi $-decomposability. The general idea is to define the two disjuncts $ \Lambda_0' $ and $ \Lambda_1' $ in a bottom-up and greedy manner, by starting with $ \Lambda \cup \{p_{m+1}^{Q_{m+1}}, \dots, p_k^{Q_k}\} $ and iteratively adding certain predicates from $ \Gamma $ with the predicate symbol replaced in a way which guarantees that in the next iteration we will still be able to construct the $ \Lambda $-proof. More precisely, in order to obtain a disjunct of $ \Lambda \cup \{p_{m+1}^{Q_{m+1}}, \dots, p_k^{Q_k}\} $, by definition it suffices to keep adding $ p_i^{Q_i'} $ to $ \Lambda $ for all $ i \in \{1, \dots, m\} $, where $ Q_i' \in \mathbb{P} $ is some suitable predicate.

We now give the details of the construction. Define \begin{align*}
	\Upsilon^0 &:= \Lambda \cup \{p_{m+1}^{Q_{m+1}}, \dots, p_k^{Q_k}\} \\
	\Upsilon_P^0 &:= \varnothing \\
	\Upsilon_=^0 &:= \varnothing
\end{align*} For $ i \in \{1, \dots, m\} $, we construct $ \Upsilon^i $ assuming $ \Upsilon^{i-1} $ is already defined, so that the following invariants \ref{claim:existence_of_lambda_proof_conditional:inv_a} and \ref{claim:existence_of_lambda_proof_conditional:inv_b} hold: \begin{enumerate}[label=\textbf{(I\arabic*)}, leftmargin=*]
	\item \label{claim:existence_of_lambda_proof_conditional:inv_a} $ \Upsilon^i \wedge p_m^P $ and $ \Upsilon^i \wedge p_m^= $ are both satisfiable
	\item \label{claim:existence_of_lambda_proof_conditional:inv_b} $ \Upsilon^i \wedge p_m^P \models \Upsilon_P^i $ and $ \Upsilon^i \wedge p_m^= \models \Upsilon_=^i $
\end{enumerate}

By Claim~\ref{claim:wlog_lambda_i_disj_lambdaeq_lambda_1mi_disj_lambdap}, invariant \ref{claim:existence_of_lambda_proof_conditional:inv_a} holds for the base case $ i = 0 $. It is trivial that $ \Upsilon_P^0 = \varnothing $ and $ \Upsilon_=^0 = \varnothing $ establish invariant \ref{claim:existence_of_lambda_proof_conditional:inv_b}.

We now give and discuss the inductive construction of $ \Upsilon^i $, $ \Upsilon_P^i $ and $ \Upsilon_=^i $ for $ i \in \{1, \dots, m\} $, assuming these predicate sets are already constructed for $ i - 1 $ so that the invariants hold. \begin{figure}[t]
	\if\zbthesismode1
	\newenvironment{lambdaproofcasesplitfigure}{\begin{tabular}{|P{2.75em}|P{6em}|P{6em}|}}{\end{tabular}}
	\else
	\newenvironment{lambdaproofcasesplitfigure}{\begin{tabular}{|P{3em}|P{7em}|P{7em}|}}{\end{tabular}}
	\fi
	\centering
	\begin{subfigure}{0.49\textwidth}
		\centering
		\begin{lambdaproofcasesplitfigure}
			\hline
			\multirow{2}*{$ Q \in \mathbb{P} $} & \multicolumn{2}{|c|}{Satisfiability} \\ \cline{2-3}
			\rule{0pt}{3ex} & $ \Upsilon^{i-1} \wedge p_i^Q \wedge p_m^P $ & $ \Upsilon^{i-1} \wedge p_i^Q \wedge p_m^= $ \\
			\hline\hline
			$ < $ & \cellcolor{LightCyan}sat & \cellcolor{LightCyan}sat \\
			$ = $ & unsat & sat \\
			$ > $ & sat & unsat \\ \hline
		\end{lambdaproofcasesplitfigure}
		\caption{\label{fig:lambda_proof_case_split:case_1}}
	\end{subfigure} \hfill
	\begin{subfigure}{0.49\textwidth}
		\centering
		\begin{lambdaproofcasesplitfigure}
			\hline
			\multirow{2}*{$ Q \in \mathbb{P} $} & \multicolumn{2}{|c|}{Satisfiability} \\ \cline{2-3}
			\rule{0pt}{3ex} & $ \Upsilon^{i-1} \wedge p_i^Q \wedge p_m^P $ & $ \Upsilon^{i-1} \wedge p_i^Q \wedge p_m^= $ \\
			\hline\hline
			$ < $ & sat & \cellcolor{LightCyan}unsat \\
			$ = $ & \cellcolor{LightCyan}unsat & sat \\
			$ > $ & unsat & sat \\ \hline
		\end{lambdaproofcasesplitfigure}
		\caption{\label{fig:lambda_proof_case_split:case_2}}
	\end{subfigure} \\[1ex]
	\begin{subfigure}{0.49\textwidth}
		\centering
		\begin{lambdaproofcasesplitfigure}
			\hline
			\multirow{2}*{$ Q \in \mathbb{P} $} & \multicolumn{2}{|c|}{Satisfiability} \\ \cline{2-3}
			\rule{0pt}{3ex} & $ \Upsilon^{i-1} \wedge p_i^Q \wedge p_m^P $ & $ \Upsilon^{i-1} \wedge p_i^Q \wedge p_m^= $ \\
			\hline\hline
			$ < $ & \cellcolor{LightCyan}unsat & sat \\
			$ = $ & \cellcolor{LightCyan}unsat & sat \\
			$ > $ & \cellcolor{LightCyan}unsat & sat \\ \hline
		\end{lambdaproofcasesplitfigure}
		\caption{\label{fig:lambda_proof_case_split:case_3}}
	\end{subfigure} \hfill
	\begin{subfigure}{0.49\textwidth}
		\centering
		\begin{lambdaproofcasesplitfigure}
			\hline
			\multirow{2}*{$ Q \in \mathbb{P} $} & \multicolumn{2}{|c|}{Satisfiability} \\ \cline{2-3}
			\rule{0pt}{3ex} & $ \Upsilon^{i-1} \wedge p_i^Q \wedge p_m^P $ & $ \Upsilon^{i-1} \wedge p_i^Q \wedge p_m^= $ \\
			\hline\hline
			$ < $ & sat & \cellcolor{LightCyan}unsat \\
			$ = $ & sat & \cellcolor{LightCyan}unsat \\
			$ > $ & sat & \cellcolor{LightCyan}unsat \\ \hline
		\end{lambdaproofcasesplitfigure}
		\caption{\label{fig:lambda_proof_case_split:case_4}}
	\end{subfigure}
	\caption{Four examples of possible ways $ \Upsilon^{i-1} \wedge p_i^Q \wedge p_m^P $ and $ \Upsilon^{i-1} \wedge p_i^Q \wedge p_m^= $ may be satisfiable (``sat'') or unsatisfiable (``unsat'') for various $ Q \in \mathbb{P} $, illustrating the intuition behind the proof by cases. More precisely, examples \ref{fig:lambda_proof_case_split:case_1}, \ref{fig:lambda_proof_case_split:case_2}, \ref{fig:lambda_proof_case_split:case_3} and \ref{fig:lambda_proof_case_split:case_4} fall into cases 1, 2, 3 and 4, respectively. The highlighted cells are the ones that determine the case handling the corresponding example.}
	\label{fig:lambda_proof_case_split}
\end{figure} We distinguish between the following four cases; for each of them, an example of a scenario handled by it is given in Figure~\ref{fig:lambda_proof_case_split}.

\begin{enumerate}[label=\textbf{Case \arabic*.}, leftmargin=*]
	\item Suppose there exists $ Q \in \mathbb{P} $ such that $ \Upsilon^{i-1} \wedge p_i^Q \wedge p_m^P $ and $ \Upsilon^{i-1} \wedge p_i^Q \wedge p_m^= $ are both satisfiable. For a concrete example, see Figure~\ref{fig:lambda_proof_case_split:case_1} wherein this case is visualized for $ Q = P_< $. We define \[
		\Upsilon^i := \Upsilon^{i-1} \cup \{p_i^Q\}
	\] and leave $ \Upsilon_P^i := \Upsilon_P^{i-1} $, $ \Upsilon_=^i := \Upsilon_=^{i-1} $ unchanged. Observe that this maintains the invariants \ref{claim:existence_of_lambda_proof_conditional:inv_a} and \ref{claim:existence_of_lambda_proof_conditional:inv_b}.
\end{enumerate}

Note that if we did not fall into case 1, then for all $ Q \in \mathbb{P} $, $ \Upsilon^{i-1} \wedge p_i^Q \wedge p_m^P $ or $ \Upsilon^{i-1} \wedge p_i^Q \wedge p_m^= $ must be unsatisfiable. We split this case into the following three subcases.

\begin{enumerate}[label=\textbf{Case \arabic*.}, leftmargin=*]
	\setcounter{enumi}{1}
	\item Suppose there exist $ Q_P, Q_= \in \mathbb{P} $ such that $ \Upsilon^{i-1} \wedge p_i^{Q_P} \wedge p_m^P $ and $ \Upsilon^{i-1} \wedge p_i^{Q_=} \wedge p_m^= $ are both unsatisfiable. For a concrete example, see Figure~\ref{fig:lambda_proof_case_split:case_2} wherein this case is visualized for $ Q_P = P_= $ and $ Q_= = P_< $. Expressing unsatisfiability as an entailment yields \begin{align*}
		\Upsilon^{i-1} \cup \{p_m^P\} &\models \neg p_i^{Q_P} \equiv \bigvee_{Q_P' \in \mathbb{P} \setminus \{Q_P\}} p_i^{Q_P'} \\
		\Upsilon^{i-1} \cup \{p_m^=\} &\models \neg p_i^{Q_=} \equiv \bigvee_{Q_=' \in \mathbb{P} \setminus \{Q_=\}} p_i^{Q_='}
	\end{align*} Hence, by the Convexity Lemma~\ref{lemma:convexity} there exist predicates $ Q_P \neq Q_P' \in \mathbb{P} $ and $ Q_= \neq Q_=' \in \mathbb{P} $ such that \begin{align}
		\label{eqn:upsiloniminus1_union_pmp_entails_piqpprime}
		\Upsilon^{i-1} \cup \{p_m^P\} &\models p_i^{Q_P'} \\
		\label{eqn:upsiloniminus1_union_pmeq_entails_piqeqprime}
		\Upsilon^{i-1} \cup \{p_m^=\} &\models p_i^{Q_='}
	\end{align} \begin{enumerate}[label=\textbf{Case 2.\arabic*.}, leftmargin=*]
		\item Suppose $ Q_P' = Q_=' $. Then we define \[
			\Upsilon^i := \Upsilon^{i-1} \cup \{p_i^{Q_P'}\}
		\] and leave $ \Upsilon_P^i := \Upsilon_P^{i-1} $, $ \Upsilon_=^i := \Upsilon_=^{i-1} $ unchanged.
		\item Suppose $ Q_P' \neq Q_=' $. Then we set \begin{align*}
			\Upsilon^i &:= \Upsilon^{i-1} \\
			\Upsilon_P^i &:= \Upsilon_P^{i-1} \cup \{p_i^{Q_P'}\} \\
			\Upsilon_=^i &:= \Upsilon_=^{i-1} \cup \{p_i^{Q_='}\}
		\end{align*}
	\end{enumerate} Clearly, (\ref{eqn:upsiloniminus1_union_pmp_entails_piqpprime}) and (\ref{eqn:upsiloniminus1_union_pmeq_entails_piqeqprime}) ensure that in both cases the invariants \ref{claim:existence_of_lambda_proof_conditional:inv_a} and \ref{claim:existence_of_lambda_proof_conditional:inv_b} are maintained.
\end{enumerate}

In essence, case 1 handles all the possible satisfiability tables (see Figure~\ref{fig:lambda_proof_case_split}) containing a row where both entries are ``sat''. If this is not the case, then every row must have an ``unsat'' entry. Case 2 handles the event when a pair of these ``unsat'' entries is in different columns. Hence, the only scenarios remaining to be handled are those where all ``unsat'' entries are in the same column (see Figures \ref{fig:lambda_proof_case_split:case_3} and \ref{fig:lambda_proof_case_split:case_4}). These scenarios are handled by the following cases 3 and 4.

\begin{enumerate}[label=\textbf{Case \arabic*.}, leftmargin=*]
	\setcounter{enumi}{2}
	\item Suppose that for all $ Q_P \in \mathbb{P} $, $ \Upsilon^{i-1} \wedge p_i^{Q_P} \wedge p_m^P $ is unsatisfiable (see Figure~\ref{fig:lambda_proof_case_split:case_3}). But then $ \Upsilon^{i-1} \wedge p_m^P $ is unsatisfiable, which contradicts the assumption that invariant \ref{claim:existence_of_lambda_proof_conditional:inv_a} holds after the $ (i - 1) $-th step of the construction. Hence, this case is not possible.
	\item Suppose that for all $ Q_= \in \mathbb{P} $, $ \Upsilon^{i-1} \wedge p_i^{Q_=} \wedge p_m^= $ is unsatisfiable (see Figure~\ref{fig:lambda_proof_case_split:case_4}). Then, similarly to case 3 above, $ \Upsilon^{i-1} \wedge p_m^= $ is unsatisfiable, which is a contradiction to invariant \ref{claim:existence_of_lambda_proof_conditional:inv_a}, so this case cannot occur.
\end{enumerate}

To sum up, at this point we have shown the existence of $ \Upsilon^m, \Upsilon_P^m $ and $ \Upsilon_=^m $ satisfying invariants \ref{claim:existence_of_lambda_proof_conditional:inv_a} and \ref{claim:existence_of_lambda_proof_conditional:inv_b}. Moreover, it follows from the above construction that the predicate sets\begin{align}
	\label{eqn:suff_cond_true_false_disj_imply_nondec_proof_p}
	\{p_m^P\} \cup \Upsilon^m \cup \Upsilon_P^m &\in \DisjOf{\Lambda_P} \subseteq \DisjOf{\Lambda} \\
	\label{eqn:suff_cond_true_false_disj_imply_nondec_proof_eq}
	\{p_m^=\} \cup \Upsilon^m \cup \Upsilon_=^m &\in \DisjOf{\Lambda_=} \subseteq \DisjOf{\Lambda}
\end{align} are indeed disjuncts of $ \Lambda $. More precisely, this is the case because we started with $ \Upsilon^0 = \Lambda \cup \{p_{m+1}^{Q_{m+1}}, \dots, p_k^{Q_k}\} $ and for all $ i \in \{1, \dots, m\} $ we added $ p_i^{Q_i'} $ for some $ Q_i' \in \mathbb{P} $ either to $ \Upsilon^i $ or to both $ \Upsilon_=^i $ and $ \Upsilon_P^i $ depending on the case we considered. All other predicates needed to obtain a disjunct are already present in $ \PiSimp{\Gamma} \subseteq \Theta \subseteq \Lambda $ and $ \{p_{m+1}^{Q_{m+1}}, \dots, p_k^{Q_k}\} $. The satisfiability of both $ \{p_m^P\} \cup \Upsilon^m \cup \Upsilon_P^m $ and $ \{p_m^=\} \cup \Upsilon^m \cup \Upsilon_=^m $ follows from the \ref{claim:existence_of_lambda_proof_conditional:inv_a} and \ref{claim:existence_of_lambda_proof_conditional:inv_b} invariants we established. Observe also that \ref{claim:existence_of_lambda_proof_conditional:inv_a} and \ref{claim:existence_of_lambda_proof_conditional:inv_b} guarantee that $ p_m^P $ (resp. $ p_m^= $) must have been added to $ \Upsilon_P^m $ (resp. $ \Upsilon_=^m $), so in (\ref{eqn:suff_cond_true_false_disj_imply_nondec_proof_p}) and (\ref{eqn:suff_cond_true_false_disj_imply_nondec_proof_eq}) the union with $ \{p_m^P\} $ and $ \{p_m^=\} $ is redundant and can be removed.

We are now ready to construct the following $ \Lambda $-proof $ (\Lambda_0, \Lambda_1, \Lambda_0', \Lambda_1') $ of non-$ \Pi $-decomposability of $ \varphi $: \begin{align}
	\label{eqn:completeness_proof_def_lambda0_prime}
	\Lambda_0' &:= \Upsilon^m \cup \Upsilon_P^m \in \DisjOf{\Lambda_P} \subseteq \DisjOf{\Lambda} \\
	\label{eqn:completeness_proof_def_lambda1_prime}
	\Lambda_1' &:= \Upsilon^m \cup \Upsilon_=^m \in \DisjOf{\Lambda_=} \subseteq \DisjOf{\Lambda} \\
	\label{eqn:completeness_proof_def_lambda0}
	\Lambda_0 &:= \Upsilon^m \cup \{p_m^P\} \subseteq \Lambda_0' \models \Lambda_0 \\
	\label{eqn:completeness_proof_def_lambda1}
	\Lambda_1 &:= \Upsilon^m \cup \{p_m^=\} \subseteq \Lambda_1' \models \Lambda_1
\end{align} By invariant \ref{claim:existence_of_lambda_proof_conditional:inv_b} (and the fact that $ \Lambda_0' \models \Lambda_0 $, $ \Lambda_1' \models \Lambda_1 $), it holds that $ \Lambda_0 \equiv \Lambda_0' $ and $ \Lambda_1 \equiv \Lambda_1' $. We now prove that all remaining properties of a $ \Lambda $-proof $ (\Lambda_0, \Lambda_1, \Lambda_0', \Lambda_1') $ are indeed satisfied.

\begin{zbclaim}
	\label{claim:lambda_0_lambda_1_pi_complex}
	$ \Lambda_0 $ and $ \Lambda_1 $ are both $ \Pi $-complex.
\end{zbclaim}
\begin{proof}
	Let $ \Theta $ be the corresponding set computed in the call to the covering algorithm which produced $ \Lambda $. Since $ \Lambda_0', \Lambda_1' \in \DisjOf{\Lambda} $ and $ \Theta \subseteq \Lambda $, observe that $ \Lambda_0' $ and $ \Lambda_1' $ can each be written as a union of a predicate set without equalities and a disjunct of $ \Theta $, which is $ \Pi $-complex by Lemma~\ref{lemma:psi_properties} \ref{lemma:psi_properties:c}. Since adding strict inequalities cannot make a predicate set $ \Pi $-simple if the resulting set is satisfiable (Lemma~\ref{lemma:adding_neq_predicates_cannot_make_pi_complex_set_pi_simple}), applying this lemma inductively yields that $ \Lambda_0' $ and $ \Lambda_1' $ are both $ \Pi $-complex. Hence, the claim follows (since $ \Lambda_0 \equiv \Lambda_0' $ and $ \Lambda_1 \equiv \Lambda_1' $).
\end{proof}

\noindent We now show that $ \Lambda_0 $ and $ \Lambda_1 $ have the same set of $ Z $-dependencies.

\begin{zbclaim}
	\label{claim:lambda_0_lindep_eq_lambda_1_lindep}
	For all $ Z \in \Pi $ it holds that \[
		\LinDep_{\pi_Z}(\Lambda_0) = \LinDep_{\pi_Z}(\Lambda_1)
	\]
\end{zbclaim}
\begin{proof}
	This can be shown in a way similar to how we proved Claim~\ref{claim:lindep_lambdaeq_eq_lindep_lambdasi} above. Let $ \Gamma_0, \Gamma_1 \in \DisjOf{\Theta} $ be disjuncts such that \begin{align*}
		\Lambda_0' = \Gamma_0 \cup (\Lambda_0' \setminus \Gamma_0) \\
		\Lambda_1' = \Gamma_1 \cup (\Lambda_1' \setminus \Gamma_1)
	\end{align*} Since $ \Lambda_0' \setminus \Gamma_0 $ and $ \Lambda_1' \setminus \Gamma_1 $ contain only strict inequality predicates, applying Lemma~\ref{lemma:psi_properties} \ref{lemma:psi_properties:b} and Theorem~\ref{thm:only_equality_predicates_can_establish_lindep_strong} (see Appendix~\ref{sec:app:lindep_facts}) yields that \[
		\LinDep_{\pi_Z}(\Lambda_0') \overset{\ref{thm:only_equality_predicates_can_establish_lindep_strong}}{=} \LinDep_{\pi_Z}(\Gamma_0) \overset{\ref{lemma:psi_properties} \ref{lemma:psi_properties:b}}{=} \LinDep_{\pi_Z}(\Gamma_1) \overset{\ref{thm:only_equality_predicates_can_establish_lindep_strong}}{=} \LinDep_{\pi_Z}(\Lambda_1')
	\] holds for all $ Z \in \Pi $. Since $ \Lambda_0 \equiv \Lambda_0' $ and $ \Lambda_1 \equiv \Lambda_1' $, the claim follows.
\end{proof}

Furthermore, note that by construction $ \Lambda_1 $ is $ p_m^= $-next to $ \Lambda_0 $. Finally, by Claim~\ref{claim:lambda_eq_disagrees_with_lambda_p_on_phi_entailment} we are guaranteed to land into precisely one of the following two cases: \begin{itemize}
	\item $ \Lambda_0' \models \Lambda_P \models \varphi $ and $ \Lambda_1' \models \Lambda_= \models \neg\varphi $
	\item $ \Lambda_0' \models \Lambda_P \models \neg\varphi $ and $ \Lambda_1' \models \Lambda_= \models \varphi $
\end{itemize} Hence, for some $ i \in \{0, 1\} $ it must be the case that $ \Lambda_i' \models \varphi $ and $ \Lambda_{1-i}' \models \neg\varphi $. Overall, we conclude that $ (\Lambda_0, \Lambda_1, \Lambda_0', \Lambda_1') $ is a $ \Lambda $-proof of non-$ \Pi $-decomposability of $ \varphi $. This completes the proof of Claim~\ref{claim:existence_of_lambda_proof_conditional} and hence of Theorem~\ref{thm:lambda_proof_soundness_completeness}.

Since $ \Psi_\Gamma $ is a set of DNF terms of $ \psi_\Gamma $, every $ \Lambda \in \Psi_\Gamma $ is a result of exactly one (possibly recursive) call to the covering algorithm. Furthermore, observe that every predicate in $ \Lambda $ either originates in $ \Theta $ or is a result of some iteration of the second loop of the covering algorithm (see Line~\ref{alg:cover:line:second_loop}). Since this loop can make at most exponentially many iterations and $ \Theta $'s representation has polynomial size regardless of the recursive call depth at which $ \Theta $ was computed, it follows that $ \Lambda $ can be encoded using at most exponentially (in the size of $ \varphi $) many bits. Hence, every $ \Lambda $-proof can also be encoded as a string of exponential length. We have shown that $ \Lambda $-proofs are sound and complete (Theorem~\ref{thm:lambda_proof_soundness_completeness}) and verifiable in time polynomial in the size of the proof, so we obtain a $ \NEXP $ upper bound for deciding non-$ \Pi $-decomposability and consequently a $ \coNEXP $ algorithm for determining variable decomposability:

\begin{theorem}
	\label{thm:conexp_upper_bound}
	Given a formula $ \varphi \in \QFLRA $ and a partition $ \Pi $, deciding the $ \Pi $-decomposability of $ \varphi $ is in $ \coNEXP $.
\end{theorem}

\subsection{Compressing non-decomposability proofs}
\label{sec:vardec:proof_compression}

The only disadvantage of the $ \Lambda $-proof based system we have constructed is the exponential worst-case size of the proofs. In this section, we analyze the reason behind this exponential blowup and show how to compress $ \Lambda $-proofs of non-$ \Pi $-decomposability into polynomial size without losing soundness, completeness and the ability to verify an alleged proof in polynomial time. This allows us to significantly strengthen the $ \coNEXP $ upper bound of Theorem~\ref{thm:conexp_upper_bound} by designing a $ \coNP $ algorithm for deciding variable decomposability.

We resume the above proof of completeness that has lead to the $ \Lambda $-proof defined in (\ref{eqn:completeness_proof_def_lambda0_prime}--\ref{eqn:completeness_proof_def_lambda1}) and analyze its structure. Let $ \Theta $ be the corresponding set computed in the call to the covering algorithm which produced $ \Lambda $. Since $ \Lambda_0', \Lambda_1' \in \DisjOf{\Lambda} $ and $ \Theta \subseteq \Lambda $, there exist $ \Gamma_0, \Gamma_1 \in \DisjOf{\Theta} $ such that $ \Lambda_0' $ and $ \Lambda_1' $ can be written as \begin{align*}
	\Lambda_0' &= \Gamma_0 \cup (\Lambda_0' \setminus \Gamma_0) \\
	\Lambda_1' &= \Gamma_1 \cup (\Lambda_1' \setminus \Gamma_1)
\end{align*} where $ \Lambda_0' \setminus \Gamma_0 $ and $ \Lambda_1' \setminus \Gamma_1 $ contain only strict inequality predicates. \begin{figure}[t]
	\centering
	\input{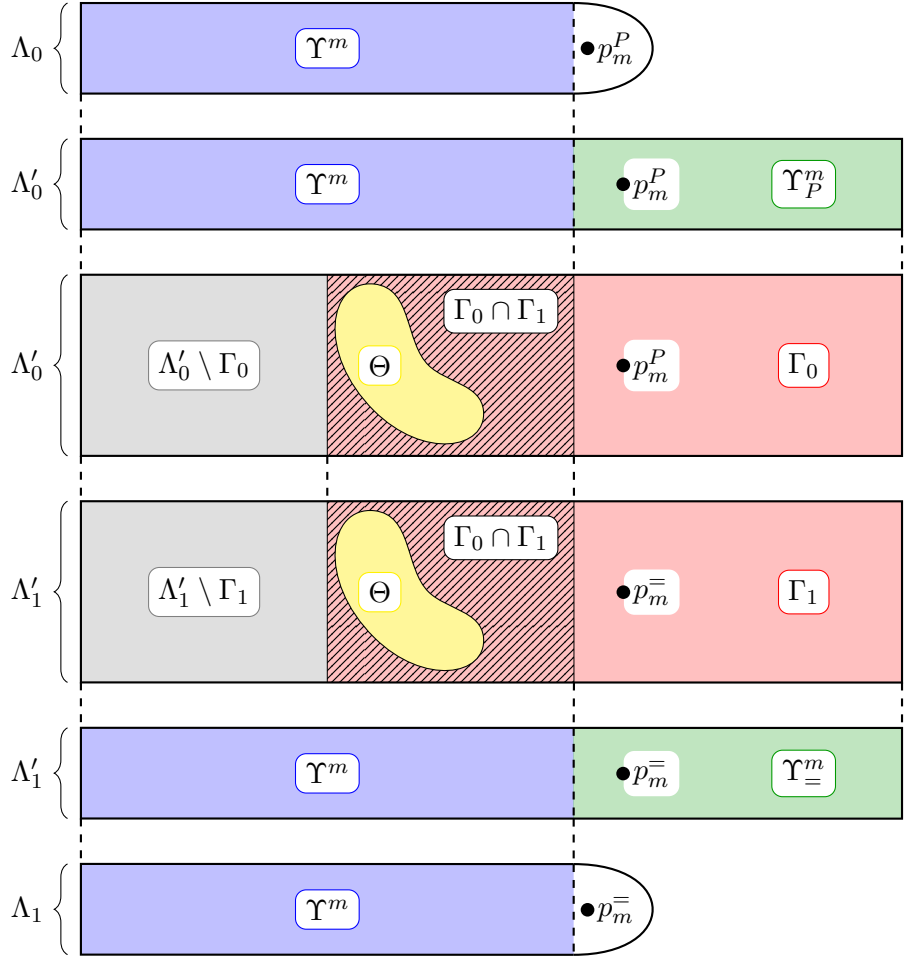}
	\caption{Six Venn diagrams illustrating the inner structure of the components $ \Lambda_0 $, $ \Lambda_1 $, $ \Lambda_0' $ and $ \Lambda_1' $ of a $ \Lambda $-proof and the relationships between specified subsets of those components. Pairs of dashed lines indicate equal subsets in different diagrams. No elements can exist where the background is white.}
	\label{fig:constructed_lambda_proof}
\end{figure} The relationships between $ \Gamma_0 $, $ \Gamma_1 $, the components of the $ \Lambda $-proof we have constructed and other relevant predicate sets are visualized in Figure~\ref{fig:constructed_lambda_proof}. We now explain and discuss these relationships in more detail.

\subsubsection{Partitioning of $ \Lambda $-proof components}

First observe that $ \Upsilon^m $ cannot agree with $ \Upsilon_P^m $ or $ \Upsilon_=^m $ on any predicate because while constructing these sets we either added $ p_i^Q $ for some $ Q \in \mathbb{P} $ to $ \Upsilon^i $, or added $ p_i $ (with the predicate symbol changed accordingly) both to $ \Upsilon_P^i $ and $ \Upsilon_=^i $. Since for every $ i $ exactly one of these decisions has been taken and $ \Upsilon_P^m $ clearly cannot agree with $ \Lambda $ on any element (because $ \Lambda $ contains only $ \Pi $-respecting predicates, whereas all elements of $ \Upsilon_P^m $ do not respect $ \Pi $), it follows that $ \Upsilon^m \cap \Upsilon_P^m = \varnothing $. Analogous reasoning also shows that $ \Upsilon^m \cap \Upsilon_=^m = \varnothing $. Note also that $ \Theta \subseteq \Gamma_0 \cap \Gamma_1 $ follows immediately from the fact that $ \Gamma_0 $ and $ \Gamma_1 $ are both disjuncts of $ \Theta $. We now use these observations to prove the following property that will play an important role later when we will be proving the correctness of the compressed proof of non-$ \Pi $-decomposability.

\begin{zbclaim}
	\label{claim:upsilonm_lamda0primeexceptgamma0_union_gamma0gamma1}
	$ \Upsilon^m = (\Lambda_0' \setminus \Gamma_0) \cup (\Gamma_0 \cap \Gamma_1) = (\Lambda_1' \setminus \Gamma_1) \cup (\Gamma_0 \cap \Gamma_1) $
\end{zbclaim}

For the proof of Claim~\ref{claim:upsilonm_lamda0primeexceptgamma0_union_gamma0gamma1}, we will need the following observation.

\begin{zbclaim}
	\label{claim:upsilonm_eq_lambda0prime_intersection_lambda1prime}
	$ \Upsilon^m = \Lambda_0' \cap \Lambda_1' $
\end{zbclaim}
\begin{proof}
	The ``$ \subseteq $'' inclusion is clear by definition of $ \Lambda_0' $ and $ \Lambda_1' $ in (\ref{eqn:completeness_proof_def_lambda0_prime}) and (\ref{eqn:completeness_proof_def_lambda1_prime}). Now observe that the inductive construction of $ \Upsilon^i $, $ \Upsilon_P^i $ and $ \Upsilon_=^i $ above explicitly ensures that $ \Upsilon_P^m $ and $ \Upsilon_=^m $ are disjoint (see case 2.2 and note that no other case modifies $ \Upsilon_P^i $ or $ \Upsilon_=^i $). This observation taken together with $ \Upsilon^m \cap \Upsilon_P^m = \varnothing $ and $ \Upsilon^m \cap \Upsilon_=^m = \varnothing $ yields the converse inclusion.
\end{proof}
\begin{proof}[Proof of Claim~\ref{claim:upsilonm_lamda0primeexceptgamma0_union_gamma0gamma1}]
	We only prove the $ \Upsilon^m = (\Lambda_0' \setminus \Gamma_0) \cup (\Gamma_0 \cap \Gamma_1) $ equality because the second equality $ \Upsilon^m = (\Lambda_1' \setminus \Gamma_1) \cup (\Gamma_0 \cap \Gamma_1) $ is symmetric and can be proven by the same reasoning (replace in the argumentation below $ \Lambda_0' $, $ \Gamma_0 $ and $ \Upsilon_P^m $ by $ \Lambda_1' $, $ \Gamma_1 $ and $ \Upsilon_=^m $, respectively). We establish $ \Upsilon^m = (\Lambda_0' \setminus \Gamma_0) \cup (\Gamma_0 \cap \Gamma_1) $ by showing that the complements of the left and right hand sides in $ \Lambda_0' $ are equal. In other words, it suffices to prove $ \Upsilon_P^m = \Gamma_0 \setminus \Gamma_1 $.
	
	For the ``$ \subseteq $'' inclusion, assume that $ p_i^Q \in \Upsilon_P^m $ holds for some $ i \in \{1, \dots, m\} $ and $ Q \in \mathbb{P} $. Clearly, since $ \Gamma_0 \in \DisjOf{\Theta} $, $ p_i^R $ must appear in $ \Gamma_0 $ for some $ R \in \mathbb{P} $. Note that $ R $ must be equal to $ Q $ because otherwise $ \Lambda_0' $ would contain contradicting predicates and be consequently unsatisfiable, which is a contradiction to Claim~\ref{claim:lambda_0_lambda_1_pi_complex}. Hence, $ p_i^Q \in \Gamma_0 $. Now suppose, for the sake of contradiction, that $ p_i^Q \in \Gamma_1 $. Then applying Claim~\ref{claim:upsilonm_eq_lambda0prime_intersection_lambda1prime} yields \[
		p_i^Q \in \Gamma_0 \cap \Gamma_1 \subseteq \Lambda_0' \cap \Lambda_1' \overset{\ref{claim:upsilonm_eq_lambda0prime_intersection_lambda1prime}}{=} \Upsilon^m
	\] and hence $ p_i^Q \in \Upsilon^m \cap \Upsilon_P^m = \varnothing $, which is a contradiction.
	
	Turning to the ``$ \supseteq $'' inclusion, suppose that $ p \in \Gamma_0 \setminus \Gamma_1 $ and, for the sake of contradiction, that $ p \in \Upsilon^m $. This immediately implies $ p \in \Lambda_1' $. Since $ p $ disrespects $ \Pi $ (because otherwise $ p \in \Theta \subseteq \Gamma_0 \cap \Gamma_1 $) but all predicates in $ \Lambda_1' \setminus \Gamma_1 $ respect $ \Pi $ (because $ \Lambda_1' \in \DisjOf{\Lambda} $), it follows that $ p \in \Gamma_1 $, which is a contradiction.
\end{proof}

Note that Claim~\ref{claim:upsilonm_lamda0primeexceptgamma0_union_gamma0gamma1} implies, in particular, that $ \Lambda_0' \setminus \Gamma_0 = \Lambda_1' \setminus \Gamma_1 $. Overall, at this point we have proven all relationships visualized in Figure~\ref{fig:constructed_lambda_proof}.


\subsubsection{Reason behind the exponential blowup}

We now analyze the sizes of various subsets of the constructed $ \Lambda $-proof components with the primary goal of understanding in detail what the actual cause of the exponential blowup is. More precisely, we study the subsets of $ \Lambda_0' $ and $ \Lambda_1' $ where exponentially many predicates from $ \Lambda \subseteq \Upsilon^m $ land. Clearly, $ \Theta $ and consequently also $ \Gamma_0 $ and $ \Gamma_1 $ all have polynomial size, so the exponential blowup cannot occur inside $ \Gamma_0 \cap \Gamma_1 \subseteq \Lambda_0' \cap \Lambda_1' \overset{\ref{claim:upsilonm_eq_lambda0prime_intersection_lambda1prime}}{=} \Upsilon^m $. Hence, the exponentially many predicates can appear only in the complement of $ \Gamma_0 \cap \Gamma_1 $, that is, in $ \Upsilon^m \setminus (\Gamma_0 \cap \Gamma_1) $. By Claim~\ref{claim:upsilonm_lamda0primeexceptgamma0_union_gamma0gamma1}, this set is equal to $ \Lambda_0' \setminus \Gamma_0 = \Lambda_1' \setminus \Gamma_1 $.

This leads us to the following idea. In order to tackle this exponential blowup, we need to find a certain kind of polynomial-size representation of $ \Lambda_0' \setminus \Gamma_0 $. Intuitively, we achieve this by constructing a predicate set entailing $ \Lambda_0' \setminus \Gamma_0 $ and then obtain a compressed version of the $ \Lambda $-proof by replacing $ \Lambda_0' \setminus \Gamma_0 = \Lambda_1' \setminus \Gamma_1 $ with that set in $ \Lambda_0 $, $ \Lambda_1 $, $ \Lambda_0' $ and $ \Lambda_1' $. Geometrically, the problem of constructing such a predicate set entailing $ \Lambda_0' \setminus \Gamma_0 $ is equivalent to defining a shape contained in $ \ModelsOf(\Lambda_0' \setminus \Gamma_0) $, which we will refer to by the name \textit{compressing shape}. Observe that $ \Lambda_0' $ is a disjunct of $ \Lambda $, therefore $ \Lambda_0' \setminus \Gamma_0 $ can contain only strict inequality predicates. Since this set is satisfiable, it can be thought of as a collection of constraints defining an open polyhedron with at most exponentially many faces. Thus, it should not be too hard to define a shape lying inside $ \ModelsOf(\Lambda_0' \setminus \Gamma_0) $ because this open polyhedron is clearly closed under taking a sufficiently small $ \varepsilon $-neighborhood of every point and any bounded shape can in principle be defined within that $ \varepsilon $-neighborhood. However, as we will see shortly, we have to be very careful when defining this shape to make sure that it can be used in a $ \Lambda $-proof without breaking any of its properties.

\subsubsection{Compressing shape}

We start turning the above intuitions into a rigorous proof by first defining a suitable compressing shape defined by a predicate set entailing $ \Lambda_0' \setminus \Gamma_0 $. We choose the open $ n $-dimensional cube (hereinafter: ``$ n $-cube'') to be the shape type. Clearly, such a cube can be characterized by a center point and a side length; it can also easily be defined using strict inequalities:

\begin{zbdefinition}[Cube around a point]
	Let $ \varphi(x_1, \dots, x_n) \in \QFLRA $ be a formula, $ v = (v_1, \dots, v_n)^\transp \in \Q^n $ be a point and $ 0 < \varepsilon \in \Q $ be a constant. We define \[
	\OpenCube{v}{\varepsilon} := \{x_i < v_i + \varepsilon/2 \mid 1 \le i \le n\} \cup \{x_i > v_i - \varepsilon/2 \mid 1 \le i \le n\}
	\] to be the set of strict inequality predicates defining an open cube around $ v $ with side length $ \varepsilon $.
\end{zbdefinition}

We now define separately the center point $ v $ and the side length $ \varepsilon $ of the $ n $-cube. Intuitively, these parameters need to be chosen so that on the one hand $ \OpenCube{v}{\varepsilon} \models \Lambda_0' \setminus \Gamma_0 $ and on the other hand the predicates in $ \OpenCube{v}{\varepsilon} $ all have a short (i.e., polynomial) encoding length.

\paragraph{Center of the cube.} As the center of the $ n $-cube we choose an arbitrary model $ v \models \Lambda_1 $ encodable using polynomially many bits. Such a model exists because, as we now argue, quantifier-free linear real arithmetic satisfies the small model property \cite{ebbinghaus:2005:finitemodeltheory, tent:2012:modeltheory, marker:2006:modeltheory, kroening:2016}. It is well-known in mathematical optimization that any polyhedron contains a point whose encoding is polynomial in the size of the string representation of the system of inequalities defining that polyhedron (see \cite[Theorem 4.4]{korte:2012:combinatorialoptimization}, \cite[Theorem 10.1]{schrijver:1998:theoryofintegerandlinearprogramming} and \cite[Lemma 5.9]{paffenholz:2010:polyhedralgeometryandlinearoptimization}). Unfortunately, these results talk only about closed polyhedra, but as we now argue, we can obtain analogous results for our setting, i.e., when strict inequalities are allowed. Given a satisfiable predicate set $ \Gamma $, we run Gaussian elimination on $ \Gamma^= $ and obtain a solution represented as a set of equations connecting leading variables with trailing ones. Similar to how variable substitution works in the well-known Fourier-Motzkin elimination algorithm \cite{kroening:2016, paffenholz:2010:polyhedralgeometryandlinearoptimization, schrijver:1998:theoryofintegerandlinearprogramming}, we let $ \Gamma' $ be obtained from $ \Gamma \setminus \Gamma^= $ by replacing all occurrences of leading variables by the corresponding linear combination of the trailing ones, which clearly guarantees a one-to-one correspondence between the solutions to $ \Gamma $ and to $ \Gamma' $. This transformation also ensures that $ \Gamma' $ defines a full-dimensional open convex polyhedron, for which it is already known that it must contain an interior point whose representation is of polynomial size (see, \textit{mutatis mutandis}, \cite[Lemma 6.2.6]{grotschel:2012:combinatorialoptimization}). Clearly, translating this point backwards into a solution to $ \Gamma $ preserves the existence of a polynomial upper bound for its size. Here it is important that the upper bound on the solution size depends only on the encoding length of any predicate appearing in $ \Lambda_0' \setminus \Gamma_0 $ and not on the size of $ \Lambda_0' \setminus \Gamma_0 $. Indeed, the size of $ \Lambda_0' \setminus \Gamma_0 $ may be exponential, but any single predicate from $ \Lambda_0' \setminus \Gamma_0 $ originates in the second \texttt{foreach} loop at Line~\ref{alg:cover:line:second_loop} of the covering algorithm as a result of algebraic computations depending only on $ \Theta $ and on the predicates appearing in the original formula $ \varphi $. We have already discussed that $ \Theta $'s representation is of polynomial size, so we conclude that the encoding of every predicate in $ \Lambda_0' \setminus \Gamma_0 $ is also of size bounded by a polynomial in the size of $ \varphi $. Hence, the small model property follows.


\paragraph{Side length of the cube.} The second parameter to be defined is the cube side length $ \varepsilon $. Since the predicates in $ \OpenCube{v}{\varepsilon} $ use $ \varepsilon $, it is important to choose $ \varepsilon $ to be a rational number having a polynomial-length encoding. On the other hand, $ \varepsilon > 0 $ should be chosen small enough so that the cube $ \ModelsOf(\OpenCube{v}{\varepsilon}) $ fits inside $ \ModelsOf(\Lambda_0' \setminus \Gamma_0) $. Let $ q_1, \dots, q_k $ be the predicates in $ \Lambda_0' \setminus \Gamma_0 $ and let $ w_1 + W_1, \dots, w + W_k $ be the affine vectorspaces of solutions to $ q_1^=, \dots, q_k^= $, respectively. 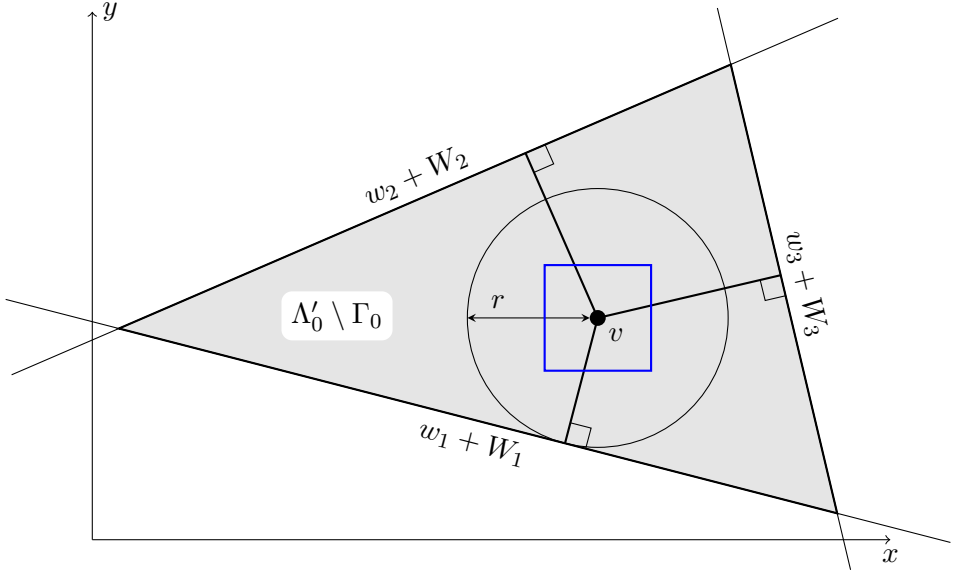
\begin{figure}[t]
	\centering
	\def\orthrectsidesize{2mm}
\begin{tikzpicture}[scale=\zbtikzbiggerscaling]
	
	\coordinate (origin) at (0, 0);
	
	\draw[->] (origin) -- (7.5,0) coordinate[label = {below:$x$}] (xmax);
	\draw[->] (origin) -- (0,5) coordinate[label = {right:$y$}] (ymax);
	
	
	\coordinate (trtr) at (6, 4.5); 
	\coordinate (trbr) at (7, 0.25); 
	\coordinate (trleft) at (0.25, 2); 
	
	\coordinate (vpoint) at (4.75, 2.1);
	
	\coordinate (vprojtop) at ($ (trleft)!(vpoint)!(trtr) $);
	\coordinate (vprojright) at ($ (trtr)!(vpoint)!(trbr) $);
	\coordinate (vprojdown) at ($ (trleft)!(vpoint)!(trbr) $);
	
	\coordinate (vprojmin) at (vprojdown);
	
	\coordinate (inscribedsquarehalfwidth) at (0.5, 0);
	\coordinate (inscribedsquarehalfheight) at (0, 0.5);
	
	\coordinate (vprojtoporthrectline) at ($ (vprojtop)!\orthrectsidesize!(vpoint) $);
	\coordinate (vprojtoporthrectedge) at ($ (vprojtop)!\orthrectsidesize!(trtr) $);
	\coordinate (vprojtoporthrectcorner) at ($ (vprojtoporthrectline) + (vprojtoporthrectedge) - (vprojtop) $);
	\coordinate (vprojbottomorthrectline) at ($ (vprojdown)!\orthrectsidesize!(vpoint) $);
	\coordinate (vprojbottomorthrectedge) at ($ (vprojdown)!\orthrectsidesize!(trbr) $);
	\coordinate (vprojbottomorthrectcorner) at ($ (vprojbottomorthrectline) + (vprojbottomorthrectedge) - (vprojdown) $);
	\coordinate (vprojrightorthrectline) at ($ (vprojright)!\orthrectsidesize!(vpoint) $);
	\coordinate (vprojrightorthrectedge) at ($ (vprojright)!\orthrectsidesize!(trbr) $);
	\coordinate (vprojrightorthrectcorner) at ($ (vprojrightorthrectline) + (vprojrightorthrectedge) - (vprojright) $);
	
	\coordinate (inscribedsquarer) at ($ (vpoint) + (inscribedsquarehalfwidth) $); 
	\coordinate (inscribedsquareu) at ($ (vpoint) + (inscribedsquarehalfheight) $); 
	\coordinate (inscribedsquared) at ($ (vpoint) - (inscribedsquarehalfheight) $); 
	\coordinate (inscribedsquarel) at ($ (vpoint) - (inscribedsquarehalfwidth) $); 
	\coordinate (inscribedsquaretl) at ($ (inscribedsquareu-|inscribedsquarel) $); 
	\coordinate (inscribedsquaretr) at ($ (inscribedsquareu-|inscribedsquarer) $); 
	\coordinate (inscribedsquarebl) at ($ (inscribedsquared-|inscribedsquarel) $); 
	\coordinate (inscribedsquarebr) at ($ (inscribedsquared-|inscribedsquarer) $); 
	
	\draw [shorten >=-4em,shorten <=-4em] (trleft) -- (trtr); 
	\draw [shorten >=-4em,shorten <=-4em] (trleft) -- (trbr);
	\draw [shorten >=-2em,shorten <=-2em] (trtr) -- (trbr);
	
	\filldraw[thick, black, fill=gray, fill opacity=0.2] (trtr) -- (trbr) -- (trleft) -- cycle;
	
	\node [draw] at (vpoint) [circle through={(vprojmin)}] (mincirc) {};
	
	\draw[thick] (vpoint) -- (vprojtop);
	\draw[thick] (vpoint) -- (vprojright);
	\draw[thick] (vpoint) -- (vprojdown);
	
	\draw (vprojtoporthrectline) -- (vprojtoporthrectcorner) -- (vprojtoporthrectedge);
	\draw (vprojbottomorthrectline) -- (vprojbottomorthrectcorner) -- (vprojbottomorthrectedge);
	\draw (vprojrightorthrectline) -- (vprojrightorthrectcorner) -- (vprojrightorthrectedge);
	
	\coordinate (mincircleft) at (mincirc.180);
	\draw[stealth-stealth] ($ (vpoint) - (0.075, 0) $) -- (mincircleft) node[near end,above]{$ r $};
	
	\fill[black] (vpoint) circle (0.075);
	
	\draw[thick, blue] (inscribedsquaretl) -- (inscribedsquaretr) -- (inscribedsquarebr) -- (inscribedsquarebl) -- cycle;
	
	\node[anchor=north west,thick] at (vpoint) {$ v $};
	
	\draw (2.3, 2.12) node [rounded corners, fill=white] {$ \Lambda_0' \setminus \Gamma_0 $};
	
	\path (trleft) edge node [sloped, below, midway] {$ w_1 + W_1 $} (trbr);
	\path (trleft) edge node [sloped, above, midway] {$ w_2 + W_2 $} (trtr);
	\path (trtr) edge node [sloped, above, midway] {$ w_3 + W_3 $} (trbr);
\end{tikzpicture}
	\caption{Example illustrating the way we construct the inscribed cube (colored blue) in two dimensions. Here, $ \Lambda_0' \setminus \Gamma_0 $ consists of three predicates $ p_1, p_2, p_3 $; solving $ p_i^= $ yields $ \ModelsOf(p_i^=) = w_i + W_i $.}
	\label{fig:inscribed_polynomial_size_cube}
\end{figure} Now we describe the way we construct $ \varepsilon $, a geometric visualization of our argumentation is given in Figure~\ref{fig:inscribed_polynomial_size_cube}. \begin{itemize}
	\item First, we project $ v $ onto every hyperplane $ w_i + W_i $. Since the predicates in $ \Lambda_0' \setminus \Gamma_0 $ are all strict inequalities, no projection can coincide with $ v $. In other words, the Euclidean distance between $ v $ and every projection must be strictly positive.
	\item Next, we fix $ w_i + W_i $ to be the affine vectorspace such that the Euclidean distance between $ v $ and its orthogonal projection onto $ w_i + W_i $ is minimal, and let $ r $ be that distance.
	\item Then, we consider the open $ n $-dimensional ball (hereinafter: ``$ n $-ball'') with center at $ v $ and radius $ r $. Clearly, all points lying on the interior of this $ n $-ball are models of $ \Lambda_0' \setminus \Gamma_0 $.
	\item Next, we argue that $ r \in \Q $ has a polynomial upper bound on the encoding length because measuring the distance between $ v $ and $ w_i + W_i $ preserves the polynomial bound on the encoding of occurring numbers and coefficients. Indeed, the construction of $ v $ ensures that it has a short encoding; $ w_i $ and vectors in the basis of $ W_i $ all have short encodings because they can be computed via Gaussian elimination in polynomial time based on the $ q_i \in \Lambda_0' \setminus \Gamma_0 $ predicate which, as already discussed, has a polynomial-size encoding.
	\item Finally, since the $ n $-ball with center at $ v $ and radius $ r $ trivially contains a cube centered around $ v $ with side length $ r/3 $, choosing $ \varepsilon := r/3 $ clearly ensures that $ \OpenCube{v}{\varepsilon} $ can be represented as a string of polynomial length and $ \OpenCube{v}{\varepsilon} \models \Lambda_0' \setminus \Gamma_0 $.
\end{itemize}

\subsubsection{Compression of the $ \Lambda $-proof}

Having defined the compressing shape, i.e., the cube, we now use it to construct the following $ (\Theta \cup \OpenCube{v}{\varepsilon}) $-proof $ (\Theta_0, \Theta_1, \Theta_0', \Theta_1') $ of non-$ \Pi $-decomposability of $ \varphi $, which is, in fact, a compressed version of the $ \Lambda $-proof we have constructed in Section~\ref{sec:vardec:nondec_proof_system}. \begin{align*}
	\Theta_0' &:= (\Gamma_0 \cap \Gamma_1) \cup \OpenCube{v}{\varepsilon} \cup \Upsilon_P^m \in \DisjOf{\Theta \cup \OpenCube{v}{\varepsilon}} \\
	\Theta_1' &:= (\Gamma_0 \cap \Gamma_1) \cup \OpenCube{v}{\varepsilon} \cup \Upsilon_=^m \in \DisjOf{\Theta \cup \OpenCube{v}{\varepsilon}} \\
	\Theta_0 &:= (\Gamma_0 \cap \Gamma_1) \cup \OpenCube{v}{\varepsilon} \cup \{p_m^P\} \subseteq \Theta_0' \models \Theta_0 \\
	\Theta_1 &:= (\Gamma_0 \cap \Gamma_1) \cup \OpenCube{v}{\varepsilon} \cup \{p_m^=\} \subseteq \Theta_1' \models \Theta_1
\end{align*}

We now verify that $ (\Theta_0, \Theta_1, \Theta_0', \Theta_1') $ indeed meets all the criteria set out in the definition of a $ \Lambda $-proof. Observe that by Claim~\ref{claim:upsilonm_lamda0primeexceptgamma0_union_gamma0gamma1}, \begin{align*}
	\Theta_0 &= (\Gamma_0 \cap \Gamma_1) \cup \OpenCube{v}{\varepsilon} \cup \{p_m^P\}
	\models (\Gamma_0 \cap \Gamma_1) \cup (\Lambda_0' \setminus \Gamma_0) \cup \{p_m^P\} \\ &\overset{\ref{claim:upsilonm_lamda0primeexceptgamma0_union_gamma0gamma1}}{\equiv} \Upsilon^m \cup \{p_m^P\} \equiv \Lambda_0 \equiv \Lambda_0' \models \Upsilon_P^m
\end{align*} and, analogously, \begin{align*}
\Theta_1 &= (\Gamma_0 \cap \Gamma_1) \cup \OpenCube{v}{\varepsilon} \cup \{p_m^=\}
\models (\Gamma_0 \cap \Gamma_1) \cup (\Lambda_0' \setminus \Gamma_0) \cup \{p_m^=\} \\ &\overset{\ref{claim:upsilonm_lamda0primeexceptgamma0_union_gamma0gamma1}}{\equiv} \Upsilon^m \cup \{p_m^=\} \equiv \Lambda_1 \equiv \Lambda_1' \models \Upsilon_P^=
\end{align*} It follows that $ \Theta_0 \models \Theta_0' $, $ \Theta_1 \models \Theta_1' $, so we conclude that $ \Theta_0 \equiv \Theta_0' $ and $ \Theta_1 \equiv \Theta_1' $ (the converse entailments are trivial).

\begin{zbclaim}
	\label{claim:theta0_theta1_pi_complex}
	$ \Theta_0 $ and $ \Theta_1 $ are both $ \Pi $-complex.
\end{zbclaim}

For the proof of Claim~\ref{claim:theta0_theta1_pi_complex}, we will need the following fact.

\begin{zbclaim}
	\label{claim:lambda_0_and_cube_sat}
	For all $ \varepsilon > 0 $, $ \Lambda_0 \wedge \OpenCube{v}{\varepsilon} $ is satisfiable.
\end{zbclaim}
\begin{proof}
	Let $ u \models \Lambda_0 $ be arbitrary but fixed ($ u $ exists by Claim~\ref{claim:lambda_0_lambda_1_pi_complex}). \begin{figure}[t]
		\centering
		\begin{tikzpicture}[scale=\zbtikzbiggerscaling]
	
	\coordinate (origin) at (0, 0);
	
	\draw[->] (origin) -- (7.5,0) coordinate[label = {below:$x$}] (xmax);
	\draw[->] (origin) -- (0,5) coordinate[label = {right:$y$}] (ymax);
	
	\coordinate (mainlinetl) at (2, 4.8); 
	\coordinate (mainlinebr) at (5.5, 0); 
	
	\coordinate (trtr) at (3, 4.5); 
	\coordinate (trbr) at (7, 0.25); 
	\coordinate (trleft) at (0.5, 1); 
	
	\coordinate (mainlinetrinttop) at (intersection of mainlinetl--mainlinebr and trleft--trtr);
	\coordinate (mainlinetrintbot) at (intersection of mainlinetl--mainlinebr and trleft--trbr);
	
	\coordinate (vpoint) at ($ (mainlinetrinttop)!0.5!(mainlinetrintbot) $);
	\coordinate (upoint) at (2, 2);
	
	\coordinate (inscribedsquarehalfwidth) at (0.75, 0);
	\coordinate (inscribedsquarehalfheight) at (0, 0.75);
	
	\coordinate (inscribedsquarer) at ($ (vpoint) + (inscribedsquarehalfwidth) $); 
	\coordinate (inscribedsquareu) at ($ (vpoint) + (inscribedsquarehalfheight) $); 
	\coordinate (inscribedsquared) at ($ (vpoint) - (inscribedsquarehalfheight) $); 
	\coordinate (inscribedsquarel) at ($ (vpoint) - (inscribedsquarehalfwidth) $); 
	\coordinate (inscribedsquaretl) at ($ (inscribedsquareu-|inscribedsquarel) $); 
	\coordinate (inscribedsquaretr) at ($ (inscribedsquareu-|inscribedsquarer) $); 
	\coordinate (inscribedsquarebl) at ($ (inscribedsquared-|inscribedsquarel) $); 
	\coordinate (inscribedsquarebr) at ($ (inscribedsquared-|inscribedsquarer) $); 
	
	\fill[fill=gray,opacity=0.2] (trleft) -- (mainlinetrinttop) -- (mainlinetrintbot) -- cycle;
	
	\draw[thick, OliveGreen] (mainlinetl) -- (mainlinebr);
	
	\draw[thick, black] (trtr) -- (trleft) -- (trbr);
	
	\draw[thick, red] (upoint) -- (vpoint);
	
	\fill[black] (vpoint) circle (0.075);
	\fill[black] (upoint) circle (0.075);
	
	\draw[thick, blue] (inscribedsquaretl) -- (inscribedsquaretr) -- (inscribedsquarebr) -- (inscribedsquarebl) -- cycle;
	
	\node[anchor=north east,thick] at (vpoint) {$ v $};
	\node[anchor=north east,thick] at (upoint) {$ u $};
	
	\draw (2.75, 1.25) node [rounded corners, fill=white] {$ \Lambda_0 $};
	
	\draw [decorate, decoration = {brace, raise=5pt, amplitude=5pt}] (inscribedsquaretr) --  (inscribedsquarebr) node[pos=0.5,right=10pt,black]{$ \OpenCube{v}{\varepsilon} $};

	\node[anchor=north west,thick, OliveGreen] at ($ (mainlinetrinttop) + (0.15, 0.1) $) {$ p_m^= $};
\end{tikzpicture}
		\caption{A two-dimensional example of how the sets defined by $ \Lambda_0 $ and $ \Lambda_1 $ may look like geometrically. Black lines depict the borders defined by linear constraints in $ \Upsilon^m $. The green line is defined by the $ p_m^= $ predicate, so $ \Lambda_1 $ can be thought of as defining the set of points lying on the green line where it is the border of the gray region corresponding to $ \Lambda_0 $. Edges of the open cube defined by $ \OpenCube{v}{\varepsilon} $ are drawn in blue.}
		\label{fig:cube_lambda0_sat}
	\end{figure}
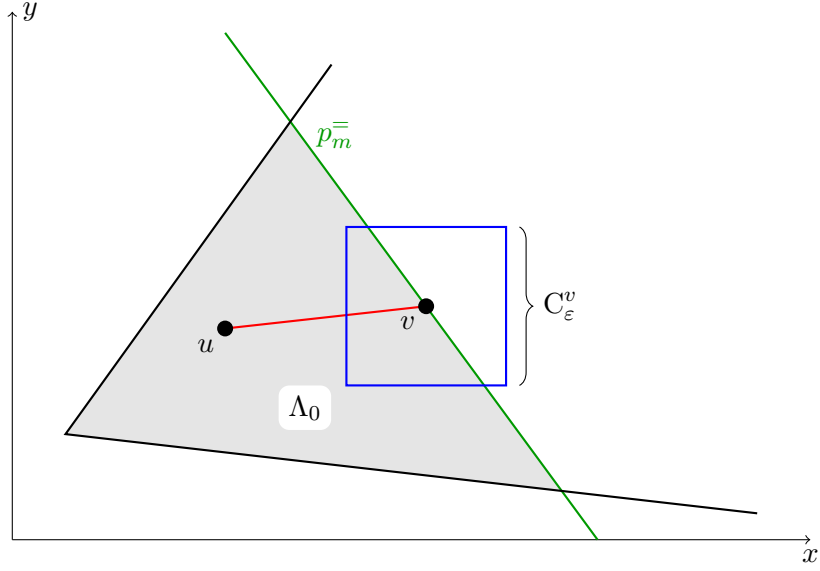 The overall idea is to prove that the line $ L := (v, u) $ must agree with the cube $ \OpenCube{v}{\varepsilon} $ on some model, our argumentation is visualized in Figure~\ref{fig:cube_lambda0_sat}. We start by showing that $ L \subseteq \ModelsOf(\Lambda_0) $. Let $ v + W $ be the affine vectorspace of solutions to $ p_m^= $. Note that $ L $ cannot agree with $ \ModelsOf(p_m^=) $ on any point because otherwise $ v + \varepsilon \cdot (u - v) \in v + W $ would hold for some $ 0 < \varepsilon < 1 $, meaning that $ u - v \in W $ and hence $ u \in v + W = \ModelsOf(p_m^=) $, which is a contradiction to $ u \models p_m^P $. In other words, we can write $ L \subseteq \ModelsOf(p_m^<) \cup \ModelsOf(p_m^>) $. Moreover, $ L $ must be a subset of $ \ModelsOf(p_m^P) $. This is because assuming the contrary would imply that $ \ModelsOf(p_m^<) \cup \ModelsOf(p_m^>) $ is closed under taking the convex hull of a point from $ \ModelsOf(p_m^<) $ and from $ \ModelsOf(p_m^>) $, which is not true due to $ v \models p_m^= $. Since $ \Upsilon^m $ defines a convex set and $ u \models \Upsilon^m $, $ v \models \Upsilon^m $, it holds that \[
		L \subseteq \ModelsOf(\Upsilon^m) \cap \ModelsOf(p_m^P) = \ModelsOf(\Lambda_0)
	\] Due to the fact that $ \OpenCube{v}{\varepsilon} $ contains only strict inequality predicates, there exists an $ 0 < \varepsilon < 1 $ such that all points in the $ \varepsilon $-neighborhood of $ v $ are still lying within the cube $ \ModelsOf(\OpenCube{v}{\varepsilon}) $. In particular, it follows that \[
		v + \varepsilon \cdot (u - v) \models \Lambda_0 \wedge \OpenCube{v}{\varepsilon}
	\]
\end{proof}
\begin{proof}[Proof of Claim~\ref{claim:theta0_theta1_pi_complex}]
	Note that $ (\Gamma_0 \cap \Gamma_1) \cup \{p_m^P\} $ and $ (\Gamma_0 \cap \Gamma_1) \cup \{p_m^=\} $ are both $ \Pi $-complex because \begin{align}
		\label{eqn:lambda_0_entails_gamma0intersectiongamma_1_union_pmp}
		\Lambda_0 &\models (\Gamma_0 \cap \Gamma_1) \cup \{p_m^P\} \\
		\label{eqn:lambda_1_entails_gamma0intersectiongamma_1_union_pmeq}
		\Lambda_1 &\models (\Gamma_0 \cap \Gamma_1) \cup \{p_m^=\}
	\end{align} and $ \Lambda_0 $, $ \Lambda_1 $ are both $ \Pi $-complex by Claim~\ref{claim:lambda_0_lambda_1_pi_complex}. Furthermore, $ \Lambda_0 \wedge \OpenCube{v}{\varepsilon} $ is satisfiable by Claim~\ref{claim:lambda_0_and_cube_sat}. Also, $ \Lambda_1 \wedge \OpenCube{v}{\varepsilon} $ is satisfiable (which is immediate by the choice of $ v $). Hence, applying the entailments in (\ref{eqn:lambda_0_entails_gamma0intersectiongamma_1_union_pmp}) and (\ref{eqn:lambda_1_entails_gamma0intersectiongamma_1_union_pmeq}) to the models of $ \Lambda_0 \wedge \OpenCube{v}{\varepsilon} $ and $ \Lambda_1 \wedge \OpenCube{v}{\varepsilon} $, respectively, yields that $ \Theta_0 $ and $ \Theta_1 $ are both satisfiable. By Lemma~\ref{lemma:adding_neq_predicates_cannot_make_pi_complex_set_pi_simple}, adding strict inequalities cannot make a predicate set $ \Pi $-simple if the resulting set is satisfiable. Hence, applying this lemma inductively yields the claim.
\end{proof}

\noindent We now show that $ \Theta_0 $ and $ \Theta_1 $ agree on the set of $ Z $-dependencies.

\begin{zbclaim}
	For all $ Z \in \Pi $ it holds that \[
	\LinDep_{\pi_Z}(\Theta_0) = \LinDep_{\pi_Z}(\Theta_1)
	\]
\end{zbclaim}
\begin{proof}
	By the same reasoning as in the proof of Claim~\ref{claim:lambda_0_lindep_eq_lambda_1_lindep} above, \[
	\LinDep_{\pi_Z}(\Gamma_0) = \LinDep_{\pi_Z}(\Gamma_1)
	\] holds for all $ Z \in \Pi $. Note that \begin{align}
		\label{eqn:theta_0_prime_eq_gamma_0_union_cube}
		\Theta_0' &= \Gamma_0 \cup \OpenCube{v}{\varepsilon} \\
		\label{eqn:theta_1_prime_eq_gamma_1_union_cube}
		\Theta_1' &= \Gamma_1 \cup \OpenCube{v}{\varepsilon}
	\end{align} Since $ \OpenCube{v}{\varepsilon} $ contains only strict inequality predicates, applying Theorem~\ref{thm:only_equality_predicates_can_establish_lindep_strong} (see Appendix~\ref{sec:app:lindep_facts}) to (\ref{eqn:theta_0_prime_eq_gamma_0_union_cube}) and (\ref{eqn:theta_1_prime_eq_gamma_1_union_cube}) yields \[
	\LinDep_{\pi_Z}(\Theta_0') = \LinDep_{\pi_Z}(\Theta_1')
	\] for all $ Z \in \Pi $. Hence, the claim follows (because $ \Theta_0 \equiv \Theta_0' $ and $ \Theta_1 \equiv \Theta_1' $).
\end{proof}

Next, observe that $ \Theta_1 $ is clearly $ p_m^= $-next to $ \Theta_0 $. The final property of a $ \Lambda $-proof, namely that $ \Theta_i' \models \varphi $ and $ \Theta_i' \models \neg\varphi $ must hold for some $ i \in \{0, 1\} $, immediately follows from the same property but for $ \Lambda_i' $ instead of $ \Theta_i' $ because \begin{align*}
	\Theta_0' \equiv \Theta_0 \models \Lambda_0' \\
	\Theta_1' \equiv \Theta_1 \models \Lambda_1'
\end{align*} We conclude that $ (\Theta_0, \Theta_1, \Theta_0', \Theta_1') $ is indeed a correct $ (\Theta \cup \OpenCube{v}{\varepsilon}) $-proof of non-$ \Pi $-decomposability of $ \varphi $. Hence, any non-$ \Pi $-decomposable formula $ \varphi $ has a $ (\Theta \cup \OpenCube{v}{\varepsilon}) $-proof of non-$ \Pi $-decomposability that can be encoded as a string of length bounded by a polynomial in the size of $ \varphi $. Combining this with the soundness of $ \Lambda $-proofs yields an $ \NP $ algorithm for determining non-$ \Pi $-decomposability. Consequently, we obtain a $ \coNP $ upper bound for deciding variable decomposability:

\begin{theorem}
	\label{thm:conp_upper_bound}
	Given a formula $ \varphi \in \QFLRA $ and a partition $ \Pi $, deciding the $ \Pi $-decomposability of $ \varphi $ is in $ \coNP $.
\end{theorem}

\begin{corollary}
	\label{cor:variable_independence_sigma2}
	Given a formula $ \varphi(x_1, \dots, x_n) \in \QFLRA $ and two of its free variables $ x_i $ and $ x_j $, deciding the independence of $ x_i $ and $ x_j $ is in $ \NP^\NP $.
\end{corollary}
\begin{proof}
	Use existential alternation to guess the partition $ \Pi $ of $ \{x_1, \dots, x_n\} $ and then use universal alternation to determine the $ \Pi $-decomposability of $ \varphi $ via the $ \coNP $ algorithm of Theorem~\ref{thm:conp_upper_bound}.
\end{proof}

\section{Lower bounds}
\label{sec:lower_bounds}

In this \zbsectionname/, we show that deciding variable decomposability is $ \coNP $-hard, even if we only consider the case when $ \Pi $ is a binary partition. We also prove a $ \coNP $ lower bound for the variable independence problem.

\subsection{Monadic decomposability}

We start by establishing $ \coNP $-hardness of deciding monadic decomposability.

\begin{theorem}
	\label{thm:mondec_conp_hard}
	Given a formula $ \varphi \in \QFLRA $, it is $ \coNP $-hard to decide the monadic decomposability of $ \varphi $.
\end{theorem}
\begin{proof}
	We construct a polynomial-time many-one reduction from the $ \coNP $-complete problem of deciding the unsatisfiability of a propositional formula $ \varphiprop(x_1, \dots, x_n) $. The reduction is similar to the one showing $ \coNP $-hardness of deciding monadic decomposability over Presburger arithmetic \cite[Lemma 3.1]{hague:2020}. Define $ \varphi(y_1, \dots, y_n, z_1, z_2) \in \QFLRA $ to be a formula obtained from $ \varphiprop $ by replacing every variable $ x_i $ with $ y_i \neq 0 $. We show that $ \varphiprop $ is unsatisfiable if and only if $ \varphi \wedge z_1 = z_2 $ is monadically decomposable. Indeed, if $ \varphiprop $ is unsatisfiable, then $ \varphi $ is equivalent to $ \bot $ meaning that $ \varphi \wedge z_1 = z_2 $ is trivially monadically decomposable.
	
	For the converse direction, suppose that $ \varphiprop $ is satisfiable and, for the sake of contradiction, that $ \varphi \wedge z_1 = z_2 $ is monadically decomposable. Since $ \varphi $ is also satisfiable and for every model of $ \varphi $ there exist infinitely many ways to assign $ z_1 $ and $ z_2 $ a value so that $ \varphi \wedge z_1 = z_2 $ is satisfied by that model, there exists an infinite family $ v_1, v_2, v_3, \dots $ of models of $ \varphi \wedge z_1 = z_2 $, such that they pairwise agree on the values they assign to the $ y_1, \dots, y_n $ variables but disagree on the values assigned to $ z_1 $ and $ z_2 $. Without loss of generality, let the decomposition of $ \varphi \wedge z_1 = z_2 $ be given as a disjunction $ \varphi \wedge z_1 = z_2 \equiv \bigvee_{\Lambda \in \Psi} \Lambda $ over a set $ \Psi $ of predicate sets. By the pigeonhole principle, there exists $ \Lambda \in \Psi $ such that for some $ i \neq j $, $ v_i \models \Lambda $ and $ v_j \models \Lambda $. Let $ v_i[z_2 \mapsto v_j(z_2)] $ denote the model obtained from $ v_i $ by changing the value $ z_2 $ gets assigned to $ v_j(z_2) $, i.e., to the value $ v_j $ assigns the variable $ z_2 $. Since $ v_i(z_1) = v_i(z_2) \neq v_j(z_2) $, it follows that $ v_i[z_2 \mapsto v_j(z_2)] \not\models z_1 = z_2 $. But on the other hand, since $ \Lambda $ is a set of monadic predicates satisfied by both $ v_i $ and $ v_j $, it follows that $ v_i[z_2 \mapsto v_j(z_2)] \models \Lambda $. Since $ \Lambda \models \varphi \wedge z_1 = z_2 \models z_1 = z_2 $, we arrive at a contradiction.
\end{proof}

We also give an alternative proof of Theorem~\ref{thm:mondec_conp_hard} showcasing how the Overspilling Theorem~\ref{thm:overspilling} can be used to obtain lower bounds.

\begin{proof}[Alternative proof of Theorem~\ref{thm:mondec_conp_hard}]
	We construct a polynomial-time many-one reduction from the $ \coNP $-complete problem of deciding the validity of a propositional formula given in DNF. Let $ \varphiprop(x_1, \dots, x_n) $ be such a formula in DNF. Without loss of generality, it can be assumed that $ \varphiprop $ is satisfiable. We define $ \varphi(x, y_1, \dots, y_n) \in \QFLRA $ to be a formula obtained from $ \varphiprop $ by replacing every variable $ x_i $ with $ y_i \ge x $.
	
	Note that if $ \varphiprop $ is valid, then so is $ \varphi $, meaning that $ \varphi \equiv \top $ is trivially monadically decomposable. Hence, to prove the theorem, it suffices to show that $ \varphi $ is not monadically decomposable if $ \varphiprop $ is not valid. As $ \varphiprop $ is satisfiable but not valid, it is possible to choose a model and an unsatisfying assignment such that they agree on the truth values they assign to all variables except exactly one, say, $ x_j $. Since $ y_i \ge x $ simulates $ x_i $ being set to true and $ y_i < x $ corresponds to $ x_i $ being set to false, it follows that there exist $ \circ_i \in \{<, \ge\} $ such that for \begin{align*}
	\psi_0 &:= \bigwedge_{\substack{i = 1 \\ i \neq j}}^n y_i \circ_i x \wedge y_j < x \\
	\psi_1 &:= \bigwedge_{\substack{i = 1 \\ i \neq j}}^n y_i \circ_i x \wedge y_j \ge x
	\end{align*} it holds that $ \psi_b \models \varphi $ and $ \psi_{1-b} \models \neg\varphi $ for some $ b \in \{0, 1\} $. Let $ \Theta_0, \Theta_1 $ be predicate sets obtained from the conjuncts of $ \psi_0 $ and $ \psi_1 $ by replacing every ``$ \ge $'' with ``$ = $'', respectively. Suppose, by way of contradiction, that $ \varphi $ is monadically decomposable. This implies, in particular, that $ \varphi $ and $ \neg\varphi $ are both $ \Pi $-decomposable for $ \Pi := \{X, Y\} $ where $ X := \{x, y_1, \dots, y_{j-1}, y_{j+1}, \dots, y_n\} $ and $ Y := \{y_j\} $.
	
	Note that $ \Theta_0 $ and $ \Theta_1 $ are both $ \Pi $-complex and it remains true that $ \Theta_b \models \varphi $ and $ \Theta_{1-b} \models \neg\varphi $ for some $ b \in \{0, 1\} $. Furthermore, by construction, $ \Theta_1 $ is $ (y_j = x) $-next to $ \Theta_0 $. We now show that $ \Theta_0 $ and $ \Theta_1 $ have the same set of $ Z $-dependencies for all $ Z \in \Pi $ because then we can apply the Overspilling Theorem~\ref{thm:overspilling} and get a contradiction.
	
	\begin{zbclaim}
		For all $ Z \in \Pi $ it holds that $ \LinDep_{\pi_Z}(\Theta_0) = \LinDep_{\pi_Z}(\Theta_1) $.
	\end{zbclaim}
	\begin{proof}
		In the proof, we rely on some utility lemmas we introduce in the Appendix. By Theorem~\ref{thm:only_equality_predicates_can_establish_lindep_strong}, it suffices to show that \[
			\LinDep_{\pi_Z}(\Theta_0^=) = \LinDep_{\pi_Z}(\Theta_1^=)
		\] Since $ \Theta_1 $ is $ (y_j = x) $-next to $ \Theta_0 $, \begin{align}
			\label{eqn:theta0_eq_union_yjx_eq_theta1_eq}
			\Theta_0^= \cup \{y_j = x\} = \Theta_1^=
		\end{align} For the sake of simplicity and convenience, we (without loss of generality) reorder the components of vectors in all vectorspaces so that \begin{itemize}
			\item $ y_j $ corresponds to the first component;
			\item precisely those $ y_i $, for which $ (y_i = x) \in \Theta_0^= $ holds, correspond to components at positions between $ 2 $ and $ k $, for some $ k \in \N $;
			\item $ x $ corresponds to the $ (n+1) $-th component;
			\item the remaining variables correspond to components at positions between $ k+1 $ and $ n $.
		\end{itemize} Then, solving $ \Theta_0^= $, $ \Theta_1^= $ and using (\ref{eqn:theta0_eq_union_yjx_eq_theta1_eq}) yields that $ \ModelsOf(\Theta_0^=) = \ImageOf(M_0) $ and $ \ModelsOf(\Theta_1^=) = \ImageOf(M_1) $, where \begin{align}
			\label{eqn:m0_mat_def}
			M_0 := \begin{pNiceMatrix}[first-col, extra-margin=2pt, code-for-first-col=\scriptscriptstyle, columns-width=2em]
				1 & 0 & 0 & 0 & \cdots & 0 & 1\\
				2 & 1 & 0 & 0 & \cdots & 0 & 0\\
				\vdots & \vdots & \vdots & \vdots & \ddots & \vdots & \vdots \\
				k   & 1 & 0 & 0 & \cdots & 0 & 0 \\
				k+1 & 0 & 1 & 0 & \cdots & 0 & 0 \\
				k+2 & 0 & 0 & 1 & \cdots & 0 & 0 \\
				\vdots & \vdots & \vdots & \vdots & \ddots & \vdots & \vdots \\
				n & 0 & 0 & 0 & \cdots & 1 & 0 \\
				n+1 & 1 & 0 & 0 & \cdots & 0 & 0
				\CodeAfter \tikz \node [highlight = (2-1) (9-5)] {};
			\end{pNiceMatrix}
		\end{align} and \begin{align}
			\label{eqn:m1_mat_def}
			M_1 := \begin{pNiceMatrix}[first-col, extra-margin=2pt, code-for-first-col=\scriptscriptstyle, columns-width=2em]
				1   & 1 & 0 & 0 & \cdots & 0\\
				2 & 1 & 0 & 0 & \cdots & 0 \\
				\vdots & \vdots & \vdots & \vdots & \ddots & \vdots \\
				k   & 1 & 0 & 0 & \cdots & 0 \\
				k+1 & 0 & 1 & 0 & \cdots & 0 \\
				k+2 & 0 & 0 & 1 & \cdots & 0 \\
				\vdots & \vdots & \vdots & \vdots & \ddots & \vdots \\
				n & 0 & 0 & 0 & \cdots & 1 \\
				n+1 & 1 & 0 & 0 & \cdots & 0
				\CodeAfter \tikz \node [highlight = (2-1) (9-5)] {} ;
			\end{pNiceMatrix}
		\end{align} Let $ \pi_X $ and $ \pi_Y $ be the projection homomorphisms outputting the components corresponding to $ X $ and $ Y $, respectively. Let furthermore $ \Pi_X $ and $ \Pi_Y $ be the matrix representations of $ \pi_X $ and $ \pi_Y $, respectively. Note that projecting onto $ Y $ corresponds to looking at the image of the first row of $ M_i $, while the projection onto $ X $ is determined by the image of the remaining $ n $ rows. More precisely, by Theorem~\ref{thm:lindep_algebraic_char}, \begin{align*}
			\LinDep_{\pi_Y}(\Theta_i^=) &\overset{\ref{thm:lindep_algebraic_char}}{=} \{U \le \Q \mid U \le \ImageOf(\Pi_Y \cdot M_i)^\bot\} \\
			&= \{U \le \Q \mid U \le \ImageOf((1))^\bot\} \\
			&= \{U \le \Q \mid U \le \gen{0}\} \\
			&= \{\gen{0}\}
		\end{align*} holds for all $ i \in \{0, 1\} $ because the first row of $ M_i $ contains a non-zero entry. Similarly, since $ M_1 $ agrees (up to zero entries) with $ M_2 $ on all rows except the first one (see (\ref{eqn:m0_mat_def}) and (\ref{eqn:m1_mat_def}) where the agreeing parts are highlighted), \[
			\ImageOf(\Pi_X \cdot M_0) = \ImageOf(\Pi_X \cdot M_1)
		\] Hence, by Theorem~\ref{thm:lindep_algebraic_char} we get \begin{align*}
			\LinDep_{\pi_X}(\Theta_0^=) &\overset{\ref{thm:lindep_algebraic_char}}{=} \{U \le \Q^n \mid U \le \ImageOf(\Pi_X \cdot M_0)^\bot\} \\
			&= \{U \le \Q^n \mid U \le \ImageOf(\Pi_X \cdot M_1)^\bot\} \\
			&\overset{\ref{thm:lindep_algebraic_char}}{=} \LinDep_{\pi_X}(\Theta_1^=)
		\end{align*} and conclude that $ \LinDep_{\pi_Z}(\Theta_0^=) = \LinDep_{\pi_Z}(\Theta_1^=) $ holds for all $ Z \in \Pi $.
	\end{proof}

	If $ \Theta_1 \models \varphi $, then $ \varphi \wedge \Theta_0 $ is satisfiable by the Overspilling Theorem~\ref{thm:overspilling}, which is a contradiction to $ \Theta_0 \models \neg\varphi $. Otherwise, $ \Theta_1 \models \neg\varphi $ must hold by Lemma~\ref{lemma:every_disj_either_true_or_false}, so applying the Overspilling Theorem~\ref{thm:overspilling} to $ \neg\varphi $ yields that $ \neg\varphi \wedge \Theta_0 $ is satisfiable, which similarly contradicts the assumption $ \Theta_0 \models \varphi $. In both cases, we have arrived at a contradiction.
\end{proof}

\subsection{Variable independence}

Since monadic decomposability is equivalent to all pairs of distinct variables being independent, we can use the hardness we have just derived in Theorem~\ref{thm:mondec_conp_hard} to show that deciding variable independence is also hard.

\begin{corollary}
	\label{cor:variable_collection_independence_conp_hard}
	Given a formula $ \varphi(x_1, \dots, x_n) \in \QFLRA $ and a set $ S \subseteq \{x_1, \dots, x_n\}^2 $, it is $ \coNP $-hard to decide whether for all $ (x_i, x_j) \in S $ the variables $ x_i $ and $ x_j $ are independent.
\end{corollary}
\begin{proof}
	We reduce from the $ \coNP $-hard problem of deciding monadic decomposability (see Theorem~\ref{thm:mondec_conp_hard}) by setting \[
		S := \{(x_i, x_j) \in \{x_1, \dots, x_n\}^2 \mid i \neq j\}
	\] and leaving the given formula $ \varphi $ unchanged. Clearly, $ \varphi $ is monadically decomposable if and only if $ x_i $ and $ x_j $ are independent for all $ i \neq j $.
\end{proof}

\begin{theorem}
	\label{thm:variable_independence_conp_hard}
	Given a formula $ \varphi(x_1, \dots, x_n) \in \QFLRA $ and two of its free variables $ x_i $ and $ x_j $, the problem of deciding the independence of $ x_i $ and $ x_j $ is $ \coNP $-hard with respect to conjunctive truth-table reductions.
\end{theorem}
\begin{proof}
	We reduce from the problem of deciding whether $ x_i $ and $ x_j $ are independent for all $ (x_i, x_j) \in S $, where $ S $ is a given collection of variable pairs (Corollary \ref{cor:variable_collection_independence_conp_hard} establishes the $ \coNP $-hardness of this problem). The conjunctive truth-table reduction calls the oracle deciding the independence of a single pair of variables for every such pair from $ S $ and then returns the conjunction of the results. Since the set of conjunctive truth-table reductions is closed under composition with many-one reductions, the reducibility of every language in $ \coNP $ to the present problem follows.
\end{proof}

\subsection{Variable decomposability}

Similar to how we used hardness of monadic decomposability to prove lower bounds for deciding variable independence, we can construct reductions proving lower bounds for variable decomposability.

\begin{lemma}
	\label{lemma:vardec_collection_conp_hard}
	Given a formula $ \varphi \in \QFLRA $ and a collection $ S $ of partitions, it is $ \coNP $-hard to decide whether $ \varphi $ is $ \Pi' $-decomposable for all $ \Pi' \in S $.
\end{lemma}
\begin{proof}
	Proposition~\ref{prop:reduction_to_binary_partitions} immediately yields the desired reduction from the problem of deciding the monadic decomposability of $ \varphi $ that we already know to be $ \coNP $-hard (Theorem~\ref{thm:mondec_conp_hard}).
\end{proof}

\begin{theorem}
	\label{thm:vardec_conp_hard}
	Given a formula $ \varphi \in \QFLRA $ and a partition $ \Pi $, it is $ \coNP $-hard to decide the $ \Pi $-decomposability of $ \varphi $, even if $ \Pi $ is a binary partition.
\end{theorem}
\begin{proof}
	We reduce from the problem of deciding whether $ \varphi $ is $ \Pi' $-decomposable for all $ \Pi' \in S $, where $ S $ is some given collection of partitions (Lemma~\ref{lemma:vardec_collection_conp_hard} establishes the $ \coNP $-hardness of this problem). We set \[
		\Pi := \bigsqcap_{\Pi' \in S} \Pi'
	\] and argue that $ \varphi $ is $ \Pi' $-decomposable for all $ \Pi' \in S $ if and only if $ \varphi $ is $ \Pi $-decomposable. Indeed, the ``$ \Rightarrow $'' direction follows from Proposition~\ref{prop:pi1_pi2_dec_implies_dec_wrt_meet_of_pi1_pi2}, while the converse implication is true because $ \Pi $ is a refinement of every $ \Pi' \in S $.
\end{proof}

\section{Experiments}
\label{sec:experiments}

\subsection{Preliminaries}

Like \cite{veanes:2017} and \cite{markgraf:2021}, we provide an implementation of the covering-based variable decomposition algorithm introduced in \zbrefsec{sec:vardec}. The implementation is in python, relies on the Z3 theorem prover \cite{z3:2008, bjorner_programming_z3:2019} and is available in the \begin{center}
	\url{https://github.com/ZeroBone/Cover}
\end{center} GitHub repository.

\subsubsection{Heuristics and optimizations}

To achieve better performance, we equipped our implementation of the covering algorithm (see Section~\ref{sec:vardec:cover} above) with the following heuristics designed to reduce both the number of disjuncts the algorithm needs to consider and the amount of recursive calls. \begin{enumerate}[label=(\alph*)]
	\item After the first loop (Line~\ref{alg:cover:line:first_loop}), the algorithm additionally tests whether $ \Theta \models \varphi $. If this is the case, it simply outputs $ \Theta $ to be the covering and proceeds normally otherwise.
	\item In the second loop, instead of enumerating all disjuncts of $ \Theta $, it suffices to consider only those which agree with $ \Upsilon $ on at least one model. Intuitively, the reason is that $ \Upsilon $ is our domain which we will use for our covering, and only disjuncts falling into that domain need to be ruled out in order to get a correct covering.
	\item We further optimize the second loop by iterating not over individual disjuncts $ \Omega \in \DisjOf{\Theta} $, but over groups of such disjuncts which agree on the equality predicates. Indeed, this does not change the semantics of the algorithm because all computations depend only on $ \Omega^= $.
	\item In the body of the second loop, at Line~\ref{alg:cover:line:second_loop:witness_vec}, instead of computing just one vector $ w $, we compute the entire basis $ w_1, \dots, w_d $ of \[
		\gen{w_1, \dots, w_d} = \ImageOf(\pi_Z \circ \lc_B)^{\bot} \cap \ImageOf(\pi_Z \circ \lc_A)
	\] Then the idea is to choose $ w $ from the basis in a way that will not lead to recursive calls capable of causing further nested recursion. More precisely, for each $ w_i $ we test whether $ \Gamma \models \pi_Z(\vec{z}) \cdot w_i \circ \pi_Z(b) \cdot w_i $ holds for some $ \circ \in \{<, >\} $. If this is the case for some $ \circ $, then we choose $ w := w_i $, set \[
		\Upsilon := \Upsilon \cup \{\pi_Z(\vec{z}) \cdot w \circ \pi_Z(b) \cdot w\}
	\] and continue searching for the next disjunct of $ \Theta $. Note that we have avoided doing a nontrivial recursive call because $ \Gamma \cup \{\pi_Z(\vec{z}) \cdot w = \pi_Z(b) \cdot w\} $ is unsatisfiable and hence $ \Pi $-simple.
	\item It may sometimes be the case that almost all $ \Omega \in \DisjOf{\Theta} $ entail $ \varphi $. Thus, even if $ \Omega $ has more $ Z $-dependencies than $ \Gamma $ and is consequently potentially distinguishable in the language of $ \Pi $-decompositions, doing the actual separation of $ \Gamma $ and $ \Omega $ via the $ \pi_Z(\vec{z}) \cdot w = \pi_Z(b) \cdot w $ predicate may be completely redundant if $ \Omega \models \varphi $. Hence, it is a good idea to first consider those $ \Omega \in \DisjOf{\Theta} $ which entail $ \neg\varphi $. Once the second loop has iterated over all such $ \Omega $, we test whether $ \Delta \vee \Theta \wedge \Upsilon \models \varphi $ and if yes, return the covering $ \Delta \vee \Theta \wedge \Upsilon $. Otherwise, we simply continue executing the second loop.
\end{enumerate}

\subsubsection{Tool comparison}

We evaluate the performance of our tool against a generic semi-decision procedure
\texttt{mondec$_1$} by Veanes et al. \cite{veanes:2017}. Their algorithm works over an arbitrary base theory and attempts to construct a monadic decomposition by splitting the space of models using two equivalence relations and maintaining path conditions, which then yield a monadic decomposition represented as an if-then-else formula. In order to be able to compare both algorithms for arbitrary partitions and not just for singleton ones, we adjusted \texttt{mondec$_1$} algorithm slightly so that it supports general partitions. The essence of our adaptation is to create a fresh variable for every block of $ \Pi $, which simulates tuples of original variables. We refer to the adapted version of the algorithm by \texttt{vardec$_1$}. Since all implementations are in python and rely on the same \texttt{z3-solver} library interacting with z3, it is fair to compare their \textit{de facto} performance.

In the remainder of this \zbsectionname/, the abbreviation \texttt{cover} denotes the double-exponential time algorithm deciding $ \Pi $-decomposability and outputting $ \Pi $-decompositions whenever they exist (see \zbrefsec{sec:vardec}, Theorem~\ref{thm:vardec_double_exponential}). We furthermore use the shorthand name \texttt{cover\_mondec} to denote the same algorithm composed with the reduction of Proposition~\ref{prop:reduction_to_binary_partitions}, allowing us to decide monadic decomposability using a method supporting only binary partitions. That is, \texttt{cover\_mondec} is a decision procedure for monadic decomposability, obtained by composing the reduction of Proposition~\ref{prop:reduction_to_binary_partitions} and the covering-based method of solving the variable decomposition problem, discussed in \zbrefsec{sec:vardec}.

\subsection{Benchmark suite}

The goal of our benchmark suite is to stress-test \texttt{cover}, \texttt{cover\_mondec}, \texttt{mondec$_1$} and \texttt{vardec$_1$} against various kinds of tricky or nontrivial instances. Since the two last algorithms do not terminate unless the formula is $ \Pi $-decomposable, we restrict the suite to certain classes of decomposable formulas we now define. Let \[
	\varphi^{\mathrm{add}}_n := \sum_{i=1}^n y_i < x < \sum_{i=1}^n y_i + 1 \wedge \bigwedge_{i=1}^n y_i = 0 \vee y_i = 2^i
\] Intuitively, $ \varphi^{\mathrm{add}}_n $ uses $ y_i $ variables to express all possible binary representations of $ \left\lfloor x \right\rfloor $, if $ \left\lfloor x \right\rfloor $ is even. Since the representation of any number as a sum of powers of two is unique, it follows that $ \varphi^{\mathrm{add}}_n $ is monadically decomposable. 

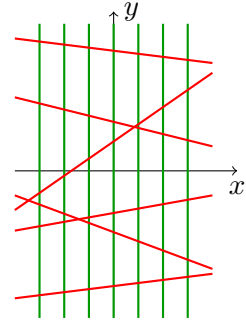
\begin{wrapfigure}{r}{0.25\textwidth}
	\if\zbthesismode1
	\vspace{-15pt}
	\else
	\vspace{-20pt}
	\fi
	\centering
	\begin{tikzpicture}[scale=0.65]	
	\draw[->] (-2,0) -- (2.5,0) coordinate[label = {below:$x$}] (xmax);
	\draw[->] (0,-3) -- (0,3.25) coordinate[label = {right:$y$}] (ymax);
	
	\draw[thick, OliveGreen] (1.5, -3) -- (1.5, 3);
	\draw[thick, OliveGreen] (1, -3) -- (1, 3);
	\draw[thick, OliveGreen] (.5, -3) -- (.5, 3);
	\draw[thick, OliveGreen] (0, -3) -- (0, 3);
	\draw[thick, OliveGreen] (-.5, -3) -- (-.5, 3);
	\draw[thick, OliveGreen] (-1, -3) -- (-1, 3);
	\draw[thick, OliveGreen] (-1.5, -3) -- (-1.5, 3);
	
	\draw[thick, red] (2, -2) -- (-2, -0.5);
	\draw[thick, red] (2, 2) -- (-2, -0.8);
	\draw[thick, red] (2, -0.5) -- (-2, -1.22);
	\draw[thick, red] (2, 0.5) -- (-2, 1.5);
	\draw[thick, red] (2, 2.2) -- (-2, 2.7);
	\draw[thick, red] (2, -2.1) -- (-2, -2.6);
	
	\end{tikzpicture}
	\caption{The set defined by $ \varphi^{\mathrm{grid2d}}_{6, 7} $ shown with green vertical lines. The formula is false at the intersections of green and red (sloping) lines.}
	\label{fig:exp:grid_formula}
	\if\zbthesismode1
	\vspace{-15pt}
	\else
	\vspace{-20pt}
	\fi
\end{wrapfigure}

We now define a new class of formulas $ \varphi^{\mathrm{grid2d}}_{n, k} $ over variables $ x $ and $ y $, whose goal is to require the inference of as many new predicates in the monadic decomposition as possible. We achieve this by generating $ k $ lines aligned along the $ y $ axis and $ n $ random lines intersecting\footnote{With probability approaching 1.} the aligned lines. Then, we construct the formula $ \varphi^{\mathrm{grid2d}}_{n, k} $ so that it is true at every point lying on an aligned line, but not on any non-aligned one (see Figure~\ref{fig:exp:grid_formula}). Thus, in order to construct a monadic decomposition, the algorithm has to somehow compute the points at which the aligned lines intersect with the non-aligned ones, which is what we aim to stress-test.

In two dimensions, all segments of the aligned lines are defined by $ \Pi $-simple predicate sets. Therefore, the $ \varphi^{\mathrm{grid}}_{n, k} $ formulas test only the effectiveness of the $ \Pi $-simple check integrated into the algorithm, whereas other parts of the covering algorithm are not tested. In order to assess, in a similar context, the performance of the covering algorithm when the predicate sets to be covered are $ \Pi $-complex, we define $ \varphi^{\mathrm{grid3d}}_k $ to be the similar class of formulas as $ \varphi^{\mathrm{grid2d}}_{n, k} $, but in three dimensions instead of two, meaning that every line becomes a 2D plane in 3D space. More precisely, $ \varphi^{\mathrm{grid3d}}_k $ defines $ k $ aligned 2D planes and is true at every point of a plane except exactly one (for every plane), which is where the plane gets intersected by two other randomly generated non-aligned planes. This construction ensures that the covering algorithm will have to execute both loops for every $ \Gamma \in \DisjPhi $. Furthermore, the second loop will have to infer a predicate separating $ \Omega $ from $ \Gamma $ and make a recursive call at least once.

\subsection{Results}

The experiments were conducted on an Intel(R) Core(TM) i7-10510U 1.80 GHz CPU with 4 cores and 16 GiB of RAM, running Ubuntu 22.04.1 LTS. The versions of \texttt{python3} and the \texttt{z3-solver} library are 3.10.6 and 4.12.1.0, respectively. The results are presented in Figure~\ref{fig:benchmark_results}.

\begin{figure}
	\centering
	\def\scale{1}
\def\scalesub{0.48}
\def\gridwidth{\textwidth}
\def\gridheight{4.5cm}
\def\colspace{0.01\textwidth}
\def\algpresvardec{red}
\def\algpresvardecnoheuristics{purple} 
\def\algpresvardecmd{orange}
\def\algpresvardecmdnoheuristics{gray}
\def\algveanesvardec{blue}
\def\algveanesmondec{OliveGreen}
\begin{subfigure}[t]{\scalesub\linewidth}
	\pgfplotstableread{experiments/add.dat}{\table}
	\begin{tikzpicture}[scale=\scale]
		\begin{axis}[
			xmin = 1, xmax = 8,
			ymin = 0, ymax = 30,
			xtick distance = 1,
			ytick distance = 5,
			scaled y ticks = false,
			grid = none,
			minor y tick num = 5,
			minor x tick num = 5,
			major grid style = {lightgray},
			minor grid style = {lightgray!25},
			width = \gridwidth,
			height = \gridheight,
			x label style={at={(axis description cs:0.5,0.1)},anchor=north},
			y label style={at={(axis description cs:0.12,0.5)},anchor=south},
			xlabel={Parameter $ n $},
			ylabel={Time (s)},
			label style={font=\tiny},
			tick label style={font=\tiny},
			legend cell align = {left},
			legend pos = south east,
			legend style={
				at={(0,0)},
				anchor=north,
				at={(axis description cs:0.3,-0.1)}
			}
			]
			
			\addplot[\algveanesvardec] table [x = {n}, y = {veanes_perf_s}] {\table};
			\addplot[\algveanesmondec] table [x = {n}, y = {veanes_perf_s_md}] {\table};
			
			\addplot[\algpresvardec] table [x = {n}, y = {presvardec_perf_s}] {\table};
			\addplot[\algpresvardecmd] table [x = {n}, y = {presvardec_perf_s_md}] {\table};
			
		\end{axis}
	\end{tikzpicture}
	\caption{Benchmark on $ \varphi^{\mathrm{add}}_n $}
	\label{fig:benchmark_results:add_time}
\end{subfigure}
\begin{subfigure}[t]{\scalesub\linewidth}
	\pgfplotstableread{experiments/add.dat}{\table}
	\begin{tikzpicture}[scale=\scale]
		\begin{axis}[
			xmin = 1, xmax = 8,
			ymin = 0, ymax = 4200,
			xtick distance = 1,
			ytick distance = 500,
			scaled y ticks = false,
			grid = none,
			minor y tick num = 5,
			minor x tick num = 5,
			major grid style = {lightgray},
			minor grid style = {lightgray!25},
			width = \gridwidth,
			height = \gridheight,
			x label style={at={(axis description cs:0.5,0.1)},anchor=north},
			y label style={at={(axis description cs:0.1,0.5)},anchor=south},
			xlabel={Parameter $ n $},
			ylabel={Decomposition size},
			label style={font=\tiny},
			tick label style={font=\tiny},
			legend cell align = {left},
			legend pos = south east,
			legend style={
				at={(0,0)},
				anchor=north,
				at={(axis description cs:0.3,-0.1)}
			}
			]
			
			\addplot[\algveanesvardec] table [x = {n}, y = {veanes_size}] {\table};
			\addplot[\algpresvardec] table [x = {n}, y = {presvardec_size}] {\table};
			
		\end{axis}
	\end{tikzpicture}
	\caption{Benchmark on $ \varphi^{\mathrm{add}}_n $}
	\label{fig:benchmark_results:add_size}
\end{subfigure}
\begin{subfigure}[t]{\scalesub\linewidth}
	\pgfplotstableread{experiments/grid.dat}{\table}
	\begin{tikzpicture}[scale=\scale]
		\begin{axis}[
			xmin = 1, xmax = 64,
			ymin = 0, ymax = 2200,
			xtick distance = 5,
			ytick distance = 200,
			scaled y ticks = false,
			grid = none,
			minor y tick num = 5,
			minor x tick num = 4,
			major grid style = {lightgray},
			minor grid style = {lightgray!25},
			width = \gridwidth,
			height = \gridheight,
			x label style={at={(axis description cs:0.5,0.1)},anchor=north},
			y label style={at={(axis description cs:0.12,0.5)},anchor=south},
			xlabel={Parameter $ n $},
			ylabel={Time (s)},
			label style={font=\tiny},
			tick label style={font=\tiny},
			legend cell align = {left},
			legend pos = south east,
			legend style={
				at={(0,0)},
				anchor=north,
				at={(axis description cs:0.3,-0.1)}
			}
			]
			
			\addplot[\algveanesmondec] table [x = {napc}, y = {veanes_perf_s}] {\table};
			
			\addplot[\algpresvardec] table [x = {napc}, y = {presvardec_perf_s}] {\table};
			
		\end{axis}
	\end{tikzpicture}
	\caption{Benchmark on $ \varphi^{\mathrm{grid2d}}_{n, 32} $}
	\label{fig:benchmark_results:grid2d_time}
\end{subfigure}
\begin{subfigure}[t]{\scalesub\linewidth}
	\pgfplotstableread{experiments/grid.dat}{\table}
	\begin{tikzpicture}[scale=\scale]
		\begin{axis}[
			xmin = 1, xmax = 64,
			ymin = 0, ymax = 55000,
			xtick distance = 5,
			ytick distance = 5000,
			scaled y ticks = false,
			grid = none,
			minor y tick num = 5,
			minor x tick num = 4,
			major grid style = {lightgray},
			minor grid style = {lightgray!25},
			width = \gridwidth,
			height = \gridheight,
			x label style={at={(axis description cs:0.5,0.1)},anchor=north},
			y label style={at={(axis description cs:0.1,0.5)},anchor=south},
			xlabel={Parameter $ n $},
			ylabel={Decomposition size},
			label style={font=\tiny},
			tick label style={font=\tiny},
			legend cell align = {left},
			legend pos = south east,
			legend style={
				at={(0,0)},
				anchor=north,
				at={(axis description cs:0.3,-0.1)}
			}
			]
			
			\addplot[\algveanesmondec] table [x = {napc}, y = {veanes_size}] {\table};
			\addplot[\algpresvardec] table [x = {napc}, y = {presvardec_size}] {\table};
			
		\end{axis}
	\end{tikzpicture}
	\caption{Benchmark on $ \varphi^{\mathrm{grid2d}}_{n, 32} $}
	\label{fig:benchmark_results:grid2d_size}
\end{subfigure}
\begin{subfigure}[t]{\scalesub\linewidth}
	\pgfplotstableread{experiments/spaces.dat}{\table}
	\begin{tikzpicture}[scale=\scale]
		\begin{axis}[
			xmin = 1, xmax = 53,
			ymin = 0, ymax = 2200,
			xtick distance = 5,
			ytick distance = 200,
			scaled y ticks = false,
			grid = none,
			minor y tick num = 5,
			minor x tick num = 4,
			major grid style = {lightgray},
			minor grid style = {lightgray!25},
			width = \gridwidth,
			height = \gridheight,
			x label style={at={(axis description cs:0.5,0.1)},anchor=north},
			y label style={at={(axis description cs:0.12,0.5)},anchor=south},
			xlabel={Parameter $ k $},
			ylabel={Time (s)},
			label style={font=\tiny},
			tick label style={font=\tiny},
			legend cell align = {left},
			legend pos = south east,
			legend style={
				at={(0,0)},
				anchor=north,
				at={(axis description cs:0.3,-0.1)}
			}
			]
			
			\addplot[\algveanesvardec] table [x = {spaces}, y = {veanes_perf_s}] {\table};
			\addplot[\algpresvardec] table [x = {spaces}, y = {presvardec_perf_s}] {\table};
			
		\end{axis}
	\end{tikzpicture}
	\caption{Benchmark on $ \varphi^{\mathrm{grid3d}}_k $}
	\label{fig:benchmark_results:grid3d_time}
\end{subfigure}
\begin{subfigure}[t]{\scalesub\linewidth}
	\pgfplotstableread{experiments/spaces.dat}{\table}
	\begin{tikzpicture}[scale=\scale]
		\begin{axis}[
			xmin = 1, xmax = 53,
			ymin = 0, ymax = 45000,
			xtick distance = 5,
			ytick distance = 5000,
			scaled y ticks = false,
			grid = none,
			minor y tick num = 5,
			minor x tick num = 4,
			major grid style = {lightgray},
			minor grid style = {lightgray!25},
			width = \gridwidth,
			height = \gridheight,
			x label style={at={(axis description cs:0.5,0.1)},anchor=north},
			y label style={at={(axis description cs:0.1,0.5)},anchor=south},
			xlabel={Parameter $ k $},
			ylabel={Decomposition size},
			label style={font=\tiny},
			tick label style={font=\tiny},
			legend cell align = {left},
			legend pos = south east,
			legend style={
				at={(0,0)},
				anchor=north,
				at={(axis description cs:0.3,-0.1)}
			}
			]
			
			\addplot[\algveanesvardec] table [x = {spaces}, y = {veanes_size}] {\table};
			\addplot[\algpresvardec] table [x = {spaces}, y = {presvardec_size}] {\table};
			
		\end{axis}
	\end{tikzpicture}
	\caption{Benchmark on $ \varphi^{\mathrm{grid3d}}_k $}
	\label{fig:benchmark_results:grid3d_size}
\end{subfigure}
\begin{center}
	\begin{tikzpicture}
		\begin{axis}[
			hide axis,
			xmin=50,
			xmax=50,
			ymin=0,
			ymax=0,
			legend columns=4,
			legend style={/tikz/every even column/.append style={column sep=0.5cm}}]
			]
			\addplot[\algpresvardec] {x};
			\addplot[\algpresvardecmd] {x};
			\addplot[\algveanesvardec] {x};
			\addplot[\algveanesmondec] {x};
			\legend{
				\texttt{cover},
				\texttt{cover\_mondec},
				\texttt{vardec$_1$},
				\texttt{mondec$_1$}
			}
		\end{axis}
	\end{tikzpicture}
	\if\zbthesismode1
	\vspace{-13.5em}
	\else
	\vspace{-16em}
	\fi
\end{center}
	\caption{Benchmark results. The size of a decomposition is the number of non-leaf nodes in the abstract syntax tree of the formula.}
	\label{fig:benchmark_results}
\end{figure}
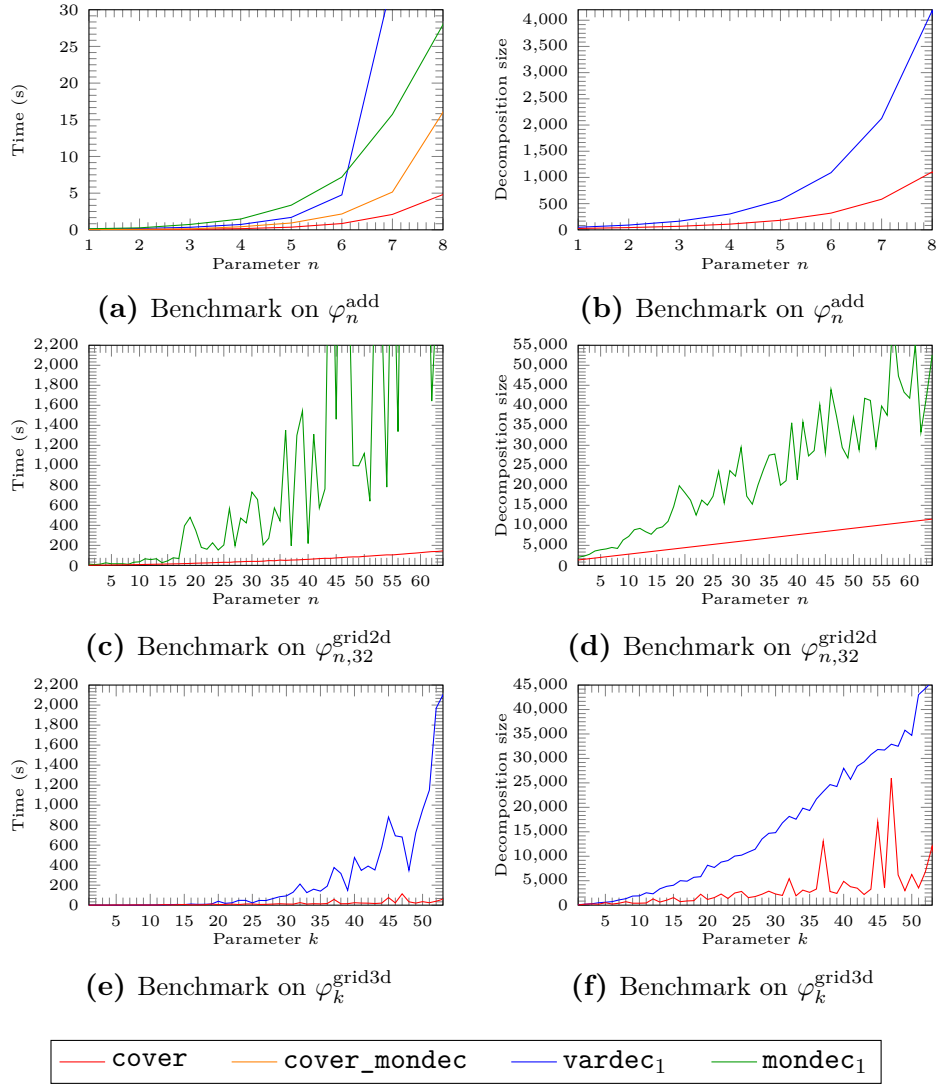

During benchmarking, \texttt{cover} and \texttt{cover\_mondec} were each called twice -- with heuristics enabled and without them. However, we observed that on the $ \varphi^{\mathrm{add}}_n $ and $ \varphi^{\mathrm{grid2d}}_{n, 32} $ families of formulas, heuristics cause hardly any improvement or degradation of running time or decomposition size (which is not surprising, given the tricky nature of the examples constructed). Thus, the red and orange lines in Figures \ref{fig:benchmark_results:add_time}, \ref{fig:benchmark_results:add_size}, \ref{fig:benchmark_results:grid2d_time} and \ref{fig:benchmark_results:grid2d_size} reflect the performance of the algorithm also with heuristics completely disabled. However, we observed that the situation in Figures \ref{fig:benchmark_results:grid3d_time} and \ref{fig:benchmark_results:grid3d_size} is substantially different in that without heuristics \texttt{cover} takes over one hour to determine the $ \Pi $-decomposability of $ \varphi^{\mathrm{grid3d}}_k $ for $ k \ge 20 $. Hence, $ \varphi^{\mathrm{grid3d}}_k $ is an example of a family of formulas demonstrating that the proposed heuristics can cause an enormous performance boost. The performance of \texttt{cover\_mondec} in Figure~\ref{fig:benchmark_results:add_time} gives encouraging evidence showing that the reduction of Proposition~\ref{prop:reduction_to_binary_partitions} is efficient in practice. Overall, the benchmarking results show that our covering-based algorithm equipped with heuristics beats the algorithm of \cite{veanes:2017} in both time and decomposition size for the considered classes of formulas.

\section{Conclusions and future work}
\label{sec:conclusions}

In this \zbdocname/, we have presented a novel model-flooding-based technique for solving the variable decomposition and independence problems over quantifier-free linear real arithmetic. We have applied this technique to derive algorithms for monadic/variable decomposability, which are optimal in theory and efficient in practice. We have also proven upper and lower bounds for the variable independence problem and have seen a way of efficiently constructing $ \Pi $-decompositions whenever they exist. For monadic and variable decomposability, our algorithms are exponentially faster compared to the best algorithms known for these problems. We also improved the current state-of-the-art algorithm for computing $ \Pi $-decompositions whenever they exist. Based on the way we obtained these results, we conclude:

\begin{center}
	\fbox{
		\parbox{0.85\textwidth}{
			Nontrivial topological properties of the set of logical statements (formulas) can be used to significantly speed up algorithms.
		}
	}
\end{center}

\subsection{Generalization and other theories}
\label{sec:conclusions:generalization}

The model-flooding-based technique we used relies at its core on quite fundamental logical properties of formulas. Thus, it has a great potential of being applicable to many other fragments of first-order logic because only two major steps are needed to adapt our proof to produce similar results. They are as follows.

\begin{itemize}
	\item Firstly and most importantly, one has to define a metric enabling us to compare satisfiable predicate sets by analyzing what dependencies between variables expressible as $ \Pi $-respecting predicates they enforce. For this, it is crucial to understand precisely what kind of connections between variables could, in principle, be established by a conjunction of predicates in the quantifier-free fragment of the given first-order theory of $ \Mstruct $. This is absolutely not surprising because monadic/variable decomposability problems are, in essence, questions about what can possibly be expressed as a Boolean combination of $ \Pi $-respecting predicates. The desired metric should behave similarly to $ \LinDep_{h} $ in the aspects that the dependencies between variables it describes should have a finite and computable basis and should satisfy the ascending chain condition in the sense that no matter which dependencies we are iteratively enforcing using predicates, we must after finitely many steps run out of dependencies which are not redundant. In the case of linear real arithmetic, this property follows from the fact that $ \Q^n $ is finitely dimensional and Noetherian. Due to the simple and fundamental nature of the said properties, we conjecture that they are satisfied by many other quantifier-free fragments of first-order logic if the metric for measuring dependencies in predicate sets has been chosen correctly.
	\item The second major step is to prove an analogue of the Overspilling Theorem~\ref{thm:overspilling}, which should establish the impossibility of separating one predicate set from another in the language of $ \Pi $-decompositions, assuming that the metric cannot distinguish the given predicate sets. Moreover, this result may impose further requirements on the given predicate sets, such as the property that $ \Theta_1 $ is $ p $-next to $ \Theta_2 $ on a predicate $ p \in \Theta_1^= $ (see Theorem~\ref{thm:overspilling}). However, it should be possible to reduce the induction we used to establish model flooding and completeness of $ \Lambda $-proofs (see Theorems \ref{thm:model_flooding} and \ref{thm:lambda_proof_soundness_completeness}) to the case when it can be without loss of generality assumed that the corresponding predicate sets satisfy the additional assumptions.
\end{itemize}

All other required changes are of technical nature and may decisively influence only the time complexity of the algorithms obtained.

\subsection{Future work}

We close this \zbdocname/ by discussing several further possible research directions:

\begin{itemize}
	\item \textit{Find new applications of variable decomposability and independence.} Although the study of these concepts is an active research topic in the field, there is a lack of applications, especially those which are not of prevailing theoretical nature. This cannot but have a demotivating effect on the community, decreasing the overall research effort.
	
	\item \textit{Establish the precise complexity of deciding variable independence over $ \QFLRA $.} In this \zbdocname/, we have established an $ \NP^\NP $ upper bound for the variable independence problem for quantifier-free linear real arithmetic (Corollary \ref{cor:variable_independence_sigma2}). We have also shown that the same problem is $ \coNP $-hard with respect to conjunctive truth-table reductions (Theorem~\ref{thm:variable_independence_conp_hard}). The upper bound does not meet the lower bound, so it remains an open problem to establish the precise complexity of determining variable independence.
	
	\item \textit{Derive an optimal generic algorithm for variable decomposability.} Our results and known results (see, e.g., \cite{hague:2020}) establish optimal algorithms for variable decomposability only for fixed fragments of first-order logic. By contrast, known algorithms that work regardless of the base theory (see, e.g., \cite{veanes:2017, libkin:2003}) are either not known to or do not achieve optimal complexity in theory. Thus, it remains an open problem to derive an optimal generic algorithm for deciding variable decomposability. A possible first step towards a solution would be to follow the ideas outlined in Section~\ref{sec:conclusions:generalization} above.
	
	\item \textit{Investigate whether parallelization can make the covering-based algorithm for producing decompositions even more efficient.} Recall that the covering-based algorithm for constructing $ \Pi $-decompositions we derived in \zbrefsec{sec:vardec} just computes coverings of disjuncts in $ \Sat(\DisjPhi) $ and then outputs their disjunction. The calls to the covering algorithm are entirely independent of each other and thus have the potential of being parallelized. Since the computation of certain coverings may well make the computation of other coverings completely unnecessary (e.g., it is not necessary to compute the covering of $ \Gamma' $ if $ \Gamma' $ entails the covering of a different $ \Gamma \in \Sat(\DisjPhi) $, as we have seen in Example~\ref{example:reduction_ex}), there is even more room for optimization. Furthermore, the iterations of the second loop of the covering algorithm (at Line~\ref{alg:cover:line:second_loop}) are also independent and thus may be parallelized. Overall, it is an exciting task to try and further optimize the algorithm by executing the covering algorithm and the iterations of its second loop in parallel.
\end{itemize}

\bibliography{lra}

\appendix

\section{Facts about linear dependencies}
\label{sec:app:lindep_facts}

In this \zbsectionname/, we establish facts about the $ \LinDep_h $ metric we rely on in the proofs of other results.

\subsection{Closure properties}

We start by showing some basic closure properties of $ \LinDep_{h, v} $.

\begin{lemma}
	\label{lemma:lindep_closed_under_addition}
	Let $ \Gamma $ be a satisfiable set of predicates, $ V $ be a $ \Q $-vectorspace and $ h : \Q^n \rightarrow V $ be a mapping. Then, for any $ v \in \Q^n $, $ \LinDep_{h, v}(\Gamma) $ is closed under addition and taking subspaces.
\end{lemma}
\begin{proof}
	Let $ U_1, U_2 \in \LinDep_{h, v}(\Gamma) $ and let $ U \le U_1 $. Since \[
	h(\ModelsOf(\Gamma) - v) \subseteq U_1^\bot \cap U_2^\bot = (U_1 + U_2)^\bot
	\] we have $ U_1 + U_2 \in \LinDep_{h, v}(\Gamma) $. Closure under taking subspaces, that is, $ U \in \LinDep_{h, v}(\Gamma) $, follows from $ h(\ModelsOf(\Gamma) - v) \subseteq U_1^\bot \le U^\bot $.
\end{proof}

\subsection{Only equality predicates can establish linear dependencies}

We now show that only equality predicates can establish linear dependencies. This fact is very powerful because it allows us to completely ignore strict inequality predicates and focus only on the equalities every time we study some properties of $ \LinDep_{h, v}(\Gamma) $. Consequently, this allows us to study only the relevant properties of affine vectorspaces instead of evenly convex polyhedral sets \cite{rodriguez:2017}, which are significantly more complicated algebraic structures.

\begin{theorem}
	\label{thm:only_equality_predicates_can_establish_lindep}
	Let $ \Gamma $ be a satisfiable set of predicates, $ V $ be a $ \Q $-vectorspace and $ h : \Q^n \rightarrow V $ be a $ \Q $-vectorspace homomorphism. Then \[
	\LinDep_{h, v}(\Gamma) = \LinDep_{h, v}(\Gamma^=)
	\] holds for all $ v \models \Gamma $.
\end{theorem}
\begin{proof}
	``$ \supseteq $'': Let $ U \in \LinDep_{h, v}(\Gamma^=) $. Since $ \Gamma \models \Gamma^= $, we have \[
	h(\ModelsOf(\Gamma) - v) \subseteq h(\ModelsOf(\Gamma^=) - v) \subseteq U^\bot
	\] Hence, $ U \in \LinDep_{h, v}(\Gamma) $.
	
	``$ \subseteq $'': We show the contrapositive of $ \LinDep_{h, v}(\Gamma) \subseteq \LinDep_{h, v}(\Gamma^=) $. It is clear that $ \gen{0} $ is in both sets, so we can, without loss of generality, consider only elements distinct from $ \gen{0} $. Let $ \gen{0} \neq \gen{w_1, \dots, w_d} \notin \LinDep_{h, v}(\Gamma^=) $. Then there exists a single vector $ w_i \in \{w_1, \dots, w_d\} $ such that $ \gen{0} \neq \gen{w_i} \notin \LinDep_{h, v}(\Gamma^=) $ because if $ \gen{w_i} \in \LinDep_{h, v}(\Gamma^=) $ was true for all $ i \in \{1, \dots, d\} $, then \[
	\sum_{i=1}^d \gen{w_i} = \gen{w_1, \dots, w_d} \in \LinDep_{h, v}(\Gamma^=)
	\] would hold since $ \LinDep_{h, v}(\Gamma^=) $ is closed under addition (Lemma~\ref{lemma:lindep_closed_under_addition}). By definition of $ \LinDep_{h, v} $, there exists $ u \models \Gamma^= $ such that \begin{align}
		\label{eqn:h_ui_minus_v_dot_w_nonzero}
		h(u - v) \cdot w_i \neq 0
	\end{align} Let $ v + W $ be the affine vectorspace of solutions to $ \Gamma^= $. Note that $ u - v \in W $ and $ \gen{0} \neq \gen{u - v} \le W $.
	Since $ \Gamma \setminus \Gamma^= $ contains only strict inequality predicates, there exists some $ \varepsilon > 0 $ such that \[
	p := v + \varepsilon \cdot \underbrace{(u - v)}_{\neq 0} \models \Gamma
	\] More precisely, $ p $ is a model of $ \Gamma $ because if some finite conjunction of strict inequalities is satisfied by $ v $, then there exists an $ \varepsilon $-neighborhood of that point $ v $, such that all those strict inequalities are still satisfied within that $ \varepsilon $-neighborhood. By linearity of $ h $ and (\ref{eqn:h_ui_minus_v_dot_w_nonzero}), we have \[
	h(p - v) \cdot w_i = \varepsilon \cdot \underbrace{(h(u - v) \cdot w_i)}_{\neq 0} \neq 0
	\] Hence, $ \gen{w_i} \notin \LinDep_{h, v}(\Gamma) $. Consequently, $ \gen{w_1, \dots, w_d} \notin \LinDep_{h, v}(\Gamma) $ because otherwise $ \gen{w_i} \in \LinDep_{h, v}(\Gamma) $ would hold since $ \LinDep_{h, v}(\Gamma) $ is closed under taking subspaces (Lemma~\ref{lemma:lindep_closed_under_addition}).
\end{proof}

\subsection{Invariance of linear dependencies under change of offset}

Now, we complement the above facts by proving that $ \LinDep_{h, v}(\Gamma) $ does not depend on the particular choice of $ v \models \Gamma^= $.

\begin{lemma}
	\label{lemma:lindep_gamma_v1_equals_lindep_gamma_v2}
	Let $ \Gamma $ be a satisfiable set of predicates, $ V $ be a $ \Q $-vectorspace and $ h : \Q^n \rightarrow V $ be a $ \Q $-vectorspace homomorphism. Then \[
	\LinDep_{h, v_1}(\Gamma) = \LinDep_{h, v_2}(\Gamma)
	\] holds for all $ v_1, v_2 \models \Gamma^= $.
\end{lemma}
\begin{proof}
	We only prove the ``$ \subseteq $'' inclusion because the proof of the converse direction is symmetric. Let $ v \models \Gamma $ and $ v + W $ be the affine vectorspace of solutions to $ \Gamma^= $. Let furthermore $ U \in \LinDep_{h, v_1}(\Gamma) $. By definition of $ \LinDep_{h, v_1} $, \begin{align}
		\label{eqn:h_modgamma_minus_v1_subsetof_ubot}
		h(\ModelsOf(\Gamma) - v_1) \subseteq U^\bot
	\end{align}
	\begin{zbclaim}
		\label{claim:hW_subsetof_Ubot}
		$ h(W) \subseteq U^\bot $
	\end{zbclaim}
	\begin{proof}
		Let $ w \in W $ be arbitrary. Since $ \Gamma \setminus \Gamma^= $ contains only strict inequality predicates and $ v \models \Gamma $, there exists some $ \varepsilon > 0 $ such that $ v + \varepsilon \cdot w \models \Gamma $. By linearity of $ h $ and (\ref{eqn:h_modgamma_minus_v1_subsetof_ubot}), \[
		h(v + \varepsilon \cdot w - v_1) = \varepsilon \cdot h(w) + h(v - v_1) \in U^\bot
		\] Note that (\ref{eqn:h_modgamma_minus_v1_subsetof_ubot}) also implies $ h(v - v_1) \in U^\bot $. Hence, $ h(w) \in U^\bot $.
	\end{proof}
	
	\noindent Claim~\ref{claim:hW_subsetof_Ubot} immediately implies $ h(v_1 - v_2) \in h(W) \subseteq U^\bot $ because \[
	v_1 - v_2 \in (v + W) - (v + W) = W - W = W
	\] We conclude that \[
	h(\ModelsOf(\Gamma) - v_2) = h(\ModelsOf(\Gamma) - v_1 + v_1 - v_2) = \underbrace{h(\ModelsOf(\Gamma) - v_1)}_{\subseteq U^\bot} + \underbrace{h(v_1 - v_2)}_{\subseteq U^\bot} \subseteq U^\bot
	\] and hence $ U \in \LinDep_{h, v_2}(\Gamma) $.
\end{proof}

Note that Lemma~\ref{lemma:lindep_gamma_v1_equals_lindep_gamma_v2} implies, in particular, that $ \LinDep_h(\Gamma) $ is well-defined.

\subsection{Stronger version of Theorem~\ref{thm:only_equality_predicates_can_establish_lindep}}

Relying on the above Theorem~\ref{thm:only_equality_predicates_can_establish_lindep} and Lemma~\ref{lemma:lindep_gamma_v1_equals_lindep_gamma_v2}, we now derive the following strengthened version of Theorem~\ref{thm:only_equality_predicates_can_establish_lindep}.

\begin{theorem}
	\label{thm:only_equality_predicates_can_establish_lindep_strong}
	Let $ \Gamma $ be a satisfiable set of predicates, $ V $ be a $ \Q $-vectorspace and $ h : \Q^n \rightarrow V $ be a $ \Q $-vectorspace homomorphism. Then \[
	\LinDep_{h, v}(\Gamma) = \LinDep_{h, v}(\Gamma^=)
	\] holds for all $ v \models \Gamma^= $.
\end{theorem}
\begin{proof}
	Fix some $ u \models \Gamma $. Then \[
	\LinDep_{h, v}(\Gamma) \overset{\ref{lemma:lindep_gamma_v1_equals_lindep_gamma_v2}}{=} \LinDep_{h, u}(\Gamma) \overset{\ref{thm:only_equality_predicates_can_establish_lindep}}{=} \LinDep_{h, u}(\Gamma^=) \overset{\ref{lemma:lindep_gamma_v1_equals_lindep_gamma_v2}}{=} \LinDep_{h, v}(\Gamma^=)
	\] holds by Theorem~\ref{thm:only_equality_predicates_can_establish_lindep} and Lemma~\ref{lemma:lindep_gamma_v1_equals_lindep_gamma_v2}.
\end{proof}

\subsection{Linear dependencies and arbitrary offsets}

Next, we establish another fact about offsets with respect to which we analyze linear dependencies. Intuitively, the following lemma says that no predicate set $ \Gamma $ can contain more linear dependencies when analyzed with respect to an arbitrary offset, compared to the dependencies obtained by analyzing with respect to offset $ v \models \Gamma $.

\begin{lemma}
	\label{lemma:arbitrary_offset_less_lindep_compared_to_good_offset}
	Let $ \Gamma $ be a satisfiable set of predicates, $ V $ be a $ \Q $-vectorspace and $ h : \Q^n \rightarrow V $ be a $ \Q $-vectorspace homomorphism. Then \[
	\LinDep_{h, u}(\Gamma) \subseteq \LinDep_{h, v}(\Gamma)
	\] holds for all $ u \in \Q^n $ and $ v \models \Gamma $.
\end{lemma}
\begin{proof}
	Let $ v + W $ be the affine vectorspace of solutions to $ \Gamma^= $ and let $ U \in \LinDep_{h, u}(\Gamma) $. By definition of $ \LinDep_{h, u}(\Gamma) $, \begin{align}
		\label{eqn:h_modgamma_minus_u_subsetof_ubot}
		h(\ModelsOf(\Gamma) - u) \subseteq U^\bot
	\end{align}
	\begin{zbclaim}
		\label{claim:hW_subsetof_Ubot_two}
		$ h(W) \subseteq U^\bot $
	\end{zbclaim}
	\begin{proof}
		Let $ w \in W $ be arbitrary. Since $ \Gamma \setminus \Gamma^= $ contains only strict inequality predicates and $ v \models \Gamma $, there exists some $ \varepsilon > 0 $ such that $ v + \varepsilon \cdot w \models \Gamma $. By linearity of $ h $ and (\ref{eqn:h_modgamma_minus_u_subsetof_ubot}), \[
		h(v + \varepsilon \cdot w - u) = \varepsilon \cdot h(w) + h(v - u) \in U^\bot
		\] Note that (\ref{eqn:h_modgamma_minus_u_subsetof_ubot}) also implies $ h(v - u) \in U^\bot $. Hence, $ h(w) \in U^\bot $.
	\end{proof}
	
	\noindent By Claim~\ref{claim:hW_subsetof_Ubot_two}, \[
	h(\ModelsOf(\Gamma^=) - v) = h(v + W - v) = h(W) \overset{\ref{claim:hW_subsetof_Ubot_two}}{\subseteq} U^\bot
	\] Hence, $ U \in \LinDep_{h, v}(\Gamma^=) $. By Theorem~\ref{thm:only_equality_predicates_can_establish_lindep}, $ U \in \LinDep_{h, v}(\Gamma) $.
\end{proof}

Relying on the facts derived above, we now prove Theorems \ref{thm:lindep_algebraic_char}, \ref{thm:union_of_same_lindep_has_that_lindep} and Lemma~\ref{lemma:adding_neq_predicates_cannot_make_pi_complex_set_pi_simple}.

\subsection{Proof of Theorem~\ref{thm:lindep_algebraic_char}}
\label{sec:app:lindep_facts:proof_lindep_algebraic_char}

We start by showing the ``$ \subseteq $'' inclusion. Let $ U \in \LinDep_{h, v}(\Gamma) $. By Theorem~\ref{thm:only_equality_predicates_can_establish_lindep_strong}, $ U \in \LinDep_{h, v}(\Gamma^=) $, so by definition of $ \LinDep_{h, v} $ and the linearity of $ h $ we have \[
h(b) - h(v) + h(\gen{b_1, \dots, b_r}) \subseteq U^\bot
\] Since $ h(\gen{b_1, \dots, b_r}) = \ImageOf(h \circ \lc_B) $, we can also write \[
h(b) - h(v) + \ImageOf(h \circ \lc_B) \subseteq U^\bot
\] Note that $ 0 \in \ImageOf(h \circ \lc_B) $ implies $ h(b) - h(v) \in U^\bot $ and hence \[
\ImageOf(h \circ \lc_B) \subseteq U^\bot - \underbrace{h(b) + h(v)}_{\in U^\bot} = U^\bot
\] Since $ h \circ \lc_B $ is a homomorphism, $ \ImageOf(h \circ \lc_B) \le U^\bot $. Hence, $ U \le \ImageOf(h \circ \lc_B)^{\bot} $.

For the converse ``$ \supseteq $'' inclusion, let $ U \le \ImageOf(h \circ \lc_B)^{\bot} $. This implies $ \ImageOf(h \circ \lc_B) \le U^{\bot} $. By Theorem~\ref{thm:only_equality_predicates_can_establish_lindep_strong} it suffices to show that $ U \in \LinDep_{h, v}(\Gamma^=) $, that is, \[
h(b - v) + h(\gen{b_1, \dots, b_r}) \subseteq U^\bot
\] Indeed, \[
h(b - v) + h(\gen{b_1, \dots, b_r}) \subseteq h(\gen{b_1, \dots, b_r}) = \ImageOf(h \circ \lc_B) \subseteq U^\bot
\] follows from $ b - v \in \gen{b_1, \dots, b_r} $ and $ \ImageOf(h \circ \lc_B) \le U^{\bot} $. This completes the proof of Theorem~\ref{thm:lindep_algebraic_char}.

\subsection{Proof of the Uniform union Theorem~\ref{thm:union_of_same_lindep_has_that_lindep}}
\label{sec:app:lindep_facts:proof_uniform_union_theorem}

Since $ \Lambda_m \models \Lambda $ for all $ m $, by Lemma~\ref{lemma:lindep_gamma_v1_equals_lindep_gamma_v2} we get \[
\LinDep_{h}(\Lambda) \overset{\ref{lemma:lindep_gamma_v1_equals_lindep_gamma_v2}}{=} \LinDep_{h, v_m}(\Lambda) \subseteq \LinDep_{h, v_m}(\Lambda_m) = \LinDep_{h}(\Lambda_m)
\] for all $ v_m \models \Lambda_m $. This proves the ``$ \subseteq $'' direction.

For the converse ``$ \supseteq $'' inclusion, let $ v_1 + W_1, \dots, v_k + W_k $ be affine vectorspaces such that $ \ModelsOf(\Lambda_i^=) = v_i + W_i $ for all $ i \in \{1, \dots, k\} $.

\begin{zbclaim}
	\label{claim:wlog_hvprime_i_in_hW_i_orth}
	We can, without loss of generality, assume that $ h(v_i) \in h(W_i)^\bot $ holds for all $ i \in \{1, \dots, k\} $.
\end{zbclaim}
\begin{proof}
	Let $ w_1, \dots, w_r \in W_i $ be vectors such that $ h(W_i) = \gen{h(w_1), \dots, h(w_r)} $. Redefine $ v_i $ to be \begin{align}
		\label{eqn:vprimeprime_constr}
		v'_i := v_i - \sum_{j=1}^{r} \frac{h(v_i) \cdot h(w_j)}{\norm{h(w_j)}^2} \cdot w_j
	\end{align} Since we are subtracting off a linear combination of $ w_1, \dots, w_r \in W_i $, the affine space $ v_i + W_i = v'_i + W_i $ does not change under the above transformation and hence $ \ModelsOf(\Lambda_i^=) = v_i + W_i $ remains true. The fact that (\ref{eqn:vprimeprime_constr}) enforces the claimed property $ h(v'_i) \in h(W_i)^\bot $ follows by construction of $ v'_i $ because \begin{align*}
		h(v'_i) = h(v_i) - \sum_{j=1}^{r} \frac{h(v_i) \cdot h(w_j)}{\norm{h(w_j)}^2} \cdot h(w_j)
	\end{align*} and $ \sum_{j=1}^{r} \frac{h(v_i) \cdot h(w_j)}{\norm{h(w_j)}^2} \cdot h(w_j) $ is the orthogonal projection of $ h(v_i) $ onto $ h(W_i) $.
\end{proof}

\begin{zbclaim}
	\label{claim:lambda_i_eq_entails_lambda_eq}
	For every $ i \in \{1, \dots, k\} $ it holds that $ \Lambda_i^= \models \Lambda^= $.
\end{zbclaim}
\begin{proof}
	Let $ w \models \Lambda_i \models \Lambda \models \Lambda^= $ and let $ w + W $ be the affine vectorspace of solutions to $ \Lambda^= $. Since $ w \models \Lambda_i^= $, it follows that $ \ModelsOf(\Lambda_i^=) = w + W_i $. Hence, it suffices to show $ W_i \le W $ because then we get \[
	\ModelsOf(\Lambda_i^=) = w + W_i \subseteq w + W = \ModelsOf(\Lambda^=)
	\] as desired. Let $ w' \in W_i $ be arbitrary. Since $ \Lambda_i \setminus \Lambda_i^= $ contains only strict inequality predicates, there exists an $ \varepsilon > 0 $ such that \[
	w + \varepsilon \cdot w' \models \Lambda_i \models \Lambda \models \Lambda^=
	\] Hence, $ w + \varepsilon \cdot w' \in w + W $ and thus $ w' \in W $, so we have shown $ W_i \le W $.
\end{proof}

\begin{zbclaim}
	\label{claim:h_w1_eq_h_w2_eq_dots_eq_h_wk}
	$ h(W_1) = h(W_2) = \dots = h(W_k) $
\end{zbclaim}
\begin{proof}
	By Theorem~\ref{thm:lindep_algebraic_char}, we have \begin{center}
		\begin{tabular}{c c c}
			$ \LinDep_{h}(\Lambda_1) $ & $ = \dots = $ & $ \LinDep_{h}(\Lambda_k) $ \\
			\rotatebox[origin=c]{90}{$ \overset{\ref{thm:lindep_algebraic_char}}{=} $} && \rotatebox[origin=c]{90}{$ \overset{\ref{thm:lindep_algebraic_char}}{=} $} \\
			$ \{U \le V \mid U \le h(W_1)^\bot\} $ & $ = \dots = $ & $ \{U \le V \mid U \le h(W_k)^\bot\} $
		\end{tabular}
	\end{center} Hence, $ h(W_1)^\bot = h(W_2)^\bot = \dots = h(W_k)^\bot $ and the claim follows.
\end{proof}

\begin{zbclaim}
	\label{claim:h_vprime_ij_all_equal}
	$ h(v_1) = h(v_2) = \dots = h(v_k) $
\end{zbclaim}

Before we prove Claim~\ref{claim:h_vprime_ij_all_equal}, we show that this claim implies the desired inclusion. Let $ U \in \LinDep_{h, v}(\Lambda_m) $ for some $ v \models \Lambda_m $. By Theorem~\ref{thm:only_equality_predicates_can_establish_lindep_strong}, $ U \in \LinDep_{h, v}(\Lambda_m^=) $, so by Lemma~\ref{lemma:lindep_gamma_v1_equals_lindep_gamma_v2} we have $ U \in \LinDep_{h, v_m}(\Lambda_m^=) $. Hence, it follows from the definition of $ \LinDep_{h, v_m}(\Lambda_m^=) $ that $ h(W_m) \subseteq U^\bot $. Applying $ h $ to \[
\ModelsOf(\Lambda) = \bigcup_{i=1}^k \ModelsOf(\Lambda_i) \subseteq \bigcup_{i=1}^k \ModelsOf(\Lambda_i^=) = \bigcup_{i=1}^k (v_i + W_i)
\] yields, by Claim~\ref{claim:h_w1_eq_h_w2_eq_dots_eq_h_wk}, \begin{align}
	\label{eqn:h_modlambda_subsetof_union_wm}
	h(\ModelsOf(\Lambda)) \subseteq \bigcup_{i=1}^k \big(h(v_i) + h(W_i)\big) \overset{\ref{claim:h_w1_eq_h_w2_eq_dots_eq_h_wk}}{=} \{h(v_1), \dots, h(v_k)\} + h(W_m)
\end{align} More precisely, here we are allowed to bring $ h(W_m) $ out of the union, simply because the $ h(W_i) $ sets are all equal by Claim~\ref{claim:h_w1_eq_h_w2_eq_dots_eq_h_wk}. By adding an offset $ v_m $ to (\ref{eqn:h_modlambda_subsetof_union_wm}) and applying Claim~\ref{claim:h_vprime_ij_all_equal}, we get \[
	h(\ModelsOf(\Lambda) - v_m) \overset{(\ref{eqn:h_modlambda_subsetof_union_wm})}{\subseteq} \{h(v_1 - v_m), \dots, h(v_k - v_m)\} + h(W_m) \overset{\ref{claim:h_vprime_ij_all_equal}}{=} h(W_m) \subseteq U^\bot
\] Hence, $ U \in \LinDep_{h, v_m}(\Lambda) $. Now observe that by Claim~\ref{claim:lambda_i_eq_entails_lambda_eq}, $ v_m \models \Lambda_m^= \models \Lambda^= $, so we can apply Lemma~\ref{lemma:lindep_gamma_v1_equals_lindep_gamma_v2} and obtain $ U \in \LinDep_{h, v}(\Lambda) $. Hence, we have proven the desired inclusion $ \LinDep_{h, v}(\Lambda_m) \subseteq \LinDep_{h, v}(\Lambda) $ assuming Claim~\ref{claim:h_vprime_ij_all_equal}, so in the remainder of the proof we show this claim.

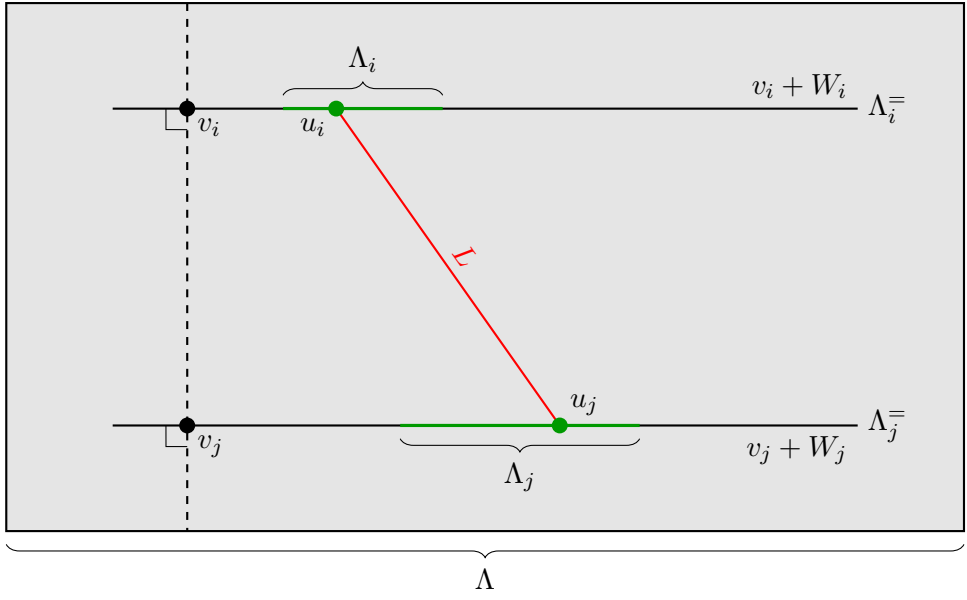
\begin{figure}[t]
	\centering
	\def\orthrectsidesize{2mm}
\begin{tikzpicture}[scale=\zbtikzbiggerscaling]
	
	\coordinate (lambdatl) at (0, 5);
	\coordinate (lambdatr) at (9, 5);
	\coordinate (lambdabl) at (0, 0);
	\coordinate (lambdabr) at (9, 0);
	
	\coordinate (lambdaih) at (0, 4);
	\coordinate (lambdajh) at (0, 1);
	
	\coordinate (lambdail) at ($ (lambdaih-|lambdatl) + (1, 0) $);
	\coordinate (lambdair) at ($ (lambdaih-|lambdatr) - (1, 0) $);
	
	\coordinate (lambdajl) at ($ (lambdajh-|lambdatl) + (1, 0) $);
	\coordinate (lambdajr) at ($ (lambdajh-|lambdatr) - (1, 0) $);
	
	\coordinate (vipoint) at ($ (lambdail)!0.1!(lambdair) $);
	\coordinate (vjpoint) at ($ (lambdajl)!0.1!(lambdajr) $);
	
	\coordinate (orthlinetop) at ($ (lambdatl-|vipoint) $);
	\coordinate (orthlinebot) at ($ (lambdabl-|vjpoint) $);
	
	\coordinate (orthlinerecttb) at ($ (vipoint)!\orthrectsidesize!(orthlinebot) $); 
	\coordinate (orthlinerecttl) at ($ (vipoint)!\orthrectsidesize!(lambdail) $); 
	\coordinate (orthlinerecttc) at ($ (orthlinerecttl) + (orthlinerecttb) - (vipoint) $);
	\coordinate (orthlinerectbb) at ($ (vjpoint)!\orthrectsidesize!(orthlinebot) $); 
	\coordinate (orthlinerectbl) at ($ (vjpoint)!\orthrectsidesize!(lambdajl) $); 
	\coordinate (orthlinerectbc) at ($ (orthlinerectbl) + (orthlinerectbb) - (vjpoint) $);
	
	\coordinate (uipoint) at ($ (lambdail)!0.3!(lambdair) $);
	\coordinate (ujpoint) at ($ (lambdajl)!0.6!(lambdajr) $);
	
	\coordinate (uitruel) at ($ (uipoint) - (0.5, 0) $);
	\coordinate (uitruer) at ($ (uipoint) + (1, 0) $);
	
	\coordinate (ujtruel) at ($ (ujpoint) - (1.5, 0) $);
	\coordinate (ujtruer) at ($ (ujpoint) + (0.75, 0) $);

	
	\filldraw[thick, black, fill=gray, fill opacity=0.2] (lambdabl) -- (lambdatl) -- (lambdatr) -- (lambdabr) -- cycle;
	\draw [decorate, decoration = {brace, raise=5pt, amplitude=5pt, mirror}] (lambdabl) --  (lambdabr) node[pos=0.5,below=10pt,black]{$ \Lambda $};
	
	\draw[thick] (lambdail) -- (lambdair);
	\node[anchor=west] at (lambdair) {$ \Lambda_i^= $};
	\node[anchor=south east] at (lambdair) {$ v_i + W_i $};
	\draw[thick] (lambdajl) -- (lambdajr);
	\node[anchor=west] at (lambdajr) {$ \Lambda_j^= $};
	\node[anchor=north east] at (lambdajr) {$ v_j + W_j $};
	
	\draw[thick, dashed] (orthlinetop) -- (orthlinebot);
	
	\fill[black] (vipoint) circle (0.075);
	\node[anchor=north west] at (vipoint) {$ v_i $};
	
	\fill[black] (vjpoint) circle (0.075);
	\node[anchor=north west] at (vjpoint) {$ v_j $};
	
	\draw (orthlinerecttb) -- (orthlinerecttc) -- (orthlinerecttl);
	\draw (orthlinerectbb) -- (orthlinerectbc) -- (orthlinerectbl);
	
	\draw[very thick, OliveGreen] (uitruel) -- (uitruer);
	\draw [decorate, decoration = {brace, raise=5pt, amplitude=5pt}] (uitruel) -- (uitruer) node[pos=0.5,above=10pt,black]{$ \Lambda_i $};
	
	\draw[very thick, OliveGreen] (ujtruel) -- (ujtruer);
	\draw [decorate, decoration = {brace, raise=5pt, amplitude=5pt, mirror}] (ujtruel) -- (ujtruer) node[pos=0.5,below=10pt,black]{$ \Lambda_j $};
	
	\path[thick, red] (uipoint) edge node [sloped, above, midway] {$ L $} (ujpoint);
	
	\fill[OliveGreen] (uipoint) circle (0.075);
	\node[anchor=north east] at (uipoint) {$ u_i $};
	
	\fill[OliveGreen] (ujpoint) circle (0.075);
	\node[anchor=south west] at (ujpoint) {$ u_j $};
	
\end{tikzpicture}
	\caption{Illustration of the geometric intuition behind the assumptions leading to a contradiction, for the special case when $ h = \id $ is the identity mapping. Horizontal black lines depict the affine vectorspaces $ v_i + W_i $ and $ v_j + W_j $ corresponding to solutions of $ \Lambda_i^= $ and $ \Lambda_j^= $, respectively. The dashed line orthogonal to both $ v_i + W_i $ and $ v_j + W_j $ illustrates the orthogonality assumption $ v_i \in W_i^\bot $, $ v_j \in W_j^\bot $ of Claim~\ref{claim:wlog_hvprime_i_in_hW_i_orth}.}
	\label{fig:same_lindep_union_contradiction}
\end{figure}

Suppose, for the sake of contradiction, that $ h(v_i) \neq h(v_j) $ holds for some $ i, j \in \{1, \dots, k\} $. Let $ u_i \models \Lambda_i $, $ u_j \models \Lambda_j $ and consider the line $ L := (u_i, u_j) $ between $ u_i $ and $ u_j $ (see Figure~\ref{fig:same_lindep_union_contradiction} for a visualization). Without loss of generality, we assume that $ u_i \neq u_j $ and consequently $ \left|L\right| = \infty $. If this is not the case, we simply pick another pair of models. If this is impossible, then it follows, in particular, that $ \Lambda_i \equiv \Lambda_j $, but it can be assumed, without loss of generality, that all such redundant disjuncts have already been removed from $ \{\Lambda_1, \dots, \Lambda_k\} $.

Intuitively, we now show that line $ L $ agrees with every space of solutions to $ \Lambda_l^= $ on at most one point. Then, we observe that for $ \Lambda $ to be correct, the set of its solutions must contain $ L $. Since $ \Lambda $ is expressible as a finite disjunction $ \Lambda \equiv \bigvee_{i = 1}^k \Lambda_i $, whereas every $ \Lambda_i $ can agree with $ L \subseteq \ModelsOf(\Lambda) $ on at most one element, we obtain a contradiction. Speaking in terms of the geometric intuition of Figure~\ref{fig:same_lindep_union_contradiction}, this argumentation rigorously proves that in order for $ \ModelsOf(\Lambda) $ to contain $ L $, we need infinitely many disjuncts $ \Lambda_l $, because every one of them can only define a space which is parallel to $ \ModelsOf(\Lambda_i^=) $ and $ \ModelsOf(\Lambda_j^=) $. We now make these vague intuitions precise.

\begin{zbclaim}
	\label{claim:vlprime_plus_Wl_int_L_at_most_one_elem}
	For every $ l \in \{1, \dots, k\} $ it holds that \[
		\left|(v_l + W_l) \cap L\right| \le 1
	\]
\end{zbclaim}
\begin{proof}
	Suppose, for the sake of contradiction, that there exist distinct $ w_1, w_2 \in W_l $ for some $ l $, such that $ v_l + w_1 \in L $ and $ v_l + w_2 \in L $. Then there exist distinct $ 0 < \lambda_1, \lambda_2 < 1 $ such that \begin{align}
		\label{eqn:vprime_l_plus_w1}
		v_l + w_1 &= \lambda_1 \cdot u_i + (1 - \lambda_1) \cdot u_j \\
		\label{eqn:vprime_l_plus_w2}
		v_l + w_2 &= \lambda_2 \cdot u_i + (1 - \lambda_2) \cdot u_j
	\end{align} Subtracting (\ref{eqn:vprime_l_plus_w1}) from (\ref{eqn:vprime_l_plus_w2}) yields \[
	0 \neq w_2 - w_1 = (\lambda_2 - \lambda_1) \cdot u_i + (\lambda_1 - \lambda_2) \cdot u_j = (\lambda_2 - \lambda_1) \cdot (u_i - u_j)
	\] Note that $ u_i \in v_i + W_i $ and $ u_j \in v_j + W_j $. Hence, by Claim~\ref{claim:h_w1_eq_h_w2_eq_dots_eq_h_wk}, \begin{align*}
		h(W_l) &\ni h(w_2 - w_1) \\
		&= h((\lambda_2 - \lambda_1) \cdot (u_i - u_j)) \\
		&= (\lambda_2 - \lambda_1) \cdot (h(u_i) - h(u_j)) \\
		&\in (\lambda_2 - \lambda_1) \cdot \big(h(v_i + W_i) - h(v_j + W_j)\big) \\
		&= (\lambda_2 - \lambda_1) \cdot \big(h(v_i - v_j) + h(W_i) - h(W_j)\big) \\
		&\overset{\ref{claim:h_w1_eq_h_w2_eq_dots_eq_h_wk}}{=} (\lambda_2 - \lambda_1) \cdot \big(h(v_i - v_j) + h(W_l)\big) \\
		&= \big((\lambda_2 - \lambda_1) \cdot h(v_i - v_j)\big) + h(W_l)
	\end{align*} To put it into words, we have shown that $ h(W_l) $ must agree with $ ((\lambda_2 - \lambda_1) \cdot h(v_i - v_j)) + h(W_l) $ on some element. This implies \[
	\underbrace{(\lambda_2 - \lambda_1)}_{\neq 0} \cdot h(v_i - v_j) \in h(W_l)
	\] and hence $ h(v_i - v_j) \in h(W_l) $. In terms of subspaces, we have \begin{align}
		\label{eqn:h_viprime_minus_vjprime_subspace_hWl}
		\gen{0} \neq \gen{h(v_i - v_j)} \le h(W_l)
	\end{align} However, by Claims \ref{claim:wlog_hvprime_i_in_hW_i_orth} and \ref{claim:h_w1_eq_h_w2_eq_dots_eq_h_wk}, \[
	h(v_i - v_j) = h(v_i) - h(v_j) \overset{\ref{claim:wlog_hvprime_i_in_hW_i_orth}}{\in} h(W_i)^\bot - h(W_j)^\bot \overset{\ref{claim:h_w1_eq_h_w2_eq_dots_eq_h_wk}}{=} h(W_l)^\bot
	\] This implies $ \gen{0} \neq \gen{h(v_i - v_j)} \le h(W_l)^\bot $, which is a contradiction to (\ref{eqn:h_viprime_minus_vjprime_subspace_hWl}).
\end{proof}

\noindent Since $ \ModelsOf(\Lambda) $ is a convex set containing $ u_i $ and $ u_j $, $ L \subseteq \ModelsOf(\Lambda) $. Hence, \begin{align}
	\label{eqn:L_subsetof_bounded_set}
	L \subseteq L \cap \ModelsOf(\Lambda) \subseteq L \cap \bigcup_{i=1}^k (v_i + W_i) = \bigcup_{i=1}^k \big((v_i + W_i) \cap L\big)
\end{align} Combining (\ref{eqn:L_subsetof_bounded_set}) with Claim~\ref{claim:vlprime_plus_Wl_int_L_at_most_one_elem} yields a contradiction \[
\infty = |L| \overset{(\ref{eqn:L_subsetof_bounded_set})}{\le} \Bigg|\bigcup_{i=1}^k \big((v_i + W_i) \cap L\big)\Bigg| \overset{\ref{claim:vlprime_plus_Wl_int_L_at_most_one_elem}}{\le} k \in \N
\] This completes the proof of the Uniform union Theorem~\ref{thm:union_of_same_lindep_has_that_lindep}.

\subsection{Proof of Lemma~\ref{lemma:adding_neq_predicates_cannot_make_pi_complex_set_pi_simple}}
\label{sec:app:lindep_facts:proof_adding_neq_predicates_cannot_make_pi_complex_set_pi_simple}

It suffices to prove \begin{align}
	\label{eqn:fixed_gamma_eq_fixed_gamma_union_p}
	\Fixes(\Gamma) = \Fixes(\Gamma \cup \{p\})
\end{align} because then $ \Gamma $ being $ \Pi $-complex implies that $ \Gamma \cup \{p\} $ is also $ \Pi $-complex. Let $ z_i $ denote the variable corresponding to the $ i $-th component in $ \Q^n $ and let $ e_i $ be the $ i $-th canonical basis vector. Fix some $ v \models \Gamma \cup \{p\} $. Since \begin{align*}
	\gen{e_i} \in \LinDep_{\id}(\Gamma) &\Leftrightarrow \ModelsOf(\Gamma) - v \subseteq \gen{e_i}^\bot \\
	&\Leftrightarrow \forall u \models \Gamma : (u - v) \in \gen{e_i}^\bot \\
	&\Leftrightarrow \forall u \models \Gamma : (u - v) \cdot e_i = 0 \\
	&\Leftrightarrow \forall u \models \Gamma : u \cdot e_i = v \cdot e_i \\
	&\Leftrightarrow z_i \in \Fixes(\Gamma)
\end{align*} and, analogously, $ \gen{e_i} \in \LinDep_{\id}(\Gamma \cup \{p\}) $ is equivalent to $ z_i \in \Fixes(\Gamma \cup \{p\}) $, in order to show (\ref{eqn:fixed_gamma_eq_fixed_gamma_union_p}), it suffices to prove $ \LinDep_{\id}(\Gamma) = \LinDep_{\id}(\Gamma \cup \{p\}) $. Indeed, this follows from Theorem~\ref{thm:only_equality_predicates_can_establish_lindep} because $ \Gamma^= = (\Gamma \cup \{p\})^= $. This completes the proof of Lemma~\ref{lemma:adding_neq_predicates_cannot_make_pi_complex_set_pi_simple}.

\section{Analysis of the covering algorithm}
\label{sec:app:cover_analysis}

\subsection{Recursion depth and termination}

We start by showing an upper bound for the recursion depth, which, as already discussed in Section~\ref{sec:vardec:cover}, implies that the covering algorithm terminates and runs in double-exponential time.

\begin{lemma}[Recursion depth and termination]
	\label{lemma:cover_equifix_rec_depth_bounded_by_n}
	The maximum recursion depth of the covering algorithm is bounded by the number of free variables appearing in $ \varphi $. In particular, it follows that the covering algorithm terminates.
\end{lemma}
\begin{proof}
	Consider the situation when during the computation of $ \cover(\Pi, \Gamma) $, the algorithm makes a recursive call $ \cover(\Pi, \Gamma') $ at Line~\ref{alg:cover:line:second_loop:rec_call} with \[
	\Gamma' := \Gamma \cup \{(\pi_Z(\vec{z}) \cdot w = \pi_Z(b) \cdot w\}
	\] If $ \Gamma' $ is unsatisfiable, then $ \Gamma' $ is $ \Pi $-simple and hence the recursive call would immediately return $ \bot $. Thus, for the analysis of the recursion depth, we assume that $ \Gamma' $ is satisfiable and fix some model $ v \models \Gamma' $. Observe that \[
	\gen{w} \in \LinDep_{\pi_Z, b}(\Gamma')
	\] holds because $ \Gamma' \models \pi_Z(\vec{z} - b) \cdot w = 0 $. Since $ \ImageOf(\pi_Z \circ \lc_A) \cap \ImageOf(\pi_Z \circ \lc_A)^\bot = \gen{0} $ and $ \gen{0} \neq \gen{w} \le \ImageOf(\pi_Z \circ \lc_A) $, it follows that $ \gen{w} \not\le \ImageOf(\pi_Z \circ \lc_A)^\bot $. By Theorem~\ref{thm:lindep_algebraic_char}, this is equivalent to $ \gen{w} \notin \LinDep_{\pi_Z, a}(\Gamma) $. Next, observe that $ \Gamma' \models \Gamma $ implies \begin{align}
		\label{eqn:lindep_gamma_v_subset_lindep_gammaprime_v}
		\LinDep_{\pi_Z, v}(\Gamma) \subseteq \LinDep_{\pi_Z, v}(\Gamma')
	\end{align} We now apply Lemmas \ref{lemma:lindep_gamma_v1_equals_lindep_gamma_v2} and \ref{lemma:arbitrary_offset_less_lindep_compared_to_good_offset} to establish \begin{center}
		\begin{tabular}{c c c c c}
			$ \LinDep_{\pi_Z, a}(\Gamma) $ & $ \overset{\ref{lemma:lindep_gamma_v1_equals_lindep_gamma_v2}}{=} \LinDep_{\pi_Z, v}(\Gamma) \overset{(\ref{eqn:lindep_gamma_v_subset_lindep_gammaprime_v})}{\subseteq} $ & $ \LinDep_{\pi_Z, v}(\Gamma') $ & $ \overset{\ref{lemma:arbitrary_offset_less_lindep_compared_to_good_offset}}{\supseteq} $ & $ \LinDep_{\pi_Z, b}(\Gamma') $ \\
			\rotatebox[origin=c]{90}{$ \notin $} && \rotatebox[origin=c]{90}{$ \in $} && \rotatebox[origin=c]{90}{$ \in $} \\
			$ \gen{w} $ && $ \gen{w} $ && $ \gen{w} $
		\end{tabular}
	\end{center} Hence, $ \LinDep_{\pi_Z, v}(\Gamma) \subsetneq \LinDep_{\pi_Z, v}(\Gamma') $ because these sets disagree on $ \gen{w} $. By Theorem~\ref{thm:lindep_algebraic_char}, there exist unique subspaces $ V_Z, V'_Z $ of $ \Q^{\left|Z\right|} $ such that \begin{center}
		\begin{tabular}{c c c}
			$ \LinDep_{\pi_Z, v}(\Gamma) $ & $ \subsetneq $ & $ \LinDep_{\pi_Z, v}(\Gamma') $ \\
			\rotatebox[origin=c]{90}{$ \overset{\ref{thm:lindep_algebraic_char}}{=} $} && \rotatebox[origin=c]{90}{$ \overset{\ref{thm:lindep_algebraic_char}}{=} $} \\
			$ \{U \mid U \le V_Z\} $ & $ \subsetneq $ & $ \{U' \mid U' \le V'_Z\} $
		\end{tabular}
	\end{center} Hence, $ V_Z \lneq V'_Z $ and consequently $ \dim{V_Z} < \dim{V'_Z} $. This shows that for every recursive call capable of causing further nested recursion, the maximum dimension of a $ Z $-dependency of $ \Gamma $, namely $ \dim{V_Z} $, is guaranteed to increase for at least one $ Z \in \Pi $ and become $ \dim{V'_Z} > \dim{V_Z} $. Note that this argument implicitly relies on the fact that, by Lemma~\ref{lemma:lindep_gamma_v1_equals_lindep_gamma_v2}, the set of linear dependencies remains the same regardless of the offset $ v \models \Gamma' $ we pick. Clearly, we have an upper bound $ \dim{V'_Z} \le |Z| $ simply because $ V'_Z \le \Q^{\left|Z\right|} $. Hence, after at most $ n = \left|X\right| + \left|Y\right| $ nested recursive calls, we are guaranteed to arrive at a $ \Gamma' $ such that the corresponding space $ V'_Z $ has dimension $ \dim{V'_Z} = |Z| $ for every $ Z \in \Pi $ unless at some point $ \Gamma' $ becomes $ \Pi $-simple whereupon no further nested recursive calls are possible. Suppose we have reached a satisfiable $ \Gamma' $, such that $ V'_Z = \Q^{\left|Z\right|} $ holds for all $ Z \in \Pi $. By definition of $ \LinDep_{\pi_Z, v}(\Gamma') \ni V'_Z $, it follows that for all $ Z \in \Pi $, \[
	\pi_Z(\ModelsOf(\Gamma') - v) \subseteq (V'_Z)^\bot = \big(\Q^{\left|Z\right|}\big)^\bot = \gen{0}
	\] Hence, $ \Gamma' $ fixes every $ Z \in \Pi $, so we conclude that $ \Gamma' $ is $ \Pi $-simple and thus no recursive calls at levels exceeding $ n $ are possible.
\end{proof}

\subsection{Invariants of loops}

Next, we prove two important invariants the covering algorithm maintains.

\begin{lemma}[Invariant of the first loop]
	\label{lemma:gamma_entails_theta}
	Let $ \Pi $ be a partition, $ \Gamma $ be a $ \Pi $-complex predicate set and $ \psi := \cover(\Pi, \Gamma) $. Then $ \Gamma \models \Theta $ is an invariant of the \texttt{foreach} loop at Line~\ref{alg:cover:line:first_loop}.
\end{lemma}
\begin{proof}
	Let $ m $ be the number of iterations the \texttt{foreach} loop at Line~\ref{alg:cover:line:first_loop} makes and let $ \Theta_i $ denote the state of $ \Theta $ after the $ i $-th iteration of the loop, for $ 0 \le i \le m $. We show $ \Gamma \models \Theta_i $ by induction on $ i $.
	
	\textbf{Base case}: For $ i = 0 $, $ \Gamma \models \Theta_0 $ follows from $ \Theta_0 = \PiSimp{\Gamma} \subseteq \Gamma $.
	
	\textbf{Inductive step}: Let $ i \le m $ be fixed. By the induction hypothesis, it suffices to show that for all $ j \in \{1, \dots, s\} $, \begin{align}
		\label{eqn:gamma_entails_pix_minus_pixa_dot_w_eq_zero}
		\Gamma \models \pi_Z(\vec{z} - a) \cdot v_j = 0
	\end{align} By definition, $ \gen{v_1, \dots, v_s} = \ImageOf(\pi_Z \circ \lc_A)^\bot $, so by Theorem~\ref{thm:lindep_algebraic_char} we have $ \gen{v_1, \dots, v_s} \in \LinDep_{\pi_Z, a}(\Gamma) $. In particular, $ \gen{v_j} \in \LinDep_{\pi_Z, a}(\Gamma) $ holds for every $ 1 \le j \le s $ because $ \LinDep_{\pi_Z, a}(\Gamma) $ is closed under taking subspaces (Lemma~\ref{lemma:lindep_closed_under_addition}). Hence, $ \pi_Z(\ModelsOf(\Gamma) - a) \subseteq \gen{v_j}^\bot $ holds for every $ j \in \{1 \dots, s\} $. This immediately implies (\ref{eqn:gamma_entails_pix_minus_pixa_dot_w_eq_zero}).
\end{proof}

\begin{lemma}[Invariant of the second loop]
	\label{lemma:gamma_entails_psi}
	Let $ \Pi $ be a partition, $ \Gamma $ be a $ \Pi $-complex predicate set and $ \psi := \cover(\Pi, \Gamma) $. Then $ \Gamma \models \Delta \vee \Theta \wedge \Upsilon $ is an invariant of the \texttt{foreach} loop at Line~\ref{alg:cover:line:second_loop}. In particular, it follows that $ \Gamma \models \psi $.
\end{lemma}
\begin{proof}
	By Lemma~\ref{lemma:cover_equifix_rec_depth_bounded_by_n}, the covering algorithm terminates and thus makes only finitely many recursive calls. Because of this, it can be assumed that Lemma~\ref{lemma:gamma_entails_psi} holds for any recursive call to $ \cover $.
	
	Let $ m $ be the number of iterations the \texttt{foreach} loop at Line~\ref{alg:cover:line:second_loop} makes and let $ \Delta_i $ and $ \Upsilon_i $ denote the values of the corresponding variables after the $ i $-th iteration of the second loop, for some $ i \in \{0, \dots, m\} $. We prove by induction on $ i $ that $ \Gamma \models \Delta_i \vee \Theta \wedge \Upsilon_i $ holds for all $ i \in \{0, \dots, m\} $.
	
	\textbf{Base case}: If $ i = 0 $, then $ \Gamma \models \Theta \equiv (\Delta_0 \vee \Theta) \wedge \Upsilon_0 $ by Lemma~\ref{lemma:gamma_entails_theta}.
	
	\textbf{Inductive step}: Let $ 0 < i \le m $ be fixed. By the induction hypothesis, we have $ \Gamma \models \Delta_{i-1} \vee \Theta \wedge \Upsilon_{i-1} $. If the condition of the \texttt{if} statement at Line~\ref{alg:cover:line:second_loop:main_if} does not hold, then neither $ \Delta_{i-1} $ nor $ \Upsilon_{i-1} $ get modified, so the invariant is trivially maintained. Otherwise, let \[
		\psi := \cover(\Pi, \Gamma \cup \{\pi_Z(\vec{z}) \cdot w = \pi_Z(b) \cdot w\})
	\] be the result of the recursive call. It suffices to show \begin{align}
		\label{eqn:gamma_entails_delta_psi_theta_neqzero}
		\Gamma \models \underbrace{\Delta_{i-1} \vee \psi \vee \Theta \wedge \Upsilon_{i-1} \wedge \pi_Z(\vec{z}) \cdot w \neq \pi_Z(b) \cdot w}_{\equiv \Delta_i \vee \Theta \wedge \Upsilon_i}
	\end{align} We prove (\ref{eqn:gamma_entails_delta_psi_theta_neqzero}) by showing that an arbitrary $ v \models \Gamma $ must be a model of the right-hand side. We distinguish between the following cases.
	
	\textbf{Case 1.} Suppose $ v \models \pi_Z(\vec{z}) \cdot w = \pi_Z(b) \cdot w $. Then \[
		v \models \Gamma \cup \{\pi_Z(\vec{z}) \cdot w = \pi_Z(b) \cdot w\} \models \psi
	\] by the assumption that Lemma~\ref{lemma:gamma_entails_psi} holds for any recursive call.
	
	\textbf{Case 2.} Suppose $ v \models \pi_Z(\vec{z}) \cdot w \neq \pi_Z(b) \cdot w $. Then \begin{align*}
		v &\models (\Delta_{i-1} \vee \Theta \wedge \Upsilon_{i-1}) \wedge \pi_Z(\vec{z}) \cdot w \neq \pi_Z(b) \cdot w \\
		&\models \Delta_{i-1} \vee \Theta \wedge \Upsilon_{i-1} \wedge \pi_Z(\vec{z}) \cdot w \neq \pi_Z(b) \cdot w
	\end{align*} holds because $ v \models \Gamma \models \Delta_{i-1} \vee \Theta \wedge \Upsilon_{i-1} $.
\end{proof}

\subsection{Properties of the produced covering}

Intuitively, we now show that the first loop of the covering algorithm at Line~\ref{alg:cover:line:first_loop} not only enforces all $ Z $-dependencies of $ \Gamma $ via predicates it adds to $ \Theta $, but also does not enforce any $ Z $-dependency absent in $ \Gamma $. In other words, $ \Theta $ has precisely the same $ Z $-dependencies as $ \Gamma $, for all $ Z \in \Pi $.

\begin{lemma}[$ \Theta $ has same $ Z $-dependencies as $ \Gamma $]
	\label{lemma:lindep_theta_eq_lindep_gamma}
	Let $ \varphi \in \QFLRA $ be a formula, $ \Pi $ be a partition, $ \Gamma \in \Sat(\DisjTrue) $ be a $ \Pi $-complex predicate set and $ \Theta $ be the corresponding set computed in $ \cover(\Pi, \Gamma) $. Then \[
		\LinDep_{\pi_Z}(\Gamma) = \LinDep_{\pi_Z}(\Theta)
	\] holds for all $ Z \in \Pi $.
\end{lemma}
\begin{proof}
	By Lemma~\ref{lemma:gamma_entails_theta}, $ \Gamma \models \Theta $, so $ \LinDep_{\pi_Z}(\Theta) \subseteq \LinDep_{\pi_Z}(\Gamma) $. For the other inclusion, note that by construction, \[
		\Theta \models \pi_Z(\vec{z} - a) \cdot v_i = 0
	\] holds for every $ i \in \{1, \dots, s\} $, where $ \gen{v_1, \dots, v_s} = \ImageOf(\pi_Z \circ \lc_A)^{\bot} $. Hence, $ \pi_Z(\ModelsOf(\Theta) - a) \subseteq \gen{v_i}^\bot $ and thus $ \gen{v_i} \in \LinDep_{\pi_Z, a}(\Theta) $ is true for all $ i $. Since $ \LinDep_{\pi_Z, a}(\Theta) $ is closed under addition (Lemma~\ref{lemma:lindep_closed_under_addition}), \[
	\ImageOf(\pi_Z \circ \lc_A)^{\bot} = \gen{v_1, \dots, v_s} = \sum_{i=1}^s \gen{v_i} \in \LinDep_{\pi_Z, a}(\Theta)
	\] Hence, by Theorem~\ref{thm:lindep_algebraic_char} and the fact that $ \LinDep_{\pi_Z, a}(\Theta) $ is closed under taking subspaces (Lemma~\ref{lemma:lindep_closed_under_addition}), \begin{align*}
		\LinDep_{\pi_Z}(\Gamma) &\overset{\ref{thm:lindep_algebraic_char}}{=} \{U \le \Q^{|Z|} \mid U \le \ImageOf(\pi_Z \circ \lc_A)^{\bot}\} \\
		&= \{U \le \gen{v_1, \dots, v_s}\} \\
		&\overset{\ref{lemma:lindep_closed_under_addition}}{\subseteq} \LinDep_{\pi_Z, a}(\Theta)
	\end{align*} Finally, note that \[
		\LinDep_{\pi_Z, a}(\Theta) \subseteq \LinDep_{\pi_Z}(\Theta)
	\] holds by Lemma~\ref{lemma:arbitrary_offset_less_lindep_compared_to_good_offset}, so we conclude that $ \LinDep_{\pi_Z}(\Gamma) \subseteq \LinDep_{\pi_Z}(\Theta) $.
\end{proof}

We now rigorously prove Lemma~\ref{lemma:psi_properties} \ref{lemma:psi_properties:b}. For the intuition behind the proof, see the corresponding discussion in Sections \ref{sec:vardec:properties_of_produced_covering} and \ref{sec:vardec:lemma_psi_properties_proof} above.

\begin{lemma}[Lemma~\ref{lemma:psi_properties} \ref{lemma:psi_properties:b}]
	\label{lemma:disjuncts_touching_lambda_have_same_z_dependencies}
	Let $ \varphi \in \QFLRA $ be a formula, $ \Pi $ be a partition, $ \Gamma \in \Sat(\DisjTrue) $ be a $ \Pi $-complex predicate set, $ \psi := \cover(\Pi, \Gamma) $ and $ \Psi $ be the set of predicate sets corresponding to the DNF terms of $ \psi $. Then for any $ \Lambda \in \Psi $ there exists a set $ \Theta $ computed in some (possibly nested) call to $ \cover $ such that for all $ \Gamma_1, \Gamma_2 \in \DisjOf{\Theta} $, the satisfiability of $ \Gamma_1 \wedge \Lambda $ and $ \Gamma_2 \wedge \Lambda $ implies \[
	\LinDep_{\pi_Z}(\Gamma_1) = \LinDep_{\pi_Z}(\Gamma_2)
	\] where $ Z \in \Pi $ is arbitrary.
\end{lemma}
\begin{proof}
	Let $ \Lambda \in \Psi $ be fixed and $ \Theta $ be the corresponding predicate set computed in the call to $ \cover $ that produced $ \Lambda $. By structural induction, we can assume that Lemma~\ref{lemma:disjuncts_touching_lambda_have_same_z_dependencies} holds if $ \Lambda $ comes from $ \Delta $, i.e., from the covering produced by any recursive call. Because of this, it suffices to show that for any $ \Omega \in \DisjOf{\Theta} $, \[
		\LinDep_{\pi_Z}(\Gamma) = \LinDep_{\pi_Z}(\Omega)
	\] holds under the assumption that $ \Omega \wedge \Lambda $ is satisfiable. We start by proving the ``$ \subseteq $'' inclusion. Since $ \Omega \models \Theta $, it follows that $ \LinDep_{\pi_Z}(\Theta) \subseteq \LinDep_{\pi_Z}(\Omega) $. Combining this with Lemma~\ref{lemma:lindep_theta_eq_lindep_gamma} yields \[
		\LinDep_{\pi_Z}(\Gamma) \overset{\ref{lemma:lindep_theta_eq_lindep_gamma}}{=} \LinDep_{\pi_Z}(\Theta) \subseteq \LinDep_{\pi_Z}(\Omega)
	\] as desired. To prove the remaining ``$ \supseteq $'' inclusion, we assume that \begin{align}
		\label{eqn:lindep_gamma_strict_subset_lindep_gammaprime}
		\LinDep_{\pi_Z}(\Gamma) \subsetneq \LinDep_{\pi_Z}(\Omega)
	\end{align} and derive a contradiction. Let $ b \in \Q^n $ be a vector and $ B := (b_1, \dots, b_r) $ be a basis such that $ \ModelsOf(\Omega^=) = b + \gen{b_1, \dots, b_r} $. Applying Theorem~\ref{thm:lindep_algebraic_char} to (\ref{eqn:lindep_gamma_strict_subset_lindep_gammaprime}) yields: \begin{center}
		\begin{tabular}{c c c}
			$ \LinDep_{\pi_Z}(\Gamma) $ & $ \overset{(\ref{eqn:lindep_gamma_strict_subset_lindep_gammaprime})}{\subsetneq} $ & $ \LinDep_{\pi_Z}(\Omega) $ \\
			\rotatebox[origin=c]{90}{$ \overset{\ref{thm:lindep_algebraic_char}}{=} $} && \rotatebox[origin=c]{90}{$ \overset{\ref{thm:lindep_algebraic_char}}{=} $} \\
			$ \{U \le \Q^{|Z|} \mid U \le \ImageOf(\pi_Z \circ \lc_A)^\bot\} $ & $ \subsetneq $ & $ \{U \le \Q^{|Z|} \mid U \le \ImageOf(\pi_Z \circ \lc_B)^\bot\} $
		\end{tabular}
	\end{center} Hence, $ \ImageOf(\pi_Z \circ \lc_A)^\bot \lneq \ImageOf(\pi_Z \circ \lc_B)^\bot $ and thus $ \ImageOf(\pi_Z \circ \lc_B) \lneq \ImageOf(\pi_Z \circ \lc_A) $. This implies \[
	\ImageOf(\pi_Z \circ \lc_B)^\bot \cap \ImageOf(\pi_Z \circ \lc_A) \neq \gen{0}
	\] It follows that the covering algorithm must have found $ \Omega \in \DisjOf{\Theta} $, computed a vector \[
	0 \neq w \in \ImageOf(\pi_Z \circ \lc_B)^\bot \cap \ImageOf(\pi_Z \circ \lc_A) \neq \gen{0}
	\] and added $ \pi_Z(\vec{z}) \cdot w \neq \pi_Z(b) \cdot w $ to $ \Upsilon $. In particular, Theorem~\ref{thm:lindep_algebraic_char} applied to $ \gen{w} \le \ImageOf(\pi_Z \circ \lc_B)^\bot $ yields that $ \gen{w} \in \LinDep_{\pi_Z, b}(\Omega) $ and thus \[
		\pi_Z(\ModelsOf(\Omega) - b) \subseteq \gen{w}^\bot
	\] This is equivalent to saying that \[
		\Omega \models \pi_Z(\vec{z} - b) \cdot w = 0
	\] Hence, $ \Omega \models \pi_Z(\vec{z}) \cdot w = \pi_Z(b) \cdot w $ and thus \begin{align}
		\label{eqn:omegalambda_entail_omega_entail_pizxw_eq_pizbw}
		\Omega \wedge \Lambda \models \Omega \models \pi_Z(\vec{z}) \cdot w = \pi_Z(b) \cdot w
	\end{align} On the other hand, by construction of the final covering at Line~\ref{alg:cover:line:return}, we have \begin{align}
		\label{eqn:omegalambda_entails_lambda_entails_upsilon_entails_pizxw_neq_pizbw}
		\Omega \wedge \Lambda \models \Lambda \models \Upsilon \models \pi_Z(\vec{z}) \cdot w \neq \pi_Z(b) \cdot w
	\end{align} Since $ \Omega \wedge \Lambda $ is satisfiable, (\ref{eqn:omegalambda_entails_lambda_entails_upsilon_entails_pizxw_neq_pizbw}) contradicts (\ref{eqn:omegalambda_entail_omega_entail_pizxw_eq_pizbw}).
\end{proof}

\section{Proof of the Overspilling Theorem~\ref{thm:overspilling}}
\label{sec:app:overspilling_proof}

\subsection{Intuition}
\label{sec:overspilling:intuition}

Before we give the rigorous proof of the Overspilling Theorem~\ref{thm:overspilling} in Section~\ref{sec:overspilling:rigorous_proof} below, we explain the main ideas by applying our argumentation to a concrete example. That is, we first give a very informal and incomplete but intuitive ``proof by example'' intending to help the reader understand important intuitions behind our argumentation in the actual rigorous proof.

\subsubsection{Preliminaries}

Let $ \Pi = \{X, Y\} $ be the partition where $ X = \{x_1, x_2, x_3\} $ and $ Y = \{y_1, y_2, y_3\} $. We fix the variable ordering $ x_1, x_2, x_3, y_1, y_2, y_3 $ -- until the end of Section~\ref{sec:overspilling:intuition} the components of all vectors correspond to values of these variables in this order. Let furthermore $ W_1 $ and $ W_2 $ be vectorspaces of solutions to $ \Theta_1^= $ and to $ \Theta_2^= $, respectively (at this point, for the sake of simplicity, we assume that the systems of equations are homogeneous, but it should be noted that in general, we would have to handle affine offsets).

\subsubsection{Bird's eye view of the main steps}

We now explain intuitions behind the main part of the proof by looking at it from a big picture perspective. The central steps in our argumentation are as follows.
\begin{enumerate}[label=(\alph*)]
	\item \label{step:overspilling:1} We prove that there exist vectors $ w_1 \in W_1 $ and $ w_2 \in W_2 \setminus W_1 $ such that for any vectorspace $ W $ satisfying \begin{align}
		\label{eqn:overspilling_step_1_property}
		\gen{w_1} \le W = \pi_X^{-1}(\pi_X(W)) \cap \pi_Y^{-1}(\pi_Y(W))
	\end{align} it holds that $ w_2 \in W $. Intuitively, vectorspaces obtained as a result of an analysis of $ \Pi $-respecting formulas satisfy (\ref{eqn:overspilling_step_1_property}). The existence of vectors $ w_1 $ and $ w_2 $ as specified above can thus be thought of as the ability to construct vectors in different vectorspaces, but so that these vectors are, in a certain sense, indistinguishable in the language of $ \Pi $-decompositions.
	\item \label{step:overspilling:2} We use $ w_1 $ to construct an infinite set $ I \subseteq \ModelsOf(\Theta_1) $ such that $ I \subseteq \gen{w_1} $. Intuitively, we need this to ensure that $ \gen{w_1} \le W $ holds in (\ref{eqn:overspilling_step_1_property}), which is needed to be able to apply step \ref{step:overspilling:1}.
	\item \label{step:overspilling:3} We then argue that for any set $ \Omega $ of $ \Pi $-respecting predicates, if $ w + W $ is the affine vectorspace of solutions to $ \Omega^= $, then \[
		W = \pi_X^{-1}(\pi_X(W)) \cap \pi_Y^{-1}(\pi_Y(W))
	\] Similarly, this is required to establish (\ref{eqn:overspilling_step_1_property}) and thus to be able to apply step \ref{step:overspilling:1}.
	\item \label{step:overspilling:4} We assume that $ \varphi $ is $ \Pi $-decomposable and consider a $ \Pi $-decomposition brought to DNF. Since $ \left|I\right| = \infty $, by the pigeonhole principle, it must be the case that some DNF term agrees with $ I $ on infinitely many models (points).
	\item \label{step:overspilling:5} Without loss of generality, we assume this DNF term to be a predicate set $ \Omega $ and let $ w + W $ be the affine vectorspace of solutions to $ \Omega^= $. Hence, steps \ref{step:overspilling:2} and \ref{step:overspilling:3} taken together with \ref{step:overspilling:1} yield that $ \gen{w_2} \le W $.
	\item \label{step:overspilling:6} By applying scaling to $ w_2 $ and handling affine offsets appropriately, we show that $ \ModelsOf(\Omega) - w $ must agree with $ W_2 $ on some point. We actually prove something slightly stronger, namely that $ \ModelsOf(\Omega) - w $ must agree with $ \ModelsOf(\Theta_2) - w $ on some point. Hence, it follows that $ \Omega \cup \Theta_2 $ is satisfiable.
	\item \label{step:overspilling:7} Since $ \Omega \models \varphi $, we conclude that $ \varphi \wedge \Theta_2 $ is satisfiable.
\end{enumerate}

\subsubsection{Discussion}

Essentially, steps \ref{step:overspilling:2} and \ref{step:overspilling:3} are needed in order to be able to apply step \ref{step:overspilling:1}, while steps \ref{step:overspilling:4}, \ref{step:overspilling:5}, \ref{step:overspilling:6} and \ref{step:overspilling:7} reduce the proof to showing a purely algebraic property. Although this reduction is crucial for the proof to work and involves many important intermediate steps, it does not rely on any fundamental ideas requiring a separate explanation. Thus, in the remainder of Section~\ref{sec:overspilling:intuition}, we only discuss step \ref{step:overspilling:1} in more detail because this is where the core ``overspilling'' effect happens.

\noindent Let \begin{align}
	\label{eqn:overspilling_m1_mat_def}
	M_1 &:= \begin{pNiceMatrix}[first-col, extra-margin=2pt, code-for-first-col=\scriptscriptstyle, columns-width=2em, vlines]
		x_1 & -1 & 2 \\
		x_2 & 2 & 2 \\
		x_3 & 1 & -2 \\
		y_1 & 0 & 2 \\
		y_2 & 1 & 0 \\
		y_3 & 1 & 1
		\CodeAfter
		\UnderBrace[yshift=3pt]{1-1}{last-1}{\scriptstyle u_1}
		\UnderBrace[yshift=3pt]{1-2}{last-2}{\scriptstyle u_2}
		\tikz \node [highlight = (1-1) (3-1)] {};
		\tikz \node [highlightblue = (4-2) (6-2)] {};
	\end{pNiceMatrix} \\[15pt]
	\label{eqn:overspilling_m2_mat_def}
	M_2 &:= \begin{pNiceMatrix}[first-col, extra-margin=2pt, code-for-first-col=\scriptscriptstyle, columns-width=2em, vlines]
		x_1 & -1 & 2 & -1 \\
		x_2 & 2 & 2 & 2 \\
		x_3 & 1 & -2 & 1 \\
		y_1 & 0 & 2 & 2\\
		y_2 & 1 & 0 & 0\\
		y_3 & 1 & 1 & 1
		\CodeAfter
		\UnderBrace[yshift=3pt]{1-1}{last-1}{\scriptstyle u_1}
		\UnderBrace[yshift=3pt]{1-2}{last-2}{\scriptstyle u_2}
		\UnderBrace[yshift=3pt]{1-3}{last-3}{\scriptstyle u}
		\tikz \node [highlight = (1-1) (3-1)] {};
		\tikz \node [highlight = (1-3) (3-3)] {};
		\tikz \node [highlightblue = (4-2) (6-2)] {};
		\tikz \node [highlightblue = (4-3) (6-3)] {};
	\end{pNiceMatrix}
\end{align}\[ \\[12pt] \] and suppose that \begin{align*}
	W_1 &= \ImageOf(M_1) \\
	W_2 &= \ImageOf(M_2)
\end{align*} In (\ref{eqn:overspilling_m1_mat_def}) and (\ref{eqn:overspilling_m2_mat_def}), every row is labeled with the corresponding variable located to the left of that row.

Since $ \Theta_1 $ is $ p $-next to $ \Theta_2 $ on an equality predicate $ p \in \Theta_1^= $, it follows that $ W_1 \le W_2 $. This is the reason why we defined the matrices so that $ M_1 $ agrees on the first two columns with $ M_2 $. Let $ u_1 $, $ u_2 $ and $ u $ be the columns of $ M_2 $ as indicated in (\ref{eqn:overspilling_m1_mat_def}) and (\ref{eqn:overspilling_m2_mat_def}). It can be further observed that essentially due to $ \Theta_1 \not\equiv \Theta_2 $, there must exist a vector $ u \in W_2 \setminus W_1 $. In other words, $ M_2 $ must have a column that is not expressible as a linear combination of the columns of $ M_1 $. Indeed, in the running example, we immediately see that the rank of $ M_1 $ (resp. $ M_2 $) is $ 2 $ (resp. $ 3 $). With some algebraic reasoning (see the proof of Claim~\ref{claim:pix_w1_eq_pix_w2} below) it is possible to show that $ \pi_X(W_1) = \pi_X(W_2) $ and $ \pi_Y(W_1) = \pi_Y(W_2) $. This is equivalent to saying that the image of the first three rows of $ M_1 $ is equal to the image of the first three rows of $ M_2 $. The same is true for last three rows of $ M_1 $ and $ M_2 $. More precisely, it holds that \begin{align}
	\label{eqn:image_m1x_eq_image_m2x}
	\ImageOf(M_{1, X}) &= \ImageOf(M_{2, X}) \\
	\label{eqn:image_m1y_eq_image_m2y}
	\ImageOf(M_{1, Y}) &= \ImageOf(M_{2, Y})
\end{align} where \begin{center}
	\begin{tabular}{|P{0.45\textwidth}|P{0.51\textwidth}|}
		\hline
		\begin{tabular}{c}
			\vspace{3pt} \\
			$ M_{1, X} := \begin{pNiceMatrix}[first-col, extra-margin=2pt, code-for-first-col=\scriptscriptstyle, columns-width=2.5em, vlines]
				x_1 & -1 & 2 \\
				x_2 & 2 & 2 \\
				x_3 & 1 & -2
				\CodeAfter
				\UnderBrace[yshift=3pt]{1-1}{last-1}{\scriptstyle \pi_X(u_1)}
				\UnderBrace[yshift=3pt]{1-2}{last-2}{\scriptstyle \pi_X(u_2)}
				\tikz \node [highlight = (1-1) (3-1)] {};
			\end{pNiceMatrix} $ \\
			\vspace{15pt}
		\end{tabular} & 
		\begin{tabular}{c}
			\vspace{3pt} \\
			$ M_{2, X} := \begin{pNiceMatrix}[first-col, extra-margin=2pt, code-for-first-col=\scriptscriptstyle, columns-width=2.5em, vlines]
				x_1 & -1 & 2 & -1 \\
				x_2 & 2 & 2 & 2 \\
				x_3 & 1 & -2 & 1
				\CodeAfter
				\UnderBrace[yshift=3pt]{1-1}{last-1}{\scriptstyle \pi_X(u_1)}
				\UnderBrace[yshift=3pt]{1-2}{last-2}{\scriptstyle \pi_X(u_2)}
				\UnderBrace[yshift=3pt]{1-3}{last-3}{\scriptstyle \pi_X(u)}
				\tikz \node [highlight = (1-1) (3-1)] {};
				\tikz \node [highlight = (1-3) (3-3)] {};
			\end{pNiceMatrix} $ \\
			\vspace{15pt}
		\end{tabular} \\ \hline
		\begin{tabular}{c}
			\vspace{3pt} \\
			$ M_{1,Y} := \begin{pNiceMatrix}[first-col, extra-margin=2pt, code-for-first-col=\scriptscriptstyle, columns-width=2.5em, vlines]
				y_1 & 0 & 2 \\
				y_2 & 1 & 0 \\
				y_3 & 1 & 1
				\CodeAfter
				\UnderBrace[yshift=3pt]{1-1}{last-1}{\scriptstyle \pi_Y(u_1)}
				\UnderBrace[yshift=3pt]{1-2}{last-2}{\scriptstyle \pi_Y(u_2)}
				\tikz \node [highlightblue = (1-2) (3-2)] {};
			\end{pNiceMatrix} $ \\
			\vspace{15pt}
		\end{tabular} & \begin{tabular}{c}
			\vspace{3pt} \\
			$ M_{2,Y} := \begin{pNiceMatrix}[first-col, extra-margin=2pt, code-for-first-col=\scriptscriptstyle, columns-width=2.5em, vlines]
				y_1 & 0 & 2 & 2\\
				y_2 & 1 & 0 & 0\\
				y_3 & 1 & 1 & 1
				\CodeAfter
				\UnderBrace[yshift=3pt]{1-1}{last-1}{\scriptstyle \pi_Y(u_1)}
				\UnderBrace[yshift=3pt]{1-2}{last-2}{\scriptstyle \pi_Y(u_2)}
				\UnderBrace[yshift=3pt]{1-3}{last-3}{\scriptstyle \pi_Y(u)}
				\tikz \node [highlightblue = (1-2) (3-2)] {};
				\tikz \node [highlightblue = (1-3) (3-3)] {};
			\end{pNiceMatrix} $ \\
			\vspace{15pt}
		\end{tabular} \\ \hline
	\end{tabular}
\end{center} In the present example, note that (\ref{eqn:image_m1x_eq_image_m2x}) and (\ref{eqn:image_m1y_eq_image_m2y}) hold trivially because the matrices $ M_{2, X} $ and $ M_{2,Y} $ contain redundant columns $ \pi_X(u) $ and $ \pi_Y(u) $, respectively (see the highlighted columns of the matrices). Moreover, we can trace back these redundant columns in the original matrices -- see the highlighted parts in (\ref{eqn:overspilling_m1_mat_def}) and (\ref{eqn:overspilling_m2_mat_def}). Note that, in general, we can only say that there exists a linearly dependent column -- it need not be equal to any of the other columns like in the present example.

Having explained the properties of the solution spaces we are dealing with, we now discuss the crux of the construction of $ w_1 \in W_1 $ and $ w_2 \in W_2 \setminus W_1 $ for step \ref{step:overspilling:1}. First, note that the statement of step \ref{step:overspilling:1} is very strong in the sense that the vectorspace $ W $ is arbitrary; the only property we know about $ W $ is formulated in (\ref{eqn:overspilling_step_1_property}). This condition immediately implies \begin{align}
	\label{eqn:pix_w1_in_pix_w_overspill}
	\pi_X(w_1) &\in \pi_X(W) \\
	\label{eqn:piy_w1_in_piy_w_overspill}
	\pi_Y(w_1) &\in \pi_X(W)
\end{align} Since we are free to define $ w_1 \in W_1 $ as we wish, (\ref{eqn:pix_w1_in_pix_w_overspill}) and (\ref{eqn:piy_w1_in_piy_w_overspill}) can be thought of as giving us the power to enforce $ \gen{\pi_X(w_1)} \le \pi_X(W) $ and $ \gen{\pi_Y(w_1)} \le \pi_Y(W) $. Thus, in order to prove step \ref{step:overspilling:1}, it suffices to construct $ w_1 \in W_1 $ and $ w_2 \in W_2 \setminus W_1 $ such that \begin{align}
	\label{eqn:pix_w2_in_gen_pix_w1_overspill}
	\pi_X(w_2) &\in \gen{\pi_X(w_1)} \\
	\label{eqn:piy_w2_in_gen_piy_w1_overspill}
	\pi_Y(w_2) &\in \gen{\pi_Y(w_1)}
\end{align} In other words, it suffices to construct $ w_1 $ and $ w_2 $ such that their projections generate equal subspaces of the projected space. We now use the running example to explain the idea behind the way we achieve this (see Section~\ref{sec:overspilling:rigorous_proof} for all details and for the reason why everything we are about to do holds in general with only minor adjustments).

Let $ u_1, \dots, u_k $ be a collection of vectors generating $ W_1 $. In the present example, $ k = 2 $ and we choose the values of $ u_1 $ and $ u_2 $ to be the columns of $ M_1 $ (see (\ref{eqn:overspilling_m1_mat_def})). The intuition is that there will always be a column in $ M_1 $ agreeing with $ u $ on the projection onto $ X $ and a column agreeing with $ u $ on the projection onto $ Y $ (this follows from (\ref{eqn:image_m1x_eq_image_m2x}) and (\ref{eqn:image_m1y_eq_image_m2y})). In this case, as can be immediately seen from the highlighting in (\ref{eqn:overspilling_m2_mat_def}), \begin{align}
	\label{eqn:pix_u1_eq_pix_u_overspill}
	\pi_X(u_1) &= \pi_X(u) \\
	\label{eqn:piy_u2_eq_piy_u_overspill}
	\pi_Y(u_2) &= \pi_Y(u)
\end{align} Now the idea is to think about (\ref{eqn:piy_u2_eq_piy_u_overspill}) as saying that subtracting off any value from both $ u_2 $ and $ u $ maintains the equality of projections onto $ Y $ (see the parts of the vectors highlighted yellow below). Intuitively, we need this to ensure that (\ref{eqn:piy_w2_in_gen_piy_w1_overspill}) holds; moreover, if we subtract off an element of $ W_1 $ from $ u \notin W_1 $, then we automatically ensure $ w_2 \notin W_2 $, which is important. Now the question is -- what to subtract off? Since we also need to ensure (\ref{eqn:pix_w2_in_gen_pix_w1_overspill}), we choose $ u_1 $ to be the vector we subtract off because then the projection onto $ X $ cancels out, meaning that (\ref{eqn:pix_w2_in_gen_pix_w1_overspill}) follows as desired. More precisely, we set \begin{align*}
	w_1 &:= u_2 - u_1 = \begin{pNiceMatrix}
		2 \\ 2 \\ -2 \\ 2 \\ 0 \\ 1
		\CodeAfter
		\tikz \node [highlightblue = (4-1) (6-1)] {};
	\end{pNiceMatrix} - \begin{pNiceMatrix}
		-1 \\ 2 \\ 1 \\ 0 \\ 1 \\ 1
		\CodeAfter
		\tikz \node [highlight = (1-1) (3-1)] {};
	\end{pNiceMatrix} = \begin{pNiceMatrix}
		3 \\ 0 \\ -1 \\ 2 \\ -1 \\ 0
		\CodeAfter
		\tikz \node [highlightyellow = (4-1) (6-1)] {};
	\end{pNiceMatrix} \in W_1 \\
	w_2 &:= u - u_1 = \begin{pNiceMatrix}
		-1 \\ 2 \\ 1 \\ 2 \\ 0 \\ 1
		\CodeAfter
		\tikz \node [highlight = (1-1) (3-1)] {};
		\tikz \node [highlightblue = (4-1) (6-1)] {};
	\end{pNiceMatrix} - \begin{pNiceMatrix}
		-1 \\ 2 \\ 1 \\ 0 \\ 1 \\ 1
		\CodeAfter
		\tikz \node [highlight = (1-1) (3-1)] {};
	\end{pNiceMatrix} = \begin{pNiceMatrix}
		0 \\ 0 \\ 0 \\ 2 \\ -1 \\ 0
		\CodeAfter
		\tikz \node [highlightyellow = (4-1) (6-1)] {};
	\end{pNiceMatrix} \in W_2 \setminus W_1
\end{align*} and observe that \begin{align*}
	\pi_X(w_2) &= \pi_X(u) - \pi_X(u_1) \overset{(\ref{eqn:pix_u1_eq_pix_u_overspill})}{=} 0 \in \gen{\pi_X(w_1)} \\
	\pi_Y(w_2) &= \pi_Y(u) - \pi_Y(u_1) \overset{(\ref{eqn:piy_u2_eq_piy_u_overspill})}{=} \pi_Y(u_2) - \pi_Y(u_1) = \pi_Y(w_1) \in \gen{\pi_Y(w_1)}
\end{align*} This establishes (\ref{eqn:pix_w2_in_gen_pix_w1_overspill}) and (\ref{eqn:piy_w2_in_gen_piy_w1_overspill}), which, as discussed above, is sufficient to prove step \ref{step:overspilling:1}. At this point, it is worth noting that this example-driven construction is a simplified version of what we actually need for the proof. The main difference is that, in general, we need to subtract off not just a single vector $ u_i $ but a certain linear combination that satisfies the same properties as $ u_1 $ in the present example. To sum up, we have shown how to construct $ w_1 $ and $ w_2 $ so that those vectors cannot be ``distinguished'' using vectorspaces satisfying (\ref{eqn:overspilling_step_1_property}). Hence, every such vectorspace must ``overspill'' in the sense that it must contain $ \gen{w_2} $ for $ w_2 \notin W_1 $. As mentioned above, other steps in our proof essentially strengthen and translate this overspilling effect from the algebraic world into the logical setting.

\subsection{Rigorous proof}
\label{sec:overspilling:rigorous_proof}

We now rigorously prove the Overspilling Theorem~\ref{thm:overspilling}. Let $ \Theta $ be the predicate set such that $ \Theta_1 = \Theta \cup \{p\} $ and $ \Theta_2 = \Theta \cup \{p^P\} $ holds for some $ P \in \{<, >\} $ and an equality predicate $ p \in \Theta_1^= $. By Theorem~\ref{thm:only_equality_predicates_can_establish_lindep_strong}, \begin{align}
	\label{eqn:lindep_pix_theta1eq_eq_lindep_pix_theta2eq}
	\LinDep_{\pi_Z}(\Theta_1^=) = \LinDep_{\pi_Z}(\Theta_2^=)
\end{align} holds for all $ Z \in \Pi $. Let $ v + W_1 $, and $ v + W_2 $ be affine vectorspaces\footnote{We can assume that the affine offsets are equal because $ \Theta_1^= \models \Theta_2^= $} of solutions to $ \Theta_1^= $ and to $ \Theta_2^= $, respectively, such that $ v \models \Theta_1 $.

\begin{zbclaim}
	\label{claim:dim_w1_le_dim_w2}
	$ \dim{W_1} < \dim{W_2} $
\end{zbclaim}
\begin{proof}
	First, observe that $ W_1 \le W_2 $ because $ \Theta_1^= = \Theta_2^= \cup \{p\} \models \Theta_2^= $. Hence, $ \dim{W_1} \le \dim{W_2} $. For the sake of contradiction, suppose that $ \dim{W_1} = \dim{W_2} $. This implies $ W_1 = W_2 $ and hence $ \Theta_1^= \equiv \Theta_2^= $. This, in turn, implies \[
	\Theta_2 \models \Theta = \Theta_2^= \cup \Theta \equiv \Theta_1^= \cup \Theta = \Theta_1
	\] Since $ \Theta_2 $ is satisfiable, it follows that $ \Theta_1 $ and $ \Theta_2 $ must agree on some model. In particular, this means that $ p \in \Theta_1^= $ and $ p^P \in \Theta_2 $ must agree on some model, which is impossible. This is a contradiction.
\end{proof}

\begin{zbclaim}
	\label{claim:pix_w1_eq_pix_w2}
	$ \pi_Z(W_1) = \pi_Z(W_2) $ holds for all $ Z \in \Pi $.
\end{zbclaim}
\begin{proof}
	Let $ Z \in \Pi $. Theorem~\ref{thm:lindep_algebraic_char} taken together with (\ref{eqn:lindep_pix_theta1eq_eq_lindep_pix_theta2eq}) yields
	\begin{center}
		\begin{tabular}{c c c}
			$ \LinDep_{\pi_Z}(\Theta_1^=) $ & $ \overset{(\ref{eqn:lindep_pix_theta1eq_eq_lindep_pix_theta2eq})}{=} $ & $ \LinDep_{\pi_Z}(\Theta_2^=) $ \\
			\rotatebox[origin=c]{90}{$ \overset{\ref{thm:lindep_algebraic_char}}{=} $} && \rotatebox[origin=c]{90}{$ \overset{\ref{thm:lindep_algebraic_char}}{=} $} \\
			$ \{U \le V \mid U \le \ImageOf(\pi_Z \circ \lc_{B_1})^{\bot}\} $ & $ = $ & $ \{U \le V \mid U \le \ImageOf(\pi_Z \circ \lc_{B_2})^{\bot}\} $ \\
			\rotatebox[origin=c]{90}{$ \in $} && \rotatebox[origin=c]{90}{$ \in $} \\
			$ \ImageOf(\pi_Z \circ \lc_{B_1})^\bot $ && $ \ImageOf(\pi_Z \circ \lc_{B_2})^\bot $
		\end{tabular}
	\end{center} for any bases $ B_1 $ of $ W_1 $ and $ B_2 $ of $ W_2 $. This implies \[
		\ImageOf(\pi_Z \circ \lc_{B_1})^\bot = \ImageOf(\pi_Z \circ \lc_{B_2})^\bot
	\] and thus $ \ImageOf(\pi_Z \circ \lc_{B_1}) = \ImageOf(\pi_Z \circ \lc_{B_2}) $. The claimed equality \begin{align*}
	\pi_Z(W_1) = \pi_Z(\ImageOf(\lc_{B_1})) &= \ImageOf(\pi_Z \circ \lc_{B_1}) \\
	&= \ImageOf(\pi_Z \circ \lc_{B_2}) = \pi_Z(\ImageOf(\lc_{B_2})) = \pi_Z(W_2)
	\end{align*} follows.
\end{proof}

Recall that we have assumed that $ \Pi $ is binary and let its elements be $ \Pi = \{X, Y\} $. Let $ (u_1, \dots, u_k) $ be a basis of $ W_1 = \gen{u_1, \dots, u_k} $ and let $ u \in W_2 \setminus W_1 $ be a vector such that $ \gen{u_1, \dots, u_k, u} \le W_2 $. Such $ u $ exists by Claim~\ref{claim:dim_w1_le_dim_w2}. Observe furthermore that $ \dim{W_1} = k \ge 1 $ because assuming $ \dim{W_1} = 0 $ would imply $ W_1 = \gen{0} $ and \[
v \in \ModelsOf(\Theta_1) \subseteq \ModelsOf(\Theta_1^=) = v + W_1 = v + \gen{0} = \{v\}
\] meaning, in particular, that $ \Theta_1 $ is $ \Pi $-simple. This contradicts the assumption that $ \Theta_1 $ is $ \Pi $-complex.

\begin{zbclaim}
	\label{claim:can_assume_pyu_eq_uj}
	At this point, we can assume, without loss of generality, that there exists $ j \in \{1, \dots, k\} $ such that $ \pi_Y(u) = \pi_Y(u_j) $. 
\end{zbclaim}
\begin{proof}
	By Claim~\ref{claim:pix_w1_eq_pix_w2}, \[
	\pi_Y(\gen{u}) \le \pi_Y(W_2) \overset{\ref{claim:pix_w1_eq_pix_w2}}{=} \pi_Y(W_1) = \pi_Y(\gen{u_1, \dots, u_k})
	\] so there must exist $ \mu_1, \dots, \mu_k \in \Q $ such that \[
	\pi_Y(u) = \sum_{i=1}^k \mu_i \cdot \pi_Y(u_i)
	\] Let $ j \in \{1, \dots, k\} $ be an index such that $ \mu_j \neq 0 $. We can assume that such an index exists because otherwise, we would have $ \pi_Y(u) = 0 $, in which case we would simply swap the values of $ X $ and $ Y $, and $ \pi_X(u) = \pi_Y(u) = 0 $ is impossible because this would contradict $ u \in W_2 \setminus W_1 $.
	
	\noindent Redefine $ u $ to be \[
	u' := \frac{1}{\mu_j} \cdot \Big(u - \sum_{\substack{i=1 \\ i \neq j}}^k \mu_i \cdot u_i\Big)
	\] This construction ensures the claimed equality \begin{align*}
	\pi_Y(u') &= \frac{1}{\mu_j} \cdot \Big(\pi_Y(u) - \sum_{\substack{i=1 \\ i \neq j}}^k \mu_i \cdot \pi_Y(u_i)\Big) \\
	&= \frac{1}{\mu_j} \cdot \Big(\pi_Y(u) + \mu_j \cdot \pi_Y(u_j) - \underbrace{\sum_{i=1}^k \mu_i \cdot \pi_Y(u_i)}_{= \pi_Y(u)}\Big) \\
	&= \pi_Y(u_j)
	\end{align*} Moreover, since in the definition of $ u' $ we are subtracting off $ \sum_{i \neq j} \mu_i \cdot u_i \in W_1 $ from $ u \in W_2 \setminus W_1 $, it follows that $ u' \in W_2 \setminus W_1 $ and $ \gen{u_1, \dots, u_k, u'} \le W_2 $. We conclude that the above transformation does not break any assumptions we make about $ u $.
\end{proof}

\noindent Now observe that, by Claim~\ref{claim:pix_w1_eq_pix_w2}, \begin{align}
\label{eqn:pixu_subspace_pxw2_eq_pxw1_eq_pxu1uk}
\pi_X(\gen{u}) \le \pi_X(W_2) \overset{\ref{claim:pix_w1_eq_pix_w2}}{=} \pi_X(W_1) = \pi_X(\gen{u_1, \dots, u_k})
\end{align} By (\ref{eqn:pixu_subspace_pxw2_eq_pxw1_eq_pxu1uk}) and the linearity of $ \pi_X $, there exist $ \lambda_1, \dots, \lambda_k \in \Q $ such that \begin{align}
\label{eqn:pixu_eq_sum_lambdai_pxui}
\pi_X(u) = \sum_{i=1}^k \lambda_i \cdot \pi_X(u_i)
\end{align} Define \begin{align}
\label{eqn:w_u_minus_sum_lambdai_ui}
w_2 := u - \sum_{i=1}^k \lambda_i \cdot u_i
\end{align} Observe that $ w_2 \notin W_1 $ because $ \sum_{i=1}^k \lambda_i \cdot u_i \in \gen{u_1, \dots, u_k} = W_1 $ but $ u \notin W_1 $. On the other hand, $ \gen{u_1, \dots, u_k, w_2} \le W_2 $ because \[
	\gen{u_1, \dots, u_k, w_2} \le \gen{u_1, \dots, u_k, u} \le W_2
\] In particular, $ w_2 \in W_2 $. Moreover, $ w_2 $ has the additional property that \begin{align}
\label{eqn:pxw2_eq_zero}
\pi_X(w_2) \overset{(\ref{eqn:w_u_minus_sum_lambdai_ui})}{=} \pi_X(u) - \sum_{i=1}^k \lambda_i \cdot \pi_X(u_i) \overset{(\ref{eqn:pixu_eq_sum_lambdai_pxui})}{=} 0
\end{align} We now define \begin{align}
\label{eqn:w1_eq_uj_minus_sum_lambdai_ui}
w_1 := u_j - \sum_{i=1}^k \lambda_i \cdot u_i
\end{align} where $ j $ is the index provided by Claim~\ref{claim:can_assume_pyu_eq_uj}. Obviously, $ w_1 \in W_1 $. \begin{figure}[t]
	\centering
	\input{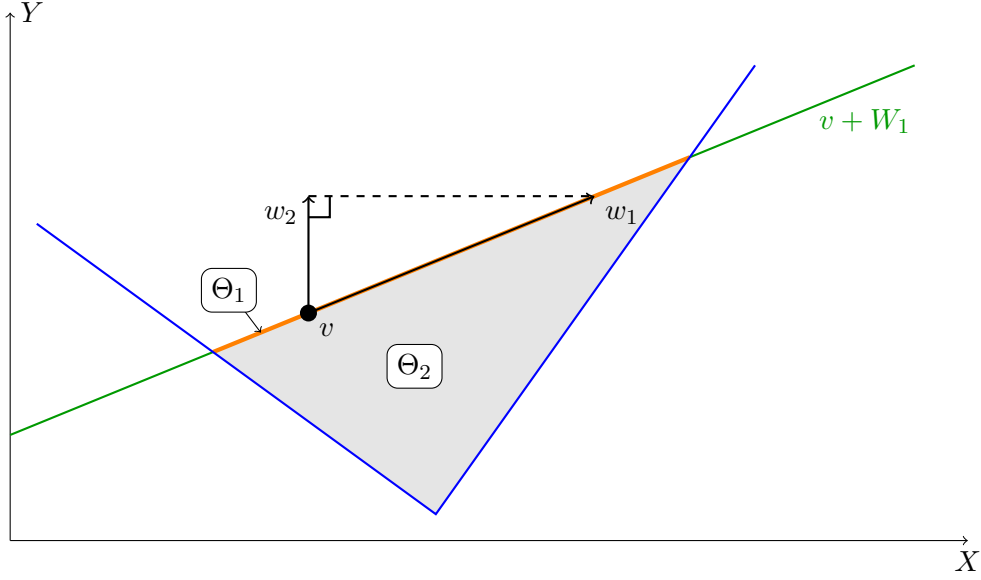}
	\caption{An example of how vectorspaces and individual vectors we introduced may look like when projected onto two dimensions. More precisely, we independently project the components of vectors corresponding to $ X \in \Pi $ and $ Y \in \Pi $. Labels at the endpoints of vectors correspond to positions relative to the initial point, not the origin. Blue lines depict the borders defined by linear constraints on which $ \Theta_1 $ and $ \Theta_2 $ agree.}
	\label{fig:overspilling_proof}
\end{figure} A way of geometrically thinking about the relationships between variables we defined so far is visualized in Figure~\ref{fig:overspilling_proof}.

\begin{zbclaim}
	\label{claim:pyw1_eq_pyw2}
	$ \pi_Y(w_1) = \pi_Y(w_2) $
\end{zbclaim}
\begin{proof}
	Combining Claim~\ref{claim:can_assume_pyu_eq_uj} with (\ref{eqn:w1_eq_uj_minus_sum_lambdai_ui}) and (\ref{eqn:w_u_minus_sum_lambdai_ui}) yields
	\begin{align*}
	\pi_Y(w_1) \overset{(\ref{eqn:w1_eq_uj_minus_sum_lambdai_ui})}{=} \pi_Y(u_j) - \sum_{i=1}^k \lambda_i \cdot \pi_Y(u_i) \overset{\ref{claim:can_assume_pyu_eq_uj}}{=} \pi_Y(u) - \sum_{i=1}^k \lambda_i \cdot \pi_Y(u_i) \overset{(\ref{eqn:w_u_minus_sum_lambdai_ui})}{=} \pi_Y(w_2)
	\end{align*} as desired.
\end{proof}

\begin{zbclaim}
	\label{claim:w1_neq_zero}
	$ w_1 \neq 0 $
\end{zbclaim}
\begin{proof}
	For the sake of contradiction, suppose that $ w_1 = 0 $. By Claim~\ref{claim:pyw1_eq_pyw2}, this means that $ \pi_Y(w_2) = \pi_Y(w_1) = 0 $. As we derived in (\ref{eqn:pxw2_eq_zero}), $ \pi_X(w_2) = 0 $. Hence, $ w_2 = 0 $, which contradicts $ w_2 \notin W_1 \ni 0 $.
\end{proof}

Having done the above preparatory work, we now get to the final part of the proof. Consider the set \[
P := (v + \gen{w_1}) \cap \ModelsOf(\Theta_1)
\] Recall that, by definition, $ v \in \ModelsOf(\Theta_1) $. Hence, $ v \in P $. Since $ v \models \Theta_1 $ and $ \Theta_1 \setminus \Theta_1^= $ contains only strict inequality predicates, there exists an $ \varepsilon > 0 $ such that \begin{align}
\label{eqn:I_def_v_v_plus_epsilon_w1}
I := (v, v + \varepsilon \cdot w_1) \subseteq P \subseteq \ModelsOf(\Theta_1)
\end{align} Note that by Claim~\ref{claim:w1_neq_zero}, $ v + \varepsilon \cdot w_1 \neq v $, so $ |I| = \infty $.

\begin{zbclaim}
	\label{claim:omega_agrees_I_on_two_els_implies_omega_theta2_sat}
	Let $ \Omega $ be a set of $ \Pi $-respecting predicates. Then $ \left|\ModelsOf(\Omega) \cap I\right| \ge 2 $ implies that $ \Omega \cup \Theta_2 $ is satisfiable.
\end{zbclaim}
\begin{proof}
	Suppose $ \left|\ModelsOf(\Omega) \cap I\right| \ge 2 $. By definition of $ I $, the two witnesses of $ \left|\ModelsOf(\Omega) \cap I\right| \ge 2 $ can be assumed to be of the form \[
	w, w + \delta' \cdot w_1 \in \ModelsOf(\Omega) \cap I
	\] for some $ w \models \Omega \cup \Theta_1 $ and $ \delta' \neq 0 $. In particular, it follows that \begin{align}
	\label{eqn:w_wplusdeltaw1_in_modomega}
	w, w + \delta' \cdot w_1 \in \ModelsOf(\Omega) \subseteq \ModelsOf(\Omega^=) = w + W
	\end{align} where $ w + W $ is the affine vectorspace of solutions to $ \Omega^= $. By (\ref{eqn:w_wplusdeltaw1_in_modomega}) and Claim~\ref{claim:w1_neq_zero}, \begin{align}
	\label{eqn:zero_neq_wq_le_W}
	\gen{0} \overset{\ref{claim:w1_neq_zero}}{\neq} \gen{w_1} \overset{(\ref{eqn:w_wplusdeltaw1_in_modomega})}{\le} W
	\end{align} Now observe that since $ \Omega $ is a set of $ \Pi $-respecting predicates, $ \Omega^= $ can be partitioned as follows. \begin{align*}
		\Omega_X &:= \{p \in \Omega^= \mid \Free(p) \subseteq X\} \\
		\Omega_Y &:= \{p \in \Omega^= \mid \Free(p) \subseteq Y\} \\
		\Omega^= &= \Omega_X \cup \Omega_Y
	\end{align*} Let $ w + W_X $ and $ w + W_Y $ be affine vectorspaces of solutions to $ \Omega_X $ and $ \Omega_Y $, respectively. Since the truth value of $ \Omega_X $ (resp. $ \Omega_Y $) depends only on the variables from $ X $ (resp. $ Y $), \begin{align*}
		\pi_Y(W_X) &= \Q^{|Y|} \\
		\pi_X(W_Y) &= \Q^{|X|}
	\end{align*} Hence, for any $ Z \in \Pi $ and $ u \in \Q^n $, $ \pi_Z(u) \in \pi_Z(W_Z) $ implies $ u \in W_Z $. This property immediately yields the ``$ \supseteq $'' direction in \begin{align}
		\label{eqn:wz_eq_pizi_piz_wz}
		W_Z = \pi_Z^{-1}(\pi_Z(W_Z))
	\end{align} where $ Z \in \Pi $ is arbitrary (note that ``$ \subseteq $'' is trivial). Furthermore, observe that since the semantic evaluation of $ \Omega_X $ (resp. $ \Omega_Y $) depends only on the variables from $ X $ (resp. $ Y $), for all $ Z \in \Pi = \{X, Y\} $ it holds that \begin{align}
		\label{eqn:piz_mod_omegaz_eq_piz_mod_omegaeq}
		\pi_Z(\ModelsOf(\Omega_Z)) = \pi_Z(\ModelsOf(\Omega^=))
	\end{align} In terms of vectorspaces, this means that \begin{align*}
		\pi_Z(w) + \pi_Z(W_Z) &= \pi_Z(w + W_Z) \\
		&= \pi_Z(\ModelsOf(\Omega_Z)) \\
		&\overset{(\ref{eqn:piz_mod_omegaz_eq_piz_mod_omegaeq})}{=} \pi_Z(\ModelsOf(\Omega^=)) \\
		&= \pi_Z(w + W) = \pi_Z(w) + \pi_Z(W)
	\end{align*} and hence \begin{align}
		\label{eqn:piz_wz_piz_w}
		\pi_Z(W_Z) = \pi_Z(W)
	\end{align} holds for all $ Z \in \Pi $. We now combine the above observations to prove the following characterization of $ W $, which will later be of crucial importance.

	\begin{zbclaim}
		\label{claim:W_eq_pixinv_WX_intersect_piyinv_WY}
		$ W = \pi_X^{-1}(\pi_X(W)) \cap \pi_Y^{-1}(\pi_Y(W)) $
	\end{zbclaim}
	\begin{proof}
		We show $ w + W = w + \big(\pi_X^{-1}(\pi_X(W)) \cap \pi_Y^{-1}(\pi_Y(W))\big) $.
		\begin{align*}
			w + W &= \ModelsOf(\Omega^=) \\
			&= \ModelsOf(\Omega_X) \cap \ModelsOf(\Omega_Y) \\
			&= (w + W_X) \cap (w + W_Y) \\
			&= w + (W_X \cap W_Y) \\
			&\overset{(\ref{eqn:wz_eq_pizi_piz_wz})}{=} w + \big(\pi_X^{-1}(\pi_X(W_X)) \cap \pi_Y^{-1}(\pi_Y(W_Y))\big) \\
			&\overset{(\ref{eqn:piz_wz_piz_w})}{=} w + \big(\pi_X^{-1}(\pi_X(W)) \cap \pi_Y^{-1}(\pi_Y(W))\big) \\
		\end{align*}
	\end{proof}
	
	\begin{zbclaim}
		\label{claim:w2_is_in_W}
		$ w_2 \in W $
	\end{zbclaim}
	\begin{proof}
		We use the characterization of $ W $ we just derived in Claim~\ref{claim:W_eq_pixinv_WX_intersect_piyinv_WY}. It is immediate that $ w_2 \in \pi_X^{-1}(\pi_X(W)) $ because \[
		\pi_X(w_2) \overset{(\ref{eqn:pxw2_eq_zero})}{=} 0 = \pi_X(0) \in \pi_X(W)
		\] holds by (\ref{eqn:pxw2_eq_zero}). By Claim~\ref{claim:pyw1_eq_pyw2}, \[
		\pi_Y(w_2) \overset{\ref{claim:pyw1_eq_pyw2}}{=} \pi_Y(w_1) \in \pi_Y(\gen{w_1}) \overset{(\ref{eqn:zero_neq_wq_le_W})}{\le} \pi_Y(W)
		\] and hence $ w_2 \in \pi_Y^{-1}(\pi_Y(W)) $. The claim $ w_2 \in W $ follows from Claim~\ref{claim:W_eq_pixinv_WX_intersect_piyinv_WY}.
	\end{proof}

	\begin{figure}[t]
		\centering
		\input{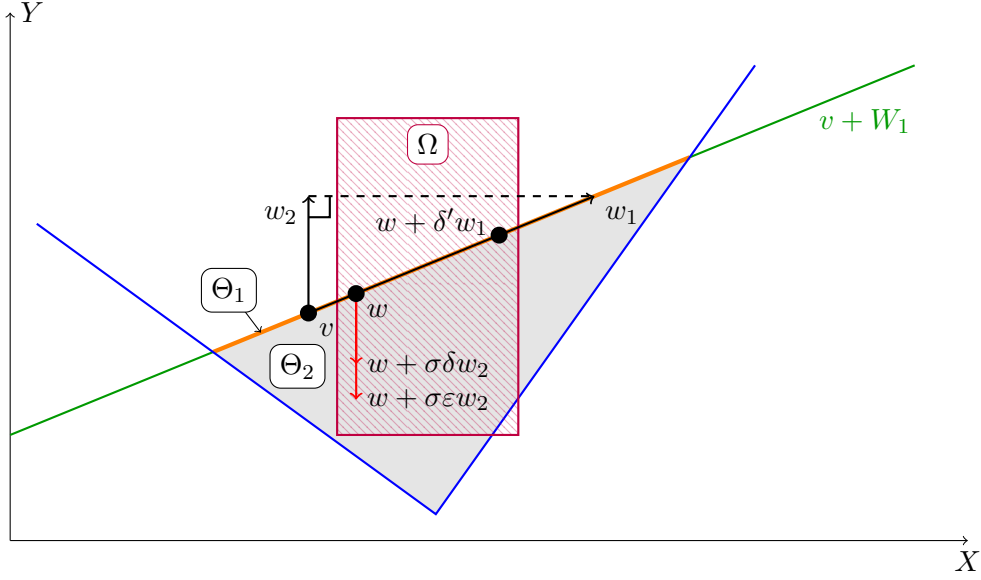}
		\caption{Figure~\ref{fig:overspilling_proof} extended with an additional visualization of $ w $, $ \Omega $ and endpoints of relevant intervals. In this example, $ \sigma = -1 $. The borders of $ \Omega $ are aligned along either $ X $ or $ Y $ because $ \Omega $ is a set of $ \Pi $-respecting predicates.}
		\label{fig:overspilling_proof_omega}
	\end{figure}
	
	\begin{zbclaim}
		\label{claim:existence_lambda_epsilon_overshoot}
		There exist $ \varepsilon > 0 $ and $ \sigma \in \{1, -1\} $ such that \[
			\varnothing \neq (w, w + \sigma \cdot \varepsilon \cdot w_2) \subseteq \ModelsOf(\Theta_2)
		\]
	\end{zbclaim}
	\begin{proof}
		Since $ w_2 \in W_2 $, $ w \models \Theta_1 \models \Theta $ and $ \Theta \setminus \Theta^= $ contains only strict inequality predicates, there exists an $ \varepsilon > 0 $ such that \[
		\varnothing \neq (w - \varepsilon \cdot w_2, w + \varepsilon \cdot w_2) \subseteq \ModelsOf(\Theta)
		\] Hence, due to $ \Theta_2 = \Theta \cup \{p^P\} $, it suffices to show the existence of $ \sigma \in \{1, -1\} $ such that \begin{align}
			\label{eqn:existence_lambda_epsilon_overshoot_suff_cond}
			\varnothing \neq (w, w + \sigma \cdot \varepsilon \cdot w_2) \subseteq \ModelsOf(p^P)
		\end{align} We first argue that $ (w, w + \varepsilon \cdot w_2) \cap \ModelsOf(p) = \varnothing $. Indeed, $ w + \varepsilon' \cdot w_2 \in \ModelsOf(p) $ for some $ 0 < \varepsilon' < \varepsilon $ would mean $ w + \varepsilon' \cdot w_2 \models \Theta_2^= \cup \{p\} = \Theta_1^= $ and consequently $ w + \varepsilon' \cdot w_2 - v \in W_1 $, but this implies a contradiction $ w_2 \in W_1 $. Hence, \begin{align}
		\label{eqn:w_plus_epsilon_w2_in_mod}
		\varnothing \neq (w, w + \varepsilon \cdot w_2) \subseteq \ModelsOf(p^<) \cup \ModelsOf(p^>)
		\end{align} so it follows that $ (w, w + \varepsilon \cdot w_2) \subseteq \ModelsOf(p^Q) $ must be true for some $ Q \in \{<, >\} $. This is because $ (w, w + \varepsilon \cdot w_2) $ is a convex set and thus assuming the contrary would imply by (\ref{eqn:w_plus_epsilon_w2_in_mod}) that $ \ModelsOf(p^<) \cup \ModelsOf(p^>) $ is closed under taking the convex hull of a point from $ \ModelsOf(p^<) $ and from $ \ModelsOf(p^>) $, but this is not true due to $ w \models p^= $. Repeating the above argumentation for $ (w, w - \varepsilon \cdot w_2) $ instead of $ (w, w + \varepsilon \cdot w_2) $ similarly yields that \begin{align}
		\label{eqn:w_minus_epsilon_w2_in_mod}
		\varnothing \neq (w, w - \varepsilon \cdot w_2) \subseteq \ModelsOf(p^Q)
		\end{align} must hold for some $ Q \in \{<, >\} $. Now observe that for every $ Q \in \{<, >\} $, \[
		(w, w + \varepsilon \cdot w_2) \subseteq \ModelsOf(p^Q) \supseteq (w, w - \varepsilon \cdot w_2)
		\] is impossible because $ p^Q $ defines a convex set and assuming \[
			(w - \varepsilon \cdot w_2, w) \cup (w, w + \varepsilon \cdot w_2) \subseteq \ModelsOf(p^Q)
		\] leads to (by taking the convex hull of the left-hand side) \[
			(w - \varepsilon \cdot w_2, w + \varepsilon \cdot w_2) \ni w \models p^Q
		\] which is a contradiction to $ w \models \Theta_1 $ because $ p \in \Theta_1 $ cannot agree with $ p^Q $ on any model. This observation, taken in conjunction with (\ref{eqn:w_plus_epsilon_w2_in_mod}) and (\ref{eqn:w_minus_epsilon_w2_in_mod}), implies the existence of $ \sigma \in \{1, -1\} $ such that (\ref{eqn:existence_lambda_epsilon_overshoot_suff_cond}) holds.
	\end{proof}
	
	At this point, we have derived everything we need to construct a model of $ \Omega \cup \Theta_2 $. Let $ \varepsilon, \sigma $ be the constants provided by Claim~\ref{claim:existence_lambda_epsilon_overshoot}. Since $ w_2 \in W $ (see Claim~\ref{claim:w2_is_in_W}), $ w \models \Omega $ and $ \Omega \setminus \Omega^= $ contains only strict inequality predicates, there exists some $ \delta > 0 $ such that \begin{align}
	\label{eqn:w_omega_interval}
	(w, w + \sigma \cdot \delta \cdot w_2) \subseteq \ModelsOf(\Omega)
	\end{align} See Figure~\ref{fig:overspilling_proof_omega} for a visualization of how $ \ModelsOf(\Omega) $ and the relevant intervals may look like geometrically. Claim~\ref{claim:existence_lambda_epsilon_overshoot} taken together with (\ref{eqn:w_omega_interval}) yields \begin{align*}
		\varnothing &\neq (w, w + \sigma \cdot \min\{\varepsilon, \delta\} \cdot w_2) \\
		&= (w, w + \sigma \cdot \varepsilon \cdot w_2) \cap (w, w + \sigma \cdot \delta \cdot w_2) \\
		&\subseteq \ModelsOf(\Theta_2) \cap \ModelsOf(\Omega) \\
		&= \ModelsOf(\Theta_2 \cup \Omega)
	\end{align*} Hence, $ \Omega \cup \Theta_2 $ is satisfiable.
\end{proof}

Since $ \varphi $ is $ \Pi $-decomposable, there exists a $ \Pi $-decomposition \begin{align}
	\label{eqn:overspilling_pi_dec}
	\varphi \equiv \bigvee_{\Omega \in D} \Omega
\end{align} where $ D $ is a finite set of sets of predicates respecting $ \Pi $ (without loss of generality). Since \[
I \overset{(\ref{eqn:I_def_v_v_plus_epsilon_w1})}{\subseteq} \ModelsOf(\Theta_1) \subseteq \ModelsOf(\varphi) \overset{(\ref{eqn:overspilling_pi_dec})}{=} \bigcup_{\Omega \in D} \ModelsOf(\Omega)
\] and $ |I| = \infty $ (see above), by the pigeonhole principle, there exists some $ \Omega \in D $ such that \[
	\left|\ModelsOf(\Omega) \cap I\right| = \infty
\] By Claim~\ref{claim:omega_agrees_I_on_two_els_implies_omega_theta2_sat}, $ \Omega \cup \Theta_2 $ is satisfiable. Since $ \Omega \models \varphi $, it follows that $ \varphi \wedge \Theta_2 $ is also satisfiable. This completes the proof of the Overspilling Theorem~\ref{thm:overspilling}.
	
\end{document}